%% file: ms.tex
\documentclass[english]{spimufcphdthesis}
 
\usepackage[utf8]{inputenc}
\usepackage{algorithm}
\usepackage{algpseudocode}
\usepackage{lscape} 
\usepackage{multirow}
\usepackage{xcolor}
\usepackage{adjustbox}
\usepackage{amsmath}
\usepackage{mathrsfs}
\usepackage{amsfonts}
\usepackage{amssymb}
\usepackage{amsthm}
\newtheorem{definition}{Definition}
\newtheorem{proposition}{Proposition}
\newtheorem{theorem}{Theorem}
\usepackage{csquotes}
\usepackage{tikz}
\usetikzlibrary{trees}

\declarethesis[Production de Données Catégorielles Respectant la Confidentialité Différentielle : Conception et Applications au Apprentissage Automatique]{Production of Categorical Data Verifying Differential Privacy: Conception and Applications to Machine Learning}{05 Janvier 2022}{Besançon}{XXX}

\addauthor[email]{Héber}{Hwang Arcolezi}
 
\addjury{Stéphane}{Prof Chrétien}{Pr\'esident}{Université de Lyon 2}
\addjury{Mathieu}{Prof Cunche}{Rapporteur}{Institut National des Sciences Appliquées de Lyon}
\addjury{Benjamin}{Prof Nguyen}{Rapporteur}{Institut National des Sciences Appliquées Centre Val de Loire}
\addjury{Mário S.}{Prof Alvim}{Examinateur}{Universidade Federal de Minas Gerais}
\addjury{Jean-François}{Prof Couchot}{Directeur de thèse}{Université Bourgogne Franche-Comté}
\addjury{Xiaokui}{Prof Xiao}{Codirecteur de thèse}{National University of Singapore}

\thesisabstract[english]{Private and public organizations regularly collect and analyze digitalized data about their associates, volunteers, clients, etc. However, because most personal data are sensitive, there is a key challenge in designing privacy-preserving systems to comply with data privacy laws, e.g., the General Data Protection Regulation. To tackle privacy concerns, research communities have proposed different methods to preserve privacy, with Differential privacy (DP) standing out as a formal definition that allows quantifying the privacy-utility trade-off. Besides, with the local DP (LDP) model, users can sanitize their data locally before transmitting it to the server. 

The objective of this thesis is thus two-fold: \textbf{O$_1$}) To improve the utility and privacy in multiple frequency estimates under LDP guarantees, which is fundamental to \textit{statistical learning}. And \textbf{O$_2$}) To assess the privacy-utility trade-off of machine learning (ML) models trained over differentially private data.
}
 
\thesiskeywords[english]{Differential privacy, Local differential privacy, Categorical data, Machine learning.}
 
\thesisabstract[french]{Les organisations privées et publiques collectent et analysent régulièrement des données numérisées sur leurs associés, volontaires, clients, etc. Cependant, comme la plupart des données personnelles sont sensibles, la conception de systèmes préservant la vie privée pour se conformer aux lois sur la confidentialité des données, par exemple le règlement général sur la protection des données, constitue un défi important. Pour résoudre les problèmes de confidentialité, les communautés de chercheurs ont proposé différentes méthodes de préservation de la confidentialité, la confidentialité différentielle (DP) se distinguant comme une définition formelle qui permet de quantifier le compromis entre confidentialité et utilité. En outre, avec le modèle de confidentialité différentielle locale (LDP), les utilisateurs peuvent sanitisé leurs données localement avant de les transmettre au serveur. 

L'objectif de cette thèse est donc double : \textbf{O$_1$}) Améliorer l'utilité et la confidentialité des estimations de fréquences multiples sous garanties LDP, ce qui est fondamental pour \textit{l'apprentissage statistique}. Et \textbf{O$_2$}) Évaluer le compromis vie privée-utilité des modèles d'apprentissage machine (ML) entraînés sur des données différentiellement privées.
}
 
\thesiskeywords[french]{Confidentialité différentielle, Confidentialité différentielle locale, Données catégorielles, Apprentissage automatique.}
 
\resetlaboratories

\addlaboratory{Laboratoire FEMTO-ST}

\begin{document}
 


\include{chapters/abstract}

\include{chapters/resume}

\include{chapters/Acknowledgements}

\tableofcontents

\mainmatter 
\include{chapters/glossaries}

\part{Thesis Introduction}

\include{chapters/chapter1}

\part{Background}

\include{chapters/chapter2}

\include{chapters/chapter3}

\part{Contribution: Improving the Utility and Privacy of LDP protocols}

\include{chapters/chapter4} 

\include{chapters/chapter7} 

\include{chapters/chapter5} 

\include{chapters/chapter6} 

\part{Contribution: Differentially Private Machine Learning Predictions}

\include{chapters/chapter91} 

\include{chapters/chapter8} 

\include{chapters/chapter9} 

\include{chapters/chapter92} 

\part{Conclusion \& Perspectives}

\include{chapters/Conclusion}

\include{chapters/publications}

\backmatter
 
 
 
 
 
 
 
 
 
\bibliographystyle{phdthesisnum}
 
\bibliography{biblio.bib}
 
 
\listoffigures
 
\listoftables
 



 
\end{document}

%% file: chapters/abstract.tex
\chapter*{Abstract}

\begingroup
\centering
{\textcolor{parttitlecolor}{\large {Production of Categorical Data Verifying Differential Privacy: Conception and Applications to Machine Learning}}}
\vspace{10mm}

{\textcolor{parttitlecolor}{\large {Héber Hwang Arcolezi}}}\\
{\textcolor{parttitlecolor}{\large {University Bourgogne Franche Comté, 2022}}}

\vspace{10mm}

{\hspace{22mm} \textcolor{parttitlecolor}{\large {Supervisors: Jean-François Couchot, Bechara Al Bouna, and Xiaokui Xiao}}}
\endgroup

\vspace{5mm}

Private and public organizations regularly collect and analyze digitalized data about their associates, volunteers, clients, etc. However, because most personal data are sensitive, there is a key challenge in designing privacy-preserving systems to comply with data privacy laws, e.g., the General Data Protection Regulation. To tackle privacy concerns, research communities have proposed different methods to preserve privacy, with Differential privacy (DP) standing out as a formal definition that allows quantifying the privacy-utility trade-off. Besides, with the local DP (LDP) model, users can sanitize their data locally before transmitting it to the server. 

The objective of this thesis is thus two-fold: \textbf{O$_1$}) To improve the utility and privacy in multiple frequency estimates under LDP guarantees, which is fundamental to \textit{statistical learning}. And \textbf{O$_2$}) To assess the privacy-utility trade-off of machine learning (ML) models trained over differentially private data.

For \textbf{O$_1$}, we first tackled the problem from two ``\textit{multiple}" perspectives, i.e., multiple attributes and multiple collections throughout time (longitudinal studies), while focusing on \textbf{utility}. Secondly, we focused our attention on the multiple attributes aspect only, in which we proposed a solution focusing on \textbf{privacy} while preserving utility. In both cases, we demonstrate through analytical and experimental validations the advantages of our proposed solutions over state-of-the-art LDP protocols. 

For \textbf{O$_2$}, we empirically evaluated ML-based solutions designed to solve real-world problems while ensuring DP guarantees. Indeed, we mainly used the \textit{input data perturbation} setting from the privacy-preserving ML literature. This is the situation in which the whole dataset is \textit{sanitized} independently (i.e., row-by-row) and, thus, we implemented LDP algorithms from the perspective of the centralized data owner. In all cases, we concluded that differentially private ML models achieve nearly the same utility metrics as non-private ones.

\textbf{KEYWORDS:} Differential privacy, Local differential privacy, Categorical data, Machine learning.

%% file: chapters/resume.tex
\chapter*{Résumé}

\begingroup
\centering
{\textcolor{parttitlecolor}{\large {Production de Données Catégorielles Respectant la Confidentialité Différentielle : Conception et Applications au Apprentissage Automatique}}}
\vspace{10mm}

{\textcolor{parttitlecolor}{\large {Héber Hwang Arcolezi}}}\\
{\textcolor{parttitlecolor}{\large {Université Bourgogne Franche Comté, 2022}}}

\vspace{10mm}

{\hspace{22mm} \textcolor{parttitlecolor}{\large {Encadrants: Jean-François Couchot, Bechara Al Bouna, et Xiaokui Xiao}}}
\endgroup

\vspace{5mm}

Les organisations privées et publiques collectent et analysent régulièrement des données numérisées sur leurs associés, volontaires, clients, etc. Cependant, comme la plupart des données personnelles sont sensibles, la conception de systèmes préservant la vie privée pour se conformer aux lois sur la confidentialité des données, par exemple le règlement général sur la protection des données, constitue un défi important. Pour résoudre les problèmes de confidentialité, les communautés de chercheurs ont proposé différentes méthodes de préservation de la confidentialité, la confidentialité différentielle (DP) se distinguant comme une définition formelle qui permet de quantifier le compromis entre confidentialité et utilité. En outre, avec le modèle de confidentialité différentielle locale (LDP), les utilisateurs peuvent sanitisé leurs données localement avant de les transmettre au serveur. 

L'objectif de cette thèse est donc double : \textbf{O$_1$}) Améliorer l'utilité et la confidentialité des estimations de fréquences multiples sous garanties LDP, ce qui est fondamental pour \textit{l'apprentissage statistique}. Et \textbf{O$_2$}) Évaluer le compromis vie privée-utilité des modèles d'apprentissage machine (ML) entraînés sur des données différentiellement privées.

Pour \textbf{O$_1$}, nous avons premièrement abordé le problème sous deux angles ``\textit{multiple}'', c'est-à-dire des attributs multiples et des collections multiples dans le temps (études longitudinales), tout en nous concentrant sur \textbf{utilité}. Deuxièmement, nous avons concentré notre attention sur l'aspect des attributs multiples uniquement, dans lequel nous avons proposé une solution axée sur la \textbf{confidentialité} tout en préservant l'utilité. Dans les deux cas, nous démontrons par des validations analytiques et expérimentales les avantages de nos solutions proposées par rapport aux protocoles LDP de pointe. 

Pour \textbf{O$_2$}, nous avons évalué empiriquement des solutions basées sur les ML conçues pour résoudre des problèmes du monde réel tout en assurant des garanties de DP. En effet, nous avons principalement utilisé le cadre \textit{perturbation des données d'entrée} de la littérature sur les ML préservant la confidentialité. Il s'agit de la situation dans laquelle l'ensemble des données est \textit{sanitisé} indépendamment (c'est-à-dire ligne par ligne) et, par conséquent, nous avons mis en œuvre des algorithmes LDP du point de vue du propriétaire centralisé des données. Dans tous les cas, nous avons conclu que les modèles ML différentiellement privés atteignent presque les mêmes mesures d'utilité que les modèles non privés.

\textbf{Mots clés:} Confidentialité différentielle, Confidentialité différentielle locale, Données catégorielles, Apprentissage automatique.

%% file: chapters/Acknowledgements.tex
\chapter*{Acknowledgements}

Primarily, I would like to express my greatest thanks to my supervisor, Professor Jean-François Couchot, for his support, leadership, and encouragement during my Ph.D. study. I am very fortunate to have had him as my supervisor and for being led toward a topic I am very passionate about. I am truly grateful for his personality as an advisor as Jean-François really cares about his students, in both academic and personal subjects, which I wish for any Ph.D. student to have. I also thank my co-supervisors Bechara Al Bouna and Xiaokui Xiao for their collaboration and support throughout this dissertation.

I would also like to thank Professors Benjamin Nguyen, Mathieu Cunche, Stéphane Chrétien, and Mário S. Alvim, who kindly accepted to be part of my dissertation jury and for their valuable suggestions on research perspectives. 

Thanks also to Denis Renaud, who leads the Orange Application for Business team in Belfort, for his continued collaboration and helpful feedback. I also thank Commandant Guillaume Royer-Fey and Capitaine Céline Chevallier from the Fire Department of Doubs and Professor Christophe Guyeux, who helped me a lot through fruitful collaboration and a lot of feedback. 

I also thank Professor Sébastien Gambs, who kindly mentored me during my research visit at the Université du Québec à Montréal, and for the opportunity to continue collaborating. I learned a lot from him and gained valuable experiences, which are important for my career as a researcher.

I am very, very grateful to Selene Cerna, a special person to me, for the many joyful moments, constant support, and for taking care of me all these years. Selene has supported me since my master's degree and was significant in my growth as a young researcher. I admire Selene for her great willingness to help and share with others, and I am fortunate to be one of those people. I learned a lot with her, both technically and through extensive discussion on research subjects, which essentially helped me during this Ph.D. study.

I also thank Zhì Háo Chen who gave me a lot of guidance through many bureaucratic processes to establish me as a foreign doctoral student in France.

Last but not least, my beloved grandparents, parents, and siblings, my biggest thank to each of you who have supported and cared for me throughout my life. From each of you, a different kind of love has been shown over the years, and I gladly consider and return all the love I can offer you all.

%% file: chapters/glossaries.tex
\makeatletter
\newcommand{\tocfill}{\cleaders\hbox{$\m@th \mkern\@dotsep mu . \mkern\@dotsep mu$}\hfill}
\makeatother
\newcommand{\abbrlabel}[1]{\makebox[3cm][l]{\textbf{#1}\ \tocfill}}
\newenvironment{abbreviations}{\begin{list}{}{\renewcommand{\makelabel}{\abbrlabel}%
\setlength{\itemsep}{0pt}}}{\end{list}}

\chapter*{List of abbreviations}
\markboth{List of abbreviations}{}

\begin{abbreviations}

\item[ACC] Accuracy
\item[ADP] Adaptive
\item[CDRs] Call Detail Records
\item[CNIL] Commission Nationale de l'Informatique et des Libertés
\item[COVID-19] Coronavirus Disease 2019
\item[DP] Differential privacy
\item[EMS] Emergency medical services
\item[FIMU] Festival International de Musiques Universitaires
\item[GDPR] General Data Protection Regulation
\item[GRR] Generalized Randomized Response
\item[LDP] Local Differential Privacy
\item[LP] Linear Program
\item[MF1] Macro F1-Score
\item[ML] Machine Learning
\item[MNO] Mobile Network Operator
\item[MSE] Mean Squared Error
\item[MS-FIMU] Mobility Scenario FIMU
\item[OBS] Orange Business Services
\item[OUE] Optimized Unary Encoding
\item[QID] Quasi-Identifier
\item[RMSE] Root Mean Square Error
\item[RR] Randomized Response
\item[Smp] Sampling
\item[Spl] Splitting
\item[SUE] Symmetric Unary Encoding
\item[UE] Unary Encoding

\end{abbreviations}

%% file: chapters/chapter1.tex
\chapter{Introduction}
\label{chap:chapter1}

\section{Introduction}

Let be given Article 12 from the Universal Declaration of Humans Right~\cite{UDHR}, which defines: ``\textit{No one shall be subjected to arbitrary interference with his privacy, family, home or correspondence, nor to attacks upon his honour and reputation. Everyone has the right to the protection of the law against such interference or attacks}." 

Notice, however, that with the advancement of technology of information (\textbf{not only correspondences anymore}), protecting individuals' privacy in the era of Big data is a significant challenge. Indeed, the explosion of the number of connected objects, mobile applications collecting and/or generating any type of data makes personal data ubiquitous and growing exponentially. 

Moreover, when collecting data in practice, one is often interested in multiple attributes of a population, i.e., \textit{multidimensional data}. For instance, in crowd-sourcing applications, the server may collect both demographic information (e.g., gender, nationality) and user habits in order to develop personalized solutions for specific groups. In addition, one generally aims to collect data from the same users throughout time (i.e., \textit{longitudinal} studies), which is essential in many situations. For example, the fact that remote antennas of mobile network operators (MNOs) have received cell phone connections may reveal a movement if the same user is identified in different antennas throughout time.

From a human point of view, data analysts can be external providers. In other words, they very rarely have the consent of the data providers (i.e., individuals concerned) to analyze the data. It is, therefore, necessary for the company providing the service to make all possible efforts to follow all the recommendations from data privacy authorities such as the General Data Protection Regulation (GPDR)~\cite{GDPR} and, particularly, make any re-identification unfeasible from a practical point of view. On the other hand, even if trusted service providers collect raw personal data, this practice can still lead to privacy breaches, i.e., the risk of information leakage is always possible even if service providers make every effort to secure the data. 

Indeed, data breaches are all too common~\cite{data_breaches}, which endanger users' privacy and can lead to substantial losses for companies under the GDPR (cf.~\cite{uber,facebook}, for example). Moreover, along with gathering data, extracting high-utility analytics through machine learning (ML) from the collected data is of great interest. Yet, even ML models trained with raw data can also indirectly reveal sensitive information~\cite{Gong2020_survey,Sarwate2013_survey} (e.g., cf.~\cite{Song2017,Shokri2017,Carlini2019}).

In addition, privacy issues appear more than ever in headlines (e.g.,~\cite{taxi,aol,tinder1,tinder2,twitch,t_mobile,experian}). To tackle privacy concerns, research communities have proposed \textbf{different methods to preserve privacy}, in which the \textbf{main goal is that anonymized data should not leak private information about any individual}~\cite{desfontaines2020}. To this end, \textit{k}-anonymity~\cite{samarati1998protecting,SWEENEY2002} and differential privacy (DP)~\cite{Dwork2006,Dwork2006DP,dwork2014algorithmic} are two well-known privacy techniques. On the one hand, \textit{k}-anonymity is very risky since it does not allow to counter intersecting and/or homogeneity attacks, for example~\cite{Machanavajjhala2006,Li2007}. On the other hand, DP has been increasingly accepted as the current standard for data privacy~\cite{DL_DP,uber_dp,census2021,dwork2014algorithmic}. However, in the originally proposed centralized DP model, queries perturbed by DP algorithms require the storage of raw databases because the noise is only added at the end of the request. As aforementioned, storing and/or sharing raw databases (as well as training ML models over raw data) is not always desirable because it is necessary to secure all access to them from both a technical and human point of view. 

To preserve privacy at the user-side, an alternative approach, namely, local differential privacy (LDP), was initially formalized in~\cite{first_ldp}. With LDP, rather than trusting in a data curator to have the raw data and sanitize it to output queries, each user applies a DP mechanism to their data before transmitting it to the data collector server. The LDP model allows collecting data in unprecedented ways and, therefore, has led to several adoptions by industry. For instance, big tech companies like Google, Apple, and Microsoft, reported the implementation of LDP mechanisms to gather statistics in well-known systems (i.e., Google Chrome browser~\cite{rappor}, Apple iOS and macOS~\cite{apple}, and Windows 10 operation system~\cite{microsoft}). 

\section{Motivation and Objectives} \label{ch1:motivation_objectives}

For the rest of this manuscript, the author will utilize \textbf{we} rather than \textbf{I} to highlight the contributions of all my collaborators (cf. Acknowledgment on page vii). Yet, \textbf{the author is the only one responsible for all errors} that may still be present on this manuscript. The work in this manuscript is based on two motivating projects. 

On the one hand, we had a preliminary collaboration with the Orange Business Services (OBS) team in Belfort, France, i.e., an MNO. The OBS team presented us \textit{an overview} of their deployed system named Flux Vision~\cite{fluxvision1}, which publishes real-time statistics on human mobility by analyzing call detail records (CDRs). The Flux Vision system motivated us \textbf{to study how to gather knowledge from the published statistics as well as to propose a distinct privacy-preserving data collection process}. More precisely, from a practical perspective, based on \textit{longitudinal} and \textit{multidimensional} OBS mobility reports, we noticed that these statistics could be improved to provide more information about mobility patterns of the individuals concerned. Thus, this is our \textbf{first objective}. Furthermore, our \textbf{second objective} is to propose a privacy-preserving CDRs processing system, which could improve the privacy of MNOs' clients. Next, from a theoretical perspective on statistical learning, our \textbf{third objective} is to improve the utility and privacy of multiple frequency estimates (i.e., multidimensional and longitudinal data collections) under LDP guarantees. 

In addition, we also worked on a collaborative framework with Selene Cerna and Christophe Guyeux, members of the AND\footnote{Algorithmique Numérique Distribuée (or, distributed digital algorithmics in English).} research team from the same research department as ours\footnote{Department of Informatics and Complex Systems (DISC in French).}. Selene Cerna holds a CIFRE thesis (N 2019/0372) with the fire department named Service Départemental d'Incendie et de Secours du Doubs (SDIS 25), i.e, an emergency medical services (EMS) in France. For the past few years, the AND team has been investigating ML-based solutions to optimize the SDIS 25 services under a strict confidentiality agreement on the SDIS 25 data. The way these data have been shared motivated us \textbf{to study the privacy-utility trade-off of ML models trained over sanitized data}. That is, we consider the case of centralized data owners (e.g., MNOs and EMS) that \textit{collect sensitive information} from individuals for both billing and/or legal purposes \textit{but do not trust} the third entity to develop decision-support systems. So, our \textbf{fourth and last objective} is to evaluate empirically the privacy-utility trade-off of different ML-based solutions trained over sanitized data. We mainly focused on the SDIS 25 data. Notice, however, that this manuscript \textit{does not} focus on the data collection nor the feature engineering processes carried out by Selene Cerna but, rather, we will present only necessary information about the dataset while \textit{focusing on the privacy-utility trade-off} analysis. 

\section{Main Contributions of this Thesis}

The main contributions of this thesis are summarized in the following:

\begin{enumerate}
    \item First, based on one-week statistical data of unions of consecutive days published by OBS~\cite{fluxvision1}, we present a method for inferring and recreating a synthetic dataset that matches the original statistical data with low mean relative error. We thus generated and published it as an open dataset (\url{https://github.com/hharcolezi/OpenMSFIMU}) such that others can use it to evaluate new privacy-preserving techniques as well as ML tasks.
    
    \item Second, by studying these aggregate statistics on human mobility, we proposed an LDP-based CDRs processing system to generate multidimensional mobility reports throughout time by offering strong privacy guarantees for each user.
    
    \item The first two studies on CDRs-based mobility reports are translated to longitudinal statistical releases about the frequency of visitors by multiple attributes. We then contribute to the \textbf{theoretical} aspect under the LDP setting. More precisely, we first focused on optimizing the \textit{utility} of LDP protocols for \textit{longitudinal} and \textit{multidimensional} frequency estimates. 
    
    \item Next, we identified a limitation of the state-of-the-art solution used for multidimensional frequency estimates with LDP, which splits users into groups instead of splitting the privacy budget. We then propose a solution to this limitation, which improves the \textit{privacy} of users while providing \textit{the same or better utility} (regarding the mean squared error metric) than the state-of-the-art solution. 
    
    \item Lastly, we empirically evaluated the privacy-utility trade-off of differentially private input perturbation-based ML models. That is, we assessed practical solutions in which data owners (e.g., MNOs and EMS) could sanitize their datasets locally before transmitting these data to untrusted parties to develop decision-support tools, with no considerable impact on the utility. 
\end{enumerate}

\section{Thesis Outline}

The rest of this manuscript is organized as follows: Chapter~\ref{chap:chapter2} presents the scientific background on data anonymization techniques. Chapter~\ref{chap:chapter3} provides the scientific background on machine learning techniques and presents the databases we will experiment on. Chapter~\ref{chap:chapter4} presents the first contribution of this manuscript, namely, an open, longitudinal, and synthetic dataset of faked virtual humans generated by an optimization approach applied to a real-life CDRs-based anonymized database. Chapter~\ref{chap:chapter7} proposes a privacy-preserving CDRs processing system to generate mobility reports longitudinally. Chapter~\ref{chap:chapter5} presents our first theoretical contribution on statistical learning with LDP. Chapter~\ref{chap:chapter6} resolves one limitation of Chapters~\ref{chap:chapter7} and~\ref{chap:chapter5} by improving the privacy of individuals while keeping the utility on statistical learning with LDP. Chapter~\ref{chap:chapter91} empirically evaluates two differentially private machine learning settings on multivariate time series forecasting. Chapter~\ref{chap:chapter8} proposes a privacy-preserving methodology to sanitize an EMS intervention dataset while allowing both statistical learning and forecasting tasks. Chapter~\ref{chap:chapter9} empirically evaluates the impact of sanitizing the location of an emergency when training ML models to predict the response time of ambulances. Chapter~\ref{chap:chapter92} empirically evaluates the impact of training ML models over anonymized data to predict the victims' mortality. Lastly, Chapter~\ref{chap:conclusion} provides a general conclusion of this work and its perspectives.

%% file: chapters/chapter2.tex
\chapter{Data Anonymization} \label{chap:chapter2}

In Chapter~\ref{chap:chapter1}, we have introduced some main concerns with regard to privacy, the motivating projects of this thesis, as well as our objectives. In this chapter, we present the background on data anonymization techniques that our work relies on. We highlight that the content of this chapter is \textbf{primarily} inspired by existing literature in books~\cite{dwork2014algorithmic,Zhu2017} and papers~\cite{SWEENEY2002,Machanavajjhala2006,tianhao2017,Andrs2013}. Appropriate references to other works are provided throughout this chapter.

\section{Introduction: Syntactic VS Algorithmic Privacy} \label{ch2:introduction} 

In the literature, many privacy models have been proposed to tackle privacy issues. In this manuscript, we consider two data privacy definitions, namely, \textit{Syntactic privacy} and \textit{Algorithmic privacy}. More specifically, the former notion tries to define a syntactic criterion that should be satisfied by the \textit{output dataset} through transforming the data. The most influential method is named \textit{k}-anonymity~\cite{samarati1998protecting,SWEENEY2002}, which was the starting point for other extensions like \textit{l}-diversity~\cite{Machanavajjhala2006} and \textit{t}-closeness~\cite{Li2007}. We introduce \textit{k}-anonymity in Section~\ref{ch2:sub_k_anon}, which will be used in Chapter~\ref{chap:chapter92}. Throughout this manuscript, we will refer to \textbf{anonymity} as a condition of being ``safe in the crowd" (i.e., anonymous).

The latter algorithmic notion considers that anonymization is a property of the \textit{algorithm}, rather than the output dataset. This is the core insight of \textit{differential privacy}~\cite{Dwork2006,Dwork2006DP}, which addresses the paradox of learning about a population while learning nothing about single individuals~\cite{dwork2014algorithmic}. One special form of DP is the \textit{non-interactive case} considered in this manuscript, which corresponds to, e.g., releasing summary statistics, the sanitized dataset, a synthetic dataset, and so on. Throughout this manuscript, we will refer to \textbf{sanitization} the fact that data anonymization was achieved through verifying DP (i.e., using a DP algorithm). In this manuscript, we consistently used differential privacy. So, we present the centralized model of DP in Section~\ref{ch2:sub_dp}, the local model of DP in Section~\ref{ch2:sub_ldp}, and a local model of DP for location privacy in Section~\ref{ch2:sub_geo_ind}. 

\section{\textit{k}-anonymity} \label{ch2:sub_k_anon}

Given a public medical database without identifiers but where age, ZIP code, ..., were present, and a $20\$$ dollars public voter records from Massachusetts, United States of America, a Ph.D. student named Latanya Sweeney was able to re-identify the Governor of Massachusetts in this medical database~\cite{sweeney2015only}. This re-identification attack took place because there was similar demographic information in both medical databases and voter list records. This way, the combination of several demographic data made people \textit{unique} in both databases, which allowed Sweeney to directly match these records in both databases.

To tackle this \textit{uniqueness} problem in data publishing, Samarati and Sweeney~\cite{samarati1998protecting,SWEENEY2002} proposed the \textit{k}-anonymity model, which requires that each released record to be indistinguishable from at least $k-1$ others. Intuitively, the larger \textit{k} is the better the privacy protection will be. On applying \textit{k}-anonymity, there is a difference between: \textit{explicit identifiers} (e.g., names), which are removed or masked to avoid direct re-identification; \textit{sensitive attributes} (e.g., disease), that might be preserved, and \textit{quasi-identifiers (QIDs)} such as age and gender, in which \textit{k}-anonymity seeks to ensure indistinguishability. We recall the definition of \textit{k}-anonymity in the following.

\begin{definition}[\textit{k}-anonymity requirement~\cite{samarati1998protecting,SWEENEY2002}] Each release of data must ensure that every combination of values of QIDs can be indistinctly matched to at least \textit{k} individuals.
\end{definition}

We also recall here an example from~\cite{Machanavajjhala2006}. Table~\ref{ch2:tab_pseudo_dataset} exhibits a pseudonymized dataset (i.e., with no direct identifiers like `name') that stores the medical record of a set of individuals. This dataset is composed of both sensitive (disease) and `non-sensitive' information like age, gender, and nationality. Table~\ref{ch2:tab_4_anonymous} exhibits a 4-anonymous version of the original data in Table~\ref{ch2:tab_pseudo_dataset}. Note that in Table~\ref{ch2:tab_4_anonymous}, there is no \textit{unique} record anymore and there are three different combinations of values grouped by $k=4$ records.

\setlength{\tabcolsep}{5pt}
\renewcommand{\arraystretch}{1.4}
\begin{table}[!ht]
    \scriptsize
    \centering
    \begin{tabular}{|r|c|c|c|c|c|}
            \cline{2-6}
            \multicolumn{1}{c|}{} &  \multicolumn{4}{c|}{Quasi Identifiers -- QIDs} & Sensitive \\
            \hline
            ID & Zip & Age & Gender & Nationality & Disease \\
            \hline
            1 & 13053 & 28 & M& Russian & Tuberculosis \\
            \hline
            2 & 13068 & 29 & M& American & Heart \\
            \hline
            3 & 13068 & 21 & F& Japanese & Viral \\
            \hline
            4 & 13053 & 23 & M& American & Viral \\
            \hline
            5 & 14853 & 49 & M& Indian & Cancer \\
            \hline
            6 & 14853 & 48 & F& Russian & Heart \\
            \hline
            7 & 14850 & 47 & M& American & Viral \\
            \hline
            8 & 14850 & 49 & F& American & Viral\\
            \hline
            9 & 13053 & 31 & M& American & Cancer \\
            \hline
            10 & 13053 & 37 & M &Indian & Cancer \\
            \hline
            11 & 13068 & 36 & F& Japanese & Cancer \\
            \hline
            12 & 13068 & 35 & F& American & Cancer \\
            \hline
            
  \end{tabular}
    \caption{An example of a pseudonymized dataset (adapted from~\cite{Machanavajjhala2006}).}
    \label{ch2:tab_pseudo_dataset}
\end{table}

\setlength{\tabcolsep}{5pt}
\renewcommand{\arraystretch}{1.4}
\begin{table}[!ht]
    \scriptsize
    \centering
    \begin{tabular}{|c|c|c|c|c|l}
    \cline{1-5}
    \multicolumn{4}{|c|}{Quasi Identifiers -- QIDs} & Sensitive &\\
    \cline{1-5}
     Zip & Age & Gender & Nationality & Disease &\\
    \cline{1-5}
    \textcolor{red}{130**} & \textcolor{red}{$[21;31[$} & \textcolor{red}{*}& \textcolor{red}{*} & \textcolor{red}{Tuberculosis} &\multirow{4}{0.5cm}{$\left. \begin{array}{r} \\ \\ \\ \\
                                                                                 \end{array} \right\} \textrm{4 individuals}$} \\
    130** & $[21;31[$ & *& * & Heart &\\
    130** & $[21;31[$ & *& * & Viral &\\
    130** & $[21;31[$ & *& * & Viral &\\
\cline{1-5}
    
    148** & $[41;50[$ & *& * & Cancer & \multirow{4}{0.5cm}{$\left. \begin{array}{r} \\ \\ \\ \\
                                                                                 \end{array} \right\} \textrm{4 individuals}$}\\
    148** & $[41;50[$ & *& * & Heart &\\
    148** & $[41;50[$ & *& * & Viral& \\
    148** & $[41;50[$ & *& * & Viral&\\
\cline{1-5}
    130** & $[31;41[$ & *& * & Cancer & \multirow{4}{0.5cm}{$\left. \begin{array}{r} \\ \\ \\ \\
                                                                                 \end{array} \right\}\textrm{4 individuals}$}\\
    130** & $[31;41[$ & *&* & Cancer &\\
    130** & $[31;41[$ & *& * & Cancer &\\
    130** & $[31;41[$ & *& * & Cancer &\\
    \cline{1-5}
\end{tabular}
    \caption{A $4$-anonymous dataset of Table~\ref{ch2:tab_pseudo_dataset} (adapted from~\cite{Machanavajjhala2006}).}
    \label{ch2:tab_4_anonymous}
\end{table}

However, several studies have pointed out limitations of the \textit{k}-anonymity model, normally resulting in a new syntactic notion of privacy such as \textit{l}-diversity~\cite{Machanavajjhala2006} and \textit{t}-closeness~\cite{Li2007}. For instance, the last four records in Table~\ref{ch2:tab_4_anonymous} exhibits the same sensitive value \textit{Cancer}. So, if an attacker with background knowledge knows someone within $[31;41[$ years old contributed to this dataset, it is obvious the disease value for this person. This is also known as \textit{homogeneity} attack. Besides, \textit{k}-anonymity does not \textit{compose}, i.e., if the same person participates in two independent \textit{k}-anonymous releases, there is no guarantee s/he will be \textit{k}-anonymous in the composition of both dataset. Suppose the person in the first row (in red color) tested positive for tuberculosis in the hospital that release the $4$-anonymous dataset of Table~\ref{ch2:tab_4_anonymous}. Although this hospital had a good laboratory, the person decides to take a second test in another hospital, which releases the $5$-anonymous dataset of Table~\ref{ch2:tab_5_anonymous}. So, if an attacker knows, e.g., that someone is 29 years old, lives in ZIP code 13012, and visited both hospitals, the \textit{unique} record that matches in both Tables~\ref{ch2:tab_4_anonymous} and~\ref{ch2:tab_5_anonymous} is the first one (also in red color). Thus, jeopardizing this user privacy since \textit{k}-anonymity does not compose.

\setlength{\tabcolsep}{5pt}
\renewcommand{\arraystretch}{1.4}
\begin{table}[!ht]
    \centering
    \scriptsize
    \begin{tabular}{|c|c|c|c|c|l}
          \cline{1-5}
          \multicolumn{4}{|c|}{Quasi Identifiers -- QIDs} & Sensitive \\
          \cline{1-5}
          Zip & Age & Gender & Nationality & Disease \\
          \cline{1-5}
          \textcolor{red}{130**} & \textcolor{red}{$<35$} & \textcolor{red}{*} & \textcolor{red}{*} & \textcolor{red}{Tuberculosis} & \multirow{5}{0.5cm}{$\left. \begin{array}{r} \\ \\ \\ \\ \\
                                                                                 \end{array} \right\}\textrm{5 individuals}$}\\
          130** & $<35$ & * & * & Diabetes \\
          130** & $<35$ & * & * & Parkinson \\
          130** & $<35$ & * & * & Parkinson  \\
          130** & $<35$ & * & * & Diabetes \\
          \cline{1-5}
          148*** & $\geq35$ & * & * & Heart & \multirow{5}{0.5cm}{$\left. \begin{array}{r} \\ \\ \\ \\ \\
                                                                                 \end{array} \right\}\textrm{5 individuals}$}\\
          148*** & $\geq35$ & * & * & Cancer \\
          148*** & $\geq35$ & * & * & Viral \\
          148*** & $\geq35$ & * & * & Cancer \\
          148*** & $\geq35$ & * & * & Cancer \\
          \cline{1-5}
        \end{tabular}
    \caption{An example of a $5$-anonymous dataset from a second hospital.}
    \label{ch2:tab_5_anonymous}
\end{table}

\section{Differential Privacy} \label{ch2:sub_dp}

Consider a database that stores the result of an infectious disease of a set of individuals (e.g., Table~\ref{ch2:tab_pseudo_dataset}). From this database, we could learn statistics about the underlying population and publish these statistics publicly. However, information might leak about specific individuals in the database, which could compromise their privacy. In theory, we would like that the global information relative to the population to be public, e.g., ``how many people tested positive for this disease". At the same time, we would like that the information of each individual to be private, i.e., not releasing ``\textit{who} tested positive for the disease". Unfortunately, this is not always possible. For instance, if each time an attacker adds or removes someone of the database and performs the query ``how many people tested positive for this disease?", in the end, it is possible to infer whose people tested positive by calculating the influence of each individual. 

One way to preserve privacy in this scenario is to add some \textit{noise} in the output of the query, which, \textit{ideally}, should not destroy the utility of the data. In other words, the challenge would be to maximize the utility of the released noisy statistics while preserving the privacy of the individuals. Differential privacy (DP)~\cite{Dwork2006,Dwork2006DP} is a formal definition that allows quantifying the privacy-utility trade-off. Indeed, rather than being a privacy property of the \textit{output} dataset (like \textit{k}-anonymity and its variants), DP is a definition that must be respected by a randomized \textit{algorithm} (i.e., algorithmic notion of privacy). 

In recent years, DP has been increasingly accepted as the current standard for data privacy with several large-scale implementations in the real-world~\cite{desfontaines_dp_real_world} (cf.~\cite{uber_dp,linkedin,microsoft,apple,facebook_dp1,facebook_dp2,rappor,aktay2020google,census2021,census,wellenius2020impacts}). One key reason is that DP addresses the paradox of learning about a population while learning nothing about single individuals~\cite{dwork2014algorithmic}. More specifically, the idea is that removing (or adding) a single row from the database should not affect \textit{much} the statistical results. A formal definition of DP is given in the following.

\begin{definition}[($\epsilon, \delta$)-Differential Privacy~\cite{dwork2014algorithmic}] \label{def:dp} Given $\epsilon>0$ and $0 \leq \delta <1$, a randomized algorithm ${\mathcal{A}: \mathcal{D} \rightarrow R}$ is said to provide {($\epsilon, \delta$)-differential-privacy (($\epsilon, \delta$)-DP}) if, for all neighbouring datasets $D_{1},D_{2} \in \mathcal{D}$ that differ on the data of one user, and for all sets $R$ of outputs:

\begin{equation}
{ \Pr[{\mathcal{A}}(D_{1})\in R]\leq e^{\epsilon } \Pr[{\mathcal{A}}(D_{2})\in R]} + \delta \textrm{.}
\label{eq:dp}
\end{equation}

\end{definition}

The additive $\delta$ on the right-side of Eq.~\eqref{eq:dp} is interpreted as a probability of failure. Normally, a common choice for $\delta$ is to set it significantly smaller than $1/n$ where $n$ is the number of users in the database~\cite{dwork2014algorithmic}. Throughout this manuscript, if $\delta=0$, we will just say that $\mathcal{A}$ is $\epsilon$-DP. 

Notice that if $\epsilon$ (a.k.a. the \textit{privacy loss} or the \textit{privacy budget}) is zero, both distributions are equal, and in this case, there is no leakage of information. This is equivalent to the privacy goal stated by Dalenius~\cite{dalenius1977towards} in 1977 as ``\textit{access to a statistical database should not enable one to learn anything about an individual that could not be learned without access}". However, respecting such a statement, as proven in~\cite{Dwork2006DP}, no utility could ever be obtained. So, we have to accept leaking \textit{some} information about individuals in order to have some utility, which is translated to increasing $\epsilon$ (i.e., \textit{privacy-utility trade-off}).

\subsection{Properties of Differential Privacy} \label{ch2:subsub_prop_dp}

Differential privacy possesses several important properties, highlighting its strength in comparison with other privacy models. For instance, with DP, there is no need to define the \textit{background knowledge} that attackers might have, which is equivalent to assuming an attacker with \textit{unlimited resources}. Besides, DP definition protects \textit{anything} associated with a single individual, e.g., their presence in the database and their sensitive information~\cite{desfontaines2020}. On the other hand, DP does not protect against attribute inference as it may leak information about individuals \textbf{not} present in the database.

In addition, DP is immune to \textit{post-processing}, which means it is not possible to make an $\epsilon$-DP mechanism less differentially private by evaluating any function $f$ of the response of the mechanism, given that there is no additional information about the database.

\begin{proposition}[Post-Processing of DP~\cite{dwork2014algorithmic}] If $\mathcal{A} : \mathcal{D} \rightarrow R$ is $\epsilon$-DP, then $f (\mathcal{A})$ is also $\epsilon$-DP for any function $f$.
\end{proposition}

Furthermore, DP also \textit{composes} well, which is one of the most powerful features of this privacy model. For instance, accounting for the \textit{overall} privacy loss consumed in a pipeline of several DP algorithms applied to the same database is feasible due to composition. We recall two types of composition below.

\begin{proposition}[Sequential Composition~\cite{dwork2014algorithmic}] \label{ch2:prop_sequential_composition} Let $\mathcal{A}_1$ be an $\epsilon_1$-DP mechanism and $\mathcal{A}_2$ be an $\epsilon_2$-DP mechanism. Then, the mechanism $\mathcal{A}_{1,2}(\mathcal{D})=\left( \mathcal{A}_1(\mathcal{D}), \mathcal{A}_2(\mathcal{D})\right)$ is $(\epsilon_1+\epsilon_2)$-DP.

\end{proposition}

\begin{proposition}[Parallel Composition~\cite{dwork2014algorithmic}] \label{ch2:prop_parallel_composition} Let $\mathcal{A}_1$ be an $\epsilon_1$-DP mechanism and $\mathcal{A}_2$ be an $\epsilon_2$-DP mechanism. Let $D_1$ and $D_2$ be arbitrary disjoint subsets of the input domain $\mathcal{D}$. Then, the mechanism $\mathcal{A}_{1,2}(\mathcal{D})=\left( \mathcal{A}_1(D_1), \mathcal{A}_2(D_2)\right)$ is $max(\epsilon_1,\epsilon_2)$-DP. 

\end{proposition}

\subsection{Differentially Private Mechanisms: Laplace and Gaussian}  \label{ch2:sub_lap_gauss}

Any mechanism that respects Definition~\ref{def:dp} can be considered differentially private. Two widely used DP mechanisms for numeric queries (i.e., functions $f : \mathcal{D} \rightarrow \mathbb{R}$) are the Laplace mechanism~\cite{Dwork2006} and the Gaussian mechanism~\cite{dwork2014algorithmic}. One important parameter that determines how accurately we can answer the queries is their \textit{sensitivity}. We recall the definition of $\ell_1$- and $\ell_2$-sensitivity and both Laplace and Gaussian mechanisms below, respectively.

\begin{definition}[$\ell_1$-sensitivity~\cite{dwork2014algorithmic}] The $\ell_1$-sensitivity of a function $f : \mathcal{D} \rightarrow \mathbb{R}$, for all neighbouring datasets $D_{1},D_{2} \in \mathcal{D}$ that differ on the data of one user, is:

\begin{equation*}
    \Delta_1 (f) = max \textrm{  } || f(D_1) - f(D_2)  ||_1
\end{equation*}

\end{definition}

\begin{definition}[$\ell_2$-sensitivity~\cite{dwork2014algorithmic}] The $\ell_2$-sensitivity of a function $f : \mathcal{D} \rightarrow \mathbb{R}$, for all neighbouring datasets $D_{1},D_{2} \in \mathcal{D}$ that differ on the data of one user, is:

\begin{equation*}
    \Delta_2 (f) = max \textrm{  } || f(D_1) - f(D_2)  ||_2
\end{equation*}

\end{definition}

\begin{definition}[Laplace mechanism~\cite{Dwork2006}] For a query function $f : D \rightarrow \mathbb{R}$ over a dataset $D \in \mathcal{D}$, the Laplace mechanism is defined as:

\begin{equation*}
    \mathcal{A}_L (D, f(.), \epsilon) = f(D) + Lap\left ( \frac{\Delta_1}{\epsilon} \right),
\end{equation*}

\end{definition}

in which $Lap(b)$ is the Laplace distribution centered around 0 and of scale $b$. The Laplace mechanism is proven to preserve $\epsilon$-DP~\cite{Dwork2006}.

\begin{definition}[Gaussian mechanism~\cite{dwork2014algorithmic}] For a query function $f : D \rightarrow \mathbb{R}$ over a dataset $D \in \mathcal{D}$ and for $\sigma = \frac{\Delta_2}{\epsilon} \sqrt{2 \ln{(1.25/\delta)}}$, the Gaussian mechanism is defined as:

\begin{equation*}
    \mathcal{A}_G (D, f(.), \epsilon, \delta) = f(D) + \mathcal{N}\left (0, \sigma^2 \right)
\end{equation*}

\end{definition}

in which $\mathcal{N}\left (0, \sigma^2 \right)$ is the normal distribution centered at 0 with variance $\sigma^2$. For $\epsilon \in (0,1)$, the Gaussian mechanism provides ($\epsilon, \delta$)-DP~\cite{dwork2014algorithmic}.

\subsection{Privacy amplification by sampling} \label{ch2:sub_sampling}

There exist scenarios in which using a random subsample of the database is sufficient to approximate the overall distribution of the original database (e.g., census data). Sampling is a fundamental tool in the design of differentially private mechanisms as there is an \textit{amplification} effect~\cite{Chaudhuri2006,Li2012,balle2018privacy,balle2020privacy,first_ldp}. For instance, amplification by sampling plays a key role in machine learning since many classes of algorithms utilize sampling strategies during the training process (e.g., differentially private stochastic gradient descent~\cite{DL_DP}).

So, why is there an amplification effect? Informally, assume we extract a random subsample from a database and, next, we apply a DP mechanism to this sampled database. Observe now that there is more \textit{uncertainty} on the output of the DP mechanism since an attacker would be, first, unable to distinguish which data samples were used and, second, there is the DP guarantee. More rigourously, Li et al.~\cite[Theorem 1]{Li2012} theoretically prove this effect. 

\begin{theorem} \label{theo:amp_sampling} \textbf{Amplification by Sampling}~\cite{Li2012}. Let $\mathcal{A}$ be an $\epsilon'$-DP mechanism and $\mathcal{S}$ to be a sampling algorithm with sampling rate $\beta$. Then, if $\mathcal{S}$ is first applied to a dataset $\mathcal{D}$, which is later sanitized with $\mathcal{A}$, the derived result satisfies $\epsilon$-DP with $\epsilon=\ln{\left( 1 + \beta (e^{\epsilon'} - 1)  \right)}$.

\end{theorem}

\section{Local Differential Privacy} \label{ch2:sub_ldp}

The centralized DP model from Section~\ref{ch2:sub_dp}, assumes that a trusted curator has access to compute on the entire raw data of users. By `trusted', we mean that curators do not misuse or leak private information from individuals. However, this assumption does not always hold in real life~\cite{data_breaches}. To preserve privacy at the user-side, an alternative approach, namely, local differential privacy (LDP), was initially formalized in~\cite{first_ldp}. With LDP, rather than trusting in a data curator to have the raw data and sanitize it to output queries, each user applies a DP mechanism to their data before transmitting it to the data collector server. A formal definition of LDP is given in the following:

\begin{definition}[$\epsilon$-Local Differential Privacy]\label{def:ldp} A randomized algorithm ${\mathcal{A}}$ satisfies $\epsilon$-local-differential-privacy ($\epsilon$-LDP) if, for any pair of input values $v_1, v_2 \in Domain(\mathcal{A})$ and any possible output $y$ of ${\mathcal{A}}$:

\begin{equation*}
    \Pr[{\mathcal{A}}(v_1) = y]\leq e^{\epsilon }\cdot \Pr[{\mathcal{A}}(v_2) = y] \textrm{.}
\label{eq:ldp}
\end{equation*}
\end{definition}

Intuitively, $\epsilon$-LDP guarantees that an attacker can not distinguish whether the true value is $v_1$ or $v_2$ (input) with high confidence (controlled by $\epsilon$) irrespective of the background knowledge one has. This is because both values have approximately the same probability to generate the same perturbed output. Similar to the centralized model of DP, LDP also enjoys the properties described in Section~\ref{ch2:subsub_prop_dp}, e.g., immunity to post-processing and composition~\cite{dwork2014algorithmic}. 

The LDP model allows collecting data in unprecedented ways and, therefore, has led to several adoptions by industry. For instance, big tech companies like Google, Apple, and Microsoft, reported the implementation of LDP mechanisms to gather statistics in well-known systems (i.e., Google Chrome browser~\cite{rappor}, Apple iOS and macOS~\cite{apple}, and Windows 10 operation system~\cite{microsoft}). Indeed, there is a rich literature on LDP models~\cite{Duchi2013,Duchi2018,Bassily2015,Cormode2021,Murakami2019,zhou2021local,bassily2017practical,rappor,microsoft,tianhao2017,kairouz2016extremal,kairouz2016discrete,wang2019,xiao2,apple,Hadamard,Chamikara2020,Fernandes2019}, and we refer the interest reader to recent survey works on LDP~\cite{Xiong2020_survey,Wang2020_survey,yang2020_survey}. 

In this manuscript, \textbf{we focus on the \textit{fundamental} problem of private frequency (or histogram) estimation} under $\epsilon$-LDP guarantees. This is a primary objective of LDP, in which the data collector decodes all the sanitized data of the users and can then estimate the number of users for each possible value. The frequency estimation task has received considerable attention in the literature~\cite{Cormode2021,Murakami2019,tianhao2017,kairouz2016discrete,Hadamard,Alvim2018,rappor,microsoft,Zhao2019,Li2020,Fernandes2019} as it is a building block for other complex tasks (e.g., heavy hitter estimation~\cite{Bassily2015,Wang2021,bassily2017practical,Bun2019}, estimating marginals~\cite{Peng2019,Zhang2018,Ren2018,Fanti2016}, frequent itemset mining~\cite{Wang2018,Qin2016}). 

Let $A_j=\{v_1,v_2,...,v_{c_j}\}$ be a set of $c_j=|A_j|$ values of a given attribute and let $\epsilon$ be the privacy budget. Each user $u_i$, for $i \in \{1,2,...,n\}$, has a value $v \in A_j$. Thus, the aggregator's goal is to estimate a $c_j$-bins histogram, including the frequency of all values in $A_j$. Algorithm~\ref{alg:general_LDP} exhibits the general procedure for frequency estimation under LDP, which includes: \texttt{Encoding} and \texttt{Randomization} at the user-side, and \texttt{Aggregation} at the server-side (i.e., the aggregator).

\begin{algorithm}
\centering
\caption{General procedure for frequency estimation under LDP}
\label{alg:general_LDP}
\begin{algorithmic}[1]
\Statex \textbf{Input :} Original data of users, privacy parameter $\epsilon$, and local randomizer $\mathcal{A}$.
\Statex \textbf{Output :} Estimated frequencies.

\# \texttt{User-side}

\State \textbf{for} each user $u_i$ ($i \in \{1,2,...,n\}$) with input value $v \in A_j$ \textbf{do}

\State  \hskip1em \texttt{Encode}($v$) into a specific format (\textbf{if needed});
\State  \hskip1em \texttt{Randomize}($v$) with $\mathcal{A}(v,\epsilon)$;
\State  \hskip1em Transmit the randomized output to the aggregator. 

\State \textbf{end for}

\# \texttt{Server-side}
\State The server \texttt{aggregates} the reported values and estimates their frequency.

\State \textbf{return :} $c_j$-bins histogram, including the frequency of all values in $A_j$.
\end{algorithmic}
\end{algorithm}

In addition, if one intends to collect data from the same population, i.e., longitudinal studies, the authors in~\cite{rappor} introduced the concept of \textit{memoization}. The idea behind \textit{memoization} is to use two steps of sanitization, where the first step uses an upper bound value of $\epsilon_{\infty}$-LDP and only outputs lower epsilon reports using this randomized data. This will be a subject of study in Chapter~\ref{chap:chapter5}. In the next three subsections, we will review state-of-the-art LDP protocols for non-longitudinal frequency estimation (a.k.a. frequency oracles). 

\subsection{Randomized response} \label{ch2:sub_RR}

Randomized response (RR) is a surveying technique proposed by Warner~\cite{Warner1965}, to provide plausible deniability for individuals responding to embarrassing questions. Suppose we want to do a survey to know ``how many people have already cheated on their partner". Due to social embarrassment, people would probably hesitate to answer this question honestly, thus lying on their answer. Instead, with RR, users would benefit from plausible deniability to their answers, following the scheme below.

Each user, throw a secret coin: 
\begin{itemize}
    \item If Tails throw the coin again (ignoring the outcome) and answer the question honestly;
    \item If Heads, then throw the coin again and answer ``Yes" if Head, and ``No" if Tail.
\end{itemize}

Notice that even if users might have answered ``Yes", we still would not be sure if they answered honestly or at random. With more details, Figure~\ref{fig:RR} illustrates the probability tree of the RR protocol with an unbiased coin (i.e., with equal probability $1/2$).

\begin{figure}[ht]
\centering
\tikzset{
  treenode/.style = {shape=rectangle, rounded corners,
                     draw, align=center,
                     top color=white},
  root/.style     = {treenode},
  env/.style      = {treenode},
  dummy/.style    = {circle,draw}
}
\tikzstyle{level 1}=[level distance=2.5cm, sibling distance=2.75cm]
\tikzstyle{level 2}=[level distance=2.5cm, sibling distance=2cm]

\begin{tikzpicture}
  [
    grow                    = right,
    edge from parent/.style = {draw, -latex},
    every node/.style       = {font=\footnotesize},
    sloped
  ]
  \node [root] {RR}
    child { node [env] {Head}
        child { node [env] {No}
          edge from parent node [below] {1/2} }
        child { node [env] {Yes}
          edge from parent node [above] {1/2} }
        edge from parent node [below] {1/2} }
    child { node [env] {Tail}
        child { node [env] {No}   
            edge from parent node [below] {Truth}}
        child { node [env] {Yes}   
            edge from parent node [above] {Truth}}
        edge from parent node [above] {1/2}};
\end{tikzpicture}
\caption{Summary of randomized response method with unbiased coins (i.e., with equal $1/2$ probability).}\label{fig:RR}
\end{figure}
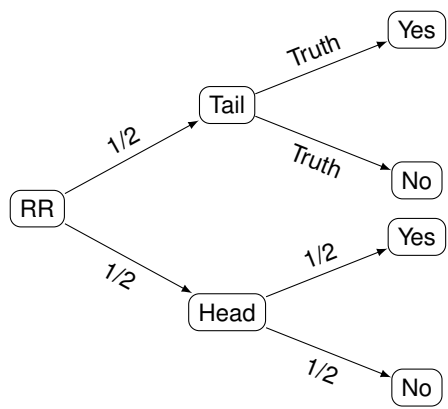

From Figure~\ref{fig:RR}, let $\mathcal{A}$ represent the RR mechanism, we can calculate the following probabilities:

\begin{gather}
   \Pr[\mathcal{A}(Yes)=Yes] = \Pr[\mathcal{A}(No)=No] = 0.75 \textrm{,}    \\       
   \Pr[\mathcal{A}(Yes)=No] = \Pr[\mathcal{A}(No)=Yes] = 0.25   \textrm{.}
\end{gather}

So, now, the objective is to estimate the frequency of ``Yes" and ``No" answers, i.e., the distribution of the original data. Let $f(v_y)$ be the proportion of \textit{true} ``Yes" answers and $N_y$ be the proportion of \textit{observed} ``Yes" answers. The following equation gives an estimated relation between these two variables:

\begin{equation*}
    N_y \approx \frac{1}{2} f(v_y) + \frac{1}{4}n \textrm{.}
\end{equation*}

The higher the number of samples $n$, with high probability, the more accurate the frequency estimation will be. In this case, $f(v_y)$ can be estimated with: 

\begin{equation*}
\hat{f}(v_y) \approx 2 N_y - \frac{1}{2}n \textrm{.}
\end{equation*}

Similarly, we can calculate the number of estimated ``No" answers. Translating the unbiased-coin RR model to DP theory, this model satisfies $\epsilon$-LDP with $\epsilon=\ln \left(\frac{0.75}{0.25}\right)=\ln (3)$~\cite{dwork2014algorithmic}. More generically, given $v \in \{0,1\}$ we can design an RR protocol to satisfy an arbitrary $\epsilon$ value (i.e., with biased coins) with the following perturbation function~\cite{kairouz2016discrete,kairouz2016extremal}:

\begin{equation*}
    \forall{y \in \{0,1\}} \Pr[\mathcal{A}_{RR(\epsilon)}(v)=y] = \begin{cases} p=\frac{e^{\epsilon}}{e^{\epsilon}+1} , \textrm{ if } y = v\\ q=\frac{1}{e^{\epsilon}+1}, \textrm{ if } y \neq v \textrm{,} \end{cases}
\end{equation*}

This satisfies $\epsilon$-LDP since $\frac{p}{q}=e^{\epsilon}$. Notice that the RR algorithm does not require any encoding technique. To estimate the normalized frequency $f(v_i)$ that a value $v_i \in V$ occurs where $V=\{v_1,v_2\}=\{0,1\}$, one calculates~\cite{kairouz2016discrete,kairouz2016extremal}:

\begin{equation}\label{eq:est_pure}
    \hat{f}(v_i) = \frac{N_i - nq}{n(p - q)} \textrm{,}
\end{equation}

in which $N_i$ is the number of times the value $v_i$ has been reported and $n$ is the total number of users. In Theorems 1 and 2 from~\cite{tianhao2017}, it is shown that $\hat{f}(v_i)$ is an \textbf{unbiased estimation} of the true frequency $f(v_i)$ (i.e., $E[\hat{f}(v_i)]=f(v_i)$), and \textbf{the variance of this estimation} is calculated as:

\begin{equation}\label{eq:var}
    Var[\hat{f}(v_i)]= \frac{q(1-q)}{n(p-q)^2} + \frac{f(v_i)(1-p-q)}{n(p-q)}  \textrm{.}
\end{equation}

Since the estimation in Eq.~\eqref{eq:est_pure} is unbiased, its variance $Var[\hat{f}(v_i)]$ is equal to the mean squared error (MSE)~\cite{mse_var} that is commonly used as an accuracy metric (e.g., cf.~\cite{Wang2020,Wang2020_post_process,Wang2021_b,li2021privacy}), also adopted throughout this manuscript. More formally, 

\begin{equation} \label{eq:mse_var}
    \begin{aligned}
        MSE &= \frac{1}{|V|} \sum_{v \in V} E \left[ \left( \hat{f}(v_i) - f(v_i) \right)^2 \right ] \\
        &= \frac{1}{|V|} \sum_{v \in V} \left( Var[\hat{f}(v_i)] + (E[\hat{f}(v_i)] - f(v_i))^2 \right)\\ 
        &= \frac{1}{|V|} \sum_{v \in V} Var[\hat{f}(v_i)] \textrm{.}
    \end{aligned}
\end{equation}

Furthermore, with no knowledge about the real frequency $f(v_i)$ and because in real life the vast majority of values appear very infrequently, we will consider $f(v_i) = 0$. Notice that this is common practice in the literature (e.g., cf.~\cite{tianhao2017,Wang2021_b}), which provides an approximation for the variance as~\cite{tianhao2017}:

\begin{equation}\label{eq:var_pure}
    Var^*[\hat{f}(v_i)] = \frac{q (1 - q)}{n(p - q)^2} \textrm{.}
\end{equation}

Replacing $p=\frac{e^{\epsilon}}{e^{\epsilon}+1}$ and $q=\frac{1}{e^{\epsilon}+1}$ into Eq.~\eqref{eq:var_pure}, the RR variance is calculated as:

\begin{equation*}
    Var^*[\hat{f}_{RR}(v_i)] = \frac{e^{\epsilon}}{n(e^{\epsilon}-1)^2} \textrm{.}
\end{equation*}

\subsection{Generalized randomized response} \label{ch2:sub_GRR}

The \textit{k}-Ary RR~\cite{kairouz2016discrete} mechanism extends RR to the case of $c_j \geq 2$ and it is also referred to as direct encoding~\cite{tianhao2017} (since no particular encoding needed) or generalized RR (GRR)~\cite{Wang2018,Wang2020_post_process,Zhang2018}. Throughout this manuscript, we will use the term GRR for this LDP protocol. Given a value $v \in A_j$, \textit{GRR($v$)} outputs the true value with probability $p$, and any other value $v'\in A_j$ such that $v' \neq v$ with probability $1-p$. More formally, the perturbation function is defined as:

\begin{equation*}
    \forall{y \in A_j} \Pr[\mathcal{A}_{GRR(\epsilon)}(v)=y] = \begin{cases} p=\frac{e^{\epsilon}}{e^{\epsilon}+c_j-1} , \textrm{ if } y = v\\ q=\frac{1}{e^{\epsilon}+c_j-1}, \textrm{ if } y \neq v \textrm{.} \end{cases}
\end{equation*}

GRR satisfies $\epsilon$-LDP since $\frac{p}{q}=e^{\epsilon}$. The estimated frequency $\hat{f}(v_i)$ that a value $v_i$ occurs for $i \in [1,c_j]$ is also calculated using Eq.~\eqref{eq:est_pure}. Replacing $p=\frac{e^{\epsilon}}{e^{\epsilon}+c_j-1}$ and $q=\frac{1}{e^{\epsilon}+c_j-1}$ into Eq.~\eqref{eq:var_pure}, the GRR variance is calculated as:

\begin{equation}\label{eq:var_grr}
    Var^*[\hat{f}_{GRR}(v_i)] = \frac{e^{\epsilon} + c_j - 2}{n(e^{\epsilon}-1)^2} \textrm{.}
\end{equation}

\subsection{Unary encoding protocols} \label{ch2:sub_UE}

Protocols based on unary encoding (UE) consist of transforming a value $v$ into a binary representation of it. So, first, for a given value $v$, $B=Encode(v)$, where $B=[0,0,...,1,0,...0]$, a $c_j$-bit array where only the $v$-th position is set to one. Next, the bits from $B$ are flipped independently, depending on parameters $p$ and $q$, to generate a sanitized vector $B'$, in which:

\begin{equation*}
    \Pr[B_i'=1] =\begin{cases} p, \textrm{ if } B_i=1 \\ q, \textrm{ if } B_i=0 \textrm{.}\end{cases}
\end{equation*}

The proof that UE-based protocols satisfy $\epsilon$-LDP for

\begin{equation}\label{eq:eps_UE}
    \epsilon = ln\left( \frac{p(1-q)}{(1-p)q} \right )  \textrm{,}
\end{equation}

is known in the literature and can be found in~\cite{rappor,tianhao2017}. In~\cite{tianhao2017} the authors presents two ways for selecting probabilities $p$ and $q$, which determines the protocol variance. One well-known UE-based protocol is the Basic One-time RAPPOR~\cite{rappor}, referred to as symmetric UE (SUE), which selects $p=\frac{e^{\epsilon/2}}{e^{\epsilon/2}+1}$ and $q=\frac{1}{e^{\epsilon/2}+1}$, where $p+q=1$ (symmetric). The estimated frequency $\hat{f}(v_i)$ that a value $v_i$ occurs for $i \in [1,c_j]$ is also calculated using Eq.~\eqref{eq:est_pure}. Replacing $p=\frac{e^{\epsilon/2}}{e^{\epsilon/2}+1}$ and $q=\frac{1}{e^{\epsilon/2}+1}$ into Eq.~\eqref{eq:var_pure}, the SUE variance is calculated as~\cite{rappor}:

\begin{equation}\label{eq:var_sue}
    Var^*[\hat{f}_{SUE}(v_i)] = \frac{e^{\epsilon/2}}{n(e^{\epsilon/2}-1)^2} \textrm{.}
\end{equation}

Moreover, rather than selecting $p$ and $q$ to be symmetric, Wang et al.~\cite{tianhao2017} proposed optimized UE (OUE), which selects parameters $p=\frac{1}{2}$ and $q=\frac{1}{e^{\epsilon}+1}$ that minimize the variance of UE-based protocols while still satisfying $\epsilon$-LDP. Similarly, the estimation method used in Eq.~\eqref{eq:est_pure} equally applies to OUE. Replacing $p=\frac{1}{2}$ and $q=\frac{1}{e^{\epsilon}+1}$ into Eq.~\eqref{eq:var_pure}, the OUE variance is calculated as~\cite{tianhao2017}:

\begin{equation}\label{eq:var_oue}
   Var^*[\hat{f}_{OUE}(v_i)] = \frac{4e^{\epsilon}}{n(e^{\epsilon}-1)^2} \textrm{.}
\end{equation}

\subsection{Adaptive LDP protocol} \label{ch2:sub_ADP}

Comparing Eq.~\eqref{eq:var_grr} with Eq.~\eqref{eq:var_oue}, elements $c_j-2+e^{\epsilon}$ is replaced by $4e^{\epsilon}$. Thus, as highlighted in~\cite{tianhao2017}, when $c_j < 3e^{\epsilon} +2$, the utility loss with GRR is lower than the one of OUE. This adaptive selection of LDP protocol has been used in many settings in the literature~\cite{Zhang2018,Wang2018}. Throughout this manuscript, we will use the term adaptive (ADP) to denote this best-effort and dynamic selection of LDP mechanism.

\section{Geo-Indistinguishability} \label{ch2:sub_geo_ind}

Geo-indistinguishability (GI)~\cite{Andrs2013} is based on a generalization of DP developed in~\cite{Chatzikokolakis2013} and has been proposed for preserving location privacy without the need of a trusted curator (e.g., a malicious location-based service), i.e., a local DP model. A mechanism satisfies $\epsilon$-GI if for any two locations $x_1$ and $x_2$ within a radius $r$, the output $y$ of them is $(\epsilon,r)$-geo-indistinguishable if we have:

\begin{equation*}\label{eq:geo}
\frac{\Pr(y|x_1)}{\Pr(y|x_2)} \leq e^{\epsilon r} \textrm{, } \forall r>0 \textrm{, } \forall y \textrm{, } \forall x_1,x_2: d(x_1,x_2)\leq r\textrm{.}
\end{equation*}

Intuitively, this means that for any point $x_2$ within a radius $r$ from $x_1$, GI forces the corresponding distributions to be at most $l=\epsilon r$ distant. In other words, the level of distinguishability $l$ increases with $r$, e.g., an attacker can distinguish that the user is in Paris rather than London but can hardly (controlled by $\epsilon$) determine the user's exact location. Although both GI and DP use the notation of $\epsilon$ to refer to the privacy budget, they cannot be compared directly because $\epsilon$ in GI contains the unit of measurement (e.g., meters).

On the continuous plane (as we consider in this manuscript), an intuitive polar Laplace mechanism has been proposed in~\cite{Andrs2013} to achieve GI, which is briefly described in the following. Rather than reporting the user's true location $x \in \mathbb{R}^2$, we report a point $y \in  \mathbb{R}^2$ generated randomly according to $D_{\epsilon}(y) =\frac{\epsilon^2}{2\pi} e^{-\epsilon d_2(x,y)}$. Algorithm~\ref{alg:GI_location} shows the pseudocode of the polar Laplace mechanism in the continuous plane. More specifically, the noise is drawn by first transforming the true location $x$ to polar coordinates. Then, the angle $\theta$ is drawn randomly between $[0,2\pi)$ (line 3), and the distance $r$ is drawn from $C^{-1}_{\epsilon}(p)$ (line 5), which is calculated using the negative branch $W_{-1}$ of the Lambert $W$ function~\cite{w_lambert}. Finally, the generated distance and angle are added to the original location.

\begin{algorithm}[!ht]
\centering
\caption{Polar Laplace mechanism in continuous plane~\cite{Andrs2013}}
\label{alg:GI_location}
\begin{algorithmic}[1]
\Statex \textbf{Input :} $\epsilon>0$, real location $x \in \mathbb{R}^2$.
\Statex \textbf{Output :}  sanitized location $y \in \mathbb{R}^2 $.

\State Draw $\theta$ uniformly in $[0,2\pi)$
\State Draw $p$ uniformly in $[0,1)$
\State Set $r = C^{-1}_{\epsilon}(p) = -\frac{1}{\epsilon} \left ( W_{-1} \left ( \frac{p-1}{e}\right) + 1 \right )$

\State \textbf{return :} $y=x+\langle r \cos{(\theta)},r \sin{(\theta)} \rangle$
\end{algorithmic}
\end{algorithm}

\section{Conclusion} \label{ch2:sub_conclusion}

In this chapter, we have revised state-of-the-art anonymization techniques. We started with the well-known \textit{k}-anonymity model, presenting its definition, an intuitive example, as well as some of its limitations. We then presented differential privacy, which is a definition that should be satisfied by a randomized algorithm. While the former satisfies a syntactic notion of privacy, i.e., the final database should satisfy ``\textit{k}-anonymity", DP is a property of the process. In addition, DP offers strong post-processing and composition properties, which are important in designing differentially private systems for real-life applications. Besides, we have presented the decentralized setting of DP, also known as local DP, in which there is no need to assume a trusted server. In the LDP setting, the aggregator already knows the users' identifiers, but not their private data. In this case, users apply a differentially private algorithm in their own device such that only perturbed data is sent to the aggregator. Also, we have presented geo-indistinguishability, which is an LDP model to protect location privacy. Geo-indistinguishability utilizes a Laplacian noise to perturb the actual location of a user before transmitting to the (un)trusted server and has received considerable attention due to its effectiveness and simplicity of implementation (e.g., Location Guard~\cite{loc_guard}). Lastly, for each DP model, we have presented the main mechanisms that will be used throughout this manuscript.

%% file: chapters/chapter3.tex
\chapter{Machine Learning and Databases Used on Experiments}
\label{chap:chapter3}

In Chapter~\ref{chap:chapter2}, we have revised the background on data anonymization techniques. In this chapter, we now briefly review the background on machine learning techniques and concepts that our work utilizes, as well as the databases we experiment on. We highlight that the content of this chapter related to machine learning is \textbf{primarily} inspired by existing literature~\cite{geron2019hands,goodfellow2016deep,james2013introduction}. Appropriate references to other works are provided throughout this chapter.

\section{Introduction to Machine Learning}

Following the definition of machine learning (ML) given by G{\'e}ron in their book~\cite{geron2019hands} ``\textit{Machine Learning is the science (and art) of programming computers so they can learn from data}." In contrast with traditional programming techniques that are based on conditional and loop statements, ML automatically learns \textit{from data}. The way of \textit{learning} ranges, e.g., from supervised, unsupervised, semi-supervised, and reinforcement learning. In this manuscript, we focus only on \textbf{supervised} learning, in which the ML algorithms also receive the desired outputs (e.g., a scalar, a label). ML supervised applications typically solve prediction and classification tasks both approached in this manuscript.

\subsection{Classification Problems}

Classification predictive modeling problems have as main goal to predict a class label. Indeed, based on a set of input $X$ the objective is to classify each sample in a given discrete label $y$. The output variables are frequently referred to as labels or categories. A classical example of a classification task is spam filters, in which a classifier is trained over emails labeled as spam or not spam. In general, depending on the objective one may want to train ML classification algorithms for binary, multiclass, or multilabel problems, for example.

\subsection{Regression Problems}

Regression predictive modeling problems have as main goal the prediction of a numerical value. More precisely, based on a set of input $X$, the objective is to predict a numerical value $y$. For example, the price of a house may be predicted by using as predictors the number of bedrooms, its area, its location, and so on. Generally, a problem with multiple inputs is often referred to as a multivariate regression problem. One special type of regression is with \textbf{ordered data}, also known as \textit{time-series} data. In these cases, the order of the samples \textit{matters}. Indeed, time-series data is a set of observations collected by repeated measures throughout time. There are many practical applications for time series data in both classification and regression problems. For example, forecasting the spread of infectious diseases~\cite{Rahimi2021}, tracking financial market indices~\cite{Sezer2020}, and forecasting human mobility~\cite{luca2020deep}, to name a few. 

\subsection{Modeling Techniques} \label{ch3:ML_models}

To select the most performing ML algorithm per problem we tackled, we generally evaluated one or more among the ML models described in the next three subsections (Section~\ref{ch3:lasso}--~\ref{ch3:dl_models}).

\subsubsection{Linear Model} \label{ch3:lasso}

In this manuscript, we only considered a regularized version of the Linear Regression model, namely, least absolute shrinkage and selection operator (LASSO)~\cite{lasso}, which is widely used for prediction purposes. The LASSO is a method of contracting the coefficients of the regression, whose ability to select a subset of variables is due to the nature of the constraint on the coefficients. Originally proposed by Tibshirani~\cite{lasso} for models using the standard least squares estimator, it has been extended to many statistical models such as generalized linear models. We used the LASSO implementation from the Scikit-learn library~\cite{scikit}.

\subsubsection{Decision Tree Algorithms}\label{ch3:trees_models}

One of the popular predictive modeling techniques used in ML is decision tree learning~\cite{Breiman2017}. Decision tree-based algorithms are often chosen for predictive modeling because of their interpretability and high performance. We evaluated two decision tree learning algorithms in this manuscript:

\begin{itemize}
    
    \item Extreme Gradient Boosting (XGBoost)~\cite{XGBoost} is a decision-tree-based ensemble ML algorithm that produces a predictive model based on an ensemble of weak predictive models (decision trees). XGBoost uses a novel regularization approach over standard gradient boosting machines, which significantly decreases the model's complexity. The system is optimized by a quick parallel tree construction and adapted to be fault-tolerant under distributed environments.
    
    \item Light Gradient Boosted Machine (LGBM)~\cite{NIPS2017_6907} is a novel gradient boosting framework, which implemented a leaf-wise strategy. This strategy significantly reduces computational speed and resource consumption in comparison to other decision tree-based algorithms.
    
\end{itemize}

\subsubsection{Artificial Neural Networks} \label{ch3:dl_models}

Another popular active research area in ML is artificial neural networks. Neural networks are the foundation of deep learning (DL), which has become a progressively popular research topic. We used the Keras library~\cite{keras} to implement all our DL models. Throughout this manuscript, we will evaluate one or more of the following DL methods:

\begin{itemize}
    \item Multilayer Perceptron (MLP) is an artificial neural network of the feedforward type~\cite{goodfellow2016deep,LeCun2015,Chollet2017}, characterized by a unidirectional flow of computation. MLPs are based on the interconnection of several units (neurons) to transmit signals, which are normally structured into three or more layers, namely, input, hidden(s), and output.
    
    \item Recurrent neural network (RNN) is a specialized class of neural networks used to process sequential data (e.g., time-series data). RNNs have at least one feedback connection that provides the ability to use contextual information when mapping between input and output sequences. In this manuscript, we have applied three state-of-the-art improvements over the standard RNN, which are described in the following:
    
    \begin{itemize}
        \item Long Short-Term Memory~\cite{LSTM} is a type of RNN that overcomes the vanishing gradient problem of standard RNNs. Inside its cell memory unit, the learning process is controlled by three gates: input, forget, and output, which give LSTM the ability to learn which data in a sequence is important to keep or to discard.
    
        \item Gated Recurrent Unit~\cite{GRU} is also a type of RNN, which works using the same principle as LSTM. GRU utilizes two gates: update and reset, which decide what information should be passed to the output.
        
        \item Bidirectional RNN (BiRNN)~\cite{BI_RNN} is a combination of two RNNs: one RNN moves forward while the other moves backward. That is, BiRNN connects two hidden layers of opposite directions to the same output. The RNN cells in a BiRNN can either be standard RNNs, LSTMs, GRUs, and so on.    
    \end{itemize}
    
\end{itemize}

\subsection{Model Selection and Hyperparameter Tuning} 

Generally, besides multiple alternatives of ML algorithms for a given task, there are as well several hyperparameters to tune in each of them. More precisely, let be given the definition from~\cite{james2013introduction}: ``\textit{The process of evaluating a model’s performance is known as model assessment, whereas the process of selecting the proper level of flexibility for a model is known as model selection.}"

Throughout this manuscript, we assess the performance scores of our ML models on a \textit{hold-out} testing set. In the following two subsections, we describe the performance metrics (Section~\ref{ch3:sub_metrics}) and the hyperparameters' optimization methods (Section~\ref{ch3:optimization}) considered in this manuscript.

\subsection{Performance Metrics} \label{ch3:sub_metrics}

Throughout this manuscript, we used common metrics from the literature to evaluate the performance of ML models. For \textbf{regression} tasks, we considered using one or more of the following metrics:

\begin{itemize}
        
    \item Root mean squared error (RMSE) measures the square root average of the squares of the errors and is calculated as: $RMSE=\frac{1}{n} \sqrt{\sum_{i=1}^{n} \left( y_i - \hat{y}_i \right)^2}$;
    
    \item Mean absolute error (MAE) measures the averaged absolute difference between real and predicted values and is calculated as: $MAE=\frac{1}{n} \sum_{i=1}^{n} |y_i - \hat{y}_i|$;

    \item Mean absolute percentage error (MAPE) measures how far the model’s predictions are off from their corresponding outputs on average and is calculated as: $MAPE=\frac{1}{n} \sum_{i=1}^{n} \left | \frac{y_i - \hat{y}_i}{y_i}\right| \cdot100\%$;
    
    \item Coefficient of determination ($R^2$) measures the proportion of the variance in the dependent variable that is predictable from the independent variable(s);
\end{itemize}

in which $y_i$ is the real output, $\hat{y}_i$ is the predicted output, and $n$ is the total number of samples, for $i \in [1,n]$. In addition, for \textbf{binary classification} tasks, we considered the following metrics: 

\begin{itemize}

    \item Accuracy (ACC) measures how many observations, both positive and negative, were correctly classified.

    \item \textit{Recall} measures how many observations out of all positive observations have been classified as positive.

    \item \textit{Precision} measures how many observations predicted as positive are indeed positive.
    
    \item Macro average F1-Score (MF1) is the harmonic mean between precision and recall with macro average, which calculates metrics for each label and finds their unweighted mean.
    
\end{itemize}

\subsection{Hyperperameter Optimization} \label{ch3:optimization}

The goal of hyperparameter optimization in ML is to discover the set of hyperparameters of a particular ML algorithm that returns the best performance measured on a \textit{hold-out} set. The search space defines the volume to be searched, with each dimension being a hyperparameter and each point representing a model configuration. In this manuscript, we mainly used Bayesian optimization (BO)~\cite{Shahriari2016} and random search optimization~\cite{bergstra2012random}. On the one hand, to apply a random search, one initially defines the search space as a bounded domain of hyperparameters values. Next, each step of the optimization randomly samples a point in that domain, builds the model, and then evaluates its performance. In the end, the random search optimization selects the most accurate method encountered during the iterative process. On the other hand, in contrast to random search, Bayesian methods track the entire set of prior evaluations of hyperparameters, which are used to build a probabilistic model of mapping hyperparameters to the likelihood of a score of an objective function. Rather than random sampling points in the domain, the goal of BO is to improve as iterations go by.
 
\section{Machine Learning with Differential Privacy} \label{ch3:sub_DP_ML}

In this manuscript, we consider two differentially private ML settings, which depend on where the DP guarantee is added. As revised in Section~\ref{ch2:subsub_prop_dp}, DP is immune to post-processing, which means that after the differentially private step, everything stays DP~\cite{dwork2014algorithmic}. The two considered settings are described in the following two subsections.

\subsection{Differentially Private Input Perturbation} \label{ch3:input_perturbation}

Input perturbation (or data perturbation) consists to the fact that \textbf{DP is added to each data sample} $\textbf{x}_i \in \mathcal{D}$. For example, let $\textbf{x}$ be a real-valued vector, then a differentially private version of it using the Laplace mechanism (cf. Section~\ref{ch2:sub_lap_gauss}) is: $\hat{\textbf{x}}=\textbf{x}+Lap(b)$. This is also true for categorical data, e.g., by randomizing each data point in $\textbf{x}$ with some LDP protocol (i.e., frequency oracle) from Section~\ref{ch2:sub_ldp}. On the one hand, input perturbation is the easiest method to apply~\cite{Sarwate2013_survey,Rubaie2019_survey} and it is independent of any ML and post-processing techniques. On the other hand, the perturbation of each sample in the dataset may have a negative impact on the utility of the trained model. 

In the literature, some works~\cite{kang2020input,Fukuchi2017} started to investigate how `input perturbation' through applying the Gaussian mechanism~\cite{dwork2014algorithmic} on data samples can guarantee $(\epsilon,\delta)$-DP on the final ML model. In~\cite{ldp_fed_ml}, the authors sanitized each sample with LDP protocols (GRR~\cite{kairouz2016discrete} for categorical data and the Piecewise mechanism~\cite{wang2019} for real-valued data) for training ML models to compare with federated learning. Indeed, there are an extensive literature on training ML models over differentially private data (e.g.,~\cite{Fan2018,MahawagaArachchige2020,Cyphers2017,zhou2021local,Chamikara2020,Dwork2014,Yilmaz2020,Arcolezi2020,Couchot2019}).

\subsection{Differentially Private Gradient Perturbation}

Another solution to guarantee DP to the trained model is perturbing intermediate values in iterative algorithms. In Chapter~\ref{chap:chapter91} of this manuscript, we considered training deep learning models with DP guarantees. In this case, the authors in~\cite{DL_DP} proposed a differentially private version of the stochastic gradient descent algorithm (DP-SGD). Indeed, DL models trained with DP-SGD provide provable DP guarantees for their input data. Two new parameters are added to the standard stochastic gradient descent algorithm, namely, \textit{clip} and \textit{noise multiplier}. The former is used to bound how much each training point can impact the model's parameters, and the latter is used to add controlled Gaussian noise to the clipped gradients in order to ensure DP guarantee to each data sample in the training dataset. There are many works in differentially private DL literature (e.g.,~\cite{tf_privacy,pytorch_privacy,Jayaraman2019,Shokri2015,Carlini2019,qu2021privacy,Phan2017}).

\section{Presentation of Databases Used on Experiments} \label{ch3:databases}

This section presents the databases shared by the OBS team (Section~\ref{ch3:orange_db}), the preprocessed SDIS 25 datasets resulting of the work carried out by our collaborator Selene Cerna (Section~\ref{ch3:sdis25_dbs}), and open datasets from the UCI ML~\cite{uci} repository (Section~\ref{ch3:sub_open_datasets}).

\subsection{Flux Vision Mobility Reports} \label{ch3:orange_db}

The first motivating project of this manuscript concerns multidimensional CDRs-based mobility reports released by OBS throughout time. On the one hand, the OBS team initially shared a database of \textbf{daily statistics for a single area} (Section~\ref{ch3:fimu_db}). We used this first database in Chapter~\ref{chap:chapter4} with the main goal of improving the utility of these data. In addition, the OBS team provided us with a more informative database of \textbf{$30$-minutes statistics for six areas of interest} (Section~\ref{ch3:paris_db}). We used this second database in Chapter~\ref{chap:chapter91} with the main goal of evaluating the privacy-utility trade-off of differentially private DL models on a multivariate time series forecasting task.

\subsubsection{Tourism Mobility Reports} \label{ch3:fimu_db}

One important use case of CDRs has been to analyze the mobility patterns of people in tourist events~\cite{Merrill2020,heerschap2014innovation,fluxvision1}. The first database at our disposal, from now on named \textbf{FIMU-DB}, regards multiple \textit{tourism statistics} on the frequency of visitors \textit{by days and by the union of consecutive days}. OBS considered `visitors' people present at least \(1\) hour between 06:00 and 23:59 of a given day of the reporting period in the area of interest. The geographical space is the area of an international music festival named \enquote{Festival International de Musique Universitaire} (FIMU). The FIMU is organized and financed by the City of Belfort, France, with the support of student associations. The $31^{st}$ edition of the FIMU occurred on the first five days of June 2017~\cite{FIMU}. 

\textbf{The FIMU-DB has seven different files}. Among them, \textit{five files} describe for each day, the cumulative number of unique visitors on the last \(Nb\) days, where \(Nb\) ranges from \(1\) to \(7\) days. These files are labeled from now on as FO\_country, FR\_geo, FR\_Gender, FR\_region, and FR\_age, where `FO' stands for foreigners and `FR' stands for French citizens.

In each file relating to French citizens, people are grouped according to their visitor category. \enquote{Resident} are people whose billing address is the administrative area around the FIMU. \enquote{French tourist} are people billed in France but not in the aforementioned category. The FO\_country file has only people grouped as \enquote{Foreign tourist} who are people with a foreign mobile phone operator.

In summary, each file aggregates people according to the \textbf{cumulative count} from 1 to 7 days (i.e., \textit{the number of people in the union of consecutive days}), and also by specific categories, which are briefly detailed below:

\begin{enumerate}
\item The FR\_Gender file contains 3,776 rows at total and distinguishes the people by \textbf{gender} (masculine, feminine, and Not Registered -- NR). Furthermore, during the analysis, we noticed very few differences in the frequency of men and women per day (about $50\%$ for both). Hence, in this study, the NR values were changed to masculine or feminine, with an equal probability of 50\%;

\item The FR\_age file contains 8,820 rows at total and groups the visitors by \textbf{age ranges} as: `$<$18', `18-24', `25-34', `35-44', `45-54', `55-64', `$>$65', and `NR';

\item The FR\_geo file contains 14,989 rows and groups the visitors in a specific category named \textbf{geolife}, divided into different socio-professional sub-categories as: `NR', 'comfortable family pavilion', `traditional rural', `comfortable family urban', `secondary residence', `popular', `dynamic rural', `growing peri-urban', `rural worker', `dynamic urban', `middle-class urban', and `low-income urban';

\item The FR\_region file contains 50,350 rows and groups the visitors in the specific category named (French) \textbf{region} as: `AUTRE 97', `Centre', `Languedoc-Roussillon', ``Provence-Alpes-Côte d'Azur", `Lorraine', `Ile-de-France', `Franche-Comté', `Midi-Pyrénées', `Corse', `Basse-Normandie', `Aquitaine', `Poitou-Charentes', `Pays de la Loire', `Nord-Pas-de-Calais', `Champagne-Ardenne', `Bourgogne', `Bretagne', `Alsace', `Rhône-Alpes', `Picardie', `Auvergne', and `Haute-Normandie';

\item The FO\_country file contains 10,832 rows and groups the foreign visitors by \textbf{country} as: `Belgium + Luxembourg', `Asia Oceania', `Netherlands', `Scandinavia', `United Kingdom', `Italy', `Spain', `China', `Other countries in Europe', `Germany', `United States', `Russia', `Swiss', `Eastern country', and `Rest of the world'.

\end{enumerate}

For instance, Table~\ref{tab:example_volume} exhibits $5$ random samples to illustrate how the volume data are grouped by geolife profiles in the FR\_geo file. In addition, Fig.~\ref{ch3:fig_example_fimu_db} illustrates the cumulative number of people for the three first consecutive FIMU's days using the same FR\_geo file (randomly replacing \# values for an integer within $1$ and $20$).

\setlength{\tabcolsep}{5pt}
\renewcommand{\arraystretch}{1.4}
\begin{table}[!ht]
\centering
\scriptsize
\caption{Number of unique visitors per geolife present on days of FIMU.}
    \begin{tabular}{c c c c c}
    \hline
    Date & Geolife & Visitor category & Cumulative days & Volume\\ \hline
    2017-06-01 &  comfortable family pavilion &  French Tourist &          7 days &   2751\\
    2017-06-02 &            low-income urban &           Resident &          4 days &   3355\\
    2017-06-03 &  comfortable family pavilion &           Resident &          3 days &      \# (i.e., $<$20) \\
    2017-06-04 &         secondary residence &  French Tourist &          1 days &     97\\
    2017-06-05 &rural worker &Resident &3 days &1,359 \\ \hline
    \end{tabular}
\label{tab:example_volume}
\end{table}

\begin{figure}[!ht]
\centering
\includegraphics[width=0.7\linewidth]{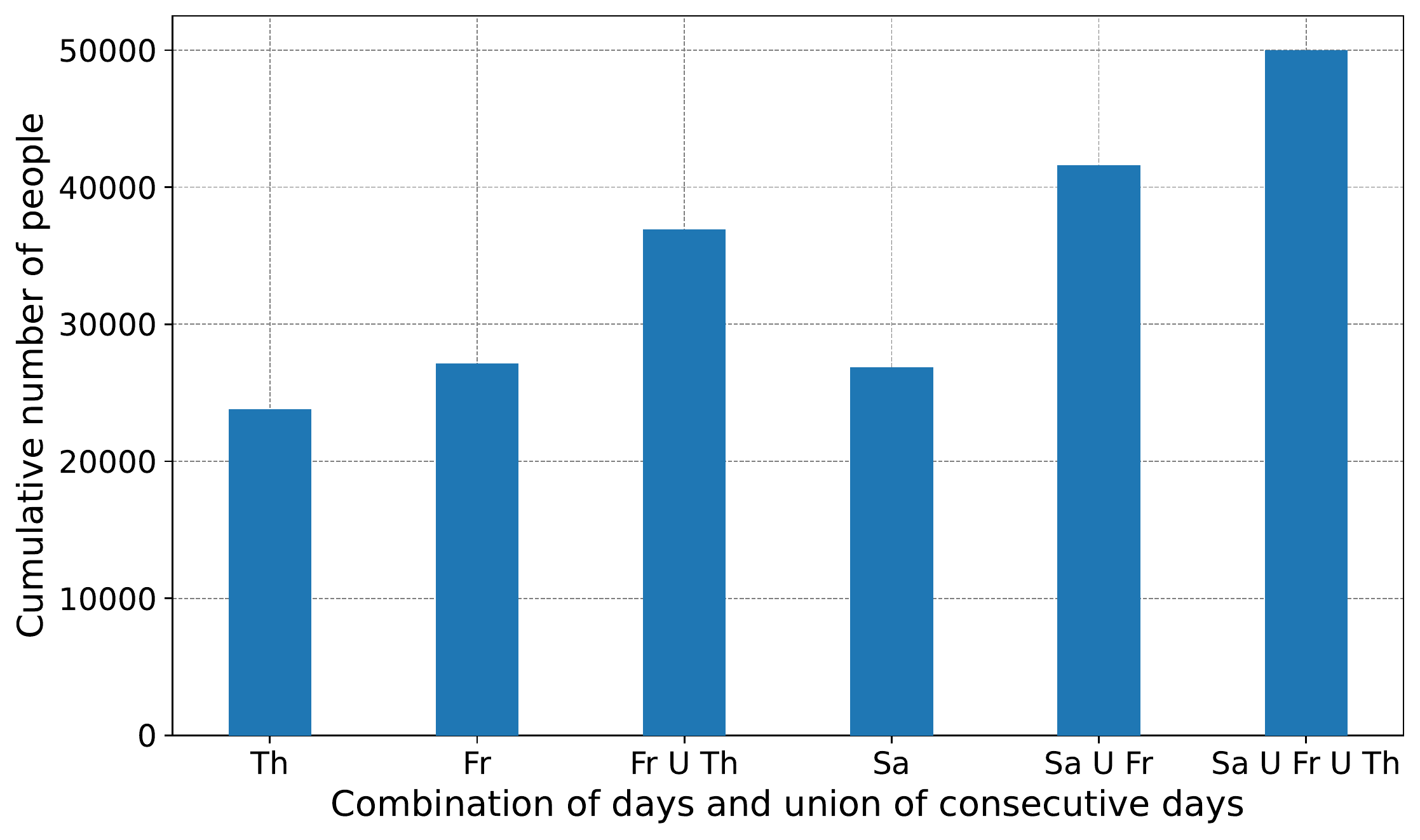}
\caption{Cumulative number of people for the three first consecutive days of FIMU, i.e., for Thursday (Tu), Friday (Fr), and Saturday (Sa). For instance, Sa U Fr means the \textbf{union} of Saturday and Friday.}
\label{ch3:fig_example_fimu_db}
\end{figure}

Furthermore, \textbf{the remaining two files} labeled from now on as Nights\_actual and Presence\_time. Unlike previous data files, these latter files \textbf{do not} consider cumulative days information, but the volume of visitors each day \((Nb=1)\). Similarly, both files classify the data by the main categories (Resident, French tourist, Foreign tourist) and by specific categories described below:

\begin{itemize}

\item The Nights\_actual file has 1,145 rows describing for each day the number of visitors who spent a night at the relevant date. Here, people are grouped by a specific category namely \textbf{sleeping area} where people spent the night. There sleeping areas are: `Agglomeration of Hericourt', `Rest Territory of Belfort', `NR', `City of Belfort', `Vosges', `Rest of Doubs', `Rest of Haute Saone', `North Haut Rhin', `Agglomeration of Belfort', `Agglomeration of Montbeliard', and `South Haut Rhin';

\item  The Presence\_time file has 1,301 rows describing for each day the number of hours where visitors were present in the area of interest. Here, people are grouped by a specific category namely \textbf{visit duration} within several sub-categories as: `Duration 2h', `Duration 3h', `Duration 4h', `Duration 5h', `Duration 6h', `Duration 7h', `Duration 8h', `Duration 9h', `Duration 10h', and `Duration 10h-18h'. For instance, `Duration 2h' matches people present between one and two hours.
\end{itemize}

We noticed that in the Nights\_actual file, the total volume of visitors per day is much less compared to the previous five files (around 4,000 on average). This means that many people did not spend the night near the city of Belfort. Therefore, considering the number of visitors per day from all other files and those in Nights\_actual, the term NR was assigned to people that did not sleep in the area of interest.

\subsubsection{Geomarketing Reports} \label{ch3:paris_db}

Another use case of CDRs is understanding people mobility during the spread of infectious diseases~\cite{Kishore2019,ebola,Oliver2020,Grantz2020,Buckee2020}. The second database at our disposal regards multiple published Flux Vision~\cite{fluxvision1} statistics for geomarketing purposes, which \textbf{were collected during the novel Coronavirus Disease 2019 (COVID-19) pandemic}~\cite{who_annouces_pandemic,Wang2020_covid} in 2020. The complete database \textit{is fully available online} in~\cite{data_paris_db}. 

We only used the file named ``\textbf{presence30min.csv}", which comprises information for two periods: from 2020-04-20 to 2020-05-03 and from 2020-08-24 to 2020-11-04. This dataset has frequency statistics by $30$ minutes (min) on the number of users by ``Zone" (i.e., $6$ regions in Paris) and by ``type" (i.e., French or foreign). The geographical space (i.e., Zone) concerns $6$ specific regions in Paris, France, named ``Commune Montreuil", ``IRIS 930480204", ``IRIS 930480205", ``IRIS 930480206", ``IRIS 930480401", and ``IRIS 930480604" in the original file. 

We applied the following preprocessing to the original ``presence30min.csv" file. We aggregated the number of users by ``type" for each of the $6$ regions, i.e., \textit{focusing only on the total number of users per the $6$ regions}. In addition, for each week, region, and $30$-min interval, we used the interquartile range technique~\cite{iqr} to detect outliers and missing data. These values were completed with the average value for that respective week, region, and $30$-min interval. We will refer to this pre-processed dataset as \textbf{Paris-DB} throughout this manuscript. 

More formally, the Paris-DB is a multivariate time series dataset $X_{(t_1,t_{\tau})}$ with aggregate number of people per $6$ regions and corresponding time period $t \in [1, \tau]$ of $30$-min intervals. That is, $X_{(t_1,t_{\tau})} = [\langle t_1, \textbf{x}_{1}\rangle, \langle t_2, \textbf{x}_{2} \rangle, ..., \langle t_{\tau}, \textbf{x}_{\tau} \rangle ]$, where $\textbf{x}_{t}$ is a vector of size $6$ in which each position represents the number of users per region at time $t \in[1,\tau]$. 

On analyzing the Paris-DB, Fig.~\ref{fig:analysis_lockdown} illustrates the \textbf{total number of people} for two 14-days periods: from the beginning of 2020-04-21 to the end of 2020-05-03 and from the beginning of 2020-09-23 to the end of 2020-10-06. The plot on the left-side corresponds to mobility analytics during the first national lockdown period in France~\cite{lockdown} because of the COVID-19 pandemic. The plot on the right-side corresponds to a period with no lockdown measures. As one can notice, there is a clear difference between the first period of analysis (low mobility activity) and the second one (high mobility activity). This type of mobility analysis provides important insights on mobility patterns for public authorities and policymakers to fight the COVID-19 pandemic, for example~\cite{Vespe2021,deAlarcon2021}. 

\begin{figure}[!ht]
    \centering
    \includegraphics[width=1\linewidth]{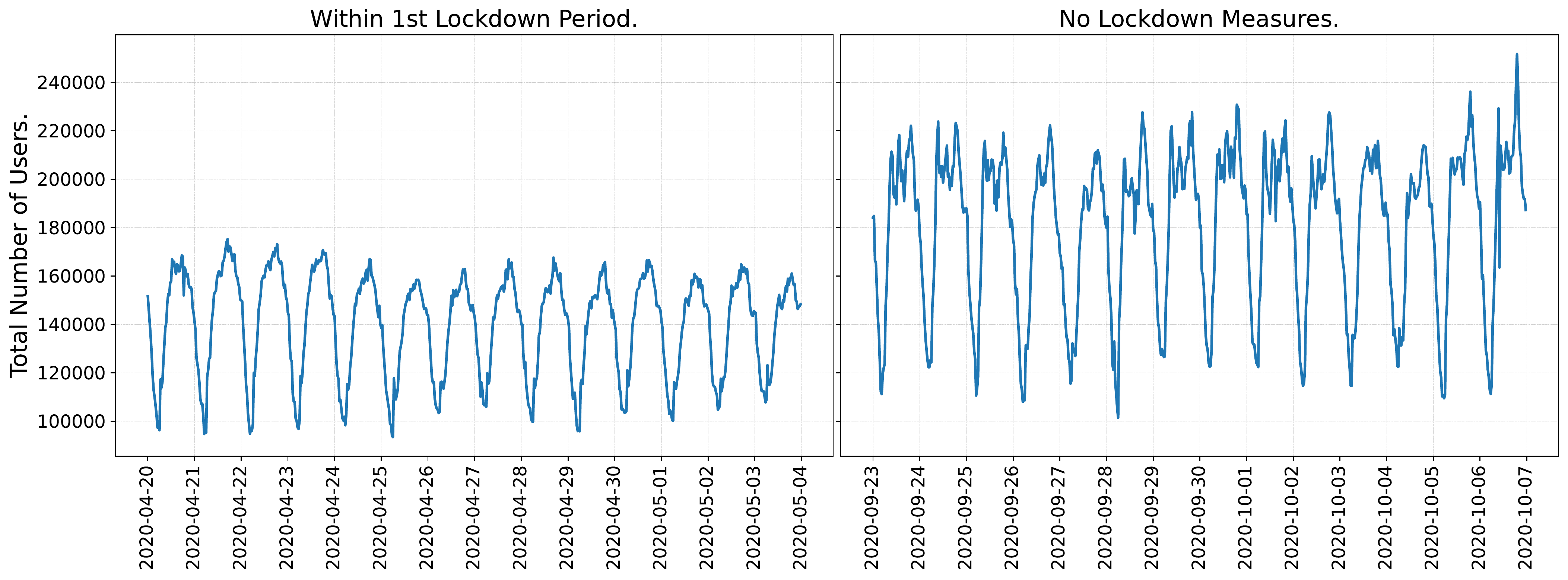}
    \caption{Aggregated human mobility analytics during two weeks within the first lockdown period in France (left-side plot) and during two weeks with no lockdown measures (right-side plot).}
    \label{fig:analysis_lockdown}
\end{figure}

\subsection{Firemen Database} \label{ch3:sdis25_dbs}

As mentioned in Chapter~\ref{chap:chapter1}, the AND team has been investigating ML-based solutions to optimize the SDIS 25 (i.e., and EMS in France) services. Our connection with the SDIS 25 is through the Ph.D. student Selene Cerna, with a CIFRE thesis (N 2019/0372) and a strict confidentiality agreement to use SDIS 25 original data. In this section, we present three datasets processed by her. These datasets have also been used by Selene Cerna to develop the ML-based solutions with original data that we use for comparison purposes, i.e., to evaluate the \textbf{privacy-utility trade-off} of our solutions.

\subsection{Interventions Data} \label{ch3:interv_dbs}

Predicting the operational demand of EMS is a way to allow their reallocation of human and material resources (e.g., cf.~\cite{Chen2016,Grekousis2019,Lstm_Cerna2019,Cerna2020_boosting,Cerna2020_b,EliasMallouhy2021}). So, the first dataset we use, from now on named \textbf{Interv-DB}, has information about $382046$ \textbf{interventions} attended by the SDIS 25 from 2006 to 2018 in $608$ cities inside the Doubs region. The \textbf{Interv-DB} has two attributes: \textit{SDate}, which is the precise ``\textit{Starting Date}" of the intervention, and the \textit{City} where the intervention took place. The way this dataset has been preprocessed will be explained in Chapter~\ref{chap:chapter8}, as part of a proposed privacy-preserving methodology that allows both statistical learning and forecasting tasks of firemen demand by region.

\subsection{Response Time Data} \label{ch3:art_dbs}

Although predicting the operational demand may help on the redeployment of resources, another solution would be to predict the response time of each ambulance, which would allow, e.g., to move from a static resource deployment plan into a dynamic one. The second dataset we use contains information about $186130$ \textbf{dispatched ambulances} from SDIS 25 centers that attended $182700$ EMS interventions from 2006 up to June 2020. After a preprocessing step carried out by Selene Cerna, the final dataset, from now on named \textbf{ART-DB}, has the following attributes:

\begin{itemize}
    \item \textbf{Temporal features.} Based on the time the SDIS 25 has been notified, a few temporal features were included, such as the: year, month, day, weekday, hour, and categorical indicators to denote holidays, end/start of the month, and end/start of the year;
    
    \item \textbf{Operation demand features.} The number of interventions attended by the SDIS 25 in the \textit{past hour} and the number of \textit{active} interventions in the current hour;
    
    \item \textbf{Traffic data.} These are prediction of traffic level for the Doubs region as indicators ranging from 1 (regular flow) to 4 (extremely difficult flow) per day from~\cite{bisonFute};
    
    \item \textbf{Weather data.} These are historical weather information from~\cite{meteoFrance} such as precipitation, temperature, wind speed, gust speed, and so on, which were added according to the hour of each intervention;
    
    \item \textbf{Location-based features.} The latitude and longitude coordinates of the intervention and of the SDIS 25 center that took charge of the intervention; the district, the city, and the zone of the intervention;
    
    \item \textbf{Computed features.} The great-circle distance~\cite{great_circle_dist} between the SDIS 25 center and the emergency scene; the estimated travel time, and the estimated driving distance. These two latter features were obtained with the open source routing machine (OSRM)~\cite{OSRM} API;
    
    \item \textbf{The scalar target variable} is the ambulance response time (ART) in minutes, which is the time measured from the SDIS 25 notification to the ambulance's arrival on the emergency scene. 
\end{itemize}

\subsection{Calls, Victims, and Operators Data} \label{ch3:calls_vic_ope_dbs}

From another point of view, identifying high urgent situations (i.e., a life-or-death situation) would allow EMS to quickly respond to victims needing priority attention (i.e., if they might die). Therefore, our third dataset, from now on named \textbf{Vic\_Mort-DB}, has information about $177883$ \textbf{victims} that the SDIS 25 attended from January 2015 to December 2020. After a preprocessing step carried out by Selene Cerna, the \textbf{Vic\_Mort-DB} has the following attributes:

\begin{itemize}
    \item \textbf{Victim data.} The victim's age, gender, and city (where the intervention occurred); 
   
   \item \textbf{(Call center) operator data.} The operator's age, gender, grade (e.g., commander, captain), and seniority (i.e., experience time in days);
    
    \item \textbf{Temporal features.} Based on the time the SDIS 25 has been notified, a few temporal features were included, such as the: hour, day, day of the week, month, and year;
    
    \item \textbf{Call/Intervention data.} The delay time to answer the phone, total call duration, delay time to diffuse the alert (i.e., to notify an SDIS 25 center), the SDIS 25 center that assisted the victim, and the type of intervention. The latter is described by 3 variables: type of operation (e.g., aid to person, fire), the subtype of operation (e.g., an emergency, fire on the public road, fire in an individual room), and the motive for departure (e.g., cardiac arrest, respiratory distress);
    
    \item \textbf{Calculated features.} Probability of mortality by \textit{motive} and by \textit{age}, which are calculated according to the learning set only; the age of the victims grouped into 8 categories; the total sum of delay time to answer the phone, call duration, and delay time to diffuse the alert; and the great-circle distance between the center and the city;
    
    \item \textbf{The target variable} is the victim's mortality, which is a binary attribute (0: alive, 1: dead). 
\end{itemize}

\subsection{Open Datasets} \label{ch3:sub_open_datasets}

For ease of reproducibility of the works carried out in Chapters~\ref{chap:chapter5} and~\ref{chap:chapter6}, we (also)\footnote{Besides the multidimensional open dataset generate in Chapter~\ref{chap:chapter4}.} considered \textbf{three multidimensional open datasets from the UCI ML repository}~\cite{uci}. These datasets were selected because they allow evaluating our solutions more practically, i.e., with typically real-world datasets. For instance, they differ on the number of users $n$, on the number of attributes $d$, on the number of values per attribute $\textbf{c}=[c_1,c_2,...,c_d]$, and on the data distribution of each attribute. These datasets are described in the following.

\begin{itemize}
    
    \item \textit{Nursery\footnote{\url{https://archive.ics.uci.edu/ml/datasets/nursery}}.} This dataset contains $n=12960$ samples and $d=8$ categorical attributes. The domain size of each attribute is $\textbf{c}=[3, 5, 4, 4, 3, 2, 3, 3, 5]$, respectively. 
    
    \item \textit{Adult\footnote{\url{https://archive.ics.uci.edu/ml/datasets/adult}}.} This dataset contains $48842$ samples extracted from the 1994 Census database. There are $14$ attributes (including the income attribute), in which $9$ are categorical and $5$ are numerical (i.e., considering `education' instead of `education-num'). After removing all samples with missing values (i.e., symbol `?'), there are $n=45222$ samples in this dataset. We only selected the categorical attributes (i.e, $d=9$). The domain size of each attribute is $\textbf{c}=[7, 16, 7, 14, 6, 5, 2, 41, 2]$, respectively. 
    
    \item \textit{Census-Income\footnote{\url{http://archive.ics.uci.edu/ml/datasets/Census-Income+\%28KDD\%29}}.} This dataset contains weighted census data from the 1994 and 1995 years. There are $40$ demographic and employment related variables (including the total person income attribute), in which $33$ are categorical and $7$ are numerical. In total, there are $n=299285$ samples in this dataset. We only selected the categorical attributes (i.e, $d=33$). The domain size of each attribute is \begin{math}\textbf{c}=[ 9, 52, 47, 17,  3,  7, 24, 15,  5, 10,  2,  3,  6,  8,  6,  6, 51, 38, 8, 10,  9, 10,  3,  4,  5, 43, 43, 43,  5,  3,  3,  3,  \\2]\end{math}, respectively.
    
\end{itemize}

\section{Conclusion}

In this section, we have revised state-of-the-art ML techniques and some concepts. We started revising supervised learning and classification and regression tasks, all three considered in this manuscript. Next, we have reviewed state-of-the-art ML algorithms ranging from linear (i.e., LASSO), decision-tree learning (i.e., LGBM and XGBoost), and deep learning (e.g., MLP, RNNs) models. We also presented the metrics that will be used to assess the models' performance, as well as two hyperparameter tuning methods (i.e., random search and Bayesian optimization). Lastly, we also reviewed how to build differentially private ML models, which fundamentally depends on where the DP guarantee is added. That is, by the post-processing property of DP~\cite{dwork2014algorithmic} (cf. Section~\ref{ch2:subsub_prop_dp}), everything after DP, stays DP. Indeed, we mainly consider in this manuscript the rigorous input perturbation setting, which sanitizes each data sample independently (i.e., row-by-row). Although utility may drop, we believe the privacy-utility trade-off is worthwhile since the sanitized dataset will be protected if data leakes~\cite{data_breaches}, and the ML model will also be differentially private, protecting the data against, e.g., data reconstruction attacks, membership inference attacks~\cite{Song2017,Shokri2017}. On the other hand, we also briefly presented another setting, namely, gradient perturbation contextualized to deep learning methods trained with the DP-SGD~\cite{DL_DP} algorithm. Lastly, we also presented the datasets we will be using in this manuscript to perform our experiments.

%% file: chapters/chapter4.tex
\chapter{MS-FIMU: A Multidimensional Dataset to Evaluate LDP protocols} \label{chap:chapter4}

In Chapters~\ref{chap:chapter1},~\ref{chap:chapter2}, and~\ref{chap:chapter3}, we have presented all main components that will be used in the rest of the contribution chapters of this manuscript. In this chapter, we start to study statistics on aggregated human mobility data generated by OBS (i.e., with the Flux Vision system~\cite{fluxvision1}). The main objective here is to instantiate a mobility scenario from these statistics and to generate a synthetic dataset that allows simulating the data collection pipeline with the privacy-preserving techniques we develop in the next three Chapters~\ref{chap:chapter7},~\ref{chap:chapter5}, and~\ref{chap:chapter6}. We emphasize that although the \textbf{exact} Flux Vision's anonymization method is \textbf{unknown} to the author, we refer to state-of-the-art methods that could give similar results.

\section{Introduction} \label{ch4:introduction}

The main objective of this chapter is to propose an approach to instantiate a mobility scenario that matches the \textbf{anonymized} dataset of mobility described in Section~\ref{ch3:fimu_db} named FIMU-DB, which was published by OBS. To generate the FIMU-DB, as stated by OBS, algorithms for data acquisition are compliant with European laws to guarantee the \textbf{anonymity} of each person. Indeed, following the GDPR~\cite{GDPR} and CNIL~\cite{CNIL}, MNOs must anonymize ``on-the-fly" CDRs used for purposes other than billing. More precisely, if CDRs are used for mobility analytics, these data must be processed within a required time interval (e.g., 15 minutes) if and only if there is a sufficient number of users present for reaching a specific level of anonymity (i.e., ``hide in the crowd"). Besides, all implicit metadata (e.g., users' IDs, timestamps) should be excluded before transmitting any data for processing. That is to say, data should be aggregated (respecting anonymity) within the required time interval and all kinds of identifiers must be excluded before any further analysis.

More specifically, to generate FIMU-DB, OBS established pre-defined indicators through generalization (e.g., age ranges, socio-professional profiles, ...). Next, an anonymity threshold of $k=20$ was defined, i.e., if there are less than $20$ users the number is masked with the symbol \#. Besides, given the number of identified Orange customers, an extrapolation algorithm was developed to estimate the real population. This latter algorithm can be seen as a perturbation-based technique to add noise to the true value. Thus, within the required time interval, OBS processed CDRs respecting $k=20$ for any interval to produce the mobility indicators per day and per the \textbf{union} of consecutive days. We notice that to generate cumulative statistics, i.e., the number of unique users by the union of days could have been done, e.g., using Bloom filters~\cite{bloom_f}. Also, we use $k$ to indicate the anonymity threshold as it follows the ``hide in the crowd" protection provided by \textit{k}-anonymity~\cite{samarati1998protecting,SWEENEY2002}. However, in our view, we believe that the OBS procedure approximates some of the privacy-preserving approaches described in~\cite{Bittau2017}.

In summary, the FIMU-DB is subject to noise resulting from the extrapolation of detected Orange clients and from the anonymization procedure to respect the GDPR and the CNIL. Besides, the FIMU-DB has information on the number of people present by the \textbf{union} of consecutive days (also referred to as `cumulative' information throughout this chapter). 

Therefore, on the one hand, with our proposed approach, we aim to improve the utility of this data, providing the number of people present by all the \textbf{intersections} of days. The mobility scenario we propose represents an invaluable source of information to the city public administration and private companies. Rather than being limited to the number of unique people present in certain regions per union of consecutive days, the scenario allows knowing if they are the same visitors or different visitors over the analyzed time period. With such specific information, companies and public administration would be able to manage their employees and equipment resources efficiently to improve accommodation and transportation systems according to peoples' mobility, thus providing better attendees comfort and security.  

Besides, we propose to recreate the instantiated mobility scenario with virtual humans, such that the synthetic dataset matches the original statistical data. Therefore, as an open dataset, one can carry out studies such as testing and improving data sanitization techniques. \textbf{For the rest of this manuscript, we will refer to the final synthetic dataset as Mobility Scenario FIMU (MS-FIMU)}, which contains $7$ categorical attributes for $88,935$ unique users along $7$ days (on average $\sim 26,000$ unique users per day). That is to say, a \textbf{longitudinal} and \textbf{multidimensional} dataset. \textbf{We invite the interested reader to visit the Github page (\url{https://github.com/hharcolezi/OpenMSFIMU}) to access the final results of this chapter and the published synthetic open dataset.}

The rest of this chapter is organized as follows.Section~\ref{ch4:sec_study_case_data_analysis} presents the study case and some challenges we faced working with real-world anonymized data. Section~\ref{ch4:sec_materials_methods} introduces the proposed approach to instantiate a precise mobility scenario and to generate synthetic data. Section~\ref{ch4:sec_resul_disc} presents the results and its discussion. Finally, Section~\ref{ch4:sec_conclusion} provides concluding remarks. The methodology presented in Section~\ref{ch4:sec_materials_methods} and the results in Section~\ref{ch4:sec_resul_disc} were published in a full paper~\cite{ms_fimu} at the 16th International Wireless Communications \& Mobile Computing Conference (IWCMC 2020). 

\section{Study Case and Data Analysis} 
\label{ch4:sec_study_case_data_analysis}

The main background for this chapter is the database named FIMU-DB from Section~\ref{ch3:fimu_db}, which has seven main files: FR\_gender, FR\_age, FR\_geo, FR\_region, FO\_country, Nights\_actual, and Presence\_time. In this section, we present the scenario in which OBS collected the data and we highlight some challenges one can face working with real-life anonymized data. 

\subsection{Study Case}
\label{subsec:study_case}

As reviewed in Section~\ref{ch3:fimu_db}, the FIMU-DB was published by OBS, which collected statistics on the frequency of users on days and union of consecutive days through analyzing mobile phone data (i.e., CDRs). The geographical space is the area of an international music festival a.k.a \enquote{Festival International de Musique Universitaire} (FIMU). 

Modeling people's mobility in such events is of great importance for public administration and private companies. Hence, we propose to model a more precise mobility scenario, including \textbf{one day before the FIMU event, the five days of the FIMU, and one day after the FIMU end. In other words, this $Nb=7$ days scenario for a 5-days event provides information for these institutions to know the number of people who got in and out of the zone of analysis before, during, and after the event.}

\subsection{Challenges with Anonymized Statistical Data}

Although the data produced by OBS are adequate for marketing purposes, conducting scientific studies using these data leads to, in our view, two challenges. First, we are unable to determine the real number of people when OBS published \#, i.e., due to the anonymity threshold $k=20$. On the one hand, one could think of excluding all \# values, which would probably compromise the utility of the data. So, in this chapter, instead of excluding this information, we considered the option of randomly replacing \# by an integer within the known range from 1 to 20. However, both solutions (excluding or randomly assigning an integer) result in different cardinalities between the seven files that describe the same population, which represents an \textbf{inconsistency}. 

For instance, Table~\ref{tab:example_database} summarizes the records of the first three days of the FIMU. In this scenario, the first day of analysis is Thursday and it has only one record labeled as `Th1'. Friday has two records labeled as `Fr1' (only Friday) and `Fr2' (Friday \textbf{OR} Thursday), respectively. And Saturday has three records labeled as `Sa1' (only Saturday), `Sa2' (Saturday \textbf{OR} Friday), and `Sa3' (Saturday \textbf{OR} Friday \textbf{OR} Thursday), respectively. \textbf{For the rest of this chapter, we will be using this notation (e.g., `Fr1', `Fr2', ...) to indicate the `cumulative' information (i.e., the union of consecutive days).}

In Table~\ref{tab:example_database}, both `FR\_geo', `FR\_region', and `FR\_age' columns present the total number of \textit{unique} French visitors aggregated in each file. This is according to the `Cum. days' attribute exemplified in Table~\ref{tab:example_volume} and after replacing randomly all \# values. The same procedure is reproduced for the other files. \textbf{Theoretically, the information from all three columns `FR\_geo', `FR\_region', and `FR\_age' should be equal as they describe the same population. However, this is not true, and the difference between files changes depending on the replacement of all \# values (unknown).}

\setlength{\tabcolsep}{5pt}
\renewcommand{\arraystretch}{1.4}
\begin{table}[!ht]
\scriptsize
\centering
\caption{Unique French visitors present over three FIMU's days.}
\begin{tabular}{cccccc}
\hline
\textbf{Label} &\textbf{Cum. days} & \textbf{FR\_geo} &\textbf{\ldots} &\textbf{FR\_region} &\textbf{FR\_age} \\ \hline
Th1  & 01 day  & 23,816 & \ldots &23,598 & 23,810  \\ \hline
Fr1 & 01 day  &27,145   & \ldots &26,945 & 27,143  \\ \hline
Fr2 & 02 days  &36,917  & \ldots &36,758 & 36,915\\ \hline
Sa1 & 01 day  & 26,894  & \ldots  &26,699 & 26,868\\ \hline
Sa2 & 02 days  & 41,615  &\ldots &41,373 & 41,589\\ \hline
Sa3 & 03 days & 50,024  & \ldots &49,823 & 49,999\\ \hline
\end{tabular}

\label{tab:example_database}
\end{table}

\section{Proposed approach} 
\label{ch4:sec_materials_methods} 

Our goal is to improve the understanding of people's mobility behavior from the number of unique visitors per day and cumulative days to find out the number of unique visitors by the intersection of days. The whole proposed approach is summarized with a flowchart depicted by Fig.~\ref{fig:flow_chart}. In this particular study, the ultimate goal is to infer the number of people who stayed in the city for one or any combination of days (i.e., the aggregate number of users in each intersection of days), considering one week, including the FIMU event. Further, once the whole mobility scenario is known, the objective is to generate samples to build a synthetic dataset with virtual people. The approach is detailed and applied in the following two subsections.

\subsection{Mobility scenario modeling}
\label{subsec:ms_modeling}

As described on the left side of Fig.~\ref{fig:flow_chart}, first, we input data with cumulative information and replace the \# values. A Boolean map is used to describe every combination of $Nb=7$ consecutive days resulting in $2^{Nb}=128$ variables. Then, each of the \(Nb(Nb+1)/2=28\) cumulative days is described as a Boolean vector with 0 (excluded) and 1 (included) values per combination of days according to the representative map. 

\begin{figure}[!ht]
\centering
\includegraphics[width=0.85\linewidth]{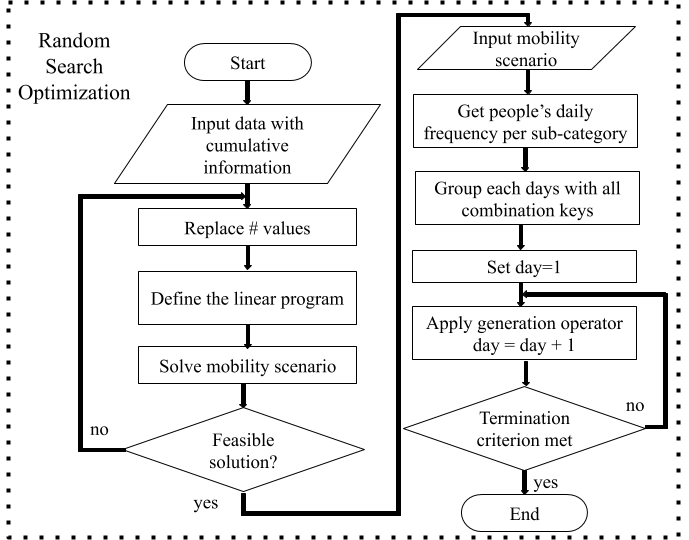}
\caption{Flowchart of the proposed algorithm to, first, instantiate a mobility scenario with information by the intersection of days and, second, to generate synthetic data.}
\label{fig:flow_chart}
\end{figure}

Then, a linear program (LP) is defined to instantiate the first feasible solution given a random initial solution, without trying to maximize or minimize any combination of days. The system constraints are the number of people per cumulative day, expressed as Boolean vectors. However, due to both problems of \# values and \textbf{inconsistencies between the cardinalities of the datasets}, rather than using the exact `known values', these problems are addressed by establishing bounds. In this case, the lower and upper bounds are the minimum and maximum values between all the datasets. The motivation for such an approach is to instantiate a feasible solution that respects the values of all available data such that the global error could minimize. More precisely, we are defining a linear constraint solver that computes an arbitrary solution within the set of feasible solutions rather than using the linear program as an optimization mechanism. In this case, the objective function of this system is just a constant (zero). Eq.(\ref{eq:min}) mathematically describes the LP as: 

\begin{equation}\label{eq:min}
\begin{array}{rl}
\displaystyle \min & 0 \textrm{,}\\
\textrm{s.t.} & \textrm{ $lb_i \leq A_{ij} x_j \leq ub_i$}\\
    & \textrm{$x_j  \geq 0$}%
\end{array}
\end{equation}

$\forall i \in [1,Nb(Nb+1)/2]$ and $\forall j \in [1,2^{Nb}]$ where $A_{ij}$ is the Boolean matrix representing the Boolean vectors $i$ and its respective days combinations $j$; $x_j$ is the number of people per combination of days; and $lb$ and $ub$ are both lower and upper bounds, respectively, which corresponds to the total number of unique users. 

Hence, instantiating a feasible solution for all the categories (Resident, French tourists, and foreign tourists) and grouping them as a unique mobility scenario provides the number of people for each combination of days. Then, with such results, the second part is retaken for generating samples of virtual humans aiming to approximate the original data. 

Before moving on to step 2 (generate synthetic data), let us consider the scenario of Table~\ref{tab:example_database} to better understand the proposed LP. Fig.~\ref{fig:example_karnaugh} illustrates the Boolean map representation of \(Nb=3\) consecutive days (Th=Thursday, Fr=Friday, Sa=Saturday, and its complements), and the example of both $Th1$ (unique visitors on Thursday) and $Sa2$ variables (unique visitors present on Saturday \textbf{OR} Friday). Notice that $x1$ represents the number of visitors that are present neither Thursday nor Friday nor Saturday. This number is obviously not known and, hence, it is not considered.

\begin{figure}[H]
\centering
\includegraphics[width=0.45\linewidth]{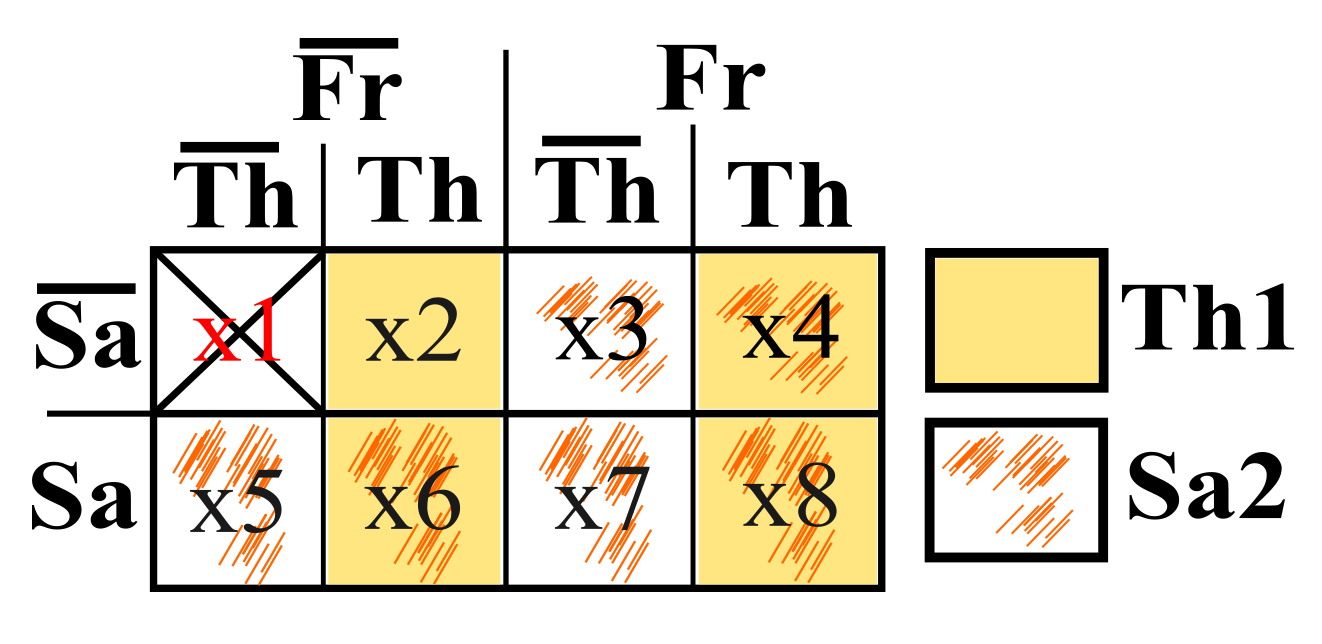}
\caption{Representation of Nb=3 days combination and illustration of both Th1 and Sa2 known values.}
\label{fig:example_karnaugh}
\end{figure}

Considering the LP in Eq.(\ref{eq:min}), Eq.(\ref{eq:ex_map3}) exhibits the $A_{ij}$ matrix according to Fig.~\ref{fig:example_karnaugh} and its lower ($lb$) and upper bound ($ub$) with values from Table~\ref{tab:example_database} relating to French citizens. 

\begin{equation}
\label{eq:ex_map3}
\left[
\begin{matrix}
Th1 \\Fr1 \\Fr2 \\Sa1 \\Sa2\\Sa3
\end{matrix}
\right]
\Longrightarrow
\left[
\begin{matrix}
23,598\\26,945\\36,758\\26,699\\41,373\\49,823\\
\end{matrix}
\right]
\leq
\left[
\begin{matrix}
0& 1& 0& 1& 0& 1& 0& 1& \\
0& 0& 1& 1& 0& 0& 1& 1& \\
0& 1& 1& 1& 0& 1& 1& 1& \\
0& 0& 0& 0& 1& 1& 1& 1& \\
0& 0& 1& 1& 1& 1& 1& 1& \\
0& 1& 1& 1& 1& 1& 1& 1&  
\end{matrix}
\right]
\left[
\begin{matrix}
x_1\\x_2\\ \vdots\\ x_{7}\\x_{8}
\end{matrix}\right]
\leq
\left[
\begin{matrix}
23,816\\27,145\\36,917\\26,894\\41,615\\50,024\\
\end{matrix}
\right]
\end{equation}

\subsection{Synthetic data generation}
\label{subsec:data_gen}

The proposed algorithm illustrated on the right side of Fig.~\ref{fig:flow_chart} is summarized in the subsequent steps. First, using the original data, the frequency of visitors present each day of the week under study is calculated for each sub-category, e.g., on the first day $50.2\%$ are men and $49.8\%$ are women. 

Next, we set up a dictionary for each day grouping its related keys of combination days; people present only Thursday are described by TT, people present both Thursday and Friday are described by TF, and so on. It is noteworthy that the same TF key appears on both Thursday and Friday dictionaries as they are the same people that attended both days in the analysis area.

Then, an iteration starts filling up each key for the first day with virtual individuals respecting the frequency of men and women, geolife categories, age groups, regions (countries for foreign tourists), the sleeping area, and the visit duration. Afterward, for the next six days, people with similar keys are directly copied from one day to another. In this case, the frequency for each category is re-calculated considering the existing people. The remaining people are then generated according to the new frequency. However, there is one exception about the attribute `visit duration', which means that people could be present more hours from one day to another (a dynamic attribute). Hence, the approach may vary the duration time of every people each day relative to the real frequency acquired from the original data.

Once the stop criterion is met, i.e., when all days have their respective virtual humans, the error is calculated by querying the generated data and comparing it to the original one. The error, total error, and error rate metrics are defined in the following.

\begin{definition}[Error]
Let \(|A|\) be the cardinality of set A and \(A_{|j}\) be the subset of \(A\) restricted to sub-category \(j\), \textit{i.e.}, $A_{|j} = \{x | x \in A, x \in j\}$. Given a set O (original data), a set G (generated data), and sub-categories \(j\) related to each specific category (i.e., from the gender category there are two sub-categories, feminine and masculine), the error is defined as
$$
    error(j) = | | G_{\vert j} | - | O_{\vert j} || \textrm{.}
$$
\end{definition}

\begin{definition}[Total Error] 
The total error TE is the sum of errors per sub-category \(j\) and per day \(i\) defined as
$$
TE = \sum_{i=1}^{n}\sum_{j=category} error(j)_i \textrm{.}
$$
\end{definition}

\begin{definition}[Error Rate]
The error rate ER is calculated considering \(j\) original datasets
$$
    ER = \frac{TE} {\sum_{j=dataset}\sum_{j=category} | O_{j \vert j} |}  \textrm{.}
$$
\end{definition}

These computations are repeated for $m$ iterations based on a \textbf{random search optimization} approach. In particular, the first parameter randomly generated is the \# values within the range 1-20, which changes the number of people per day and, consequently, per file at every iteration. In addition, considering the LP in Eq.(\ref{eq:min}), an initial solution is randomly generated such that the linear constraint solver can provide a different mobility scenario at each iteration. Then, the error rate metric is calculated. Finally, the mobility scenario and synthetic dataset with the smallest error rate are recovered as a final solution.

Some motivations for such a random search approach are described as follows. First, an initial attempt to model our problem as a linear program resulted in an infinity number of solutions. And second, as aforementioned, the \# problem due to privacy constraints had to be handled, resulting in different cardinalities for the datasets. 

\section{Results and Discussion}
\label{ch4:sec_resul_disc}

To carry out this work, we used the Pyeda Python package~\cite{pyeda} for Boolean algebra operations. We applied the mixed-integer nonlinear programming solver from the Gekko package~\cite{gekko} to the proposed mobility scenario in Eq.(\ref{eq:min}). The Faker package~\cite{faker} assigned fake French names for French citizens and default names (United States) for foreign tourists. All algorithms were implemented in Python 3. In order to run our codes, we used a machine with Intel (R) Core (TM) i7-8650 CPU @ 1.90GHz and 32GB RAM using Debian 10. In the next two subsections, we present our results.

\subsection{Mobility scenario}
\label{subsec:resul_ms}

The random search algorithm performs $5,000$ evaluations of $m=100$ iterations in parallel to search for the most representative distribution of people over the week of interest. This is a suitable way to ensure a convergence pattern towards a global minimum due to probabilistic properties. At the end of 22 minutes, the random search stops, and the dataset is recovered with an error rate less than $8.1\%$ at evaluation $1,050$ and iteration $79$. 

We summarize the final mobility scenario in Table~\ref{ch4:tab_final_ms}, which presented the smallest error rate. In Table~\ref{ch4:tab_final_ms}, each day of the week is represented in an abbreviated format, e.g., Sunday -- $Su$ and its complement by $\overline{Su}$. Besides, Fig.~\ref{fig:error} depicts the decreasing error rate function based on the number of iterations. Lastly, for illustration purposes, Table~\ref{tab:ex_error} presents the number of visitors for both real and synthetic datasets (FR\_age) and the absolute error for three sub-categories of ages on the first day of interest.

\setlength{\tabcolsep}{5pt}
\renewcommand{\arraystretch}{1.4}
\begin{table}[!ht]
    \scriptsize
    \centering
    \begin{tabular}{c c c c| c| c| c| c| c| c| c| c }
    \hline
    \multicolumn{4}{c|}{\multirow{3}{*}{Days combination}} & \multicolumn{4}{c}{$\overline{Fr}$} & \multicolumn{4}{c}{$Fr$}  \\ \cline{5-12}
    & & & & \multicolumn{2}{c|}{$\overline{Th}$} & \multicolumn{2}{c|}{$Th$} & \multicolumn{2}{c}{$\overline{Th}$} & \multicolumn{2}{|c}{$Th$} \\ \cline{5-12}
    & & & & \multicolumn{1}{c|}{$\overline{We}$} & \multicolumn{1}{c|}{$We$} & \multicolumn{1}{c|}{$\overline{We}$} & \multicolumn{1}{c|}{$We$} & \multicolumn{1}{c|}{$\overline{We}$} & \multicolumn{1}{c|}{$We$} & \multicolumn{1}{c|}{$\overline{We}$} & \multicolumn{1}{c}{$We$} \\ \hline
    
    \multirow{8}{*}{$\overline{Tu}$} & \multirow{4}{*}{$\overline{Mo}$} & \multirow{2}{*}{$\overline{Su}$} & $\overline{Sa}$ & - & 4851 & 4378 &1527 &1801 &1701 & 786 & 3450 \\ \cline{4-12}
    & &  & $Sa$ & 4791 & 234 & 87 & 266 & 1748 & 48 & 417 & 893 \\ \cline{3-12}
    & &\multirow{2}{*}{$Su$} &$\overline{Sa}$ & 9695 & 228 & 199 & 508 & 341 & 92 & 506 & 1220  \\ \cline{4-12}
    & & & $Sa$ & 2171 & 287 & 74 & 73 & 4237 & 103 & 1109 & 1229 \\ \cline{2-12}
    & &\multirow{2}{*}{$\overline{Su}$} &$\overline{Sa}$ & 5937 & 183 & 49 & 207 & 97 & 36 & 67 & 233 \\  \cline{4-12}
    & & & $Sa$ & 592 & 100 & 103 & 42 & 63 & 116 & 63 & 80 \\ \cline{3-12}
    & &\multirow{2}{*}{$Su$} &$\overline{Sa}$ & 7380 & 71 & 34 & 56 & 71 & 89 & 77 & 22 \\ \cline{4-12}
    & & & $Sa$ & 256 & 51 & 96 & 49 & 27 & 52 & 94 & 61 \\ \hline
    
    \multirow{8}{*}{$Tu$} & \multirow{4}{*}{$\overline{Mo}$} & \multirow{2}{*}{$\overline{Su}$} & $\overline{Sa}$ & 7052 & 446 & 213 & 787 & 1163 & 35 & 679 & 775 \\ \cline{4-12}
    & &  & $Sa$ & 441 & 59 & 104 & 71 & 62 & 94 & 106 & 99 \\ \cline{3-12}
    & &\multirow{2}{*}{$Su$} &$\overline{Sa}$ & 1004 & 110 & 53 & 70 & 85 & 87 & 52 & 53  \\ \cline{4-12}
    & & & $Sa$ & 42 & 94 & 50 & 91 & 93 & 38 & 51 & 36\\ \cline{2-12}
    & &\multirow{2}{*}{$\overline{Su}$} &$\overline{Sa}$ & 159 & 309 & 72 & 325 & 442 & 67 & 396 & 94 \\  \cline{4-12}
    & & & $Sa$ & 111 & 76 & 89 & 35 & 71 & 34 & 102 & 434 \\ \cline{3-12}
    & &\multirow{2}{*}{$Su$} &$\overline{Sa}$ & 434 & 84 & 71 & 41 & 112 & 67 & 89 & 149 \\ \cline{4-12}
    & & & $Sa$ & 4176 & 61 & 71 & 93 & 211 & 74 & 506 & \textbf{176} \\ \hline
    
    \end{tabular}
    \caption{Final result of using our methodology, which produces a mobility scenario with the frequency of users per day and per the intersection of days.}
    \label{ch4:tab_final_ms}
\end{table}

\begin{figure}[H]
    \centering
    \includegraphics[width=0.75\linewidth]{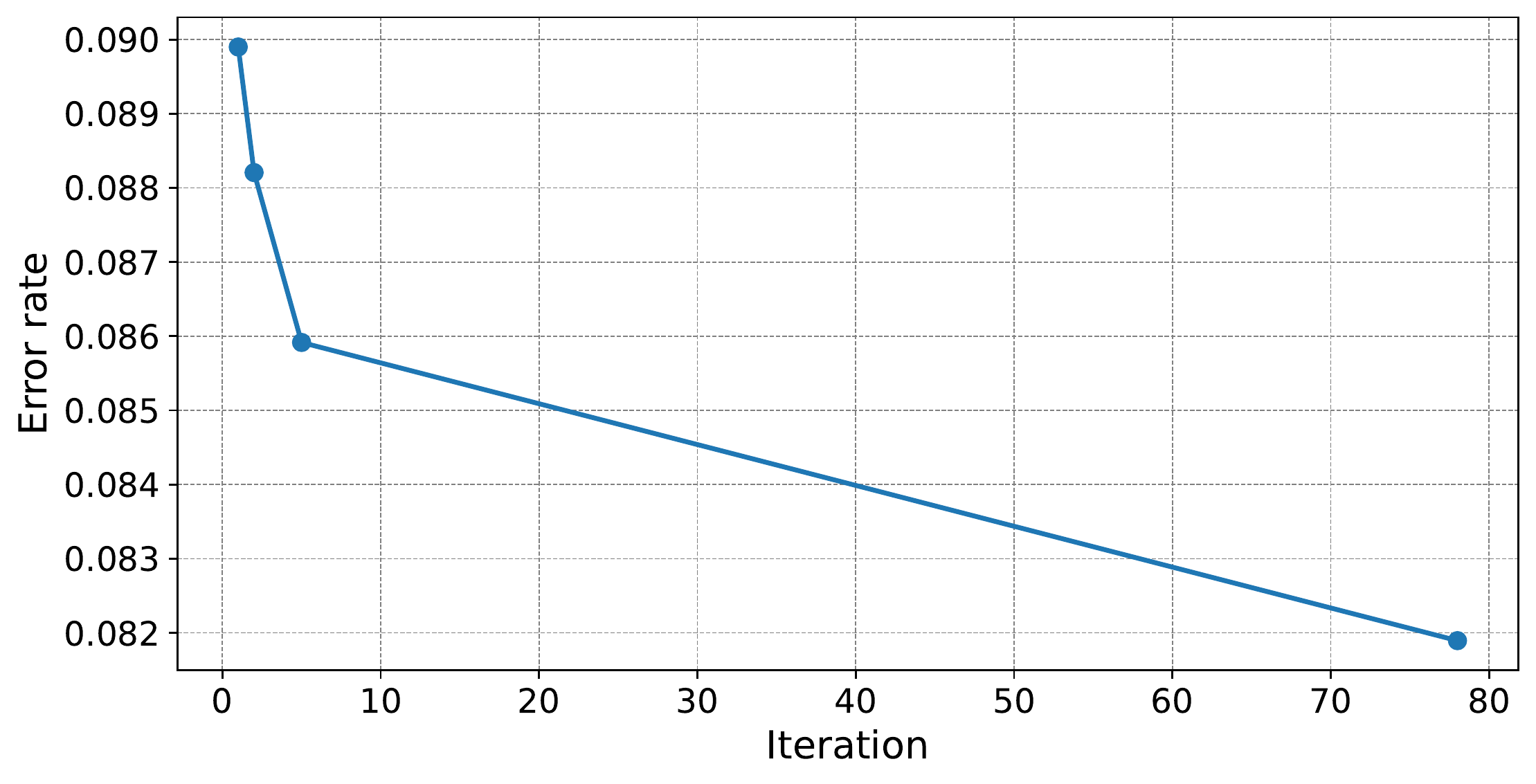}
    \caption{Decreasing error rate function through iterations.}
    \label{fig:error}
\end{figure}

\setlength{\tabcolsep}{5pt}
\renewcommand{\arraystretch}{1.4}
\begin{table}[H]
    \scriptsize
    \centering
    \caption{Number of visitors per dataset and absolute error for each sub-category of ages on the first day.}
    \label{tab:ex_error}
    \begin{tabular}{cccc}
        \hline
        \textbf{Age group} & \textbf{Real data} & \textbf{Synthetic data} &\textbf{  Absolute Error  }\\
        \hline
        18-24 &2,312 & 2,319 &7 (0.3\%) \\\hline
        35-44 & 3,230 & 3,215 &15 (0.46\%)\\\hline
        \(>65\) &3,483 & 3,439 &44 (1.26\%)\\\hline
    \end{tabular}
\end{table}

\subsection{Synthetic data} \label{ch4:info_ms_fimu}

In the end, an open dataset is proposed with an associative table whose primary key is (Person ID, Date ID) combination, which specifies the visit duration information, as shown in Table~\ref{tab:associative}. These two individual keys are linked to two other tables. The first, exemplified by Table~\ref{tab:personal_info}, contains specific information about people, for instance, fake French names, geolife profile, and region. The second table maps the days under analysis, from the first to the last day respectively as follows: \{1: 2017-05-31\}, ..., \{7: 2017-06-06\}.

\setlength{\tabcolsep}{5pt}
\renewcommand{\arraystretch}{1.4}
\begin{table}[!ht]
    \scriptsize
    \centering
    \caption{Final format of the generated dataset.}
    \label{tab:associative}
    \begin{tabular}{cccc}
        \hline
        \textbf{Index} & \textbf{Person ID} & \textbf{Date ID} & \textbf{Visit Duration}\\
        \hline
        1 &5385 & 2 & 6h\\\hline
        2 & 234 & 5 & 4h\\
        \hline
    \end{tabular}
\end{table}

\setlength{\tabcolsep}{5pt}
\renewcommand{\arraystretch}{1.4}
\begin{table}[!ht]
    \centering
    \scriptsize
    \caption{Table with personal information about individuals.}
    \label{tab:personal_info}
    \begin{tabular}{c c c c c c c c}
    \hline
    \textbf{Person ID} & \textbf{Name} & \textbf{Gender} &    \textbf{Age} &   \ldots & \textbf{Visitor category} &  \textbf{Region}   \\
    \hline
     91 &  Adrien Clement &      M &  45-54 &       \ldots &   French tourist &         Alsace \\\hline
     32947 &         Grégoire Didier &      M &  25-34 &  \ldots &   French tourist &  Franche-Comté \\\hline
     53990 &           Marie Le Lemaitre &      F &  25-34 &  \ldots &         Resident &  Franche-Comté \\\hline
     58664 &      Michelle-Céline Marion &      F &  25-34 &  \ldots &         Resident &  Franche-Comté  \\\hline
    \end{tabular}
\end{table}

We recall here the information about all the attributes of the MS-FIMU dataset below.

\begin{itemize}
    \item Static:
    \begin{itemize}
        \item \textbf{Visitor Category} with 3 values: `Resident', `Foreign tourist', and `French tourist';
        
        \item \textbf{Gender} with 3 values: `M' (masculine), `F' (feminine), and `NR' (Not Registered, e.g., for foreign people);
        
        \item \textbf{Age} with 8 values: `$<$18', `18-24', `25-34', `35-44', `45-54', `55-64', `$>$65', and `NR';
        
        \item \textbf{Geolife} with 12 values: `NR', 'comfortable family pavilion', `traditional rural', `comfortable family urban', `secondary residence', `popular', `dynamic rural', `growing peri-urban', `rural worker', `dynamic urban', `middle-class urban', and `low-income urban';
        
        \item \textbf{Region} with 37 values (countries for foreign people): `Belgium + Luxembourg', `Asia Oceania', `Netherlands', `Scandinavia', `United Kingdom', `Italy', `Spain', `China', `Other countries in Europe', `Germany', `United States', `Russia', `Swiss', `Eastern country', `Rest of the world', `AUTRE 97', `Centre', `Languedoc-Roussillon', ``Provence-Alpes-Côte d'Azur", `Lorraine', `Ile-de-France', `Franche-Comté', `Midi-Pyrénées', `Corse', `Basse-Normandie', `Aquitaine', `Poitou-Charentes', `Pays de la Loire', `Nord-Pas-de-Calais', `Champagne-Ardenne', `Bourgogne', `Bretagne', `Alsace', `Rhône-Alpes', `Picardie', `Auvergne', and `Haute-Normandie';
        
        \item \textbf{Sleeping Area} with 11 values: `Agglomeration of Hericourt', `Rest Territory of Belfort', `NR', `City of Belfort', `Vosges', `Rest of Doubs', `Rest of Haute Saone', `North Haut Rhin', `Agglomeration of Belfort', `Agglomeration of Montbeliard', `South Haut Rhin'.
    \end{itemize}
    
    \item Dynamic: 
        \begin{itemize}
        \item \textbf{Visit Duration} with 10 values: `Duration 2h', `Duration 3h', `Duration 4h', `Duration 5h', `Duration 6h', `Duration 7h', `Duration 8h', `Duration 9h', `Duration 10h', `Duration 10h-18h'.
    \end{itemize}
\end{itemize}

The motivation to release the synthetic open dataset with an associative table is to facilitate its improvement through adding more information about the population. Therefore, the associative table will remain unaltered, while more attributes can be added to the table with specific information about people. The generated dataset is available for anyone to freely access, use, modify, and share for any purpose at the aforementioned Github page (\url{https://github.com/hharcolezi/OpenMSFIMU}).

\subsection{Discussion and Related Work}

In the literature, several studies on human mobility show that humans follow particular patterns with a high probability of predictability~\cite{deMontjoye2013}. Hence, there is a high interest in understanding how people move. However, taking into account users' privacy, research emerges using synthetic and open data to solve such a problem. For example, in~\cite{Kashiyama2017}, the authors provided an approach for creating an open people mass movement dataset. In~\cite{Caiati2016}, the authors studied the use of open data for building and validating a realistic urban mobility model. The authors in~\cite{Pappalardo2017} developed a framework for the generation of individual human mobility trajectories with realistic Spatio-temporal patterns. Finally, the authors in~\cite{Kong2018} proposed a mobility dataset generation method of social vehicles traveling. 

All the aforementioned works treated a different problem from ours, which corresponds to different data types available they have. In our case, there were only statistical mobility indicators with information about the unique number of people per day and per the union of consecutive days. We then proposed a solution based on linear programming (linear constraint solver) to instantiate a feasible solution and, thus, reconstruct virtual humans based on statistics. We also notice that the authors in~\cite{Alaggan2017} used a similar linear constraint solver to their problem.

Regarding our solution described in Section~\ref{ch4:sec_materials_methods} and summarized in Fig.~\ref{fig:flow_chart}, one can notice that we have split the problem into two steps. Indeed, solving a single linear program considering the number of intersections $2^{Nb}=128$ for each sub-category (e.g., masculine or feminine) of each category (e.g., gender) would probably require a large number of variables and, thus, it was not considered in this chapter. In addition, we consider that virtual humans have `static' values for five attributes, i.e., each person has always the same geolife profile; people are from one unique region, they normally sleep in the same zone, and so on. The exception is for the `Visit duration' attribute, which was considered `dynamic' since people can vary the number of hours they stay in the FIMU per day.

Hence, as noticed in Fig.~\ref{fig:error} and Table~\ref{tab:ex_error}, the error metrics are very low when querying people in each sub-category from the generated data, compared to the original one. In other words, the result, which is one of many possible scenarios, closely describes how people behave during the week of interest. With such results, it is possible to assert with a reasonable amount of accuracy how many people were present in each combination (\textbf{intersection}) of 7 days, which is a more precise mobility scenario than just knowing the number of unique people per day or cumulative days (\textbf{union}).

For instance, from the final mobility scenario, and by querying the generated dataset, we can find out how many foreign tourists, French tourists, and residents are present only one or several days at the FIMU event, as well as their specific information such as socio-professional profile, region or countries, age groups, gender, and so on. For illustrative purposes, it is possible to know that from $176$ visitors present during all week (see Table.~\ref{ch4:tab_final_ms}, highlighted in bold), $153$ are residents, $20$ are French tourists, and $3$ are foreign tourists.

Moreover, it is noticed that foreign tourists were present normally at one unique day or at most three consecutive days, which is consistent with reality. Indeed, foreign people come to the FIMU for few days and usually have no `gaps' between days, such as one day present, the other not, and the next yes. Additionally, the premise of assigning one unique sleeping area for each visitor seems to indicate that the approach is consistent.

Such specific information on human mobility would be valuable for local communities and for accommodation and transportation companies, which would allow them to learn how people behave during a time period in a particular area. For instance, if one has information about the presence of foreign tourists on a specific combination of days and if they do not change much their sleeping place, accommodation companies can improve their future strategies to assist this population. Similarly, tourism companies would be more prepared knowing that most of the people present during the week are residents while tourists are rather present during the weekend of the FIMU event.

\section{Conclusion}
\label{ch4:sec_conclusion}

This chapter proposes an approach to infer and recreate synthetic data that provides a precise mobility scenario based on one-week statistical data of \textbf{unions of consecutive days} made available by~\cite{fluxvision1}. Our improved mobility scenario presents specific information about people present on one or several combinations of days (i.e., \textbf{all intersections of days}). The proposed approach is generic enough to apply to other mobility scenarios that rely on databases with information about the cumulative number of unique people for days (i.e., the union of consecutive days). Moreover, the proposed approach overcomes challenges due to data acquisition with anonymization techniques such as an anonymity threshold ($k=20$ in this case) and extrapolation algorithms.

The results show that the proposal can be efficiently applied to generate a synthetic dataset with specific information about people present in a certain region, for instance, attending the FIMU~\cite{FIMU} as was the case in this study. Finally, the generated and open dataset named MS-FIMU closely matches the original statistics with a low error rate, which substantiates the proposed approach. One direct \textbf{use case of MS-FIMU} is to evaluate differentially private cardinality estimation methods on longitudinal studies (e.g.,~\cite{Alaggan2017,app_blip,Alaggan2018}). Besides, \textbf{experimenting with machine learning tasks} could also be considered. Lastly, one can also \textbf{evaluate the effectiveness of new LDP protocols} on multidimensional and longitudinal frequency estimates, as we present in Chapters~\ref{chap:chapter5} and~\ref{chap:chapter6}, or other complex tasks such as marginal estimation (e.g.,~\cite{Shen2021,Peng2019,Zhang2018,Ren2018,Fanti2016}).

%% file: chapters/chapter7.tex
\chapter{LDP-Based System to Generate Mobility Reports from CDRs} \label{chap:chapter7}

In Chapters~\ref{chap:chapter1},~\ref{chap:chapter3}, and~\ref{chap:chapter4}, we have reviewed mobility reports published by OBS Flux Vision system~\cite{fluxvision1}. These mobility reports are, in other words, longitudinal statistics releases about the frequency of visitors by multiple attributes (e.g., as in~\cite{aktay2020google,heerschap2014innovation} too). Although current data privacy legislations require \textit{anonymity} ``on-the-fly" to collect CDRs for human mobility analytics, we posed ourselves two questions: \textbf{Q$_1$) what if} MNOs do not trust the data analyzers (e.g., third party companies working on mobility analytics)? Or \textbf{Q$_2$) what if} future data privacy legislations require different privacy protections than ``anonymity on-the-fly", e.g., demanding ``\textbf{sanitization} on-the-fly"? Indeed, while the former \textit{anonymity} protection provides syntactic privacy through the ``safe in the crowd" concept, the latter \textit{sanitization} protection provides algorithmic privacy by using a DP model, as we defined in Section~\ref{ch2:introduction} for this manuscript. These questions \textbf{Q$_1$} and \textbf{Q$_2$} initially motivated the core contributions of this chapter where we propose an LDP-based CDRs processing system to generate mobility reports, following the objective of the OBS Flux Vision system. We invite the reader to refer to Chapter~\ref{chap:chapter2} for the background on LDP. Lastly, we highlight that although we refer to our proposal as \textit{LDP-based}, this is a centralizer data owner (i.e., MNOs) that applies the LDP protocol on its servers, thus, providing $\epsilon$-DP guarantees for users.

 \section{Introduction} \label{ch7:sec_introduction}

We start by recalling some requirements of data privacy regulations on collecting and analyzing CDRs for human mobility analytics. For instance, although MNOs have the right and duty to store CDRs, according to the GDPR~\cite{GDPR}, it does not mean MNOs have the right to use the collected \textit {raw data} for other purposes. Besides, the CNIL~\cite{CNIL}, in France, requires that CDRs used for human mobility analysis (i.e., for purposes other than billing) to be \textit{anonymized on-the-fly} (i.e., ``hide in the crowd"). One reason behind this is because users cannot sanitize their data locally since CDRs are automatically generated on MNOs' servers through the use of a service (e.g., making/receiving phone calls). Lastly, MNOs cannot process data for mobility analytics containing users' IDs or a hashed version of them as they are still unique IDs.

The purpose of this chapter is, thus, to propose a privacy-preserving system for human mobility analytics through mobile phone CDRs. This way, MNOs can benefit from such an important data source while respecting their clients' privacy and following major recommendations of data protection authorities such as the GDPR and CNIL. Indeed, we intend to analyze human mobility through \textit{multidimensional} and \textit{longitudinal} statistical data releases (e.g., as in~\cite{aktay2020google,heerschap2014innovation,fluxvision1}). Throughout this paper, ``multidimensional" refers to data with multiple $d>1$ attributes. For instance, as shown in Section~\ref{ch3:fimu_db} and Chapter~\ref{chap:chapter4}, from subscription data, MNOs can gather: gender, age range (date of birth), and county origin (invoice address). From CDRs there is the coarse location (antennas that handled the service) and if it is ``roaming data" (foreign mobile) or not. Besides, ``longitudinal" refers to data with temporal information, i.e., analyzing human mobility over time (a.k.a. \textit{continuous monitoring} in software literature~\cite{rappor,microsoft,Erlingsson2019}). 

Therefore, we propose an LDP-based CDRs processing system for such purpose that goes beyond \textit{``anonymity on-the-fly"}, i.e., with \textit{``sanitization on-the-fly"}. The main reason to use the local DP model is that it allows sanitizing each sample independently while providing strong privacy guarantees (see Section~\ref{ch2:sub_ldp}). So, while MNOs CDRs processing systems utilizes each users' raw information (e.g., gender masculine, age range $<$18, ...), we propose that an LDP version of each users' data be used instead, i.e., LDP(masculine), LDP($<$18). In fact, we assume that MNOs CDRs processing systems have, first, pre-defined the mobility indicators (e.g., gender, age-ranges, nationality, ...) they want to release.  

Fig.~\ref{ch7:fig_system_model} illustrates our system model overview and the considered trust boundary (in dashed line). The first entity, namely, \textbf{users}, refers to MNOs' clients, which are not able to sanitize their data locally when using a service (e.g., exchanging SMS). The second entity is the MNOs themselves, which are \textbf{data holders} and must ensure \textit{``sanitization on-the-fly"} through an LDP mechanism \textit{for all CDRs used for analyzing human mobility}. The third entity is the \textbf{data processor} (considered as an \textit{untrusted} analyzer in our model), which processes data to generate statistical indicators. The last entity is the \textbf{data consumers}, which have access only to released statistics. 

\begin{figure}[!ht]
\centering
\includegraphics[width=1\linewidth]{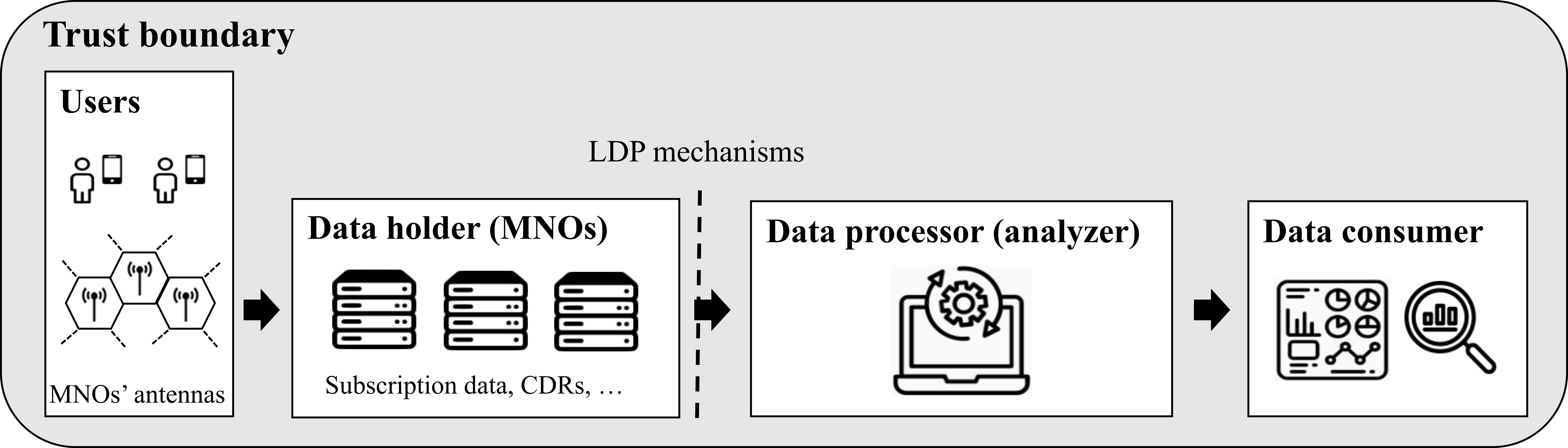}
\caption{Overview of our system model with an LDP-based privacy-preserving solution to sanitize each users' data on-the-fly before transmitting to the analyzer.} \label{ch7:fig_system_model}
\end{figure}

With more details, each time a user makes a call, or sends SMS, or connects to the internet ..., a CDR is generated and is stored by MNOs offline for billing and legal purposes, along with their subscription data (e.g., invoice address). This way, instead of MNOs transmit the raw data of this user (according to the pre-defined indicators), this data should be processed by an $\epsilon$-LDP algorithm in the MNOs servers, where $\epsilon$ is a public parameter, before transmitting it to the data processor. \textbf{Besides, similar to MNOs CDRs processing systems, the users' IDs should be excluded before transmitting any data to the analyzer. Thus, improving the users' privacy.}

Therefore, the data processor would only store uncorrelated (i.e., no IDs) $\epsilon$-LDP data. At the end of the period of analysis, the analyzer can aggregate these data to produce statistics through multiple frequency estimation, which depends on the LDP algorithm and the public parameter $\epsilon$. Notice that \textbf{with our proposal, both users and MNOs are safeguarded as no \textit{raw data} will be collected to analyze human mobility anymore.} \textit{However, $\epsilon$-LDP values must not result in \textbf{indirect unique identifiers}.} Indeed, if one can detect a unique $\epsilon$-LDP value for many days, it would violate the privacy of this user as s/he could be easily tracked away. 

So, in this chapter, we propose to use the GRR~\cite{kairouz2016discrete} LDP mechanism, which corresponds to the situation where no particular encoding is chosen. In other words, with GRR, $\epsilon$-LDP private reports will become anonymous depending on the size of the attribute (e.g., feminine or masculine, for the gender attribute), thus allowing a longitudinal collection of data. Lastly, we propose to generate mobility reports similar to the FIMU-DB of Section~\ref{ch3:fimu_db}, i.e., multidimensional frequency estimates by day and by the union of consecutive days (``\textit{cumulative frequency estimates}").

The remainder of this chapter is organized as follows. In Section~\ref{ch7:sec_multi_grr}, We extended the analytical analysis of GRR for multidimensional frequency estimates. Next, we explain our proposed LDP-based CDRs processing system in Section~\ref{ch7:sec_prop_metho}. In Section~\ref{ch7:sec_results}, we present our results, its discussions, and review related work. Lastly, in Section~\ref{ch7:sec_conclusion} we present the concluding remarks. The results of Section~\ref{ch7:sec_multi_grr} and a preliminary version of the proposed LDP-based CDRs processing system in Section~\ref{ch7:sec_prop_metho} with its results were published in a full paper~\cite{Arcolezi2021} at the 15th IFIP Summer School on Privacy and Identity Management.

\section{Multidimensional Frequency Estimates with GRR}\label{ch7:sec_multi_grr}

In the literature, there are few works for collecting multidimensional data with LDP based on random sampling (i.e., dividing users in groups)~\cite{xiao2,wang2019,Duchi2018,Wang2021_b,tianhao2017}. This technique reduces both dimensionality and communication costs, which will also be the focus of this chapter. Let $d\geq2$ be the total number of attributes, $\textbf{c}=[c_1,c_2,...,c_d]$ be the domain size of each attribute, $n$ be the number of users, and $\epsilon$ be the whole privacy budget. An intuitive solution is splitting (\textit{Spl}) the privacy budget, i.e., assigning $\epsilon/d$ for each attribute. The other solution is based on uniformly sampling without replacement (\textit{Smp}) only $r$ attribute(s) out of $d$ possible ones, i.e., assigning $\epsilon/r$ per attribute. Notice that both solutions satisfy $\epsilon$-LDP according to the sequential composition theorem
~\cite{dwork2014algorithmic}. More visually, Fig.~\ref{ch7:fig_spl_smp} illustrates both \textit{Spl} and \textit{Smp} solutions, with $r=1$ for \textit{Smp}.

\begin{figure}[!ht]
    \centering
    \includegraphics[width=1\linewidth]{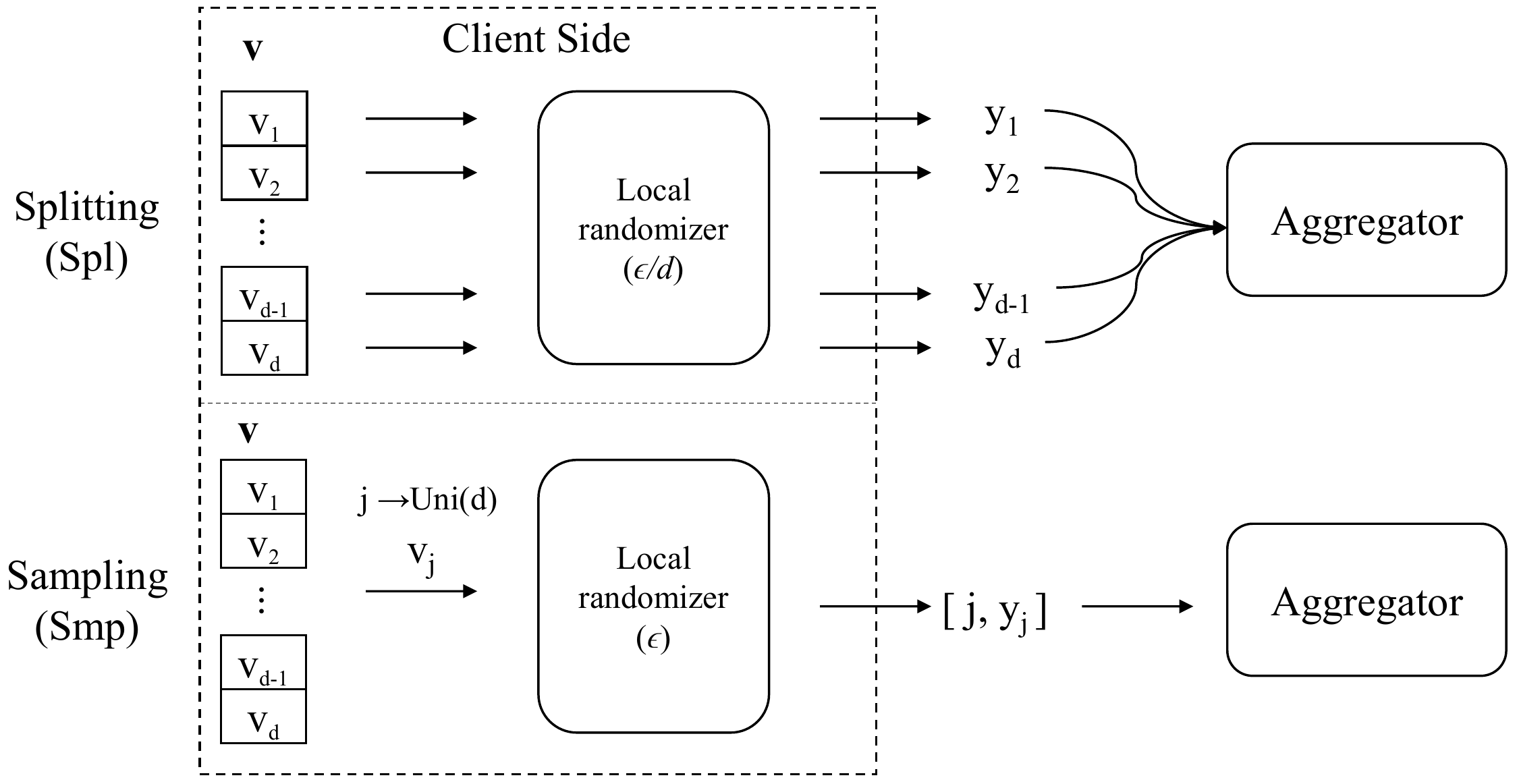}
    \caption{State-of-the-art solutions for multidimensional frequency estimates under $\epsilon$-LDP guarantees, where $Uni(d)=Uniform(\{1,2,...,d\})$.}
    \label{ch7:fig_spl_smp}
\end{figure}

For the first case, \textit{Spl}, replacing $\epsilon$ by $\epsilon/d$ in Eq.~\ref{eq:var_grr}, gives the variance ($\sigma^{2}_{1}$) of GRR as:

\begin{equation} \label{eq:var_spl_grr}
    \sigma^{2}_{1,GRR} = \frac{e^{\epsilon/d} + c_j - 2}{n(e^{\epsilon/d}-1)^2} \textrm{.}
\end{equation}

For the second case, \textit{Smp}, the number of users per attribute is reduced to $nr/d$. Thus, replacing $n$ by $nr/d$ and $\epsilon$ by $\epsilon/r$ in Eq.~\ref{eq:var_grr}, gives the variance ($\sigma^2_2$) of GRR as:

\begin{equation} \label{eq:var_smp_grr}
    \sigma^{2}_{2,GRR} = \frac{d(e^{\epsilon/r} + c_j - 2)}{nr(e^{\epsilon/r}-1)^2} \textrm{.}
\end{equation}

Obviously, if $r=d$ in Eq.~\eqref{eq:var_smp_grr}, one has Eq.~\eqref{eq:var_spl_grr}. Practically, the objective is reduced to finding $r$, which minimizes $\sigma^{2}_{2,GRR}$. This way, to find the optimal $r$, we first multiply $\sigma^{2}_{2,GRR}$ in Eq.~\eqref{eq:var_smp_grr} by $\epsilon$. Without loss of generality, minimizing $\sigma^{2}_{2,GRR}$ is equivalent to minimizing $\frac{\epsilon e^{\epsilon/r}}{r(e^{\epsilon/r}-1)^2}$. Hence, let $x=r/\epsilon$ be the independent variable, $\sigma^2_{2,GRR}$ can be rewritten as $y=\frac{1}{x}\cdot \frac{e^{1/x}}{(e^{1/x}-1)^2}$ as a function over $x$. It is not hard to prove that $y$ is an increasing function w.r.t. $x$ and, hence, we have a minimum and optimal when $r=1$ (a single attribute per user). We highlight that this is a common result in the LDP literature obtained for different protocols and contexts~\cite{xiao2,wang2019,Wang2021_b,tianhao2017,Jianyu2020,Wang2021,erlingsson2020encode,bassily2017practical}.

\section{LDP-Based Collection of CDRs for Mobility Reports} \label{ch7:sec_prop_metho}

In this section, according to the system overview in Fig.~\ref{ch7:fig_system_model}, we detail our LDP-based solution (Section ~\ref{ch7:sub_prop_metho}) regarding the \textit{Cumulative frequency estimates} scenario outlined in the introduction and its limitations (Section ~\ref{ch7:sub_limitations}). 

\subsection{Proposed methodology}\label{ch7:sub_prop_metho}

Fig.~\mbox{\ref{ch7:fig_flow_chart}} illustrates the overview of our LDP-based CDRs processing system applied to generate mobility reports by days and by the \textit{union} of consecutive days in a flow chart. Without loss of generality, we present our methodology for days, but it can be extended to any timestamp one desires. 

\begin{figure}[!ht]
\centering
\includegraphics[width=1\linewidth]{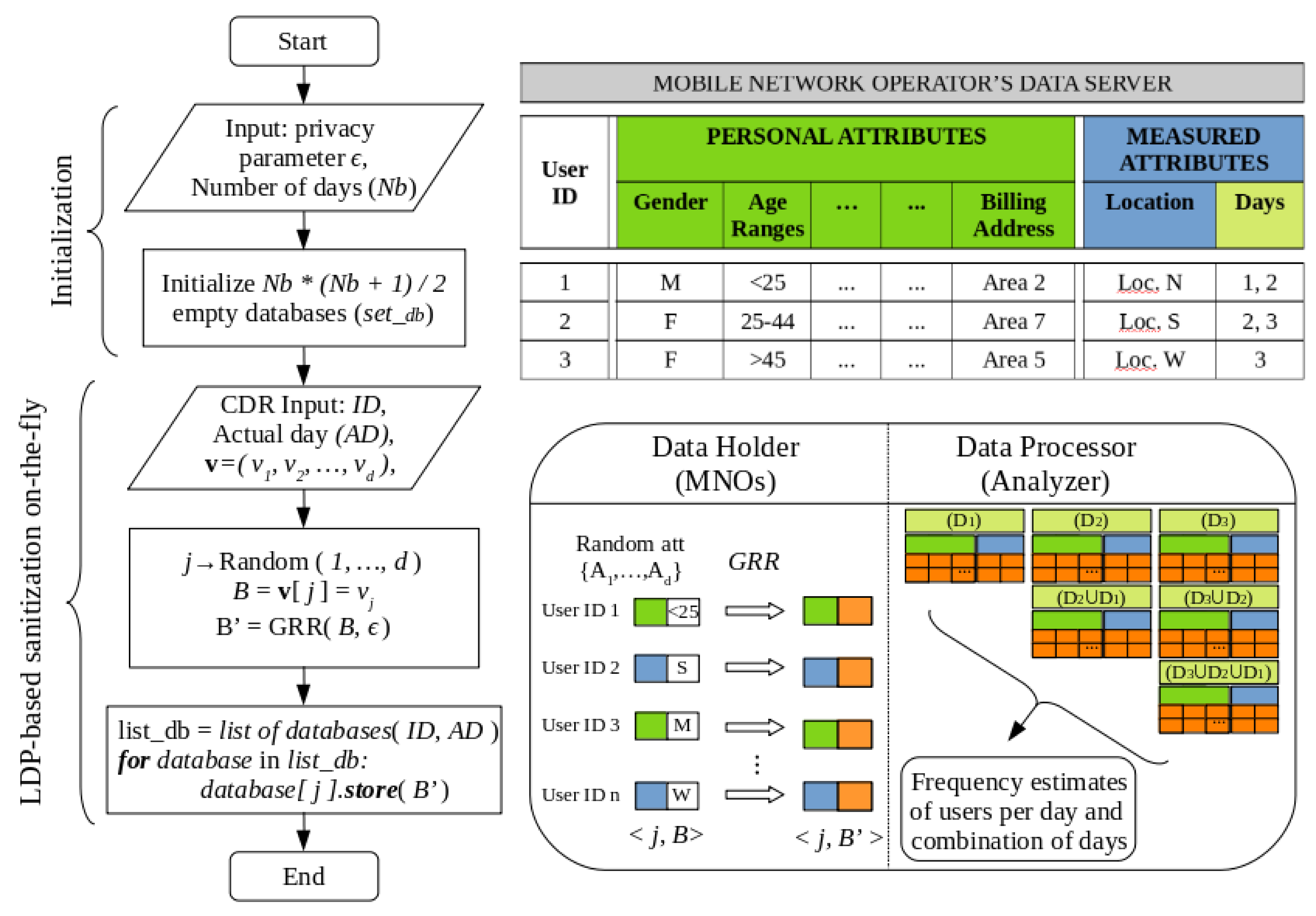}
\caption{Overview of our LDP-based CDRs processing system to generate mobility reports by days and by the \textit{union} of consecutive days.} \label{ch7:fig_flow_chart}
\end{figure}

\begin{enumerate}
    \item \textbf{Initialization.} According to the left side of Fig.~\mbox{\ref{ch7:fig_flow_chart}}, MNOs should define the privacy guarantee $\epsilon$, which is uniform for all users. Let $Nb$ be the whole period of analysis (e.g., total number of days) known a priori, e.g., before an event like the FIMU. So, the data processor should initialize $Nb(Nb+1)/2$ empty databases, \textit{which corresponds to all days and union of consecutive days}. For instance, if $Nb=3$ one will have $set_{db}=\{D_1, D_2, D_2\cup D_1, D_3, D_3\cup D_2, D_3\cup D_2\cup D_1\}$. 

    \item \textbf{LDP-based sanitization on-the-fly.} Similar to MNOs CDRs processing systems (e.g.,~\cite{fluxvision1}), MNOs will continue to be responsible for applying a privacy-preserving technique. In our proposal, the privacy-preserving technique corresponds to an LDP-based sanitization model on-the-fly using the GRR~\cite{kairouz2016discrete} protocol explained in Section~\ref{ch2:sub_ldp}. Besides, we assume that MNOs store information about their clients such that each user $u_i$ ($1 \le i\le n$) has a discrete-encoded tuple record $\textbf{v}=(v_1,v_2,...,v_d)$, which contains the values of $d$ categorical attributes $A=\{A_1,A_2,...,A_d\}$ (according to the pre-defined mobility indicators, e.g., as the table in the right side of Fig.~\ref{ch7:fig_flow_chart}). Since we have multiple attributes, we adopt the \textit{Smp} solution from~\ref{ch7:sec_multi_grr}, which randomly samples a \textbf{single attribute per user} and uses the whole privacy budget $\epsilon$ to sanitize it. For the rest of this chapter, we will refer to this solution as Smp[GRR].
    
    Therefore, we propose that MNOs apply GRR a \textit{single time} (i.e., once and for all) for each users' sampled data $B=v_j$ and consistently use the sanitized value $B'$ for all future reports $\langle j, B' \rangle$. In other words, MNOs would not use any raw data anymore but, rather, an $\epsilon$-LDP version of their clients' data. Since GRR does not utilize any particular encoding, the uncorrelated $\epsilon$-LDP values could be made `anonymous' within all other reports, thus, allowing longitudinal data collection with no risk of creating a `unique ID'. Notice that our solution can not ensure ``anonymity on-the-fly", but instead, the $\epsilon$-LDP values could probably be ``hidden in the crowd" depending on the domain size of the attributes. 
        
    Moreover, on the MNOs' side, each CDR contains metadata such as the user's ID and timestamp (Actual Day -- AD). Hence, for each user, MNOs calculate a $list_{db}$ that represents which databases (days and union of consecutive days) the $\epsilon$-LDP data should be stored by the data processor (with no ID). For instance, the $list_{db}$ can be calculated by knowing the days this user ``was present" (by CDRs) or, similarly, by using Bloom filters~\cite{bloom_f} to de-duplicate the users' presence throughout days. We later explain in an example how to calculate $list_{db}$.
   
    \item \textbf{Generating statistics.} Throughout the analysis period, the data processor can estimate the frequency of the population for all $d$ attributes for the database of each day and the combinations of past consecutive days. Finally, at the end of the analysis period, the analyst will have $Nb(Nb+1)/2$ databases, with the estimated frequencies for all $d$ attributes in each combination of \textit{union} of consecutive days.
    
\end{enumerate}

\textbf{Example to calculate $list_{db}$.} To calculate the $list_{db}$ for each user, consider the right side of Fig.~\ref{ch7:fig_flow_chart}, which has data for $Nb=3$ days. First, let Actual Day $AD=1$ (the first day of analysis). So, user $ID=1$ is detected by the MNO and his $list_{db}=\{D_1, D_2\cup D_1, D_3\cup D_2\cup D_1\}$. The reason behind this is that if this user does not appear anymore, we have considered his $\epsilon$-LDP report in the whole analysis. Next, let $AD=2$. For the same user $ID=1$, the MNO knows he was present in both two days, hence, his $list_{db}=\{D_2, D_3\cup D_2\}$ as the previous day his $\epsilon$-LDP report was already stored in $D_2\cup D_1$ and $D_3\cup D_2\cup D_1$. And, for the user $ID=2$, her $list_{db}=\{D_2, D_2\cup D_1, D_3\cup D_2, D_3\cup D_2\cup D_1\}$ to guarantee her $\epsilon$-LDP report is considered in each past union and future ones in the case she does not show up anymore. Without loss of generality the same procedure is applied when $AD=3$.

\subsection{Limitations}\label{ch7:sub_limitations}

The first key limitation we see in our methodology is the storage factor, which is due to collecting users' data per day and union of consecutive days. For instance, data processors need to initialize $Nb(Nb+1)/2$ empty databases where if one wishes to analyze an enhanced detailed scenario, it grows up very fast (i.e., with at least an $Nb^2/2$ factor). However, this scenario is only intended in special mobility analytics cases, e.g., tourism events, natural disasters, following up the spread of diseases, etc. In addition, there is high power for computation and powerful tools to deal with big data nowadays. One way to smoothen this problem in, e.g., daily scenarios, is to exclude the stored data after retrieving statistics. 

Further, similar to the FIMU-DB explained in Chapter~\ref{chap:chapter3}, there is an important loss of information by not calculating the intersection of users through days. That is, we propose to compute the number of users per union of consecutive days as it may have very few users per intersection (see our enhanced mobility scenario in Table~\ref{ch4:tab_final_ms} of Chapter~\ref{chap:chapter4}). The latter would not produce accurate frequency estimations due to the LDP formulation, which is data-hungry. At first glance, one can surely compute the pair-wise intersection for any two days in the analysis period using $|A \cap B| = |A| + |B| - |A\cup B|$. One possibility of solving the whole problem is to use the methodology developed in Chapter~\ref{chap:chapter4}, which models our proposed mobility scenario (days and union of consecutive days) as a linear program to find a solution for all intersections. Besides, for the case where one can have sufficient data samples per pair-, triple-, ..., and $Nb$-wise intersections, one can easily extend our methodology for such a case. However, the storage factor is even bigger as data processors would have to initialize $2^{Nb}-1$ empty databases (all combinations of intersections of days).

Lastly, the \textit{single time} sanitization step implies always reporting the same sanitized value $B'$ for the unique sampled attribute, which can be effective in the cases where the true client's data does not vary (static)~\cite{rappor,microsoft}. On the other hand, a measured attribute such as location is dynamic. Therefore, for the users who sample a dynamic attribute, for each different value, a new sanitized value would be generated, thus accumulating the privacy budget $\epsilon$ by the sequential composition theorem~\cite{dwork2014algorithmic}. Yet, in our privacy-preserving architecture (Fig.~\ref{ch7:fig_system_model}), the collected/stored $\epsilon$-LDP reports are `uncorrelated' from users, as no ID will be stored. Thus, improving the privacy of users.

\section{Results and Discussion} \label{ch7:sec_results}

In this section, we present the setup of our experiments in Section~\ref{ch7:sub_setup_experiments}. Next, we report the results in Section~\ref{sub:results_LDP_CDRs} obtained by applying our proposed methodology in the MS-FIMU dataset generated in Chapter~\ref{chap:chapter4}. Lastly, we discuss our work and review related work in Section \ref{ch7:sub_discussion}. 

\subsection{Setup of Experiments} \label{ch7:sub_setup_experiments}

\textbf{Environment.} All algorithms were implemented in Python 3.8.8 with NumPy 1.19.5 and Numba 0.53.1 libraries. The codes we developed for the preliminary results in paper~\cite{Arcolezi2021} are available in a Github repository\footnote{\url{https://github.com/hharcolezi/ldp-protocols-mobility-cdrs}.}. In all experiments of this manuscript, we report average results over 100 runs as LDP algorithms are randomized.

\textbf{Dataset.} We experimented with the MS-FIMU dataset from Chapter~\ref{chap:chapter4}. In these experiments, we excluded the data from `Foreign tourist' users regarding the `Visitor category' attribute. Hence, the filtered dataset aggregates a population of $87,098$ unique French users with $6$ attributes, where $5$ are static (`Visitor category' excluded) and $1$ is dynamic, along $Nb=7$ days (on average $\sim 26,000$ unique users per day). For more details about the attributes of this dataset, please refer to Section~\ref{ch4:info_ms_fimu}. Notice that the `Region' attribute only considers $22$ regions in France since we excluded Foreign people.

\textbf{Evaluation and metrics.} Let $Nb=7$ days be the whole analysis period, we then have $Nb(Nb+1)/2=28$ databases considering each day and union of consecutive days combination as $set_{db}=\{D_1, ..., D_3\cup D_2\cup D_1,..., D_7\cup D_6\cup...\cup D_1\}$. Notice that, at the same time, we can evaluate the privacy-utility trade-off according to data size, i.e., each day has around $26,700$ unique users, while the last union of consecutive days $D_7\cup D_6\cup...\cup D_1$ has all $87,098$ users. 

We vary the privacy parameter in the range $\epsilon=[1,2,3,4,5,6]$, which is within range of values experimented in the literature for multidimensional data (e.g., in~\cite{wang2019} the range is $\epsilon=[0.5,...,4]$ and in~\cite{Wang2021_b} the range is $\epsilon=[0.1,...,10]$). 

To evaluate our results, we use the mean squared error (MSE) metric averaged per the number of attributes $d$. Thus, for each attribute $j$ at time $t \in [1,Nb]$, we compute for each value $v_i \in A_j$ the estimated frequency $\hat{f}(v_i)$ and the real one $f(v_i)$ and calculate their differences. More precisely,

\begin{equation}
    MSE_{avg} = \frac{1}{d} \sum_{j \in [1,d]} \frac{1}{|A_j|} \sum_{v_i \in A_j}(f(v_i) - \hat{f}(v_i) )^2 \textrm{.}
\end{equation}

\textbf{Methods evaluated.} We consider for evaluation the two solutions from Section~\ref{ch7:sec_multi_grr}: 
\begin{itemize}
    \item Spl[GRR]: Splitting the privacy budget over the number of attributes $d$, i.e., for each user, send all value with $\epsilon/d$-LDP. 
    
    \item Smp[GRR]: Sampling a single attribute and send it with the whole privacy budget, i.e., for each user, send a sampled value with $\epsilon$-LDP. This is the solution adopted in our LDP-based CDRs processing system (cf. Fig.~\ref{ch7:fig_flow_chart}).
\end{itemize}

\subsection{Cumulative frequency estimates results}\label{sub:results_LDP_CDRs}

Fig.~\ref{ch7:fig_mean_mse_vhs} illustrate for both Spl[GRR] and Smp[GRR] methods, the averaged $MSE_{avg}$ per the number of days $Nb$ (y-axis) according to the privacy parameter $\epsilon$ (x-axis). With more details, Fig.~\ref{ch7:fig_results_rmse_all} illustrates for both methods the $MSE_{avg}$ results (y-axis) according to the privacy budget $\epsilon$ for each day and the union of consecutive days (x-axis), e.g., `321' refers to $D_3\cup D_2\cup D_1$. Lastly, for the sake of illustration, Fig.~\ref{ch7:fig_freq_est} illustrates multidimensional frequency estimates for a single day ($D_7$) and for the union of all consecutive days ($D_7\cup D_6\cup...\cup D_1$) using the adopted Smp[GRR] solution and $\epsilon=1$.

\begin{figure}[!ht]
\centering
\includegraphics[width=0.65\linewidth]{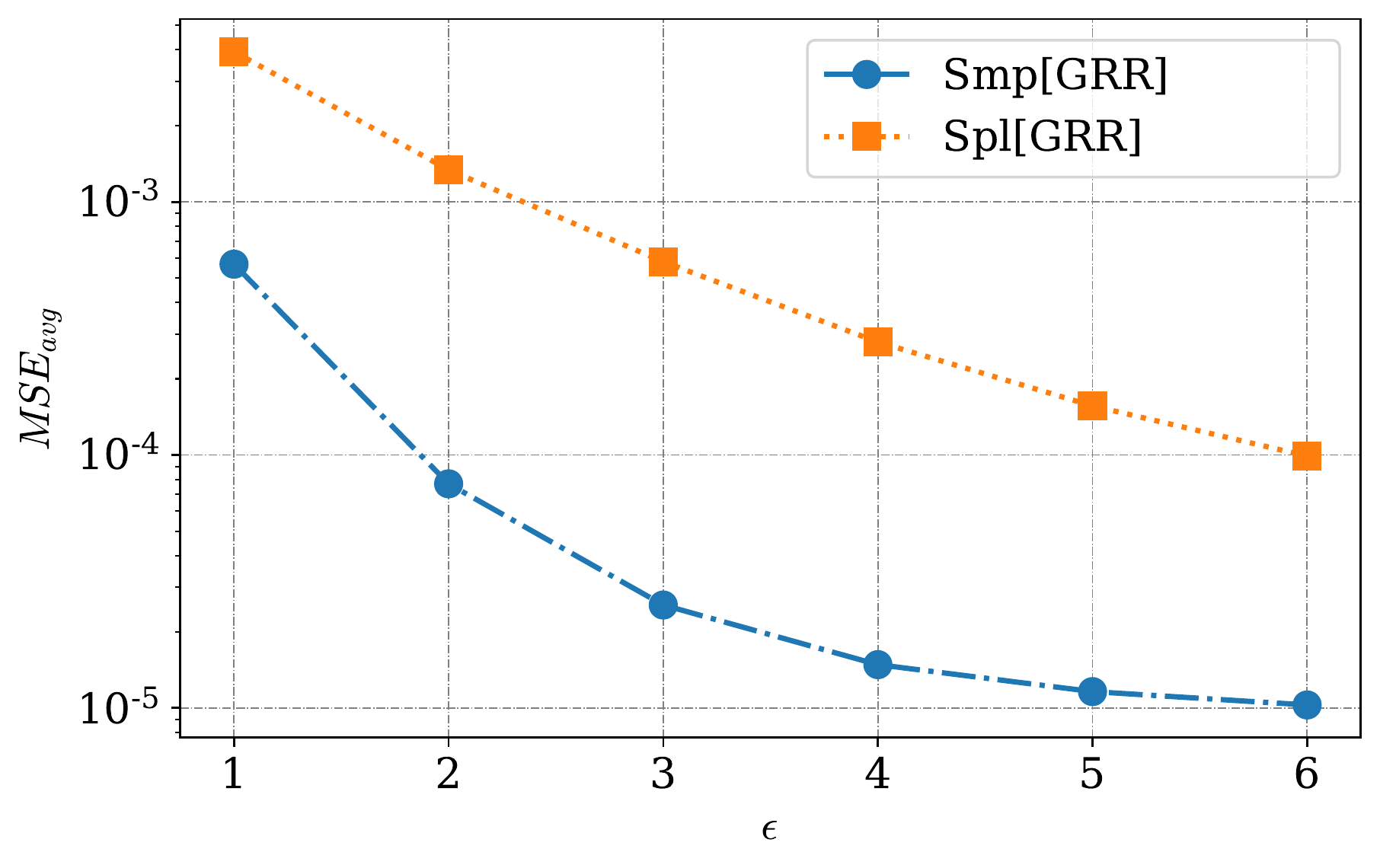}
\caption{Averaged $MSE_{avg}$ per the number of days $Nb$ (y-axis) varying $\epsilon$ (x-axis) on the MS-FIMU dataset comparing Spl[GRR] and Smp[GRR].} \label{ch7:fig_mean_mse_vhs}
\end{figure}

\begin{figure}[!ht]
\centering
\includegraphics[width=1\linewidth]{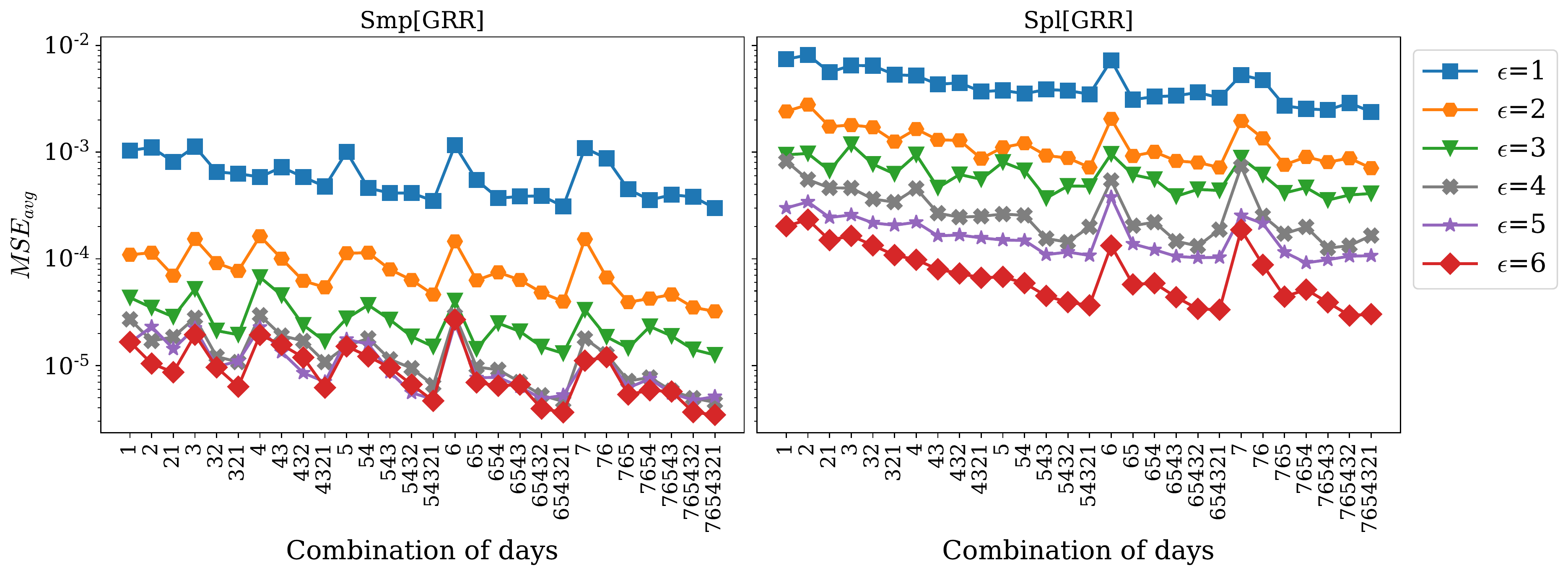}
\caption{$MSE_{avg}$ (y-axis) analysis comparing Smp[GRR] (left-side plot) and Spl[GRR] (right-side plot) by varying the privacy budget $\epsilon$ on each combination of days (x-axis) individually.} \label{ch7:fig_results_rmse_all}
\end{figure}

\begin{figure}[!ht]
\centering
\includegraphics[width=1\linewidth]{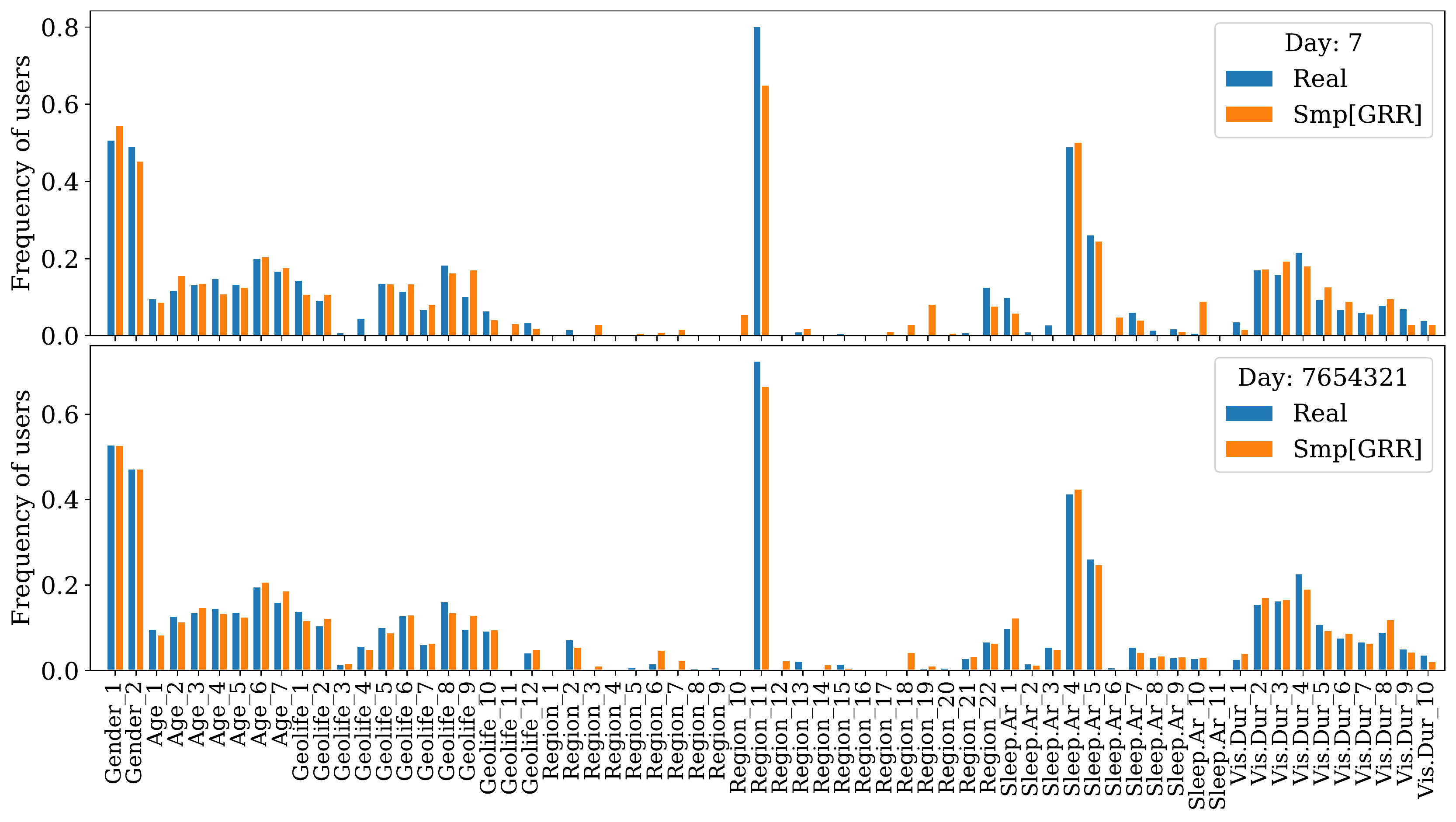}
\caption{Comparison between real and estimated frequencies for a single day ($D_7$) and to the union of all consecutive days ($D_7\cup D_6\cup...\cup D_1$) using the adopted Smp[GRR] solution and $\epsilon=1$.} \label{ch7:fig_freq_est}
\end{figure}

As one can notice in Fig.~\ref{ch7:fig_mean_mse_vhs}, overall, the proposed Smp[GRR] solution adopted in our LDP-based CDRs processing system consistently and considerably outperforms the baseline Spl[GRR]. In Fig.~\ref{ch7:fig_results_rmse_all}, except for $\epsilon=1$, the curves of Smp[GRR] are under even to the best one of Spl[GRR] using the highest privacy budget $\epsilon=6$. As also highlighted in the literature~\mbox{\cite{wang2019,xiao2,tianhao2017}}, privacy budget splitting is sub-optimal, which leads to higher estimation error. Indeed, in a multidimensional setting, the combination of privacy budget splitting and high numbers of values in a given attribute (e.g., $Region$ with 22 values) leads to lower data utility even for high privacy regimes. On the other hand, the Smp[GRR] solution based on random sampling uses the whole privacy budget to a single attribute, and this problem is, hence, minimized. However, there is also an error provided by the sampling technique, which is due to observing a sample instead of the entire population.

Moreover, in Figs.~\ref{ch7:fig_results_rmse_all} and~\ref{ch7:fig_freq_est}, it is noteworthy that the $MSE_{avg}$ decreases as the data size increases. Intuitively, this is due to LDP, which requires a large amount of data to guarantee a good balance of noise. In our case, single days (e.g., $D_7$) have less users comparing to the union of all consecutive days (e.g., $D_7\cup D_6\cup...\cup D_1$) and, hence, single days are generally the peak-values in Fig.~\ref{ch7:fig_results_rmse_all}. Indeed, these results are consistent with Eqs.~\eqref{eq:var_spl_grr} and~\eqref{eq:var_smp_grr}, where the variances are decreasing functions over the number of users $n$. Yet, these peak values are smoothed using Smp[GRR], which induces less error by sampling a single attribute for each user.

Lastly, we highlight that the objective of our experiments was to measure the accuracy loss (based on the $MSE_{avg}$ error metric) of using our LDP-based ``sanitization on-the-fly" system in comparison with the original statistics produced by an ``anonymity on-the-fly" based system. As shown in the results, accurate multidimensional frequency estimates could be achieved for practical purposes with strong privacy guarantees (see, e.g., Fig.~\ref{ch7:fig_freq_est} with $\epsilon=1$). On the other hand, in terms of the overall privacy budget $\epsilon$ per user, in the worst case, the sequential composition theorem~\cite{dwork2014algorithmic} applies for each data release. As also pointed out in~\cite[Section 8.4, Table 2]{linkedin} and in~\cite{desfontaines_dp_real_world}, real-world DP systems utilize $\epsilon$ as large as the ones experimented in this manuscript on daily basis. Thus, some future implementation of our LDP-based CDRs processing system to generate mobility reports is a potential perspective.

\subsection{Discussion and Related Work}\label{ch7:sub_discussion}

As reviewed in Chapters~\ref{chap:chapter1} and~\ref{chap:chapter2}, mobile phone CDRs have been largely used to analyze human mobility in several contexts, e.g., the spread of infectious diseases~\cite{Kishore2019,Lu2012,ebola,deAlarcon2021,Grantz2020,Vespe2021}, natural disasters~\cite{Hong2018,Dujardin2020,Lu2012}, tourism~\cite{fluxvision1,heerschap2014innovation,Merrill2020}, and so on. However, concerning privacy, de Montjoye et al.~\cite{deMontjoye2013} show that humans follow particular patterns, which allows predicting human mobility with high accuracy. For instance, in a dataset of $1.5$ million users, the authors showed that $95\%$ of this population can be re-identified using four approximate locations and their timestamps. Besides, Zang and Bolot~\cite{Zang2011} have performed extensive experiments showing that the anonymization of location data from CDRs using \textit{k}-anonymity~\cite{samarati1998protecting,SWEENEY2002} leads to privacy risks. Further, in non-technical papers, de Montjoye et al.~\cite{deMontjoye2018} discuss the conscientious use of mobile phone data for mobility analytics, and Buckee~\cite{Buckee2014} highlights both the importance of collecting CDRs to analyzing human mobility in low-income countries and the privacy concerns that rise up. 

Because of these privacy issues, MNOs tend to publish aggregated mobility data~\cite{deAlarcon2021,Xu2017,Vespe2021,Tu2018,fluxvision1}, e.g., the number of users by coarse location at a given timestamp or the number of users in a single location (cf. Section~\ref{ch3:fimu_db}). However, as recent studies have shown, even aggregated mobility data can be subject to membership inference attacks~\cite{Pyrgelis2017,Pyrgelis2020} and users' trajectory recovery attack~\cite{Tu2018,Xu2017}. More precisely, the later authors in~\cite{Tu2018,Xu2017} showed that their attack reaches accuracies as high as $73\% \sim 91\%$, suggesting generalization and perturbation through DP~\cite{Dwork2006,Dwork2006DP,dwork2014algorithmic} as a means to mitigate this attack. Therefore, it is vital to deploying systems that allow analyzing human mobility (e.g., through CDRs) with strong privacy-preserving guarantees. 

With these elements in mind and with the motivating questions \textbf{Q$_1$} and \textbf{Q$_2$} from the beginning of this chapter, we have proposed a solution beyond ``anonymity on-the-fly" since aggregated location data are still at risk of leaking private information. Indeed, our solution considers ``sanitization on-the-fly" with an LDP protocol, in which rather than transmitting aggregated raw data for the analyzer, we propose that MNOs sanitize each users' data independently (as if it was made by the user) and send it to the \textit{untrusted analyzer}. 

As we present in this chapter, implementing the Smp[GRR] solution in our methodology could ensure that $\epsilon$-LDP private reports will not become indirect \textit{unique IDs}. The reason behind this is because no particular encoding is used with GRR and, thus, $\epsilon$-LDP values are generic to any user. So, it is possible to utilize the sanitized value in longitudinal studies if the domain size of attributes is not big. Besides, notice that each time users connect, MNOs will always report the same attribute out of $d$ possible ones. That is, even though users appear all days in the analysis (in this dataset $\sim 0.2\%$ of users), MNOs will never report the remaining $d-1$ attributes, which were not sampled. Lastly, our solution would also safeguard MNOs as no \textit{raw data} would be shared with the analyzer for the purpose of human mobility analysis, but, rather, $\epsilon$-LDP values that are robust to post-processing. One clear limitation of our LDP-based CDRs processing system is that the recent privacy amplification by shuffling~\cite{Balle2019,Erlingsson2019,erlingsson2020encode,Wang2020,li2021privacy} does not apply. Although all users' IDs are excluded, the signals' order is not hidden due to ``sanitization on-the-fly". That is, the $\epsilon$-LDP reports are not aggregated in ``batches" to provide some ``anonymity" and profit from amplification. Therefore, extending our solution to the shuffle DP model is a potential and intended perspective.

\section{Conclusion} \label{ch7:sec_conclusion}

This chapter investigated the problem of collecting and analyzing CDRs-based data to generate multidimensional frequency estimates throughout time. We proposed an LDP-based CDRs processing system as an extension of ``anonymity on-the-fly" to satisfy ``\textit{sanitization} on-the-fly", thus, providing higher privacy guarantees for each user. With our proposal, we can have preliminary answers to the motivating questions \textbf{Q$_1$} and \textbf{Q$_2$} highlighted at the beginning of this chapter. That is, such a privacy-preserving system would allow MNOs to share the sanitized data with \textit{untrusted analyzers}, with a more strict setting that allows sanitizing each data independently on-the-fly. As shown in the results, the proposed LDP-based CDRs processing system using Smp[GRR] achieves accurate multidimensional frequency estimates for practical purposes (\textit{cf.} Fig.~\ref{ch7:fig_freq_est}, for example), proving its effectiveness in producing mobility reports as the original ones from OBS. 

On the one hand, this is because GRR has low utility loss for attributes with small domain sizes. On the other hand, if MNOs intend to pre-define a mobility indicator on a higher domain (e.g., the number of people in each $\sim 1,000$ bus stops of a given city), other protocols like OUE~\cite{tianhao2017} could provide higher data utility, as its variance does not depend on the domain size. However, since OUE is based on unary-encoding (cf. Section~\ref{ch2:sub_ldp}), it would probably generate a sanitized value similar to a \textit{unique ID}. In other words, analyzers would be able to use the \textit{unique} OUE-based reports to track individuals across many days. One possible solution would be using two rounds of sanitization (i.e., \textit{memoization}~\cite{rappor,microsoft}), also mentioned in Section~\ref{ch2:sub_ldp}. Indeed, this is one of the core contributions of the next Chapter~\ref{chap:chapter5}, which investigates how to improve the utility of LDP protocols for longitudinal (based on memoization) and multidimensional frequency estimates.

%% file: chapters/chapter5.tex
\chapter{Multidimensional Frequency Estimates Over Time With LDP: Utility Focus}   \label{chap:chapter5}

In Chapter~\ref{chap:chapter7}, we focused on a more \textbf{practical} perspective for the problem of generating multidimensional mobility reports throughout time from CDRs. In this chapter, we abstracted this problem and, thus, we contribute on the \textbf{theoretical} aspect by optimizing \textbf{the utility} of LDP protocols for \textit{longitudinal} and \textit{multidimensional} frequency estimates. This way, the more the estimated frequencies approximate the real ones, the more ML models can take advantage of when performing learning/prediction tasks~\cite{ElSalamouny2020}. Notice that \textbf{our solutions are generic to any LDP application scenario} (e.g., collecting user behavior in software~\cite{rappor,microsoft,apple}). We invite the reader to refer to Chapter~\ref{chap:chapter2} for the background on LDP. 

\section{Introduction} \label{ch5:sec_introduction}

In this chapter, we focus on the problem of private frequency (or histogram) estimation of multiple attributes throughout time with LDP. As in previous Chapter~\ref{chap:chapter7}, we assume there are $d$ attributes $A=\{A_1,A_2,...,A_d\}$, where each attribute $A_j$ with a discrete domain has a specific number of values $c_j=|A_j|$. Each user $u_i$ for $i \in \{1,2,...,n\}$ has a tuple $\textbf{v}^{(i)}=(v^{(i)}_{1},v^{(i)}_{2},...,v^{(i)}_{d})$, where $v^{(i)}_{j}$ represents the value of attribute $A_j$ in record $\textbf{v}^{(i)}$. Thus, for each attribute $A_j$ at time $t \in [1,\tau]$, the aggregator's goal is to estimate a $c_j$-bins histogram, including the frequency of all values in $A_j$. 

On tackling both longitudinal and multidimensional settings, one needs to consider the allocation of the privacy budget, which can grow extremely quickly due to the composition theorem~\cite{dwork2014algorithmic}. So, first, we focus on solving the multidimensional aspect with a random sampling-based solution~\cite{xiao2,wang2019,Duchi2018,Wang2021_b}, also used in Chapter~\ref{chap:chapter7}. Next, we considered the \textit{memoization}-based framework~\cite{rappor,microsoft,erlingsson2020encode} to solve the longitudinal setting, which allows having an upper bound to the privacy budget. In both cases, we extended the analysis of three state-of-the-art protocols, namely, GRR~\cite{kairouz2016discrete}, OUE~\cite{tianhao2017}, and SUE~\cite{rappor}, presented in Section~\ref{ch2:sub_ldp}. Thus, combining the optimal cases of each setting, we propose a new solution named \underline{A}daptive \underline{L}DP for \underline{LO}ngitudinal and \underline{M}ultidimensional \underline{FRE}quency \underline{E}stimates (ALLOMFREE). We demonstrate through experimental validations using four real-world datasets the advantages of ALLOMFREE over state-of-the-art protocols~\cite{rappor,tianhao2017}, with a gain of accuracy, on average, ranging from $10\%$ up to $55\%$ with the analyzed range of $\epsilon$-LDP guarantees. 

The rest of this chapter is organized as follows. In Section~\ref{ch5:sec_multidimensional}, we extend the analysis of OUE and SUE to multidimensional data collections. In Section~\ref{ch5:sec_longitudinal} we present the \textit{memoization}-based framework for longitudinal data collections, the extension and analysis of longitudinal GRR and longitudinal UE-based protocols; the numerical evaluation of their performance, and we present our ALLOMFREE solution. In Section~\ref{ch5:sec_results_discussion}, we present experimental results, discuss our results and review related work. Lastly, in Section~\ref{ch5:sec_conc}, we present the concluding remarks. The development in Sections~\ref{ch5:sec_multidimensional} and~\ref{ch5:sec_longitudinal} and the results presented in Section~\ref{ch5:sec_results_discussion} were submitted as part of a full article~\cite{Arcolezi2021_allomfree} to the Digital Communications and Networks journal.

\section{Multidimensional Frequency Estimates with LDP}\label{ch5:sec_multidimensional}

As reviewed in Section~\ref{ch7:sec_multi_grr}, there are mainly two solutions for collecting multidimensional data with LDP (see Fig.~\ref{ch7:fig_spl_smp}). In this section, we will follow the same development used in Section~\ref{ch7:sec_multi_grr} for two other protocols, namely, SUE and OUE. Let $d\geq2$ be the total number of attributes, $\textbf{c}=[c_1,c_2,...,c_d]$ be the domain size of each attribute, $n$ be the number of users, and $\epsilon$ be the whole privacy budget. 

For the first case, \textit{Spl}, replacing $\epsilon$ by $\epsilon/d$ in Eqs.~\eqref{eq:var_sue} and~\eqref{eq:var_oue} give the variances ($\sigma^{2}_{1}$) of SUE and OUE, respectively, as:

\begin{equation} \label{ch5:eq_var_spl_ue}
\begin{split}
\sigma^{2}_{1,SUE} & = \frac{e^{\epsilon/2d}}{n(e^{\epsilon/2d}-1)^2}  \textrm{,}   \\
\sigma^{2}_{1,OUE} & = \frac{4e^{\epsilon/d}}{n(e^{\epsilon/d}-1)^2} \textrm{.}
\end{split}
\end{equation}

For the second case, \textit{Smp}, the number of users per attribute is reduced to $nr/d$. Thus, replacing $n$ by $nr/d$ and $\epsilon$ by $\epsilon/r$ in Eqs.~\ref{eq:var_sue} and~\ref{eq:var_oue} give the variances ($\sigma^2_2$) of SUE and OUE, respectively, as:

\begin{equation} \label{ch5:eq_var_smp_ue}
\begin{split}
\sigma^{2}_{2,SUE} & = \frac{d(e^{\epsilon/2r})}{nr(e^{\epsilon/2r}-1)^2} \textrm{,}\\
\sigma^{2}_{2,OUE} & = \frac{d(4e^{\epsilon/r})}{nr(e^{\epsilon/r}-1)^2} \textrm{.}
\end{split}
\end{equation}

Obviously, if $r=d$ in Eq.~\eqref{ch5:eq_var_smp_ue}, one has Eq.~\eqref{ch5:eq_var_spl_ue}. Practically, the objective is reduced to finding $r$, which minimizes $\sigma^2_2$ for each protocol. This way, to find the optimal $r$ for each protocol, we first multiply each $\sigma^2_2$ in Eq.~\eqref{ch5:eq_var_smp_ue} by $\epsilon$. Without loss of generality, minimizing $\sigma^{2}_{2,SUE}$ and $\sigma^{2}_{2,OUE}$ is equivalent to minimizing $\frac{\epsilon e^{\epsilon/2r}}{r(e^{\epsilon/2r}-1)^2}$ and $\frac{\epsilon e^{\epsilon/r}}{r(e^{\epsilon/r}-1)^2}$ (similar to GRR in Section~\ref{ch7:sec_multi_grr}), respectively. Hence, let $x=r/\epsilon$ be the independent variable, $\sigma^2_{2,OUE}$ can be rewritten as $y_1=\frac{1}{x}\cdot \frac{e^{1/x}}{(e^{1/x}-1)^2}$ and $\sigma^2_{2,SUE}$ can be rewritten as $y_2=\frac{1}{x}\cdot \frac{e^{1/2x}}{(e^{1/2x}-1)^2}$ as functions over $x$. It is not hard to prove that both $y_1$ and $y_2$ are increasing functions w.r.t. $x$ and, hence, we have a minimum and optimal when $r=1$ (a single attribute per user) for both protocols too. 

\textbf{Therefore, in this chapter, we adopt the multidimensional setting \textit{Smp} with $r=1$}. In this setting, users tell the data collector which attribute was sampled, and its perturbed value ensuring $\epsilon$-LDP by applying either GRR or UE-based protocols; the data analyst server would not receive any information about the remaining $d-1$ attributes. .

\section{Longitudinal Frequency Estimates with LDP}\label{ch5:sec_longitudinal}

In this section, we present the \textit{memoization}-based framework for longitudinal data collections (Section~\ref{sub:memoization}). Next, we present the analysis of longitudinal GRR (Section~\ref{subsub:l_de}) and longitudinal UE-based protocols (Section~\ref{subsub:l_ue}). Lastly, we evaluate numerically the extended longitudinal protocols (Section~\ref{sub:analysis_long}) and we propose our ALLOMFREE solution (Section~\ref{sub:allomfree}).

\subsection{Memoization-based data collection with LDP} \label{sub:memoization}

In the literature, many works study how to collect and analyze categorical data longitudinally based on \textit{memoization}~\cite{rappor,microsoft,erlingsson2020encode}. The key idea behind memoization is using two sanitization processes. The first round ($RR_1$) replaces the real value $B$ with a sanitized one $B'$ with a higher epsilon ($\epsilon_{\infty}$). Whenever one intends to report $B$, $B'$ shall be reused to produce other sanitized versions $B''$ with lower epsilon values. Notice that the second sanitization ($RR_2$) is a \textit{must} to avoid `averaging attacks', in which adversaries can reconstruct the true value from multiple sanitized versions of it. This technique allows achieving privacy over time with an upper bound value of $\epsilon_{\infty}$-LDP.

Let $A_j=\{v_1,v_2,...,v_{c_j}\}$ be a set of $c_j=|A_j|$ values of a given attribute and let $\epsilon$ be the privacy budget. In this chapter, for both $RR_1$ and $RR_2$ steps, we will apply either GRR, SUE, or OUE. The unbiased estimator in Eq.~\eqref{eq:est_pure} for the frequency $f(v_i)$ of each value $v_i$ for $i \in [1,c_j]$ is now extended to: 

\begin{equation}\label{eq:est_longitudinal}
    \hat{f}_L(v_i) = \frac{N_i - nq_1(p_2-q_2) - nq_2}{n(p_1-q_1)(p_2-q_2)} \textrm{,}
\end{equation}

in which $N_i$ is the number of times the value $v_i$ has been reported, $n$ is the total number of users, $p_1$ and $q_1$ are the parameters used by an LDP protocol for $RR_1$, and $p_2$ and $q_2$ are the parameters used by an LDP protocol for $RR_2$. 

\begin{theorem} \label{theo:est_long} The estimation result $\hat{f}_L(v_i)$ in Eq.~\eqref{eq:est_longitudinal} is an unbiased estimation of $f (v_i)$ for any value $v_i \in A_j$.
\end{theorem}

\begin{proof}
\begin{equation*}
\begin{aligned}
    E[\hat{f}_L(v_i)] &= E\left[ \frac{N_i - nq_1(p_2-q_2) + nq_2}{n(p_1-q_1)(p_2-q_2)} \right] \\
    &= \frac{E[Ni]}{n(p_1-q_1)(p_2-q_2)}  -  \frac{ q_1(p_2-q_2) - q_2}{(p_1-q_1)(p_2-q_2)}  \textrm{.}
\end{aligned}
\end{equation*}

Let us focus on 

\begin{equation*}
\begin{aligned}
    E[N_i] &= n f(v_i) \left(p_{1} p_{2} + q_{2} \left(1 - p_{1}\right)\right) \\
    &+ n \left(1 - f(v_i)\right) \left(p_{2} q_{1} + q_{2} \left(1 - q_{1}\right)\right)\textrm{.}
\end{aligned}
\end{equation*}

Thus,

\begin{equation*}
    E[\hat{f}_L(v_i)] = f(v_i) \textrm{.}
\end{equation*}
\end{proof}

\begin{theorem} \label{theo:variance_grr_allomfree} The variance of the estimation in Eq.~\eqref{eq:est_longitudinal} is:

\begin{equation}\label{var:longitudinal}
\begin{gathered}
    Var(\hat{f}_L(v_i))  = \frac{\gamma (1-\gamma)}{n (p_1-q_1)^2 (p_2-q_2)^2} \textrm{, where} \\
    \gamma = f(v_i) \left( 2 p_{1} p_{2} - 2 p_{1} q_{2} + 2 q_{2} - 1 \right) + p_{2} q_{1} + q_{2} (1 - q_{1}) \textrm{.}
\end{gathered}
\end{equation}

\end{theorem}

\begin{proof}
Thanks to Eq.~\eqref{eq:est_longitudinal} we have

\begin{equation*}
Var(\hat{f}_L(v_i)) = 
\frac{Var(N_i)}{n^2 (p_1-q_1)^2 (p_2-q_2)^2}  \textrm{.}
\end{equation*}

Since $N_i$ is the number of times the value $v_i$ is observed, it can be defined as $N_i = \sum_{z=1}^n X_z$ where $X_z$ is equal to 1 if the user $z$, 
$1 \le z \le n$ reports value $v_i$, and 0 otherwise. We thus have 
$
Var(N_i) 
= \sum_{z=1}^n Var(X_z) 
= n Var(X)$. Since all the users are independent,

\begin{equation*}
P(X = 1) = P(X^2 = 1) = f(v_i) \left( 2 p_{1} p_{2} - 2 p_{1} q_{2} + 2 q_{2} - 1 \right) + p_{2} q_{1} + q_{2} (1 - q_{1}) = \gamma \textrm{.}
\end{equation*}

We thus have $Var(X)= \gamma - \gamma^2 = \gamma(1 - \gamma) $ and, finally,

\begin{equation*} 
Var(\hat{f}_L(v_i)) =
\frac{\gamma (1-\gamma)}{n (p_1-q_1)^2 (p_2-q_2)^2}.
\end{equation*}
\end{proof}

In this chapter, we will use the \textit{approximate variance}, in which $f(v_i)=0$ in Eq.~\eqref{var:longitudinal}, which gives:

\begin{equation}\label{var:aprox_longitudinal}
    Var^*(\hat{f}_L(v_i))  =  \frac{\left(p_{2} q_{1} - q_{2} \left(q_{1} - 1\right)\right) \left(- p_{2} q_{1} + q_{2} \left(q_{1} - 1\right) + 1\right)}{n (p_1-q_1)^2 (p_2-q_2)^2} \textrm{.}
\end{equation}

\subsection{Longitudinal GRR (L-GRR): definition and $\epsilon$-LDP study}\label{subsub:l_de}

Let $V=\{v_1,v_2,...,v_{c_j}\}$ be a set of $c_j$ values of a given attribute and let $v_i$ be the real value. We now describe an extension of GRR for longitudinal studies; we refer to this protocol as L-GRR for the rest of this chapter. First, GRR does not require any particular encoding (direct encoding~\cite{tianhao2017}). Next, there are two rounds of sanitization, $RR_1$ and $RR_2$ applying GRR, described in the following.

\begin{enumerate}
    \item $RR_1[GRR]$: Memoize a value $B'$ such that
    \begin{equation*}
    B'=
    \begin{cases}
      v_i, & \text{with probability}\ p_1 \textrm{,}\\
      v_{k\neq v_i}, & \text{with probability}\ q_1=\frac{1-p_1}{c_j-1}  \textrm{,}\\
    \end{cases}
  \end{equation*}

  in which $p_1$ and $q_1$ control the level of longitudinal $\epsilon_{\infty}$-LDP. The value $B'$ shall be reused as the basis for all future reports on the real value $v_i$.
  \item $RR_2[GRR]$: Generate a reporting $B''$ such that
  \begin{equation*}\label{eq:perm}
    B''=
    \begin{cases}
      B', & \text{with probability}\ p_2 \textrm{,}\\
      v_{k\neq B'}, & \text{with probability}\ q_2=\frac{1-p_2}{c_j-1}  \textrm{,}\\
    \end{cases}
  \end{equation*}
  in which $B''$ is the report to be sent to the server.
\end{enumerate}

Visually, Fig.~\ref{fig:tree_l_grr} illustrates the probability tree of the L-GRR protocol. In the first round of sanitization, $RR_1$, our proposed L-GRR applies GRR with $p_1=Pr[ B'=v_i | B=v_i ] =\frac{e^{\epsilon_{\infty}}}{e^{\epsilon_{\infty}}+c_j-1}$ and $q_1=Pr[ B'=v_i | B=v_{k\neq i} ] =\frac{1-p_1}{c_j-1}=\frac{1}{e^{\epsilon_{\infty}}+c_j-1}$ (\colorbox{yellow}{highlighted} in the middle of Fig.~\ref{fig:tree_l_grr}), where $c_j=|A_j|$. As discussed in Section~\ref{ch2:sub_GRR}, this \textit{permanent} memoization satisfies $\epsilon_{\infty}$-LDP since $\frac{p_1}{q_1}=e^{\epsilon_{\infty}}$, which is the upper bound.

\begin{figure}[!ht]
\centering
\tikzstyle{level 1}=[level distance=3cm, sibling distance=2cm]
\tikzstyle{level 2}=[level distance=2.6cm, sibling distance=0.75cm]
\tikzstyle{bag} = [text width=4em, text centered]
\tikzstyle{end} = [circle, minimum width=2pt,fill, inner sep=0pt]
\begin{tikzpicture}[grow=right, sloped]
\node[bag] {$B=v_i$}
child {
node[bag] {$B'=v_{k\neq i}$}
child {
node[end, label=right:
{\colorbox{yellow}{$B''=v_i$}}] {}
edge from parent
node[below] {$q_2$}
}
child {
node[end, label=right:
{B''=$v_{k\neq i}$}] {}
edge from parent
node[above] {$p_2$}
}
edge from parent
node[below] {$q_1$}
}
child {
node[bag] {\colorbox{yellow}{$B'=v_{i}$}}
child {
node[end, label=right:
{$B''=v_{k\neq i}$}] {}
edge from parent
node[below] {$q_2$}
}
child {
node[end, label=right:
{\colorbox{yellow}{$B''=v_i$}}] {}
edge from parent
node[above] {$p_2$}
}
edge from parent
node[above] {$p_1$}
};
\end{tikzpicture}
\begin{tikzpicture}[grow=right, sloped]
\node[bag] {$B=v_{k\neq i}$}
child {
node[bag] {\colorbox{yellow}{$B'=v_{i}$}}
child {
node[end, label=right:
{$B''=v_{k\neq i}$}] {}
edge from parent
node[below] {$q_2$}
}
child {
node[end, label=right:
{\colorbox{yellow}{$B''=v_{i}$}}] {}
edge from parent
node[above] {$p_2$}
}
edge from parent
node[below] {$q_1$}
}
child {
node[bag] {$B'=v_{k\neq i}$}
child {
node[end, label=right:
{\colorbox{yellow}{$B''=v_{i}$}}] {}
edge from parent
node[below] {$q_2$}
}
child {
node[end, label=right:
{$B''=v_{k\neq i}$}] {}
edge from parent
node[above] {$p_2$}
}
edge from parent
node[above] {$p_1$}
};
\end{tikzpicture}
\caption{Probability tree for two rounds of sanitization using GRR (L-GRR).} \label{fig:tree_l_grr}
\end{figure}
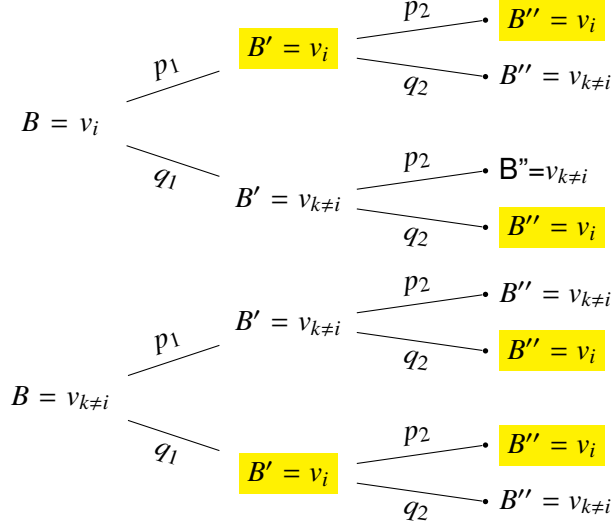

On the other hand, with a single collection of data, the attacker's knowledge of $v_i$ comes only from $B''$, which is generated using two randomization steps with GRR. This provides a higher level of privacy protection~\cite{rappor}. From Fig.~\ref{fig:tree_l_grr}, we can obtain the following conditional probabilities:

\begin{equation*}
    \Pr[ B'' | B ] = 
    \begin{cases}
        \Pr[ B''=v_i | B=v_i ] = p_1 p_2 + q_1 q_2 \\
        \Pr[ B''=v_{k\neq i} | B=v_i ] =  p_1 q_2 + q_1 p_2\\
        \Pr[ B''=v_i | B=v_{k\neq i} ] = p_1 q_2 + q_1 p_2 \\
        \Pr[ B''=v_{k\neq i} | B=v_{k\neq i} ] = p_1 p_2 + q_1 q_2 \\
    \end{cases}
\end{equation*}

Let $p_s=\Pr[ B''=v_i | B=v_i ]$ and $q_s = \Pr[ B''=v_i | B=v_{k\neq i} ]$ (\colorbox{yellow}{highlighted} in far right of Fig.~\ref{fig:tree_l_grr}), with the second round of sanitization, $RR_2[GRR]$, our proposed L-GRR protocol satisfies $\epsilon_1$-LDP since $\frac{p_s}{q_s}=e^{\epsilon_1}$. Notice that $\epsilon_{1}$ corresponds to a single report (lower bound) and its extension to infinity reports is limited by $\epsilon_{\infty}$ (upper bound) since $RR_2[GRR]$ uses as input the output of $RR_1[GRR]$. More specifically, the calculus of $\epsilon_1$ for L-GRR is:

\begin{equation}\label{eq:e1_grr}
     \epsilon_{1} = \ln{  \left ( \frac{p_1 p_2 + q_1 q_2}{p_1 q_2 + q_1 p_2}\right)}  
\end{equation}

in which $p_1=\frac{e^{\epsilon_{\infty}}}{e^{\epsilon_{\infty}}+c_j-1}$, $q_1=\frac{1-p_1}{c_j-1}$, and both $p_2$ and $q_2$ are selectable according with $\epsilon_{\infty}$, $\epsilon_{1}$, and $c_j$, calculated as:

\begin{equation} \label{ch5:eq_p2_lgrr}
\begin{gathered}
    p_2 = \frac{e^{\epsilon_{1} + \epsilon_{\infty}} - 1}{- c_{j} e^{\epsilon_{1}} + \left(c_{j} - 1\right) e^{\epsilon_{\infty}} + e^{\epsilon_{1}} + e^{\epsilon_{1} + \epsilon_{\infty}} - 1}  \textrm{,}\\
    q_2  = \frac{1-p_2}{c_j - 1} \textrm{.}
\end{gathered}
\end{equation}

The estimated frequency $\hat{f}_L(v_i)$ that a value $v_i$ occurs for $i \in [1,c_j]$ is calculated using Eq.~\eqref{eq:est_longitudinal}. Lastly, one can calculate the L-GRR approximate variance by replacing the resulting $p_1,q_1,p_2,q_2$ parameters into Eq.~\eqref{var:aprox_longitudinal}.

\subsection{Longitudinal UE (L-UE): definition and $\epsilon$-LDP study}\label{subsub:l_ue}

We now describe UE-based protocols for longitudinal studies; we refer to this protocol as L-UE for the rest of this chapter. Let $V=\{v_1,v_2,...,v_{c_j}\}$ be a set of $c_j$ values of a given attribute and let $v_i$ be the real value. First, $Encode(v)=B$ (unary encoding), where $B=[0,0,...,1,0,...0]$, a $c_j$-bit array where only the $v$-th position is set to one. Next, there are two rounds of sanitization, $RR_1$ and $RR_2$ applying UE-based protocols, described in the following.

\begin{enumerate}
    \item $RR_1[UE]$: For each bit $i$, $1\le i \le c_j$ in $B$, memoize a value $B'$ such that
    \begin{equation*}
    P(B'_i=1)=
    \begin{cases}
      p_1, & \text{if}\ B_i=1 \textrm{ and}\\
      q_1, & \text{if}\ B_i=0 \textrm{,}
    \end{cases}
  \end{equation*}
  in which $p_1$ and $q_1$ control the level of longitudinal $\epsilon_{\infty}$-LDP. The value $B'$ shall be reused as the basis for all future reports on the real value $v_i$.
  \item $RR_2[UE]$: For each bit $i$, $1\le i \le c_j$ in $B'$, generate a reporting $B''$ that
  \begin{equation*}
    P(B''_i=1)=
    \begin{cases}
      p_2, & \text{if}\ B'_i=1 \textrm{ and}\\
      q_2, & \text{if}\ B'_i=0 \textrm{,}
    \end{cases}
  \end{equation*}
  in which $B''$ is the report to be sent to the server.
\end{enumerate}

Visually, Fig.~\ref{fig:tree_l_ue} illustrates the probability tree of the L-UE protocol. \textbf{One natural question emerges: how to select the parameters $\{p_1,q_1,p_2,q_2\}$ in order to optimize the utility of this L-UE protocol?} One can see $RR_1[UE]$ as a \textit{permanent} sanitization and $RR_2[UE]$ as a `small' perturbation to avoid averaging attacks and keep privacy over time. 

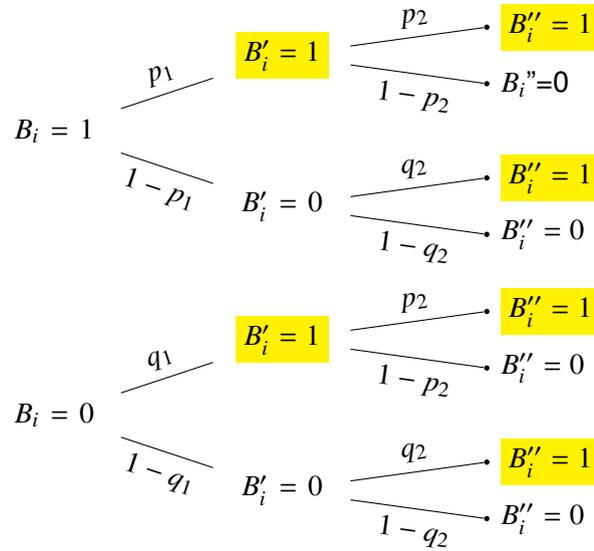
\begin{figure}[!ht]
\centering
\tikzstyle{level 1}=[level distance=3cm, sibling distance=2cm]
\tikzstyle{level 2}=[level distance=2.7cm, sibling distance=0.75cm]
\tikzstyle{bag} = [text width=4em, text centered]
\tikzstyle{end} = [circle, minimum width=2pt,fill, inner sep=0pt]
\begin{tikzpicture}[grow=right, sloped]
\node[bag] {$B_i=1$}
child {
node[bag] {$B_i'=0$}
child {
node[end, label=right:
{$B_i''=0$}] {}
edge from parent
node[below] {$1-q_2$}
}
child {
node[end, label=right:
{\colorbox{yellow}{$B_i''=1$}}] {}
edge from parent
node[above] {$q_2$}
}
edge from parent
node[below] {$1-p_1$}
}
child {
node[bag] {\colorbox{yellow}{$B_i'=1$}}
child {
node[end, label=right:
{$B_i$''=0}] {}
edge from parent
node[below] {$1-p_2$}
}
child {
node[end, label=right:
{\colorbox{yellow}{$B_i''=1$}}] {}
edge from parent
node[above] {$p_2$}
}
edge from parent
node[above] {$p_1$}
};
\end{tikzpicture}
\begin{tikzpicture}[grow=right, sloped]
\node[bag] {$B_i=0$}
child {
node[bag] {$B_i'=0$}
child {
node[end, label=right:
{$B_i''=0$}] {}
edge from parent
node[below] {$1-q_2$}
}
child {
node[end, label=right:
{\colorbox{yellow}{$B_i''=1$}}] {}
edge from parent
node[above] {$q_2$}
}
edge from parent
node[below] {$1-q_1$}
}
child {
node[bag] {\colorbox{yellow}{$B_i'=1$}}
child {
node[end, label=right:
{$B_i''=0$}] {}
edge from parent
node[below] {$1-p_2$}
}
child {
node[end, label=right:
{\colorbox{yellow}{$B_i''=1$}}] {}
edge from parent
node[above] {$p_2$}
}
edge from parent
node[above] {$q_1$}
};
\end{tikzpicture}
\caption{Probability tree for two rounds of sanitization using UE (L-UE).} \label{fig:tree_l_ue}
\end{figure}

Based on SUE and OUE, we are then left with four options: two known solutions that strictly use only OUE or SUE parameters in both sanitization steps and two proposed settings that combine both OUE and SUE. These four L-UE protocols are summarized below: 

\begin{enumerate}[I]
    \item both sanitization steps with OUE (L-OUE);
    \item both sanitization steps with SUE (L-SUE);
    \item starting with OUE and then with SUE (L-OSUE);
    \item starting with SUE and then with OUE (L-SOUE);
\end{enumerate}

in which, L-SUE is the well-known Basic-RAPPOR protocol~\cite{rappor}, L-OUE is the state-of-the-art OUE protocol~\cite{tianhao2017} with memoization, and both L-OSUE and L-SOUE are proposed in this chapter. 

As presented in~\cite{tianhao2017}, the OUE variance in Eq.~\eqref{eq:var_oue} is smaller than the SUE variance in Eq.~\eqref{eq:var_sue} and, therefore, the former can provide higher utility than the latter for $RR_1$. On the other hand, we argue that OUE might be too strict for $RR_2$ since the parameter $p_2=1/2$ is constant. Thus, we hypothesize that option III (i.e., L-OSUE) is the most suitable one. Without loss of generality, \textbf{the following analyses are done only for L-OSUE}, which can be easily extended to any of the other combinations. 

In the first round of sanitization, $RR_1$, our solution L-OSUE applies OUE with $p_1=Pr[ B_i^{'} =1 | B_i=1 ] =\frac{1}{2}$ and $q_1=Pr[ B_i^{'} =1 | B_i=0 ] =\frac{1}{e^{\epsilon_{\infty}}+1}$ (\colorbox{yellow}{highlighted} in the middle of Fig.~\ref{fig:tree_l_ue}). As discussed in Section~\ref{ch2:sub_UE}, this \textit{permanent} memoization satisfies $\epsilon_{\infty}$-LDP since $\frac{p_1(1-q_1)}{(1-p_1)q_1}=e^{\epsilon_{\infty}}$, which is the upper bound. 

Following the same development as for L-GRR, on the other hand, with a single collection of data, the attacker's knowledge of $B=Encode(v)$ comes only from $B''$, which is generated using two randomization steps with OUE and SUE, respectively. This provides a higher level of privacy protection~\cite{rappor}. From Fig.~\ref{fig:tree_l_ue}, we can obtain the following conditional probabilities according to each bit $i \in [1,c_j]$:

\begin{equation*}
    \Pr[ B_i'' | B_i ] = 
    \begin{cases}
        \Pr[ B_i''= 1| B_i = 1] = p_1 p_2 + (1 - p_1) q_2 \\
        \Pr[ B_i''= 0| B_i = 1] = p_1 (1 - p_2) + (1 - p_1) (1 - q_2) \\
        \Pr[ B_i''= 1| B_i = 0] = q_1 p_2 + (1 - q_1) q_2 \\
        \Pr[ B_i''= 0| B_i = 0] = q_1 (1 - p_2) + (1 - q_1) (1 - q_2) \\
    \end{cases}
\end{equation*}

Let $p_s=\Pr[ B_i'' =1 | B_i =1]$ and $q_s = \Pr[ B_i'' =1 | B_i =0 ]$ (\colorbox{yellow}{highlighted} in far right of Fig.~\ref{fig:tree_l_ue}), with the second round of sanitization, $RR_2[SUE]$, our proposed L-OSUE protocol satisfies $\epsilon_1$-LDP since $\frac{p_s(1-q_s)}{(1-p_s)q_s}=e^{\epsilon_{1}}$. Notice that $\epsilon_{1}$ corresponds to a single report (lower bound) and its extension to infinity reports is limited by $\epsilon_{\infty}$ (upper bound) since $RR_2[SUE]$ uses as input the output of $RR_1[OUE]$. More specifically, the calculus of $\epsilon_1$ for L-OSUE (or L-UE protocols in general) is:

\begin{equation}\label{eq:e1_ue}
    \epsilon_{1} = \ln{  \left (  \frac{\left(p_{1} p_{2} - q_{2} \left(p_{1} - 1\right)\right) \left(p_{2} q_{1} - q_{2} \left(q_{1} - 1\right) - 1\right)}{\left(p_{2} q_{1} - q_{2} \left(q_{1} - 1\right)\right) \left(p_{1} p_{2} - q_{2} \left(p_{1} - 1\right) - 1\right)}  \right)} \textrm{,}
\end{equation}

in which, for L-OSUE, we have $p_1=\frac{1}{2}$, $q_1=\frac{1}{e^{\epsilon_{\infty}}+1}$, and both $p_2$ and $q_2$ are symmetric ($p_2+q_2 = 1$) and selectable according to $\epsilon_{\infty}$ and $\epsilon_1$, calculated as: 

\begin{equation}  \label{ch5:eq_p2_losue}
    \begin{gathered}
    p_2 = \frac{1 - e^{\epsilon_{1} + \epsilon_{\infty}}}{e^{\epsilon_{1}} - e^{\epsilon_{\infty}} - e^{\epsilon_{1} + \epsilon_{\infty}} + 1} \textrm{,}\\
    q_2 = 1 - p_2 \textrm{.}
    \end{gathered}
\end{equation}

Similarly, the estimated frequency $\hat{f}_L(v_i)$ that a value $v_i$ occurs for $i \in [1,c_j]$ is calculated using Eq.~\eqref{eq:est_longitudinal}. Lastly, one can calculate the L-OSUE (or L-UE protocols in general) approximate variance by replacing the resulting $p_1,q_1,p_2,q_2$ parameters into Eq.~\eqref{var:aprox_longitudinal}.

\subsection{Numerical evaluation of L-GRR and L-UE protocols}\label{sub:analysis_long}

In this subsection, we evaluate numerically the approximate variance of all developed longitudinal protocols, namely, L-GRR and the four UE-based options namely L-OUE, L-SUE, L-OSUE, and L-SOUE, respectively. As aforementioned, once defined both $\epsilon_{\infty}$ and $\epsilon_1$ privacy guarantees, one can obtain the parameters $p_1$ and $q_1$ depending on $\epsilon_{\infty}$, and the parameters $p_2$ and $q_2$ depending on both $\epsilon_{\infty}$ and $\epsilon_1$ (and the domain size $c_j$ for L-GRR) as given in Eq.~\eqref{ch5:eq_p2_lgrr} for L-GRR and in Eq.~\eqref{ch5:eq_p2_losue} for L-OSUE. 

Next, once computed the parameters $\{p_1,q_1,p_2,q_2\}$, one can calculate the approximate variance with Eq.~\eqref{var:aprox_longitudinal} for each protocol. In other words, following our proposal, one has to set both the upper ($\epsilon_{\infty}$) and lower ($\epsilon_1$) bounds of the privacy guarantees. For example, let $\epsilon_{\infty} = 2$, one might want that the first $\epsilon_1$-LDP report to have high privacy such as $\epsilon_1=0.1$, i.e., $\epsilon_1=0.05\epsilon_{\infty}$ (\textbf{we will use this percentage notation to set up the privacy guarantees}).

Table~\ref{tab:analysis_var} exhibits numerical values of the approximate variance using Eq.~\eqref{var:aprox_longitudinal} for all longitudinal protocols with $n=10000$, $\epsilon_{\infty}=[0.5, 1.0, 2.0, 4.0]$ (as in~\cite{tianhao2017}), and $\epsilon_1 = \{0.6\epsilon_{\infty},0.5\epsilon_{\infty},0.4\epsilon_{\infty},0.3\epsilon_{\infty},0.2\epsilon_{\infty},0.1\epsilon_{\infty}\}$. For values of $\epsilon_1$ higher than $0.6\epsilon_{\infty}$, neither L-OUE nor L-SOUE could satisfy some values of $\epsilon_1$ because of the constant $p_2=1/2$ in $RR_2$. Yet, it is not desirable to have higher values of $\epsilon_1$ and, thus, we did not consider values above $0.6\epsilon_{\infty}$ in our analysis. Besides, Table~\ref{tab:analysis_var_non_long} exhibits numerical values for non-longitudinal GRR, OUE, and SUE protocols, which allows evaluating how utility degrades with a second step of sanitization. 

\setlength{\tabcolsep}{5pt}
\renewcommand{\arraystretch}{1.4}
\begin{table}[!ht]
    \scriptsize
    \centering
    \caption{Numerical values of Eq.~\eqref{var:aprox_longitudinal} (i.e., $Var^*[\hat{f}_L(v_i)]$) for L-GRR and L-UE protocols with different $\epsilon_{\infty}$ and $\epsilon_1$ privacy guarantees, following $\epsilon_1 = \{0.6\epsilon_{\infty},0.5\epsilon_{\infty},0.4\epsilon_{\infty},0.3\epsilon_{\infty},0.2\epsilon_{\infty},0.1\epsilon_{\infty}\}$, respectively.}
    \begin{tabular}{c| c| c| c| c| c| c| c| c}\hline
    \multirow{2}{*}{$\epsilon_{1}$} &  \multirow{2}{*}{Privacy Guarantees} &  \multicolumn{3}{c|}{L-GRR} &  \multicolumn{4}{c}{L-UE}  \\ \cline{3-9}
    & & $c_j=2$ & $c_j=32$ & $c_j=2^{10}$ & L-OSUE & L-SUE & L-SOUE & L-OUE\\ \hline
     
     \multirow{4}{*}{$0.6\epsilon_{\infty}$} &$\epsilon_{\infty}=0.5,\epsilon_{1}=0.30$ & 0.001103 &     0.980969 &       26706 &  0.004411 &  0.004436 &  0.005306 &  0.005549 \\
     & $\epsilon_{\infty}=1.0,\epsilon_{1}=0.60$ &    0.000270 &     0.125036 &        3153 &  0.001078 &  0.001103 &  0.001234 &  0.001347 \\
     & $\epsilon_{\infty}=2.0,\epsilon_{1}=1.20$ & 0.000062 &     0.006327 &         117 &  0.000247 &  0.000270 &  0.000264 &  0.000310 \\    
     & $\epsilon_{\infty}=4.0,\epsilon_{1}=2.40$ &    0.000011 &     0.000078 &           0.25903 &  0.000044 &  0.000062 &  0.000045 &  0.000057 \\ \hline
     
     \multirow{4}{*}{$0.5\epsilon_{\infty}$} &$\epsilon_{\infty}=0.5,\epsilon_{1}=0.25$ & 0.001592 &     2.088372 &       60218 &  0.006367 &  0.006392 &  0.007336 &  0.007611 \\ 
     &$\epsilon_{\infty}=1.0,\epsilon_{1}=0.50$ & \textbf{0.000392} &     0.268074 &        7198 &  \textbf{0.001567} &  \textbf{0.001592} &  0.001740 &  0.001872   \\ 
     &$\epsilon_{\infty}=2.0,\epsilon_{1}=1.00$ & \textbf{0.000092} &     0.013926 &         281 &  \textbf{0.000368} &  \textbf{0.000392} &  0.000389 &  0.000447   \\ 
     &$\epsilon_{\infty}=4.0,\epsilon_{1}=2.00$ & \textbf{0.000018} &     0.000188 &           0.74088 &  \textbf{0.000072} &  \textbf{0.000092} &  0.000073 &  0.000092   \\ \hline
     
     \multirow{4}{*}{$0.4\epsilon_{\infty}$} &$\epsilon_{\infty}=0.5,\epsilon_{1}=0.20$ &  0.002492 &     4.530779 &      135874 &  0.009967 &  0.009992 &  0.011012 &  0.011324 \\
     &$\epsilon_{\infty}=1.0,\epsilon_{1}=0.40$ &  0.000617 &     0.586823 &       16443 &  0.002467 &  0.002492 &  0.002658 &  0.002812\\ 
     &$\epsilon_{\infty}=2.0,\epsilon_{1}=0.80$ &  0.000148 &     0.031552 &         673 &  0.000593 &  0.000617 &  0.000617 &  0.000690  \\ 
     &$\epsilon_{\infty}=4.0,\epsilon_{1}=1.60$ &  0.000032 &     0.000484 &           2.12772 &  0.000127 &  0.000148 &  0.000128 &  0.000156  \\ \hline
    
     \multirow{4}{*}{$0.3\epsilon_{\infty}$} &$\epsilon_{\infty}=0.5,\epsilon_{1}=0.15$ & 0.004436 &    10 &      329836 &  0.017744 &  0.017769 &  0.018863 &  0.019214 \\
     &$\epsilon_{\infty}=1.0,\epsilon_{1}=0.30$ & 0.001103 &     1.398568 &       40412 &  0.004411 &  0.004436 &  0.004620 &  0.004799 \\
     &$\epsilon_{\infty}=1.0,\epsilon_{1}=0.60$ & 0.000270 &     0.078202 &        1737 &  0.001078 &  0.001103 &  0.001106 &  0.001198 \\
     &$\epsilon_{\infty}=2.0,\epsilon_{1}=1.20$ & 0.000062 &     0.001389 &           6 &  0.000247 &  0.000270 &  0.000248 &  0.000291 \\ \hline
          
     \multirow{4}{*}{$0.2\epsilon_{\infty}$} &$\epsilon_{\infty}=0.5,\epsilon_{1}=0.10$ &  0.009992 &    30 &      972656 &  0.039967 &  0.039992 &  0.041148 &  0.041536 \\
     &$\epsilon_{\infty}=1.0,\epsilon_{1}=0.20$ & 0.002492 &     4.080052 &      120651 &  0.009967 &  0.009992 &  0.010190 &  0.010394  \\ 
     &$\epsilon_{\infty}=2.0,\epsilon_{1}=0.40$ & 0.000617 &     0.237925 &        5443 &  0.002467 &  0.002492 &  0.002498 &  0.002610   \\ 
     &$\epsilon_{\infty}=4.0,\epsilon_{1}=0.80$ & 0.000148 &     0.004939 &          24 &  0.000593 &  0.000617 &  0.000595 &  0.000659   \\ \hline
     
     \multirow{4}{*}{$0.1\epsilon_{\infty}$} & $\epsilon_{\infty}=0.5,\epsilon_{1}=0.05$ & 0.039992 &   154 &     4941829 &  0.159967 &  0.159992 &  0.161191 &  0.161608 \\ 
     & $\epsilon_{\infty}=1.0,\epsilon_{1}=0.10$ & 0.009992 &    20 &      620584 &  0.039967 &  0.039992 &  0.040201 &  0.040424 \\ 
     & $\epsilon_{\infty}=2.0,\epsilon_{1}=0.20$ & 0.002492 &     1.255550 &       29356 &  0.009967 &  0.009992 &  0.010000 &  0.010130  \\ 
     & $\epsilon_{\infty}=4.0,\epsilon_{1}=0.40$ & 0.000617 &     0.030494 &         156 &  0.002467 &  0.002492 &  0.002469 &  0.002560  \\ \hline
    \end{tabular}
    \label{tab:analysis_var}
\end{table}

\setlength{\tabcolsep}{5pt}
\renewcommand{\arraystretch}{1.4}
\begin{table}[!ht]
    \scriptsize
    \centering
    \caption{Numerical values of Eq.~\ref{eq:var} (i.e., $Var^*[\hat{f}(v_i)]$) for the non-longitudinal GRR, OUE, and SUE protocols with different $\epsilon_{\infty}$ privacy guarantees.}
    \begin{tabular}{c | c| c| c| c| c}\hline
    $\epsilon_{\infty}$ &GRR($c_j=2$)  &GRR($c_j=32$)  &GRR($c_j=2^{10}$)  &OUE  &SUE \\ \hline
    $\epsilon_{\infty}=0.5$    &\textbf{0.000392} &     0.007520 &           0.243240 &  \textbf{0.001567} &  \textbf{0.001592} \\ \hline
    $\epsilon_{\infty}=1.0$    &\textbf{0.000092} &     0.001108 &           0.034707 &  \textbf{0.000368} &  \textbf{0.000392} \\ \hline
    $\epsilon_{\infty}=2.0$    &\textbf{0.000018} &     0.000092 &           0.002522 &  \textbf{0.000072} &  \textbf{0.000092} \\ \hline
    $\epsilon_{\infty}=4.0$    &0.000002 &     0.000003 &           0.000037 &  0.000008 &  0.000018  \\ \hline
    \end{tabular}
    \label{tab:analysis_var_non_long}
\end{table}

\textbf{From Table~\ref{tab:analysis_var}, one can notice that L-GRR presents the smallest variance values for binary attributes (i.e., when $c_j=2$).} On the other hand, L-GRR is also the most sensitive to change in privacy parameters $\epsilon_{\infty}$ and $\epsilon_1$ when $c_j$ is large, which leads to much higher variance than when using a non-longitudinal GRR in Table~\ref{tab:analysis_var_non_long}. Similar to non-longitudinal GRR, this increase in the variance is due to the number of values $c_j$, which decreases the probability $p$ of reporting the true value. With two rounds of sanitization, it further deteriorates the accuracy of the L-GRR protocol getting to extremely high values, e.g., see L-GRR$(c_j=2^{10})$. Interestingly, when $c_j=2$ in Table~\ref{tab:analysis_var}, the variance of L-GRR with $\epsilon_1=0.5\epsilon_{\infty}$ is a lagged version of the variance values given by the non-longitudinal GRR in Table~\ref{tab:analysis_var_non_long}. This effect is also observed for both L-SUE (cf. SUE in Table~\ref{tab:analysis_var_non_long}) and L-OSUE (cf. OUE in Table~\ref{tab:analysis_var_non_long}) protocols, which use symmetric probabilities on $RR_2$ (i.e., $p_2+q_2=1$). We highlighted these values in \textbf{bold font}. However, for L-GRR, this is not true for other values of $c_j$, whose further analysis is beyond the scope of this chapter.

On the other hand, L-UE protocols avoid having a variance that depends on $c_j$ by encoding the value into the unary representation, which results in a constant variance no matter the size of the attribute. To complement the results of Table~\ref{tab:analysis_var}, Fig.~\ref{fig:analysis_var} illustrates numerical values of the approximate variance for L-UE protocols with $\epsilon_1=\{0.3\epsilon_{\infty}, 0.6\epsilon_{\infty}\}$. With the four options I-IV analyzed, on high privacy regimes, L-OSUE and L-SUE have similar performance while \textit{always} favoring the proposed L-OSUE one. On lower privacy regimes, our proposed protocols L-SOUE and L-OSUE have similar performance, which outperform both L-OUE and L-SUE protocols. As shown in our experiments, the L-OUE protocol has the worst performance among the four options analyzed, with the exception of high values for $\epsilon_{\infty}$ (see the plot on the bottom of Fig.~\ref{fig:analysis_var}), when it has performance superior or similar to L-SUE. Indeed, for L-OUE, selecting $p_2=1/2$ for the second sanitization step is too strict, which results in higher variance value. \textbf{Therefore, by comparing the approximate variances, the best option for L-UE protocols, in terms of utility, is starting with OUE and then with SUE as we propose in this chapter, i.e., L-OSUE.}

\begin{figure}[!ht]
    \centering
    \includegraphics[width=0.7\linewidth]{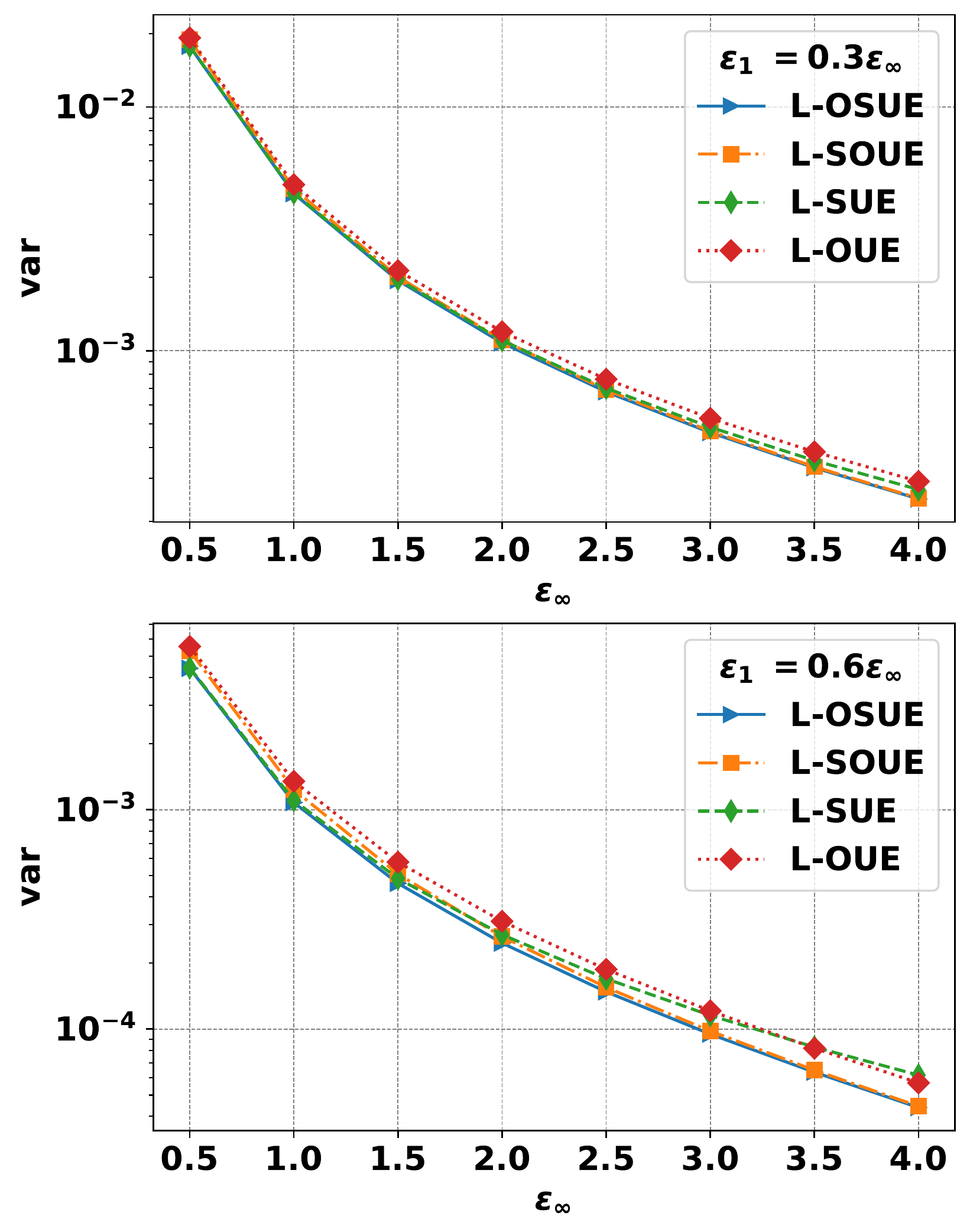}
    \caption{Numerical values of $Var^*[\hat{f}_L(v_i)]$ for L-UE protocols with $\epsilon_{1}=0.3\epsilon_{\infty}$ (plot on the top) and with $\epsilon_{1}=0.6\epsilon_{\infty}$ (plot on the bottom).}
    \label{fig:analysis_var}
\end{figure}

\subsection{The \textit{ALLOMFREE} algorithm}\label{sub:allomfree}

Let $A=\{A_1,A_2,...,A_d\}$ be a set of $d$ attributes with domain size $\textbf{c}=[c_1,c_2,...,c_d]$, $\mathbb{A}=\{\textit{L-GRR},\textit{L-OSUE}\}$ be a set of optimal longitudinal LDP protocols, and $\epsilon_{\infty}$ and $\epsilon_1$ be the longitudinal and \textit{single-report} privacy guarantees, respectively. Each user $u_i$, for $1 \leq i \leq n$, holds a tuple $\textbf{v}^{(i)}=(v^{(i)}_{1},v^{(i)}_{2},...,v^{(i)}_{d})$, i.e., a private value per attribute. From now on, we will simply omit the index notation $\textbf{v}^{(i)}$ and use $\textbf{v}$ in the analysis as we focus on one arbitrary user $u_i$ here. For each attribute $j \in [1,d]$ (we slightly abuse the notation and use $j$ for $A_j$) at time $t \in [1,\tau]$, the aggregator aims to estimate the frequencies of each value $v \in A_j$. 

\textbf{Client-Side.} In a multidimensional setting with different domain sizes for each attribute, a dynamic selection of longitudinal LDP protocols is preferred. As mentioned in Section~\ref{ch5:sec_multidimensional}, we propose that each user randomly sample $r=Uniform(1,2,...,d)$ to select a single attribute $A_r$. Given $c_r$ (the domain size), $\epsilon_{\infty}$, and $\epsilon_1$, one calculates the parameters $fp_{L-GRR}=\{p_1,q_1,p_2,q_2\}$ and $fp_{L-OSUE}=\{p_1,q_1,p_2,q_2\}$, for L-GRR and L-OSUE, respectively (cf. Eq.~\eqref{ch5:eq_p2_lgrr} and Eq.~\eqref{ch5:eq_p2_losue}). Next, with $fp_{L-GRR}$ and $fp_{L-OSUE}$, one calculates the approximate variances $Var^*[\hat{f}_{L_{(\textit{L-GRR})}}]$ for L-GRR and $Var^*[\hat{f}_{L_{(\textit{L-OSUE})}}]$ for L-OSUE with Eq.~\eqref{var:aprox_longitudinal}. Lastly, to select L-GRR as the local randomizer, we are then left to evaluate if $Var^*[\hat{f}_{L_{(\textit{L-GRR})}}] \leq Var^*[\hat{f}_{L_{(\textit{L-OSUE})}}]$. Therefore, the first round of sanitization ensures a \textit{permanent memoization} $B'$ that is always used for the second round of sanitization to generate $B''$ each time $t \in [1,\tau]$ the user will report the real value $B$. We call our solution \underline{A}daptive \underline{L}DP for \underline{LO}ngitudinal and \underline{M}ultidimensional \underline{FRE}quency \underline{E}stimates (ALLOMFREE), which is summarized in Algorithm~\ref{alg:allomfree} as a pseudocode. 

\begin{algorithm}[!ht]
\caption{User-side algorithm of ALLOMFREE.}
\label{alg:allomfree}
\begin{algorithmic}[1]
\State \textbf{Input :} $\textbf{v} = [v_1,v_2,..., v_d]$, $\textbf{c}=[c_1,c_2,...,c_d]$, $\mathbb{A}=\{\textit{L-GRR},\textit{L-OSUE}\}$, $\epsilon_{\infty}$, $\epsilon_1$, number of reports $\tau$. 

\State $r \gets Uniform(\{1,2,...,d \})$ \Comment{Select attribute only once}

\State $B \gets v_r$

\State $fp_{L-GRR} \gets p_1=\frac{e^{\epsilon_{\infty}}}{e^{\epsilon_{\infty}}+k_r-1},q_1=\frac{1-p_1}{k_r-1}, p_2 = \frac{e^{\epsilon_{1} + \epsilon_{\infty}} - 1}{- k_r e^{\epsilon_{1}} + \left(k_r - 1\right) e^{\epsilon_{\infty}} + e^{\epsilon_{1}} + e^{\epsilon_{1} + \epsilon_{\infty}} - 1}, q_2  = \frac{1-p_2}{k_r - 1}$ \Comment{Get $p_2$ and $q_2$ with Eq.~\eqref{ch5:eq_p2_lgrr}.}

\State $fp_{L-OSUE} \gets p_1=\frac{1}{2}, q_1=\frac{1}{e^{\epsilon_{\infty}}+1}, p_2 = \frac{1 - e^{\epsilon_{1} + \epsilon_{\infty}}}{e^{\epsilon_{1}} - e^{\epsilon_{\infty}} - e^{\epsilon_{1} + \epsilon_{\infty}} + 1}, q_2 = 1 - p_2$ \Comment{Get $p_2$ and $q_2$ with Eq.~\eqref{ch5:eq_p2_losue}.}

\State \textbf{if}  $Var^*[\hat{f}_{L_{(\textit{L-GRR})}}](fp_{L-GRR}) \leq Var^*[\hat{f}_{L_{(\textit{L-OSUE})}}](fp_{L-OSUE})$ : \Comment{Check variances with Eq.~\eqref{var:aprox_longitudinal}})

\State  \hskip1em $\mathcal{A} \gets \textrm{L-GRR}$ \Comment{Select L-GRR as local randomizer}

\State \textbf{else}

\State  \hskip1em $\mathcal{A} \gets \textrm{L-OSUE}$ \Comment{Select L-OSUE as local randomizer}

\State $B' \gets \mathcal{A}(B, \epsilon_{\infty}, c_r)$ \Comment{First round of sanitization (permanent memoization)}

\State \textbf{for} $t \in [1,\tau]$ \textbf{do}

\State  \hskip1em $B''= \mathcal{A}(B',\epsilon_{1}, c_r)$ \Comment{Second round of sanitization}

\State \textbf{end for}

\State  \textbf{send :} $(t,\langle r, B''\rangle)$ for $t \in [1,\tau]$ 

\end{algorithmic}
\end{algorithm}

The intuition of ALLOMFREE is as follows. By requiring each user to submit only 1 attribute with the whole privacy budget, it reduces both the variance incurred as well as the communication cost. Also, since we developed the calculus of the approximate variance in Eq.~\eqref{var:aprox_longitudinal} for the proposed longitudinal protocols (L-GRR and L-OSUE), ALLOMFREE can adaptively select the protocol with a smaller variance value to optimize the data utility. Therefore, ALLOMFREE utilizes optimal solutions for both multidimensional and longitudinal data collection settings developed in Sections~\ref{ch5:sec_multidimensional} and~\ref{ch5:sec_longitudinal} of this manuscript, respectively.

\textbf{Server-Side.} On the server-side, for each attribute $j\in[1,d]$ at time $t \in [1,\tau]$, the estimated frequency $\hat{f}_L(v_i)$ that a value $v_i$ occurs for $i \in [1,c_j]$ is calculated using Eq.~\eqref{eq:est_longitudinal}.

\textbf{Privacy analysis.} On the one hand, according to the analysis in Subsections~\ref{subsub:l_de} and \ref{subsub:l_ue}, Alg.~\ref{alg:allomfree} satisfies $\epsilon$-LDP with upper $\epsilon_{\infty}$ (infinity reports) and lower $\epsilon_1$ (a single report) bounds as it uses either L-GRR or L-OSUE to sanitize a single attribute per user. \textbf{Notice that, to ensure users' privacy over time and to avoid the sequential composition theorem~\cite{dwork2014algorithmic}, each user must always report the same unique attribute $A_r$}. In addition, the privacy of a user decreases gracefully according to the number of LDP reports $t \leq \tau $ that an adversary has gained access to, which is calculated as~\cite{Naor2020,erlingsson2020encode}: 

\begin{equation}\label{eq:eps_long}
\epsilon_t=\ln{\left (\frac{e^{\epsilon_{{\infty}}+t\epsilon_{1}} + 1}{e^{\epsilon_{{\infty}}}+e^{t\epsilon_{1}}} \right)} \leq \min \{ \epsilon_{{\infty}}, t\epsilon_{1}\}    \textrm{.}
\end{equation} 

\section{Results and Discussion} \label{ch5:sec_results_discussion}

In this section, we present the setup of our experiments in Section~\ref{sub:setup_allomfree}, the results with real-world data in Section~\ref{sub:results_allomfree}, and a general discussion in Section~\ref{ch5:discussion_allomfree} with related work and limitations.

\subsection{Setup of experiments} \label{sub:setup_allomfree}

The main goal of our experiments is to evaluate the proposed longitudinal LDP protocols on multidimensional frequency estimates a single time, i.e., satisfying $\epsilon_1$-LDP (as in~\cite{rappor,Vidal2020,Kim2018}, for example).

\textbf{Environment.} All algorithms were implemented in Python 3.8.8 with NumPy 1.19.5 and Numba 0.53.1 libraries. The codes we developed and used for all experiments are available in a Github repository\footnote{\url{https://github.com/hharcolezi/ldp-protocols-mobility-cdrs}.}. In all experiments, we report average results over 100 runs as LDP algorithms are randomized. 

\textbf{Methods evaluated.} We consider for evaluation the following solutions and protocols: 

\begin{itemize}

    \item Solution \textit{Smp} (cf. Section~\ref{ch5:sec_multidimensional}), which randomly samples a single attribute to send with the whole privacy budget. We will experiment with the state-of-the-art protocols, namely, L-SUE and L-OUE, and with our extended protocols L-OSUE and L-SOUE; 
    
    \item Our ALLOMFREE solution (cf. Alg.~\ref{alg:allomfree}), which also randomly samples a single attribute to send with the whole privacy budget but adaptively select the optimal protocol, i.e., either L-GRR or L-OSUE.

\end{itemize}

\textbf{Experimental evaluation and metrics.} We vary the longitudinal privacy parameter in the range $\epsilon_{\infty}=[0.5, 1, ..., 3.5, 4]$ with $\epsilon_1 = \{0.3\epsilon_{\infty}, 0.6 \epsilon_{\infty}\}$ to compare our experimental results with numerical ones from Section~\ref{sub:analysis_long}. Notice that this range of privacy guarantees is commonly used in the literature for multidimensional data (e.g., in~\cite{wang2019} the range is $\epsilon=[0.5,...,4]$ and in~\cite{Wang2021_b} the range is $\epsilon=[0.1,...,10]$). 

Since the estimator in Eq.~\eqref{eq:est_longitudinal} is unbiased (cf. Theorem~\ref{theo:est_long}), the variance of our protocols is equal to the MSE that is commonly used in practice as an accuracy metric~\cite{Wang2020,Wang2020_post_process,Wang2021_b,li2021privacy} (cf. Eq.~\eqref{eq:mse_var}). So, to evaluate our results, we use the MSE metric averaged per the number of attributes $d$ \textbf{in a single data collection $\tau=1$, i.e., with $\epsilon_1$-LDP}. Thus, for each attribute $j$, we compute for each value $v_i \in A_j$ the estimated frequency $\hat{f}(v_i)$ and the real one $f(v_i)$ and calculate their differences. More precisely,

\begin{equation}
    MSE_{avg} = \frac{1}{\tau} \sum_{t \in [1,\tau]} \frac{1}{d} \sum_{j \in [1,d]} \frac{1}{|A_j|} \sum_{v_i \in A_j}(f(v_i) - \hat{f}(v_i) )^2 \textrm{.}
\end{equation}

\textbf{Datasets.} For ease of reproducibility, we conduct our experiments on four multidimensional open datasets. We briefly recall here the datasets from Section~\ref{ch3:sub_open_datasets} and the generated one in Chapter~\ref{chap:chapter4}.

\begin{itemize}
    \item \textit{Nursery.} A dataset from the UCI machine learning repository~\cite{uci} with $d=9$ categorical attributes and $n=12960$ samples. The domain size of each attribute is $\textbf{c}=[3, 5, 4, 4, 3, 2, 3, 3, 5]$, respectively. 
    
    \item \textit{Adult.} A dataset from the UCI machine learning repository~\cite{uci} with $d=9$ categorical attributes and $n=45222$ samples after cleaning the data. The domain size of each attribute is $\textbf{c}=[7, 16, 7, 14, 6, 5, 2, 41, 2]$, respectively. 
    
    \item \textit{MS-FIMU.} The dataset developed in Chapter~\ref{chap:chapter4} in which we select $d=6$ categorical attributes (all static attributes, i.e., the dynamic `Visit duration' attribute was not used). The domain size of each attribute is $\textbf{c}=[3, 3, 8, 12, 37, 11]$ (cf. Section~\ref{ch4:info_ms_fimu}), respectively, and there are $n=88935$ samples.
    
    \item \textit{Census-Income.} A dataset from the UCI machine learning repository~\cite{uci} with $d=33$ categorical attributes and $n=299285$ samples. The domain size of each attribute is \begin{math}\textbf{c}=[ 9, 52, 47, 17,  3,  7, 24, ..., 43,  5,  3,  3,  3,  2]\end{math}, respectively. 

\end{itemize}

\subsection{Results}  \label{sub:results_allomfree}

Our experiments were conducted on four real-world datasets with varied parameters for $n$, $d$, and $\textbf{c}$, which allowed evaluating our solutions more practically. Fig.~\ref{fig:results_nursery_allomfree} (\textit{Nursery}), Fig.~\ref{fig:results_adults_allomfree} (\textit{Adult}), Fig.~\ref{fig:results_vhs_allomfree} (\textit{MS-FIMU}), and Fig.~\ref{fig:results_census_allomfree} (\textit{Census-Income}) illustrate for all evaluated protocols, averaged $MSE_{avg}$ (y-axis) according to the longitudinal privacy parameter $\epsilon_{\infty}$ (x-axis) with $\epsilon_1 = 0.3\epsilon_{\infty}$ (right-side plot) and with $\epsilon_1 = 0.6\epsilon_{\infty}$ (left-side plot), respectively.

As one can notice in the results, for all datasets, ALLOMFREE consistently and considerably outperforms the state-of-the-art protocols, namely, L-SUE (a.k.a. Basic-RAPPOR)~\cite{rappor} and L-OUE (that uses OUE~\cite{tianhao2017} twice). Indeed, the difference on performance between ALLOMFREE and the other longitudinal LDP protocols increases according to the privacy guarantees, i.e., for high $\epsilon_{\infty}$ and $\epsilon_1$ values the gap is bigger. This is, first, because in all datasets there are attribute(s) with small domain size (e.g., $c_j=2$ or $c_j=3$), in which L-GRR can provide smaller variance values than L-UE protocols (cf. Section~\ref{sub:analysis_long}). Secondly, by selecting adequately the probabilities $p_1,q_1,p_2,q_2$ for the L-UE protocol (i.e., L-OSUE) also optimizes data utility. Thus, since there is a way to measure the approximate variance of the extended protocols (i.e., Eq.~\eqref{var:aprox_longitudinal}), given the sampled attribute, ALLOMFREE adaptively selects one of the optimized protocol (i.e., L-GRR or L-OSUE) whose smaller variance improves the data utility.

In addition, among the L-UE protocols applied individually, the experimental results with multidimensional data approximate the numerical results with a single attribute from Section~\ref{sub:analysis_long}. For instance, the proposed L-OSUE provides similar or improved performance than L-SUE while always outperforming L-OUE. Besides, L-SOUE always outperforms L-OUE too, achieving similar performance than L-OSUE and L-SUE in low privacy regimes (i.e., high $\epsilon$ values). As we have already shown in Section~\ref{sub:analysis_long}, even though OUE has higher utility than SUE for one-time collection~\cite{tianhao2017}, applying OUE twice does not provide higher utility.

\begin{figure}[H]
    \centering
    \includegraphics[width=1\linewidth]{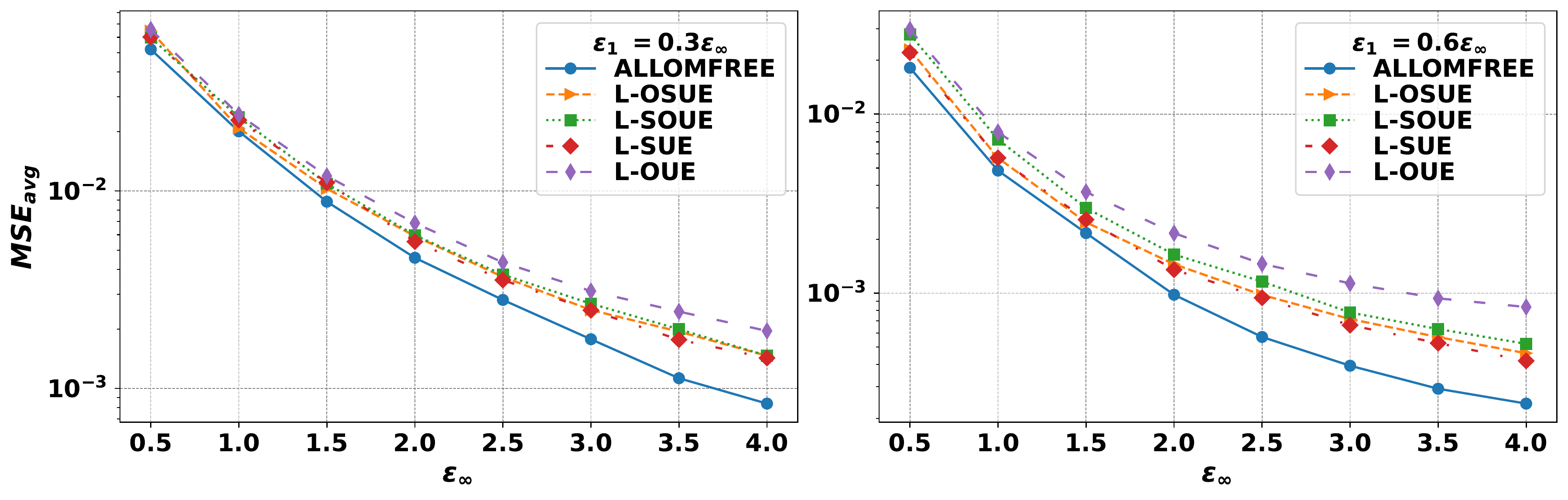}\\
    \caption{Averaged MSE varying $\epsilon_{\infty}$ with $\epsilon_1 = 0.3\epsilon_{\infty}$ (left-side plot) and with $\epsilon_1 = 0.6\epsilon_{\infty}$ (right-side plot) on the \textit{Nursery} dataset.}
    \label{fig:results_nursery_allomfree}
\end{figure}

\begin{figure}[H]
    \centering
    \includegraphics[width=1\linewidth]{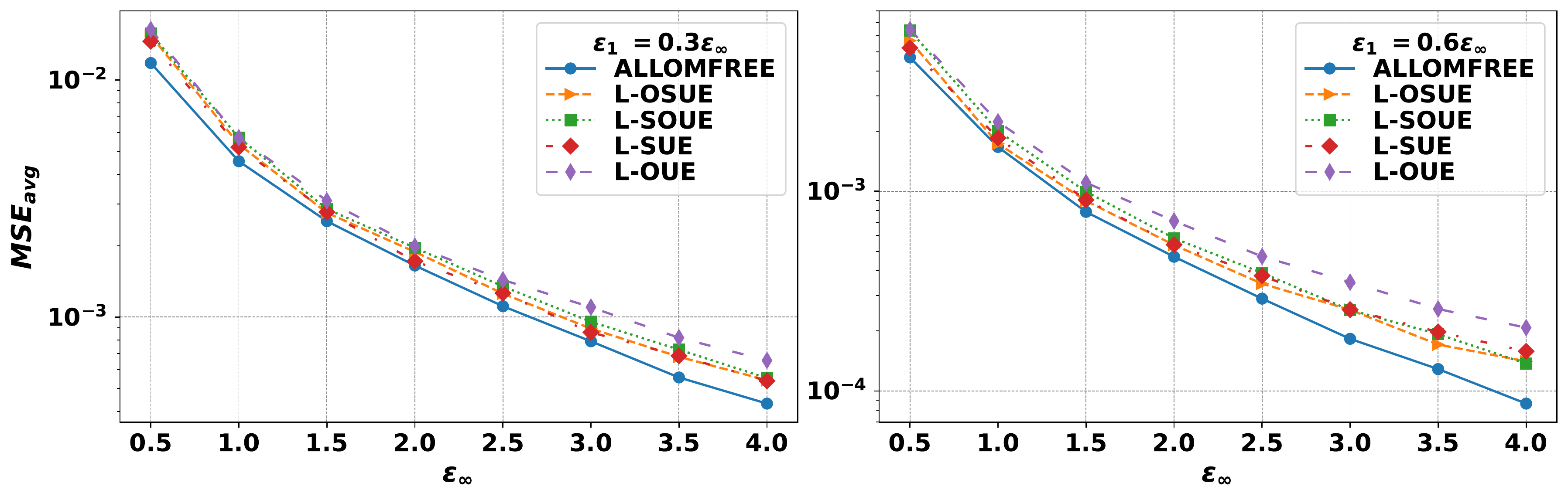}\\
    \caption{Averaged MSE varying $\epsilon_{\infty}$ with $\epsilon_1 = 0.3\epsilon_{\infty}$ (left-side plot) and with $\epsilon_1 = 0.6\epsilon_{\infty}$ (right-side plot) on the \textit{Adult} dataset.}
    \label{fig:results_adults_allomfree}
\end{figure}

\begin{figure}[H]
    \centering
    \includegraphics[width=1\linewidth]{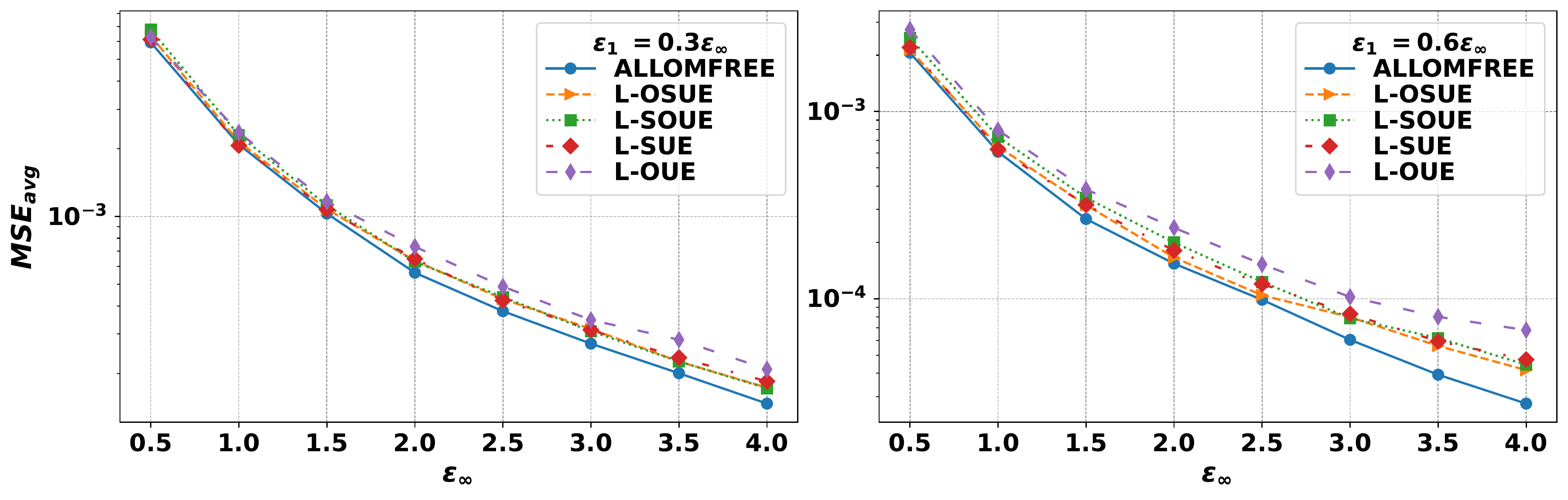}\\
    \caption{Averaged MSE varying $\epsilon_{\infty}$ with $\epsilon_1 = 0.3\epsilon_{\infty}$ (left-side plot) and with $\epsilon_1 = 0.6\epsilon_{\infty}$ (right-side plot) on the \textit{MS-FIMU} dataset.}
    \label{fig:results_vhs_allomfree}
\end{figure}

\begin{figure}[H]
    \centering
    \includegraphics[width=1\linewidth]{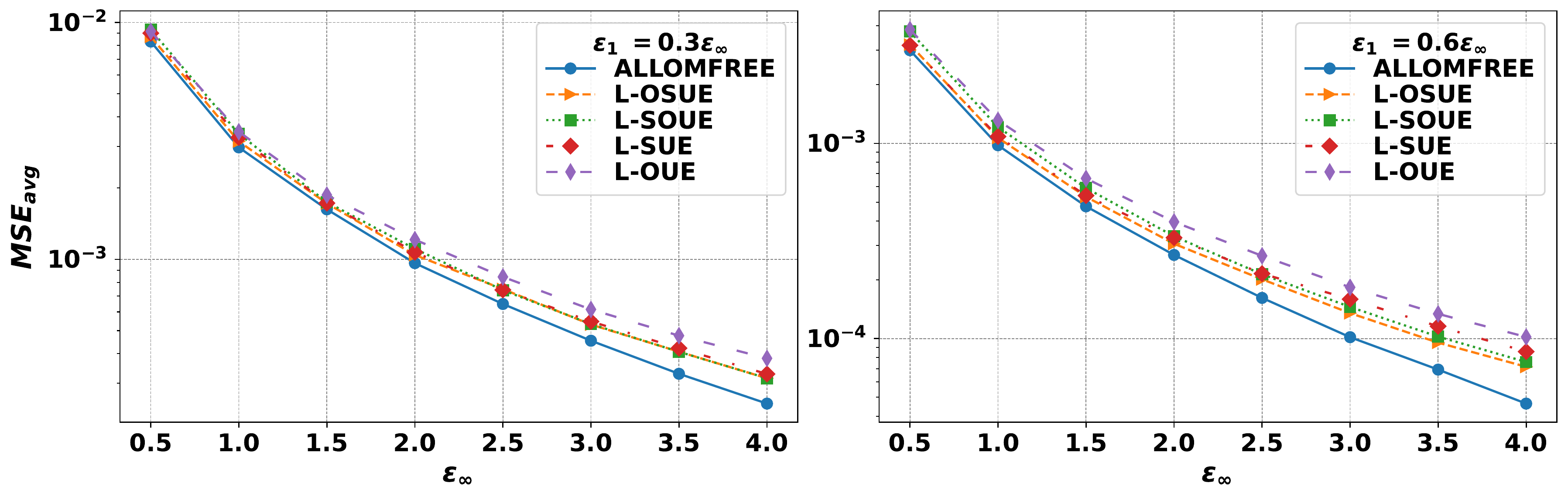}
    \caption{Averaged MSE varying $\epsilon_{\infty}$ with $\epsilon_1 = 0.3\epsilon_{\infty}$ (left-side plot) and with $\epsilon_1 = 0.6\epsilon_{\infty}$ (right-side plot) on the \textit{Census-Income} dataset.}
    \label{fig:results_census_allomfree}
\end{figure}

To complement the results of Figs.~\ref{fig:results_nursery_allomfree} --~\ref{fig:results_census_allomfree}, Table~\ref{tab:utility_results_03} ($\epsilon_1 = 0.3\epsilon_{\infty}$) and Table~\ref{tab:utility_results_06} ($\epsilon_1 = 0.6\epsilon_{\infty}$) exhibit for all datasets and $\epsilon_{\infty}$ guarantees the following utility metrics:

\begin{equation} \label{eq:utility}
\begin{gathered}
    \mathscr{U}_{\textit{L-SUE}} = \frac{MSE_{avg_{(\textit{L-SUE})}} - MSE_{avg_{(\textit{ALLOMFREE})}}}{MSE_{avg_{(\textit{L-SUE})}}} \textrm{,} \\
    \mathscr{U}_{\textit{L-OUE}} = \frac{MSE_{avg_{(\textit{L-OUE})}} - MSE_{avg_{(\textit{ALLOMFREE})}}}{MSE_{avg_{(\textit{L-OUE})}}} \textrm{,}
\end{gathered}
\end{equation}

in which $\mathscr{U}_{\textit{L-SUE}}$ and $\mathscr{U}_{\textit{L-OUE}}$ represent the accuracy gain of ALLOMFREE over the state-of-the-art L-SUE and L-OUE protocols, respectively.

\setlength{\tabcolsep}{5pt}
\renewcommand{\arraystretch}{1.4}
\begin{table}[!ht]
    \scriptsize
    \centering
    \caption{Accuracy gain of ALLOMFREE over the state-of-the-art L-SUE and L-OUE protocols for all datasets with $\epsilon_1 = 0.3\epsilon_{\infty}$, measured with the $\mathscr{U}_{\textit{L-SUE}}$ and $\mathscr{U}_{\textit{L-OUE}}$ metrics expressed in $\%$.}
    \begin{tabular}{lrrrrrrrr}
    \hline
    \multirow{2}{*}{$\epsilon_{\infty}$} & \multicolumn{2}{c}{\textit{Nursery}}  &  \multicolumn{2}{c}{\textit{Adult}}  &  \multicolumn{2}{c}{\textit{MS-FIMU}}  &  \multicolumn{2}{c}{\textit{Census-Income}}  \\ 
    & $\mathscr{U}_{\textit{L-SUE}}$ & $\mathscr{U}_{\textit{L-OUE}}$ & $\mathscr{U}_{\textit{L-SUE}}$ & $\mathscr{U}_{\textit{L-OUE}}$ & $\mathscr{U}_{\textit{L-SUE}}$ & $\mathscr{U}_{\textit{L-OUE}}$ & $\mathscr{U}_{\textit{L-SUE}}$ & $\mathscr{U}_{\textit{L-OUE}}$ \\
    \hline
    0.5 &  13.51 &  20.63 &  19.03 &  27.73 &   3.03 &   5.43 &   7.84 &   9.48 \\
    1.0 &  12.36 &  17.75 &  12.77 &  20.44 &   1.01 &  11.57 &   9.21 &  14.08 \\
    1.5 &  19.95 &  25.86 &   8.47 &  18.01 &   4.13 &  11.55 &   5.82 &  12.92 \\
    2.0 &  17.18 &  33.24 &   4.11 &  17.16 &  13.22 &  23.44 &  10.06 &  20.41 \\
    2.5 &  20.70 &  35.40 &  11.93 &  22.54 &  10.41 &  22.25 &  12.77 &  23.15 \\
    3.0 &  28.69 &  42.98 &   8.35 &  28.22 &  13.07 &  21.56 &  17.07 &  26.21 \\
    3.5 &  36.19 &  54.02 &  18.97 &  32.02 &  14.78 &  29.10 &  22.02 &  30.96 \\
    4.0 &  41.24 &  57.16 &  19.81 &  34.25 &  20.38 &  29.64 &  24.99 &  35.60 \\\hline
    Mean &  23.73 &  35.88 &  12.93 &  25.05 &  10.00 &  19.32 &  13.72 &  21.60 \\
    \hline
    \end{tabular}
    \label{tab:utility_results_03}
\end{table}

\setlength{\tabcolsep}{5pt}
\renewcommand{\arraystretch}{1.4}
\begin{table}[!ht]
    \scriptsize
    \centering
    \caption{Accuracy gain of ALLOMFREE over the state-of-the-art L-SUE and L-OUE protocols for all datasets with $\epsilon_1 = 0.6\epsilon_{\infty}$, measured with the $\mathscr{U}_{\textit{L-SUE}}$ and $\mathscr{U}_{\textit{L-OUE}}$ metrics expressed in $\%$.}
    \begin{tabular}{lrrrrrrrr}
    \hline
    \multirow{2}{*}{$\epsilon_{\infty}$} & \multicolumn{2}{c}{\textit{Nursery}}  &  \multicolumn{2}{c}{\textit{Adult}}  &  \multicolumn{2}{c}{\textit{MS-FIMU}}  &  \multicolumn{2}{c}{\textit{Census-Income}}  \\ 
    & $\mathscr{U}_{\textit{L-SUE}}$ & $\mathscr{U}_{\textit{L-OUE}}$ & $\mathscr{U}_{\textit{L-SUE}}$ & $\mathscr{U}_{\textit{L-OUE}}$ & $\mathscr{U}_{\textit{L-SUE}}$ & $\mathscr{U}_{\textit{L-OUE}}$ & $\mathscr{U}_{\textit{L-SUE}}$ & $\mathscr{U}_{\textit{L-OUE}}$ \\
    \hline
    0.5 &  17.82 &  38.84 &  10.42 &  27.46 &   6.41 &  24.79 &   5.65 &  21.61 \\
    1.0 &  14.99 &  38.97 &   9.83 &  25.14 &   2.97 &  23.32 &   9.79 &  25.46 \\
    1.5 &  15.88 &  41.05 &  12.90 &  28.59 &  16.00 &  30.52 &  11.88 &  28.05 \\
    2.0 &  27.52 &  54.69 &  12.95 &  33.78 &  14.81 &  35.65 &  18.45 &  32.31 \\
    2.5 &  39.59 &  60.96 &  23.28 &  38.50 &  17.71 &  35.34 &  24.89 &  39.11 \\
    3.0 &  40.64 &  65.32 &  28.59 &  47.95 &  27.26 &  40.97 &  36.12 &  44.48 \\
    3.5 &  44.39 &  68.73 &  34.85 &  50.00 &  33.69 &  50.94 &  40.01 &  48.18 \\
    4.0 &  42.24 &  71.13 &  45.26 &  58.33 &  41.83 &  59.47 &  45.85 &  54.44 \\\hline
    Mean &  30.38 &  54.96 &  22.26 &  38.72 &  20.08 &  37.62 &  24.08 &  36.70 \\
    \hline
    \end{tabular}
    \label{tab:utility_results_06}
\end{table}

From Tables~\ref{tab:utility_results_03} and~\ref{tab:utility_results_06}, one can notice that ALLOMFREE considerably improves the quality of the frequency estimates in comparison with the state-of-the-art L-SUE and L-OUE protocols. On average, ALLOMFREE improves the results of L-SUE at least $10\%$ with the \textit{MS-FIMU} dataset in Table~\ref{tab:utility_results_03} and at most $30.38\%$ with the \textit{Nursery} dataset in Table~\ref{tab:utility_results_06} for the privacy guarantees $\epsilon_{\infty}$ and $\epsilon_1$ analyzed. Similarly, on average, ALLOMFREE improves the results of L-OUE at least $19.32\%$ with the \textit{MS-FIMU} dataset in Table~\ref{tab:utility_results_03} and at most $54.96\%$ with the \textit{Nursery} dataset in Table~\ref{tab:utility_results_06}. The highest gain of accuracy was about $\sim 71\%$, achieved with the \textit{Nursery} dataset when $\epsilon_{\infty}=4$ in Table~\ref{tab:utility_results_06} in comparison with the L-OUE protocol. Finally, as one can note, with higher values of $\epsilon_1$, ALLOMFREE will provide much higher utility than the other protocols.

\subsection{Discussion and Related Work}  \label{ch5:discussion_allomfree}

Frequency estimation is a fundamental primitive in LDP and has received considerable attention for a single attribute in both theoretical and application perspectives~\cite{yang2020_survey,Wang2020_survey,Xiong2020_survey} (see, e.g.,~\cite{Murakami2019,Hadamard,Wang2021_b,Fernandes2019,Cormode2021,tianhao2017,kairouz2016discrete,Alvim2018,rappor,microsoft,Kim2018,Arcolezi2020,Vidal2020,alvim2017,kairouz2016extremal,Zhao2019,Li2020,Wang2017}). However, most studies for collecting multidimensional data with LDP mainly focused on numerical data~\cite{Xiong2020_survey} (e.g., cf.~\cite{xiao2,wang2019,Duchi2018,Wang2021_b}) or other complex tasks with categorical data, e.g., marginal estimation~\cite{Shen2021,Peng2019,Zhang2018,Ren2018,Fanti2016} and analytical/range queries~\cite{Jianyu2020,Xu2020,Gu2019,Cormode2019}. For instance, in~\cite{xiao2,wang2019}, the authors propose sampling-based LDP mechanisms for real-valued data (named Harmony and Piecewise Mechanism) and applied these protocols in a multidimensional setting using state-of-the-art LDP mechanisms from~\cite{Bassily2015,tianhao2017} for categorical data. Regarding multidimensional frequency estimates, in~\cite{tianhao2017}, the authors prove for the optimal local hashing protocol that sending $1$ attribute with the whole privacy budget $\epsilon$ results in less variance than splitting the privacy budget for $d=2$ attributes, i.e., with $\epsilon/2$. More generically, this is true for any number of attributes $d$ for the GRR protocol, as we have shown analytically and empirically in Chapter~\ref{chap:chapter7}, and for both OUE and SUE protocols, as shown in Section~\ref{ch5:sec_multidimensional}.

Besides, most frequency estimation academic literature focuses on single data collection. To address longitudinal data collections, in~\cite{rappor,microsoft,erlingsson2020encode}, the authors proposed LDP protocols based on two rounds of sanitization, i.e., \textit{memoization}, which was also adopted in this chapter. In the literature, some works~\cite{Kim2018,Vidal2020} applied L-SUE (a.k.a. Basic-RAPPOR~\cite{rappor}) and L-OUE (i.e., OUE~\cite{tianhao2017} two times) for longitudinal frequency estimates. However, rather than strictly using only SUE or OUE twice, we prove that the optimal combination is starting with OUE and then with SUE (i.e., L-OSUE). The privacy guarantees of chaining two LDP protocols has been further studied in~\cite{Naor2020,erlingsson2020encode}, which results in Eq.~\eqref{eq:eps_long}. Indeed, both ``multiple" settings combined (i.e., many attributes and several collections throughout time), imposes several challenges, in which this paper, proposes the first solution named ALLOMFREE under LDP.

Indeed, both ``multiple" settings combined (i.e., many attributes and several collections throughout time), imposes several challenges, in which this chapter, proposes the first solution named ALLOMFREE under $\epsilon$-LDP. Yet, concerning the privacy guarantees of ALLOMFREE, the memoization step is certainly effective for longitudinal privacy to the cases where the true client's data does not vary (static) or vary very slowly or in an uncorrelated manner~\cite{rappor}. In many application scenarios, gender, age-ranges, nationality, and other demographic data are generally static or vary hardly ever. On the other hand, for dynamic attributes such as location or the time spent in the application, this is not the case. Therefore, for each different value, a new memoized value would be generated, thus accumulating the privacy budget $\epsilon_{\infty}$ by the sequential composition theorem~\cite{dwork2014algorithmic}. 

\section{Conclusion} \label{ch5:sec_conc}

This chapter investigates the problem of collecting multidimensional data throughout time (i.e., longitudinal studies) for the fundamental task of frequency estimation under $\epsilon$-LDP guarantees. We extended the analysis of three state-of-the-art LDP protocols, namely, GRR~\cite{kairouz2016discrete}, OUE~\cite{tianhao2017}, and SUE~\cite{rappor} (cf. Section~\ref{ch2:sub_ldp}) for both longitudinal and multidimensional frequency estimates. On the one hand, for all three protocols, we theoretically prove that randomly sampling a single attribute per user improves data utility, which is an extension of common results in the LDP literature~\cite{erlingsson2020encode,tianhao2017,Jianyu2020,Wang2021,Zhang2018,Arcolezi2021}. 

On the other hand, in the literature, both SUE and OUE protocols have been extended (and also applied~\cite{Kim2018,Vidal2020}) to longitudinal studies based on the concept of \textit{memoization}~\cite{rappor,microsoft}, i.e., L-SUE and L-OUE, respectively. However, we numerically and experimentally show that combining both protocols provides higher data utility, i.e., starting with OUE and then with SUE (L-OSUE) minimizes the variance incurred rather than using SUE or OUE twice. In addition, for the first time, we also extended GRR for longitudinal studies (i.e., L-GRR), which provides higher data utility than the other protocols based on unary encoding for attributes with small domain sizes. 

We also notice that in a multidimensional setting with different domain sizes for each attribute, a dynamic selection of longitudinal LDP protocols is preferred. Therefore, we also proposed a new solution named \underline{A}daptive \underline{L}DP for \underline{LO}ngitudinal and \underline{M}ultidimensional \underline{FRE}quency \underline{E}stimates (ALLOMFREE), which combines all the aforementioned results. More specifically, ALLOMFREE randomly samples a single attribute to send with the whole privacy budget and adaptively selects the optimal protocol, i.e., either L-GRR or L-OSUE. 

To validate our proposal, we conducted a comprehensive and extensive set of experiments on four real-world open datasets. Under the same privacy guarantee, results show that ALLOMFREE consistently and considerably outperforms the state-of-the-art L-SUE~\cite{rappor} and L-OUE~\cite{tianhao2017} protocols in the quality of the frequency estimates, with a gain of accuracy, on average, ranging from $10\%$ up to $55\%$.

Lastly, we highlight that ALLOMFREE is based on the multidimensional \textit{Smp} solution, which randomly samples a single attribute out of $d$ ones to send it with $\epsilon$-LDP. However, aggregators (who are also seen as attackers) are aware of the sampled attribute and its LDP value, which is protected by a ``less strict" $e^{\epsilon}$ probability bound (rather than $e^{\epsilon/d}$). Indeed, in some cases, using the \textit{Smp} solution may be ``unfair" with some users, e.g., users that randomly sample a demographic attribute (e.g., age) might be less concerned to report their data than those whose sampled attribute is socially ``more" sensitive (e.g., disease, location, most common web page). Investigating how to deal with this ``unfair" issue on multidimensional frequency estimates is the main \textit{goal} of the next Chapter~\ref{chap:chapter6}.

%% file: chapters/chapter6.tex
\chapter{Multidimensional Frequency Estimates With LDP: Privacy Focus}  \label{chap:chapter6}

In Chapter~\ref{chap:chapter5}, we tackled both multidimensional and longitudinal settings for the fundamental task of frequency estimation under $\epsilon$-LDP guarantees. In this chapter, we continue contributing to the \textbf{theoretical} aspect dedicating our efforts to the multidimensional setting only. Indeed, the sampling-based solution for multidimensional frequency estimates used in Chapters~\ref{chap:chapter7} and~\ref{chap:chapter5} (i.e., \textit{Smp}), focuses on optimizing \textbf{the utility}. However, this solution considers that all attributes have equal weight in terms of privacy, which (generally) is not the case in real life. For example, in health data collection, people who randomly sample the disease attribute might hesitate to share their data in comparison with others that randomly sample, e.g., age. This idea extends to other application scenarios, e.g., in software monitoring applications with the ``favorite webpage" attribute, and so on. Therefore, in this chapter, we propose a solution for multidimensional frequency estimates under $\epsilon$-LDP guarantees, which improves the \textbf{privacy} of users while providing \textbf{the same or better performance} than the state-of-the-art \textit{Smp} solution.

\section{Introduction} \label{ch6:sec_introduction}

We start recalling the problem statement here. As in previous chapters, we assume there are $d$ attributes $A=\{A_1,A_2,...,A_d\}$, where each attribute $A_j$ with a discrete domain has a specific number of values $|A_j|=c_j$. Each user $u_i$ for $i \in \{1,2,...,n\}$ has a tuple $\textbf{v}^{(i)}=(v^{(i)}_{1},v^{(i)}_{2},...,v^{(i)}_{d})$, where $v^{(i)}_{j}$ represents the value of attribute $A_j$ in record $\textbf{v}^{(i)}$. Thus, for each attribute $A_j$, the analyzer's goal is to estimate a $c_j$-bins histogram, including the frequency of all values in $A_j$.

As presented in Chapters~\ref{chap:chapter7} and~\ref{chap:chapter5}, there are mainly two solutions for satisfying LDP by randomizing $\textbf{v}$, namely, \textit{Spl} and \textit{Smp}. We will also omit the index notation $\textbf{v}^{(i)}$ and use $\textbf{v}$ in the analysis as we focus on one arbitrary user $u_i$ here. Although the \textit{Smp} solution adds sampling error, in the literature~\cite{wang2019,xiao2,tianhao2017,Arcolezi2021,Wang2021_b,Jianyu2020,Wang2021,erlingsson2020encode,bassily2017practical} and in previous Chapters~\ref{chap:chapter7} and~\ref{chap:chapter5}, \textit{Smp} has proven to provide higher data utility than the former \textit{Spl} solution. 

However, as aforementioned, aggregators (who are also seen as attackers) are aware of the sampled attribute and its LDP value, which is protected by a ``less strict" $e^{\epsilon}$ probability bound (rather than $e^{\epsilon/d}$). In other words, while both solutions provide $\epsilon$-LDP, we argue that using the \textit{Smp} solution may be unfair with some users. For instance, on collecting multidimensional health records (i.e., demographic and clinical data), users that randomly sample a demographic attribute (e.g., gender) might be less concerned to report their data than those whose sampled attribute is ``disease" (e.g., if positive for human immunodeficiency viruses - HIV).

This way, there is a privacy-utility trade-off between the \textit{Spl} and \textit{Smp} solutions. With these elements in mind, we formulate the \textbf{problematic of this chapter} as: \textit{For the same privacy budget $\epsilon$, is there a solution for multidimensional frequency estimates that provides better data utility than Spl and more protection than Smp?}

Thus, we intend to solve the aforementioned problematic by answering the following question: \textit{What if the sampling result (i.e., the selected attribute) was \textbf{not} disclosed with the aggregator?} Thus, since the sampling step randomly selects an attribute $j \in [1,d]$ (we slightly abuse the notation and use $j$ for $A_j$), we propose that users \textbf{add uncertainty about the sampled attribute} through generating $d-1$ \textit{fake data}, i.e., one for each non-sampled attribute. 

We call our solution \textit{\underline{R}andom \underline{S}ampling plus \underline{F}ake \underline{D}ata (RS+FD)}. On the one hand, since RS+FD introduces some \textit{uncertainty} in the view of the aggregator, we remarked that users' \textit{privacy} is amplified by sampling~\cite{Chaudhuri2006,Li2012,balle2018privacy,balle2020privacy,first_ldp}. Besides, we integrate two state-of-the-art LDP protocols, namely, GRR~\cite{kairouz2016discrete} and OUE~\cite{tianhao2017} for single attribute frequency estimation, both presented in Section~\ref{ch2:sub_ldp}, into our RS+FD solution to propose four protocols. We demonstrate through experimental validations using four real-world datasets the advantages of our protocols with RS+FD over the state-of-the-art \textit{Spl} and \textit{Smp} solutions.

The rest of this chapter is organized as follows. In Section~\ref{ch6:sec_prop_metho}, we introduce our RS+FD solution, the integration of state-of-the-art LDP mechanisms within RS+FD, and their analysis. In Section~\ref{ch6:sec_results}, we present experimental results. Lastly, in Section~\ref{ch6:sec_conclusion}, we present the concluding remarks. The proposed RS+FD solution in Section~\ref{ch6:sec_prop_metho} and the results presented in Section~\ref{ch6:sec_results} were published in a full paper~\cite{Arcolezi2021_rs+fd} at the 30th International Conference on Information and Knowledge Management (CIKM 2021).

\section{Random Sampling Plus Fake Data (RS+FD)} \label{ch6:sec_prop_metho}

In this section, we present the overview of our RS+FD solution (Section~\ref{ch6:sub_overview}), and the integration of the local randomizers presented in Section~\ref{ch2:sub_ldp} within RS+FD (Subsections~\ref{ch6:sub_rs+fd_grr},~\ref{ch6:sub_rs+fd_oue}, and~\ref{ch6:sub_rs+fd_adp}).

\subsection{Overview of RS+FD} \label{ch6:sub_overview}

Fig.~\ref{ch6:fig_system_overview} illustrates the overview of our proposed RS+FD solution in comparison with the aforementioned known solutions, namely, \textit{Spl} and \textit{Smp}, which is detailed in the following. 

\begin{figure}[!ht]
    \centering
    \includegraphics[width=1\linewidth]{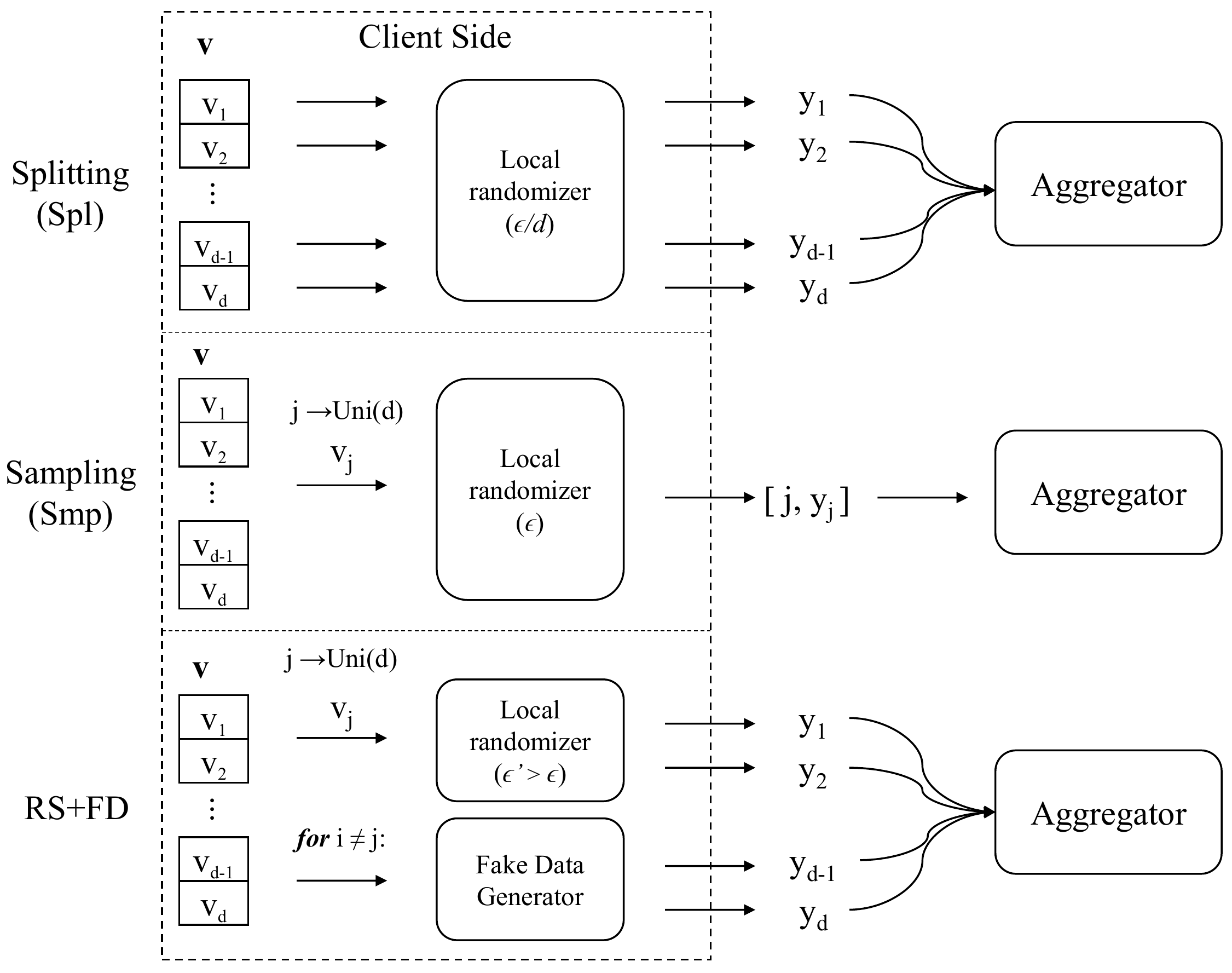}
    \caption{Overview of our random sampling plus fake data (RS+FD) solution in comparison with two known solutions, namely, \textit{Spl} and \textit{Smp}, where $Uni(d)=Uniform(\{1,2,...,d\})$.}
    \label{ch6:fig_system_overview}
\end{figure}

We consider the local DP model, in which there are two entities, namely, users and the aggregator (an untrusted curator). Let $n$ be the total number of users, $d$ be the total number of attributes, $\textbf{c}=[c_1,c_2,...,c_d]$ be the domain size of each attribute, $\mathcal{A}$ be a local randomizer, and $\epsilon$ be the whole privacy budget. Each user holds a tuple $\textbf{v}=(v_1,v_2,...,v_d)$, i.e., a private value per attribute.

\textbf{Client-Side.} The client-side is split into two steps, namely, local randomization and fake data generation (cf. Fig.~\ref{ch6:fig_system_overview}). Initially, each user samples a unique attribute $j$ uniformly at random and applies an LDP mechanism to its value $v_j$. Indeed, RS+FD is generic to be applied with any existing LDP mechanisms (e.g., GRR~\cite{kairouz2016discrete}, UE- or hash-based protocols~\cite{rappor,tianhao2017}, Hadamard Response~\cite{Hadamard}). Next, for each $d-1$ non-sampled attribute $i$, the user generates one random fake data. Finally, each user sends the (LDP or fake) value of each attribute to the aggregator, i.e., a tuple $\textbf{y}=(y_1,y_2,...,y_d)$. This way, the sampling result is not disclosed with the aggregator and, thus, an amplified privacy budget $\epsilon' \geq \epsilon$ can be used. In summary, Alg.~\ref{alg:rs+fd} exhibits the pseudocode of our RS+FD solution. 

\begin{algorithm}

\caption{\underline{R}andom \underline{S}ampling plus \underline{F}ake \underline{D}ata (RS+FD)}
\label{alg:rs+fd}
\begin{algorithmic}[1]

\Statex \textbf{Input :} tuple $\textbf{v} = (v_1,v_2,..., v_d)$, domain size of attributes $\textbf{c}=[c_1,c_2,...,c_d]$, privacy parameter $\epsilon$, local randomizer $\mathcal{A}$. 
\Statex \textbf{Output :} sanitized tuple $\textbf{y}=(y_1,y_2,...,y_d)$.

\State $\epsilon' = \ln{\left( d \cdot (e^{\epsilon} - 1) + 1 \right)}$ \Comment{amplification by sampling~\cite{Li2012}} 

\State $j \gets Uniform(\{1,2,...,d \})$ \Comment{Selection of attribute to sanitize}

\State $B_j \gets v_j$

\State $y_j \gets \mathcal{A}(B_j, c_j, \epsilon')$ \Comment{sanitize data of the sampled attribute}

\State \textbf{for} $i \in \{1,2,...,d\} \setminus \{j\}$ \textbf{do}\Comment{non-sampled attributes} 

\State  \hskip1em $y_i \gets \textit{Uniform}(\{1,...,c_i\}) $ \Comment{generate fake data}

\State \textbf{end for}

\Statex \textbf{return :} $\textbf{y}=(y_1,y_2,...,y_d)$ \Comment{sampling result is not disclosed}
\end{algorithmic}
\end{algorithm}

\textbf{Aggregator.} For each attribute $j \in [1,d]$, the aggregator performs frequency (or histogram) estimation on the collected data by removing bias introduced by the local randomizer and fake data.

\textbf{Privacy analysis.} Let $\mathcal{A}$ be any existing LDP mechanism, Algorithm~\ref{alg:rs+fd} satisfies $\epsilon$-LDP, in a way that $\epsilon'=\ln{\left( d \cdot (e^{\epsilon} - 1) + 1 \right)}$. Indeed, we observe that our scenario is equivalent to sampling a dataset $\mathcal{D}$ without replacement with sampling rate $\beta=\frac{1}{d}$ in the centralized setting of DP, which enjoys privacy amplification (cf. Section~\ref{ch2:sub_sampling}). More specifically, let a trusted curator in the centralized DP setting randomly split a dataset $\mathcal{D}$ in $d$ disjoint subsets $D_1,D_2,...,D_d$, i.e., each with $n/d$ non-overlapping users. Next, let the trusted curator perform frequency estimation in each subset $D \in \mathcal{D}$ with $\epsilon'$-DP. Therefore, invoking Theorem~\ref{theo:amp_sampling} (amplification by sampling) and Proposition~\ref{ch2:prop_parallel_composition} (parallel composition), all frequency estimation queries satisfy $\epsilon$-DP with $\epsilon=\ln{\left( 1 + \beta (e^{\epsilon'} - 1)  \right)}$ where $\beta=1/d$. In our case, with the local model, users sanitize their data locally with a DP model. This way, to satisfy $\epsilon$-LDP with RS+FD, an amplified privacy parameter $\epsilon' \geq \epsilon$ can be used.

\textbf{Limitations.} Similar to other sampling-based methods for collecting multidimensional data under LDP~\cite{Duchi2018,xiao2,wang2019,Wang2021_b}, our RS+FD solution also entails \textit{sampling error}, which is due to observing a sample instead of the entire population. In addition, in comparison with the \textit{Smp} solution, RS+FD requires more computation on the user side because of the fake data generation part. Yet, communication cost is still equal to the \textit{Spl} solution, i.e., each user sends one message per attribute. Lastly, while RS+FD utilizes an amplified $\epsilon' \geq \epsilon$, there is also bias generated from uniform fake data that may require a sufficient number of users $n$ to eliminate the noise.

\subsection{RS+FD with GRR} \label{ch6:sub_rs+fd_grr}

\textbf{Client side.} Integrating GRR as the local randomizer $\mathcal{A}$ into Alg.~\ref{alg:rs+fd} (RS+FD[GRR]) requires no modification. Initially, on the client-side, each user randomly samples an attribute $j$. Next, the value $v_j$ is sanitized with GRR (cf. Section~\ref{ch2:sub_GRR}) using the size of the domain $c_j$ and the privacy parameter $\epsilon'$. In addition, for each non-sampled $d-1$ attribute $i$, the user also generates fake data uniformly at random according to the domain size $c_i$. Lastly, the user transmits the sanitized tuple $\textbf{y}$, which includes the LDP value of the true data ``hidden" among fake data. Visually, Fig.~\ref{fig:prob_tree_rsfd_grr} illustrates the probability tree of the RS+FD[GRR] protocol.

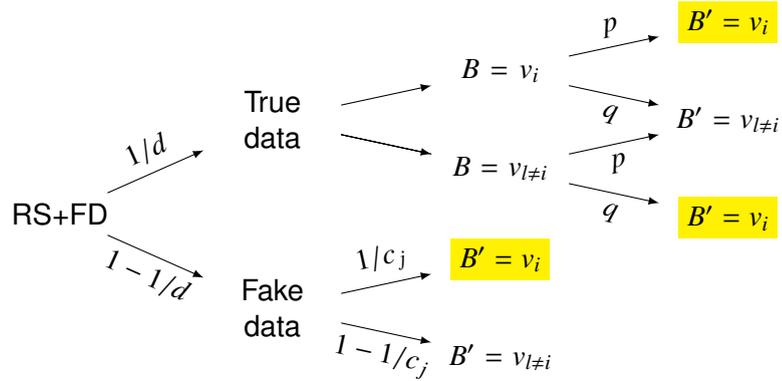
\begin{figure}[!ht]
\centering
\tikzset{
  treenode/.style = {shape=rectangle, rounded corners,
                     draw,align=center,
                     top color=white},
  root/.style     = {treenode},
  env/.style      = {treenode},
  dummy/.style    = {circle,draw}
}
\tikzstyle{level 1}=[level distance=2.8cm, sibling distance=2.5cm]
\tikzstyle{level 2}=[level distance=3cm, sibling distance=1.3cm]
\tikzstyle{bag} = [text width=4em, text centered]
\tikzstyle{end} = [circle, minimum width=2pt,fill, inner sep=0pt]
\begin{tikzpicture}
  [
    grow                    = right,
    edge from parent/.style = {draw, -latex},
    sloped
  ]
  \node [bag] {RS+FD}
    child { node [bag] {Fake data}
        child { node [bag] {$B'=v_{l\neq i}$}
          edge from parent node [below] {$1-1/c_j$} }
        child { node [bag] {\colorbox{yellow}{$B'=v_i$}}
          edge from parent node [above] {$1/c_j$} }
        edge from parent node [below] {$1-1/d$} }
    child { node [bag] {True data}
        child { node [bag] {$B=v_{l\neq i}$}  
         child { node [bag] {\colorbox{yellow}{$B'=v_i$}}
            edge from parent node [below] {$q$}}
            child { node [bag] { \textcolor{white}{$B'=v_{l\neq i}$}}
            edge from parent node [below] {$p$}}
            edge from parent node [below] {}
            edge from parent node [below] {}}
        child { node [bag] {$B=v_i$} 
            child { node [bag] {$B'=v_{l\neq i}$}
            edge from parent node [below] {$q$}}
            child { node [bag] {\colorbox{yellow}{$B'=v_i$}}
            edge from parent node [above] {$p$}}
            edge from parent node [above] {}}
        edge from parent node [above] {$1/d$}};
\end{tikzpicture}
\caption{Probability tree for the RS+FD[GRR] protocol.} 
\label{fig:prob_tree_rsfd_grr}
\end{figure}

\textbf{Aggregator RS+FD[GRR].} On the server-side, according to the probability tree in Fig.~\ref{fig:prob_tree_rsfd_grr}, for each attribute $j\in[1,d]$, the aggregator estimates $\hat{f}(v_i)$ for the frequency of each value $i \in [1,c_j]$ as:

\begin{equation}\label{eq:est_rs+fd_grr}
    \hat{f}(v_i) = \frac{N_i dc_j - n(d - 1 + qc_j)}{nc_j(p-q)} \textrm{,} 
\end{equation}

in which $N_i$ is the number of times the value $v_i$ has been reported, $p=\frac{e^{\epsilon'}}{e^{\epsilon'} + c_j - 1}$, and $q=\frac{1-p}{c_j-1}$.

\begin{theorem} \label{theo:est_grr} For $j\in[1,d]$, the estimation result $\hat{f}(v_i)$ in Eq.~\eqref{eq:est_rs+fd_grr} is an unbiased estimation of $f (v_i)$ for any value $v_i \in A_j$.
\end{theorem}

\begin{proof}

\begin{equation*}
\begin{aligned}
    E[\hat{f}(v_i)] &= E\left[ \frac{N_i d c_j - n(d - 1 + q c_j)}{n c_j (p-q)} \right] \\
    &= \frac{d }{n(p-q)} E[Ni] -  \frac{ d - 1 + q c_j}{c_j (p-q)}  \textrm{.}
\end{aligned}
\end{equation*}

Let us focus on 

\begin{equation*}
\begin{aligned}
    E[N_i] &= \frac{1}{d} \left( p n f (v_i) + q (n - n f (v_i))\right)  + \frac{d-1}{d c_j} n\\
    &= \frac{n}{d} \left(f (v_i)(p-q) + q   + \frac{d-1}{c_j} \right)   \textrm{.}
\end{aligned}
\end{equation*}

Thus,

\begin{equation*}
    E[\hat{f}(v_i)] = f(v_i) \textrm{.}
\end{equation*}
\end{proof}

\begin{theorem} \label{theo:variance_grr} The variance of the estimation in Eq.~\eqref{eq:est_rs+fd_grr} is:

\begin{equation}\label{var:rs+fd_grr}
\begin{gathered}
    \operatorname{VAR}(\hat{f}(v_i)) = \frac{d^2 \gamma (1-\gamma)}{n (p-q)^2} \textrm{, where} \\
    \gamma = \frac{1}{d} \left( q + f(v_i) (p-q) + \frac{(d-1)}{c_j} \right ) \textrm{.}
\end{gathered}
\end{equation}

\end{theorem}

\begin{proof}
Thanks to Eq.~\eqref{eq:est_rs+fd_grr} we have

\begin{equation*}
\operatorname{VAR}(\hat{f}(v_i)) = 
\frac{\operatorname{VAR}(N_i) d^2}{n^2 (p-q)^2}  \textrm{.}
\end{equation*}

Since $N_i$ is the number of times value $v_i$ is observed, it can be defined as $N_i = \sum_{z=1}^n X_z$ where $X_z$ is equal to 1 if the user $z$, 
$1 \le z \le n$ reports value $v_i$, and 0 otherwise. We thus have 
$
\operatorname{VAR}(N_i) 
= \sum_{z=1}^n \operatorname{VAR}(X_z) 
= n \operatorname{VAR}(X)$, since all the users are independent. \textcolor{black}{According to the probability tree in Fig.~\ref{fig:prob_tree_rsfd_grr}, }
\[
P(X = 1) = P(X^2 = 1) = \gamma = \frac{1}{d} \left( q + f(v_i) (p-q) + \frac{(d-1)}{c_j} \right )  \textrm{.}
\]
We thus have $\operatorname{VAR}(X)= \gamma - \gamma^2 = \gamma(1 - \gamma) $ and, finally,

\begin{equation*} \label{var:generic}
\operatorname{VAR}(\hat{f}(v_i)) =
\frac{d^2 \gamma (1-\gamma)}{n (p-q)^2}.
\end{equation*}
\end{proof}

\subsection{RS+FD with OUE}\label{ch6:sub_rs+fd_oue}

\textbf{Client side.} To use UE-based protocols (OUE in our work) as local randomizer $\mathcal{A}$ in Alg.~\ref{alg:rs+fd}, there is, first, a need to define the fake data generation procedure. We propose two solutions: (i) RS+FD[OUE-z] in Alg.~\ref{alg:rs+fd_oue_z}, which applies OUE to $d-1$ zero-vectors (i.e., vectors with only zeros, e.g., $[0,0,...,0,0]$), and (ii) RS+FD[OUE-r] in Alg.~\ref{alg:rs+fd_oue_r}, which applies OUE to $d-1$ one-hot-encoded fake data (uniform at random). Visually, Figs.~\ref{fig:prob_tree_rsfd_oue_z} and~\ref{fig:prob_tree_rsfd_oue_r} illustrate the probability trees of the RS+FD[OUE-z] and RS+FD[OUE-r] protocols, respectively.

\begin{algorithm}[H]
\caption{RS+FD[OUE-z]}
\label{alg:rs+fd_oue_z}
\begin{algorithmic}[1]
\Statex \textbf{Input :} tuple $\textbf{v} = (v_1,v_2,..., v_d)$, domain size of attributes $\textbf{c}=[c_1,c_2,...,c_d]$, privacy parameter $\epsilon$, local randomizer OUE. 
\Statex \textbf{Output :} sanitized tuple $\textbf{B}'=(B_1',B_2',...,B_d')$.

\State $\epsilon' = \ln{\left( d \cdot (e^{\epsilon} - 1) + 1 \right)}$ \Comment{amplification by sampling~\cite{Li2012}} 

\State $j \gets Uniform(\{1,2,...,d \})$ \Comment{Selection of attribute to sanitize}

\State $B_j=Encode(v_j)=[0,0,...,1,0,...0]$ \Comment{one-hot-encoding}

\State $B_j' \gets OUE(B_j, \epsilon')$ \Comment{sanitize real data with OUE}

\State \textbf{for} $i \in \{1,2,...,d\}\setminus \{j\}$ \textbf{do}\Comment{non-sampled attributes} 

\State  \hskip1em $B_i \gets [0,0,...,0] $ \Comment{initialize zero-vectors}

\State  \hskip1em $B_i' \gets OUE(B_i, \epsilon')$  \Comment{randomize zero-vector with OUE}

\State \textbf{end for}

\Statex \textbf{return :} $\textbf{B}'=(B_1',B_2',...,B_d')$ \Comment{sampling result is not disclosed}
\end{algorithmic}
\end{algorithm}

\begin{algorithm}[H]
\caption{RS+FD[OUE-r]}
\label{alg:rs+fd_oue_r}
\begin{algorithmic}[1]
\Statex \textbf{Input :} tuple $\textbf{v} = (v_1,v_2,..., v_d)$, domain size of attributes $\textbf{c}=[c_1,c_2,...,c_d]$, privacy parameter $\epsilon$, local randomizer OUE. 
\Statex \textbf{Output :} sanitized tuple $\textbf{B}'=(B_1',B_2',...,B_d')$.

\State $\epsilon' = \ln{\left( d \cdot (e^{\epsilon} - 1) + 1 \right)}$ \Comment{amplification by sampling~\cite{Li2012}} 

\State $j \gets Uniform(\{1,2,...,d \})$ \Comment{Selection of attribute to sanitize}

\State $B_j=Encode(v_j)=[0,0,...,1,0,...0]$ \Comment{one-hot-encoding}

\State $B_j' \gets OUE(B_j, \epsilon')$ \Comment{sanitize real data with OUE}

\State \textbf{for} $i \in \{1,2,...,d\}\setminus \{j\}$ \textbf{do}\Comment{non-sampled attributes} 

\State  \hskip1em $y_i \gets \textit{Uniform}(\{1,...,c_i\}) $ \Comment{generate fake data}

\State  \hskip1em $B_i \gets Encode(y_i) $ \Comment{one-hot-encoding}

\State  \hskip1em $B_i' \gets OUE(B_i, \epsilon')$  \Comment{randomize fake data with OUE}

\State \textbf{end for}

\Statex \textbf{return :} $\textbf{B}'=(B_1',B_2',...,B_d')$ \Comment{sampling result is not disclosed}
\end{algorithmic}
\end{algorithm}

\begin{figure}[!ht]
\centering
\tikzstyle{level 1}=[level distance=2.7cm, sibling distance=2.66cm]
\tikzstyle{level 2}=[level distance=2.5cm, sibling distance=1.6cm]
\tikzstyle{level 3}=[level distance=2.5cm, sibling distance=0.75cm]
\tikzstyle{bag} = [text width=3em, text centered]
\tikzstyle{end} = [circle, minimum width=2pt,fill, inner sep=0pt]

\begin{tikzpicture}
  [
    grow                    = right,
    edge from parent/.style = {draw, -latex},
    sloped
  ]
  \node [bag] {RS+FD}
    child { node [bag] {Fake data}
        child { node [bag] {$B_i=0$}  
         child { node [bag] {$B_i'=0$}
            edge from parent node [below] {$1-q$}}
            child { node [bag] {\colorbox{yellow}{$B_i'=1$}}
            edge from parent node [above] {$q$}}
            edge from parent node [below] {}
            edge from parent node [below] {}} 
            edge from parent node [below] {$1-1/d$}}
    child { node [bag] {True data}
        child { node [bag] {$B_i=0$}  
         child { node [bag] {$B_i'=0$}
            edge from parent node [below] {$1-q$}}
            child { node [bag] {\colorbox{yellow}{$B_i'=1$}}
            edge from parent node [above] {$q$}}
            edge from parent node [below] {}
            edge from parent node [below] {}}
        child { node [bag] {$B_i=1$} 
            child { node [bag] {$B_i'=0$}
            edge from parent node [below] {$1-p$}}
            child { node [bag] {\colorbox{yellow}{$B_i'=1$}}
            edge from parent node [above] {$p$}}
            edge from parent node [above] {}}
        edge from parent node [above] {$1/d$}};
\end{tikzpicture}
\caption{Probability tree for the RS+FD[OUE-z] protocol.} 
\label{fig:prob_tree_rsfd_oue_z}
\end{figure}

\begin{figure}[!ht]
\centering

\tikzstyle{level 1}=[level distance=2.7cm, sibling distance=3.5cm]
\tikzstyle{level 2}=[level distance=3cm, sibling distance=1.7cm]
\tikzstyle{level 3}=[level distance=3cm, sibling distance=0.7cm]
\tikzstyle{bag} = [text width=3em, text centered]
\tikzstyle{end} = [circle, minimum width=2pt,fill, inner sep=0pt]

\begin{tikzpicture}
  [
    grow                    = right,
    edge from parent/.style = {draw, -latex},
    sloped
  ]
  \node [bag] {RS+FD}
    child { node [bag] {Fake data}
        child { node [bag] {$B_i=0$}  
         child { node [bag] {$B_i'=0$}
            edge from parent node [below] {$1-q$}}
            child { node [bag] {\colorbox{yellow}{$B_i'=1$}}
            edge from parent node [above] {$q$}}
            edge from parent node [below] {$1-1/c_j$}
            edge from parent node [below] {}}
        child { node [bag] {$B_i=1$} 
            child { node [bag] {$B_i'=0$}
            edge from parent node [below] {$1-p$}}
            child { node [bag] {\colorbox{yellow}{$B_i'=1$}}
            edge from parent node [above] {$p$}}
            edge from parent node [above] {$1/c_j$}}
        edge from parent node [below] {$1-1/d$} }
    child { node [bag] {True data}
        child { node [bag] {$B_i=0$}  
         child { node [bag] {$B_i'=0$}
            edge from parent node [below] {$1-q$}}
            child { node [bag] {\colorbox{yellow}{$B_i'=1$}}
            edge from parent node [above] {$q$}}
            edge from parent node [below] {}
            edge from parent node [below] {}}
        child { node [bag] {$B_i=1$} 
            child { node [bag] {$B_i'=0$}
            edge from parent node [below] {$1-p$}}
            child { node [bag] {\colorbox{yellow}{$B_i'=1$}}
            edge from parent node [above] {$p$}}
            edge from parent node [above] {}}
        edge from parent node [above] {$1/d$}};
\end{tikzpicture}
\caption{Probability tree for the RS+FD[OUE-r] protocol.} 
\label{fig:prob_tree_rsfd_oue_r}
\end{figure}
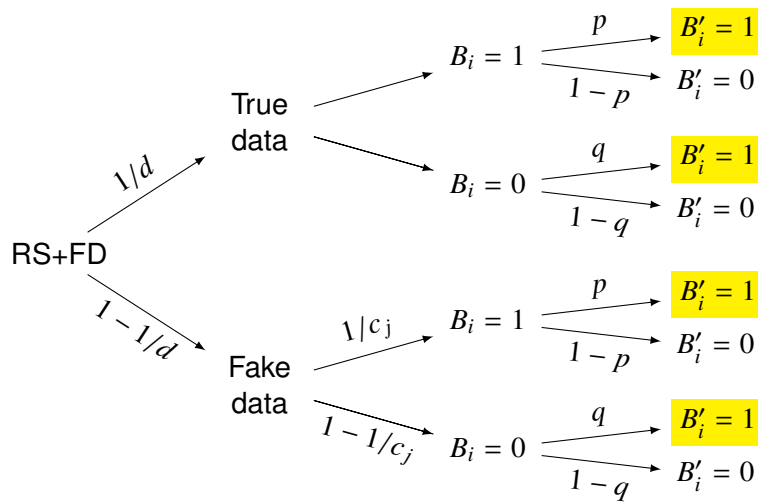

\textbf{Aggregator RS+FD[OUE-z].} On the server-side, if fake data are generated with OUE applied to zero-vectors as in Alg.~\ref{alg:rs+fd_oue_z}, according to the probability tree in Fig.~\ref{fig:prob_tree_rsfd_oue_z}, for each attribute $j\in[1,d]$, the aggregator estimates $\hat{f}(v_i)$ for the frequency of each value $i \in [1,c_j]$ as:

\begin{equation}\label{eq:est_rs+fd_oue_zeros}
    \hat{f}(v_i) =  \frac{d(N_i  - nq)}{n(p-q)}  \textrm{,}
\end{equation}

in which $N_i$ is the number of times the value $v_i$ has been reported, $n$ is the total number of users, $p=\frac{1}{2}$, and $q=\frac{1}{e^{\epsilon'}+1}$.

\begin{theorem} \label{theo:est_oue_z} For $j\in[1,d]$, the estimation result $\hat{f}(v_i)$ in Eq.~\eqref{eq:est_rs+fd_oue_zeros} is an unbiased estimation of $f (v_i)$ for any value $v_i \in A_j$.
\end{theorem}

\begin{proof}
\begin{equation*}
\begin{aligned}
     E[\hat{f}(v_i)] &=  E\left[\frac{d(N_i  - nq)}{n(p-q)} \right] =   \frac{d(E[N_i]  - nq)}{n(p-q)} \\
     &= \frac{d}{n(p-q)}E[N_i]  - \dfrac{dq}{p-q}.
\end{aligned}
\end{equation*}

We have successively 

\begin{equation*}
\begin{aligned}
      E[N_i] &= \frac{n}{d} \left( p  f (v_i) + q (1 -  f (v_i))\right)  + \frac{(d-1)nq}{d}\\
     &= \frac{n}{d} \left( f(v_i)(p-q)   + dq \right) \textrm{.}
\end{aligned}
\end{equation*}

Thus,

\begin{equation*}
    E[\hat{f}(v_i)]  =  f(v_i)\textrm{.}
\end{equation*}
\end{proof}

\begin{theorem} \label{theo:variance_oue_z} The variance of the estimation in Eq.~\eqref{eq:est_rs+fd_oue_zeros} is:

\begin{equation}\label{var:rs+fd_oue_z}
\begin{gathered}
    \operatorname{VAR}(\hat{f}(v_i)) = \frac{d^2 \gamma (1-\gamma)}{n (p-q)^2} \textrm{, where} \\
    \gamma = \frac{1}{d} \left( dq + f(v_i) (p-q) \right) \textrm{.}
\end{gathered}
\end{equation}

\end{theorem}

The proof for Theorem~\ref{theo:variance_oue_z} follows the \textit{Proof} of Theorem~\ref{theo:variance_grr} and is omitted here. \textcolor{black}{In this case, $\gamma$ follows the probability tree in Fig.~\ref{fig:prob_tree_rsfd_oue_z}}.

\textbf{Aggregator RS+FD[OUE-r].} Otherwise, if fake data are generated with OUE applied to one-hot-encoded random data as in Alg.~\ref{alg:rs+fd_oue_r}, according to the probability tree in Fig.~\ref{fig:prob_tree_rsfd_oue_r}, for each attribute $j\in[1,d]$, the aggregator estimates $\hat{f}(v_i)$ for the frequency of each value $i \in [1,c_j]$ as:

\begin{equation}\label{eq:est_rs+fd_oue_random}
    \hat{f}(v_i) =  \frac{N_i d c_j - n\left[ qc_j + (p-q)(d-1) + qc_j(d-1)) \right]}{nc_j(p-q)} \textrm{,} 
\end{equation}

in which $N_i$ is the number of times the value $v_i$ has been reported, $p=\frac{1}{2}$, and $q=\frac{1}{e^{\epsilon'}+1}$.

\begin{theorem} \label{theo:est_oue_r} For $j\in[1,d]$, the estimation result $\hat{f}(v_i)$ in Eq.~\eqref{eq:est_rs+fd_oue_random} is an unbiased estimation of $f (v_i)$ for any value $v_i \in A_j$.
\end{theorem}

\begin{proof}
\begin{equation*}
\begin{aligned}
     E[\hat{f}(v_i)] &=  E\left[\frac{N_i d c_j - n\left[ qc_j + (p-q)(d-1) + qc_j(d-1)) \right]}{nc_j(p-q)}\right] \\
     &= \frac{d E[N_i]}{n(p-q)} - \frac{(p-q)(d-1) + qdc_j }{c_j(p-q)} .
\end{aligned}
\end{equation*}

We have successively 

\begin{equation*}
\begin{aligned}
     E[N_i] &= \frac{n}{d} \left( p  f (v_i) + q (1 -  f (v_i))\right)  + \frac{n(d-1)}{d}(\frac{p}{c_j} + \frac{c_j-1}{c_j}q)\\
     &=  \frac{n}{d} \left( f(v_i)(p-q) + q)\right)  + \frac{n(d-1)}{dc_j}(p-q  + c_jq) \textrm{.}
\end{aligned}
\end{equation*}

Thus,

\[
E[\hat{f}(v_i)]  =  f(v_i) \textrm{.}
\] 
\end{proof}

\begin{theorem} \label{theo:variance_oue_r} The variance of the estimation in Eq.~\eqref{eq:est_rs+fd_oue_random} is:

\begin{equation}\label{var:rs+fd_oue_r}
\begin{gathered}
    \operatorname{VAR}(\hat{f}(v_i)) = \frac{d^2 \gamma (1-\gamma)}{n (p-q)^2} \textrm{, where} \\
    \gamma = \frac{1}{d} \left(q + f(v_i) (p-q) + \frac{(d-1)}{c_j}\left( p + (c_j-1)q \right)\right) \textrm{.}
\end{gathered}
\end{equation}

\end{theorem}

The proof for Theorem~\ref{theo:variance_oue_r} follows the \textit{Proof} of Theorem~\ref{theo:variance_grr} and is omitted here. \textcolor{black}{In this case, $\gamma$ follows the probability tree in Fig.~\ref{fig:prob_tree_rsfd_oue_r}.}

\subsection{Analytical analysis: RS+FD with ADP} \label{ch6:sub_rs+fd_adp}

As shown in Chapter~\ref{chap:chapter5}, in a multidimensional setting with different domain sizes for each attribute, a dynamic selection of LDP mechanisms is preferred. In this chapter, we also analyze the \textit{approximate variances} $\operatorname{VAR}_1$ for RS+FD[GRR] in Eq.~\eqref{var:rs+fd_grr} and $\operatorname{VAR}_2$ for RS+FD[OUE-z] in Eq.~\eqref{var:rs+fd_oue_z}, in which $f(v_i)=0$. Assume there are $d\geq 2$ attributes with domain size $\textbf{c}=[c_1,c_2,...,c_d]$ and a privacy budget $\epsilon'$. For each attribute $j$ with domain size $c_j$, to select RS+FD[GRR], we are then left to evaluate if $\operatorname{VAR}_1 \leq \operatorname{VAR}_2$. This is equivalent to check whether, 

\begin{equation}\label{ineq:variance}
\frac{d^2\gamma_1 (1-\gamma_1)}{n(p_1-q_1)^2} - \frac{d^2\gamma_2 (1-\gamma_2)}{n(p_2-q_2)^2} \leq 0 \textrm{,}
\end{equation}

in which
$p_1 = \frac{e^{\epsilon'}}{e^{\epsilon'} + c_j - 1}$, 
$q_1 = \frac{1-p_1}{c_j-1}$,
$p_2=\frac{1}{2}$, 
$q_2=\frac{1}{e^{\epsilon'}+1}$,
$\gamma_1 = 
\frac{1}{d}
\left( 
q_1 + \frac{d-1}{c_j}
\right)$, and 
$
\gamma_2 =q_2$. 
\textbf{In other words, if Eq.~\eqref{ineq:variance} is verified, the utility loss is lower with RS+FD[GRR]; otherwise, RS+FD[OUE-z] should be selected. Throughout this chapter, we will refer to this dynamic selection of our protocols as RS+FD[ADP].}

For the sake of illustration, Fig.~\ref{fig:surface_variance} illustrates a 3D visualization of $\frac{d^2\gamma_1 (1-\gamma_1)}{n(p_1-q_1)^2} - \frac{d^2\gamma_2 (1-\gamma_2)}{n(p_2-q_2)^2}$, i.e., the left side of Eq.~\eqref{ineq:variance}, by fixing $\epsilon'=\ln(3)$ and $n=10000$, and by varying $d\in[2, 10]$ and $c_j \in [2, 20]$, which are common values for real-world datasets (cf. Section~\ref{ch6:sub_setup}). In this case, one can notice in Fig.~\ref{fig:surface_variance} that neither RS+FD[GRR] nor RS+FD[OUE-z] will always provide the lowest variance value, which reinforces the need for an adaptive mechanism. For instance, with the selected parameters, for lower values of $c_j$, RS+FD[GRR] can provide lower estimation errors even if $d$ is large. On the other hand, as soon as the domain size starts to grow, e.g., $c_j\geq 10$, one is better off with RS+FD[OUE-z] even for small values of $d\geq 3$, as its variance in Eq.~\eqref{var:rs+fd_oue_z} does not depend on $c_j$.

\begin{figure}[H]
    \centering
    \includegraphics[width=0.6\linewidth]{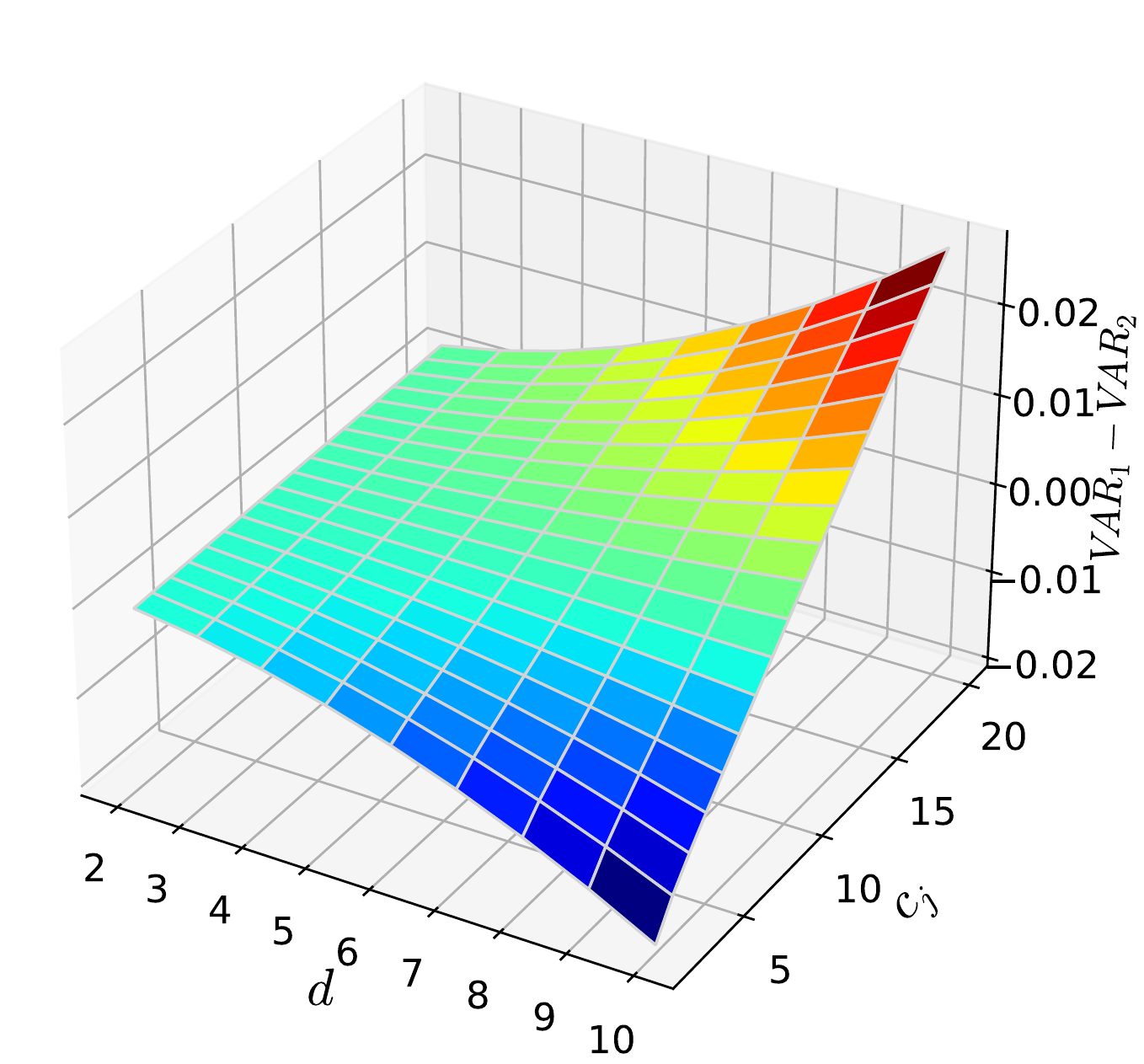}
    \caption{Analytical evaluation of Eq.~\eqref{ineq:variance} that allows a dynamic selection between RS+FD[GRR] with variance $\operatorname{VAR}_1$ and RS+FD[OUE-z] with variance $\operatorname{VAR}_2$. Parameters were set as $\epsilon'=\ln(3)$, $n=10000$, $d\in[2,10]$, and $c_j\in[2,20]$.}
    \label{fig:surface_variance}
\end{figure}

\section{Experimental Validation} \label{ch6:sec_results}

In this section, we present the setup of our experiments in Section~\ref{ch6:sub_setup}, the results with synthetic data in Section~\ref{ch6:sub_synthetic_data}, and the results with real-world data in Section~\ref{ch6:sub_real_data}.

\subsection{Setup of experiments} \label{ch6:sub_setup}

\textbf{Environment.} All algorithms were implemented in Python 3.8.5 with NumPy 1.19.5 and Numba 0.53.1 libraries. The codes we developed and used for all experiments are available in a Github repository\footnote{\url{https://github.com/hharcolezi/ldp-protocols-mobility-cdrs}}. In all experiments, we report average results over 100 runs as LDP algorithms are randomized.

\textbf{Synthetic datasets.} Our first set of experiments are conducted on six synthetic datasets. The distribution of values in each attribute follows an uniform distribution, for all synthetic datasets.

\begin{itemize}
    \item For the first two synthetic datasets, we fix the number of attributes $d=5$ and the domain size of each attribute as $\textbf{c}=[10,10,...,10]$ (uniform), and vary the number of users as $n=50000$ and $n=500000$.
    \item Similarly, for the third and fourth synthetic datasets, we fix the number of attributes $d=10$ and the domain size of each attribute as $\textbf{c}=[10,10,...,10]$ (uniform), and vary the number of users as $n=50000$ and $n=500000$.
    \item Lastly, for the fifth and sixth synthetic datasets, we fix the number of users as $n=500000$. Next, we set the number of attributes $d=10$ with domain size of each attribute as $\textbf{c}=[10,20,...,90,100]$ for one dataset, and we set the number of attributes $d=20$ with domain size of each attribute as $\textbf{c}=[10,10,20,20,...,100,100]$ for the other.
\end{itemize}

\textbf{Real-world datasets.} In addition, we also conduct experiments on four real-world open datasets with non-uniform distributions. We briefly recall here the datasets from Section~\ref{ch3:sub_open_datasets} and the generated one in Chapter~\ref{chap:chapter4}.

\begin{itemize}
    \item \textit{Nursery.} A dataset from the UCI machine learning repository~\cite{uci} with $d=9$ categorical attributes and $n=12960$ samples. The domain size of each attribute is $\textbf{c}=[3, 5, 4, 4, 3, 2, 3, 3, 5]$, respectively. 
    
    \item \textit{Adult.} A dataset from the UCI machine learning repository~\cite{uci} with $d=9$ categorical attributes and $n=45222$ samples after cleaning the data. The domain size of each attribute is $\textbf{c}=[7, 16, 7, 14, 6, 5, 2, 41, 2]$, respectively. 
    
    \item \textit{MS-FIMU.} The dataset developed in Chapter~\ref{chap:chapter4} in which we select $d=6$ categorical attributes (all static attributes, i.e., the dynamic `Visit duration' attribute was not used). The domain size of each attribute is $\textbf{c}=[3, 3, 8, 12, 37, 11]$ (cf. Section~\ref{ch4:info_ms_fimu}), respectively, and there are $n=88935$ samples.
    
    \item \textit{Census-Income.} A dataset from the UCI machine learning repository~\cite{uci} with $d=33$ categorical attributes and $n=299285$ samples. The domain size of each attribute is \begin{math}\textbf{c}=[ 9, 52, 47, 17,  3,  7, 24, ..., 43,  5,  3,  3,  3,  2 ]\end{math}, respectively. 

\end{itemize}

\textbf{Evaluation and metrics.} We vary the privacy parameter in a logarithmic range as $\epsilon=[\ln (2),\ln (3),...,\ln (7)]$, which is within range of values experimented in the literature for multidimensional data (e.g., in~\cite{wang2019} the range is $\epsilon=[0.5,...,4]$ and in~\cite{Wang2021_b} the range is $\epsilon=[0.1,...,10]$). 

Because our estimators in Eq.~\eqref{eq:est_rs+fd_grr}, Eq.~\eqref{eq:est_rs+fd_oue_zeros}, and Eq.~\eqref{eq:est_rs+fd_oue_random} are unbiased, their variance is equal to the MSE (cf. Eq.~\eqref{eq:mse_var}), which is commonly used in practice as an accuracy metric~\cite{Wang2020,Wang2020_post_process,Wang2021_b,li2021privacy}. So, to evaluate our results, we use the MSE metric averaged per the number of attributes $d$ to evaluate our results. Thus, for each attribute $j$, we compute for each value $v_i \in A_j$ the estimated frequency $\hat{f}(v_i)$ and the real one $f(v_i)$ and calculate their differences. More precisely,

\begin{equation}
    MSE_{avg} = \frac{1}{d} \sum_{j \in [1,d]} \frac{1}{|A_j|} \sum_{v \in A_j}(f(v_i) - \hat{f}(v_i) )^2 \textrm{.}
\end{equation}

\textbf{Methods evaluated.} We consider for evaluation the following solutions (cf. Fig.~\ref{ch6:fig_system_overview}) and protocols: 

\begin{itemize}

    \item Solution \textit{Spl}, which splits the privacy budget per attribute $\epsilon/d$ with a best-effort approach using the adaptive mechanism presented in Section~\ref{ch2:sub_ADP}, i.e., Spl[ADP]. 
    
    \item Solution \textit{Smp}, which randomly samples a single attribute and use all the privacy budget $\epsilon$ also with the adaptive mechanism, i.e., Smp[ADP].
    
    \item Our solution RS+FD, which randomly samples a single attribute and uses an amplified privacy budget $\epsilon' \geq \epsilon$ while generating fake data for each $d-1$ non-sampled attribute: 
    
    \begin{itemize}
        \item RS+FD[GRR] (Alg.~\ref{alg:rs+fd} with GRR as local randomizer $\mathcal{A}$);
        \item RS+FD[OUE-z] (Alg.~\ref{alg:rs+fd_oue_z});
        \item RS+FD[OUE-r] (Alg.~\ref{alg:rs+fd_oue_r});
        \item RS+FD[ADP] presented in Section~\ref{ch6:sub_rs+fd_adp} (i.e., adaptive choice between RS+FD[GRR] and RS+FD[OUE-z]).
    \end{itemize}

\end{itemize}

\subsection{Results on synthetic data} \label{ch6:sub_synthetic_data}

Our first set of experiments were conducted on six synthetic datasets. Fig.~\ref{ch6:fig_results_syn1_syn2} (first two synthetic datsets), Fig.~\ref{fig:results_syn3_syn4} (third and fourth synthetic datsets), and Fig.~\ref{fig:results_syn5_syn6} (last two synthetic datasets) illustrate for all methods, the averaged $MSE_{avg}$ (y-axis) according to the privacy parameter $\epsilon$ (x-axis). 

\textbf{Impact of the number of users.} In both Fig.~\ref{ch6:fig_results_syn1_syn2} and Fig.~\ref{fig:results_syn3_syn4}, one can notice that the $MSE_{avg}$ decreases with respect to the number of users $n$. More precisely, with the datasets we experimented, the $MSE_{avg}$ decreases (approximately) one order of magnitude by increasing $n$ in one order of magnitude too. In comparison with \textit{Smp}, the noise in our RS+FD solution comes mainly from fake data as it uses an amplified $\epsilon' \geq \epsilon$. \textit{This suggests that, in some cases, with appropriately high number of user $n$, our solutions may most likely provide higher data utility than the state-of-the-art \textit{Smp} solution (e.g., cf. Fig.~\ref{fig:results_syn5_syn6}).}

\textbf{Impact of the number of attributes.} One can notice the effect on increasing $d$ comparing the results of Fig.~\ref{ch6:fig_results_syn1_syn2} ($d=5$) and Fig.~\ref{fig:results_syn3_syn4} ($d=10$) while fixing $n$ and $\textbf{c}$ (uniform number of values). For instance, even though there are twice the number of attributes, the accuracy (measured with the averaged MSE metric) does not suffer much. \textit{This is because the amplification by sampling ($\frac{e^{\epsilon'}-1}{e^{\epsilon}-1} = \frac{1}{\beta}$~\cite{Li2012}) depends on the sampling rate $\beta=\frac{1}{d}$, which means that the more attributes one collects, the more the $\epsilon'$ is amplified, i.e., $\epsilon'=\ln{\left( d \cdot (e^{\epsilon} - 1) + 1 \right)}$; thus balancing data utility.}

Besides, in Fig.~\ref{fig:results_syn5_syn6}, one can notice a similar pattern, i.e., increasing the number of attributes from $d=10$ (left-side plot) to $d=20$ (right-hand plot), with varied domain size $\textbf{c}$, resulted in only a slightly loss of performance. This, however, is not true for the \textit{Spl} solution, for example, in which the $MSE_{avg}$ increased much more in order of magnitude than our RS+FD solution.

\begin{figure}[!htb]
    \centering
    \includegraphics[width=1\linewidth]{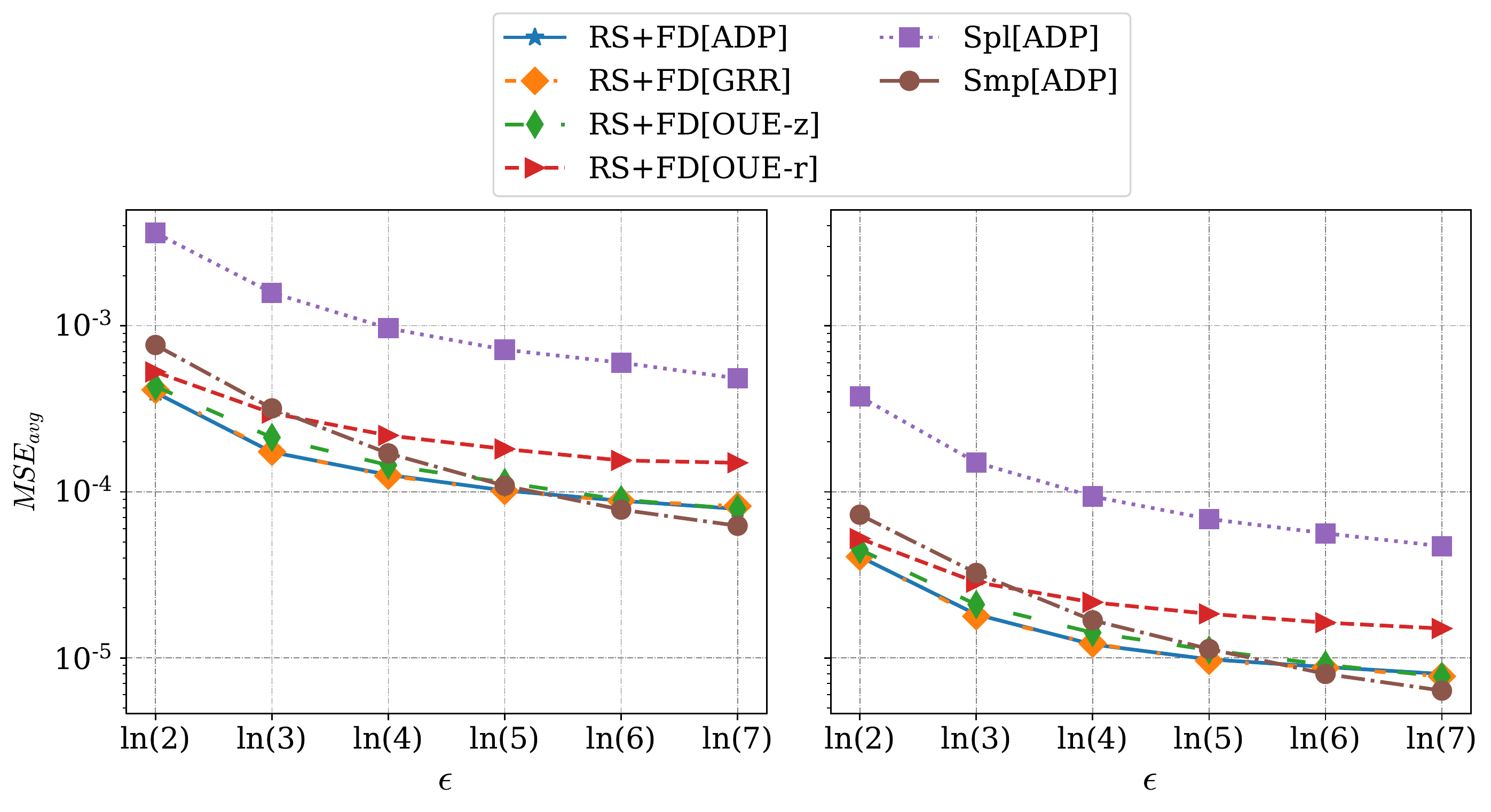}
        \caption{Averaged MSE varying $\epsilon$ on the \textit{synthetic} datasets with $d=5$, uniform domain size $\textbf{c}=[10,10,...,10]$, and $n=50000$ (left-side plot) and $n=500000$ (right-side plot).}\label{ch6:fig_results_syn1_syn2}
\end{figure}

\begin{figure}[!htb]
    \centering
    \includegraphics[width=1\linewidth]{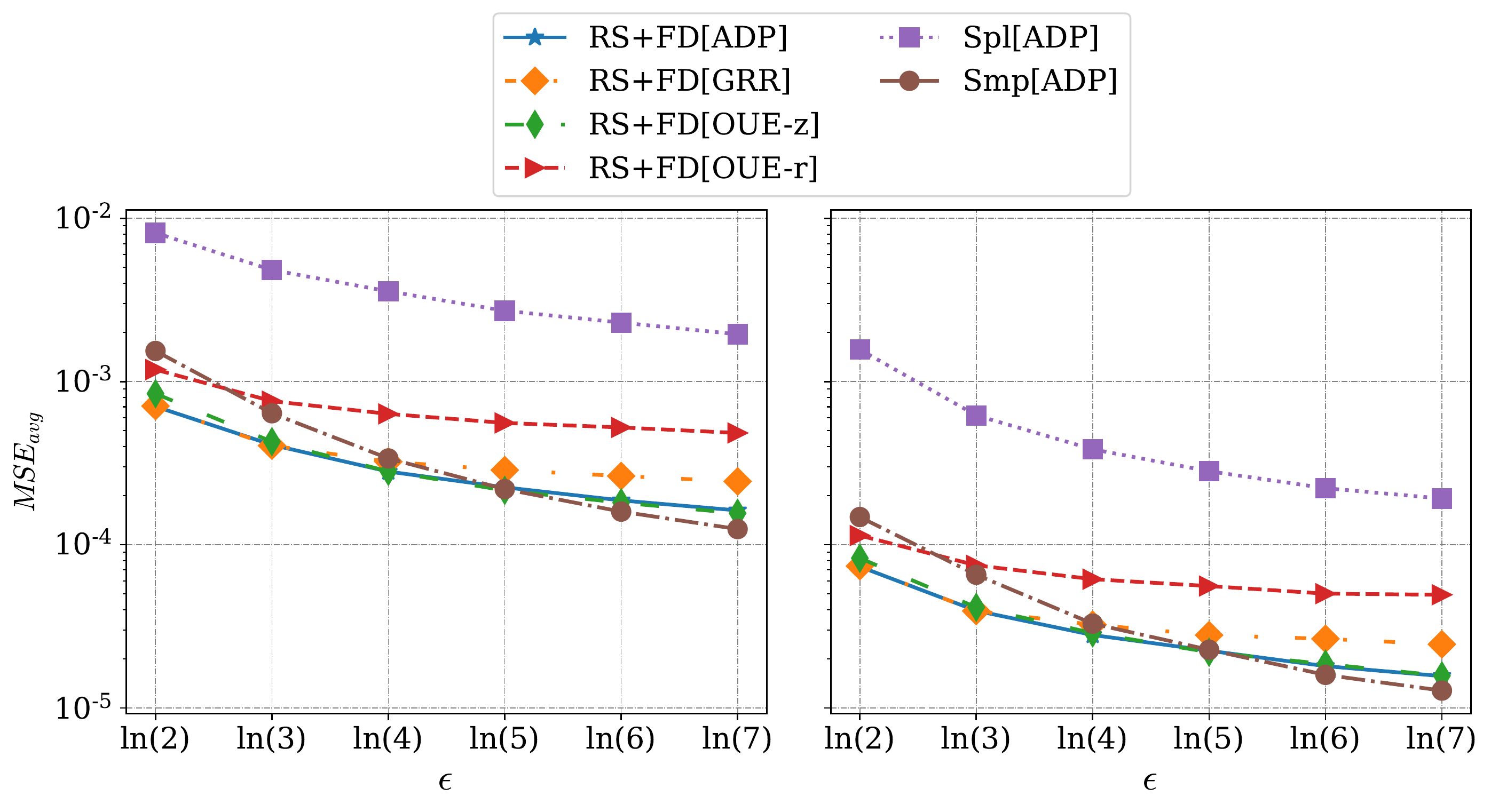}
        \caption{Averaged MSE varying $\epsilon$ on the \textit{synthetic} datasets with $d=10$, uniform domain size $\textbf{c}=[10,10,...,10]$, and $n=50000$ (left-side plot) and $n=500000$ (right-side plot).}\label{fig:results_syn3_syn4}
\end{figure}

\begin{figure}[!htb]
    \centering
    \includegraphics[width=1\linewidth]{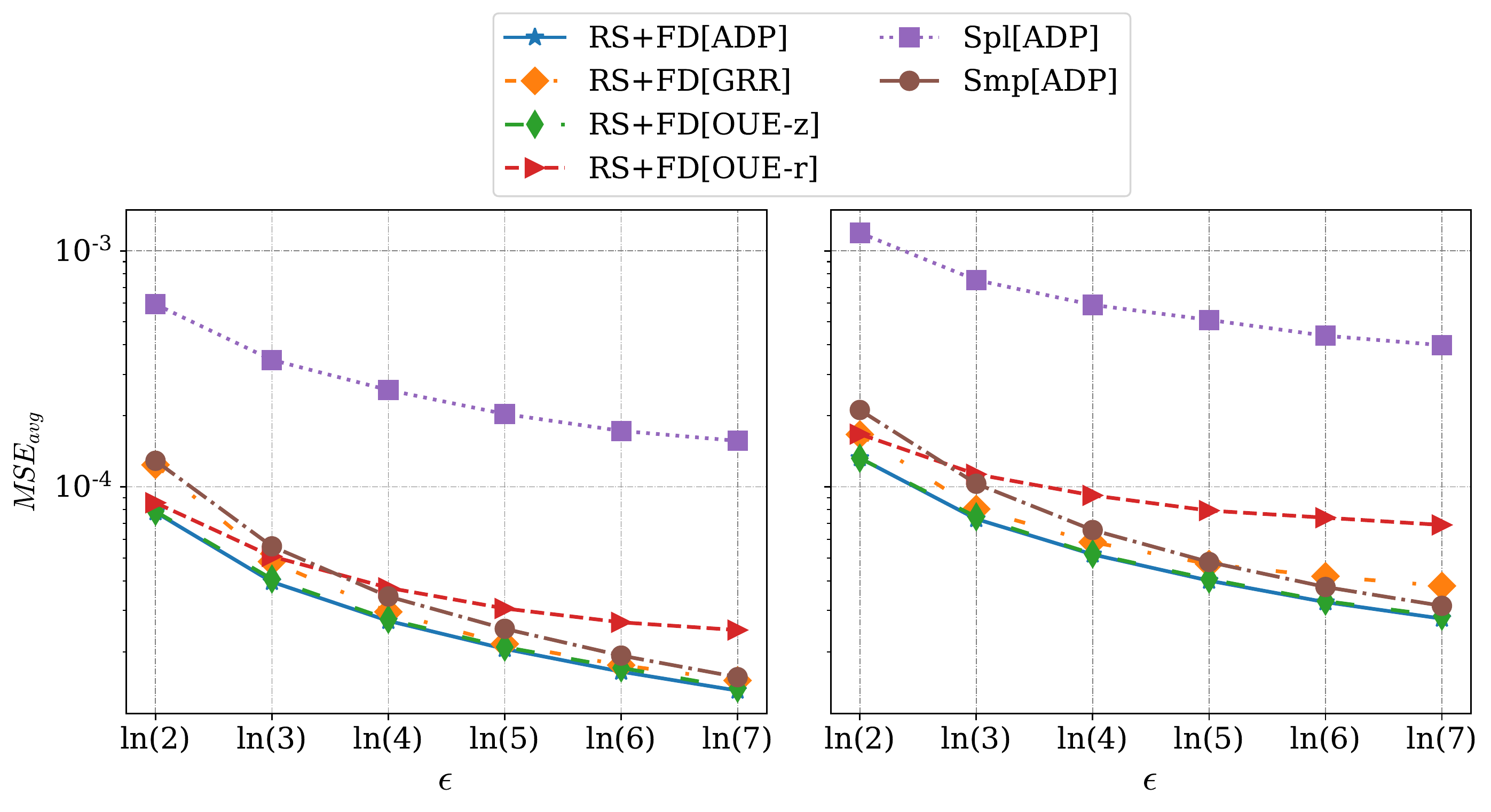}
        \caption{Averaged MSE varying $\epsilon$ on the \textit{synthetic} datasets with $n=500000$: the first with $d=10$ and domain size $\textbf{c}=[10,20,...,90,100]$ (left-side plot), and the other with $d=20$ and domain size $\textbf{c}=[10,10,20,...,100,100]$ (right-side plot).}\label{fig:results_syn5_syn6}
\end{figure}

\textbf{Comparison with existing solutions.} From our experiments, one can notice that the \textit{Spl} solution always resulted in more estimation error (i.e., higher $MSE_{avg}$) than our RS+FD solution and than the \textit{Smp} solution, which is in accordance with other works~\cite{wang2019,xiao2,tianhao2017,Arcolezi2021,Wang2021_b}. Besides, our RS+FD[GRR], RS+FD[OUE-z], and RS+FD[ADP] protocols achieve smaller estimation error (i.e., lower $MSE_{avg}$) or nearly the same $MSE_{avg}$ than the \textit{Smp} solution with a best-effort adaptive mechanism Smp[ADP], which uses GRR for small domain sizes $k$ and OUE for large ones. Although this is not true with RS+FD[OUE-r], it still provides less estimation error than Spl[ADP] while ``hiding" the sampled attribute from the aggregator.

\textbf{Globally, on high privacy regimes (i.e., low values of $\epsilon$), our RS+FD solution consistently outperforms the other two solutions \textit{Spl} and \textit{Smp}. By increasing $\epsilon$, Smp[ADP] starts to outperform RS+FD[OUE-r] while achieving similar performance than our RS+FD[GRR], RS+FD[OUE-z], and RS+FD[ADP] solutions.} In addition, one can notice in Fig.~\ref{fig:results_syn3_syn4}, for example, the advantage of RS+FD[ADP] over our protocols RS+FD[GRR] and RS+FD[OUE-z] applied individually, as it adaptively selects the protocol with the smallest \textit{approximate variance} value.

\subsection{Results on real world data} \label{ch6:sub_real_data}

Our second set of experiments were conducted on four real-world datasets with varied parameters for $n$, $d$, and $\textbf{c}$. Fig.~\ref{fig:results_nursery} (\textit{Nursery}), Fig.~\ref{fig:results_adults} (\textit{Adult}), Fig.~\ref{fig:results_vhs} (\textit{MS-FIMU}), and Fig.~\ref{ch6:fig_results_census} (\textit{Census-Income}) illustrate for all methods, averaged $MSE_{avg}$ (y-axis) according to the privacy parameter $\epsilon$ (x-axis).

\begin{figure}[!ht]
    \centering
    \includegraphics[width=0.625\linewidth]{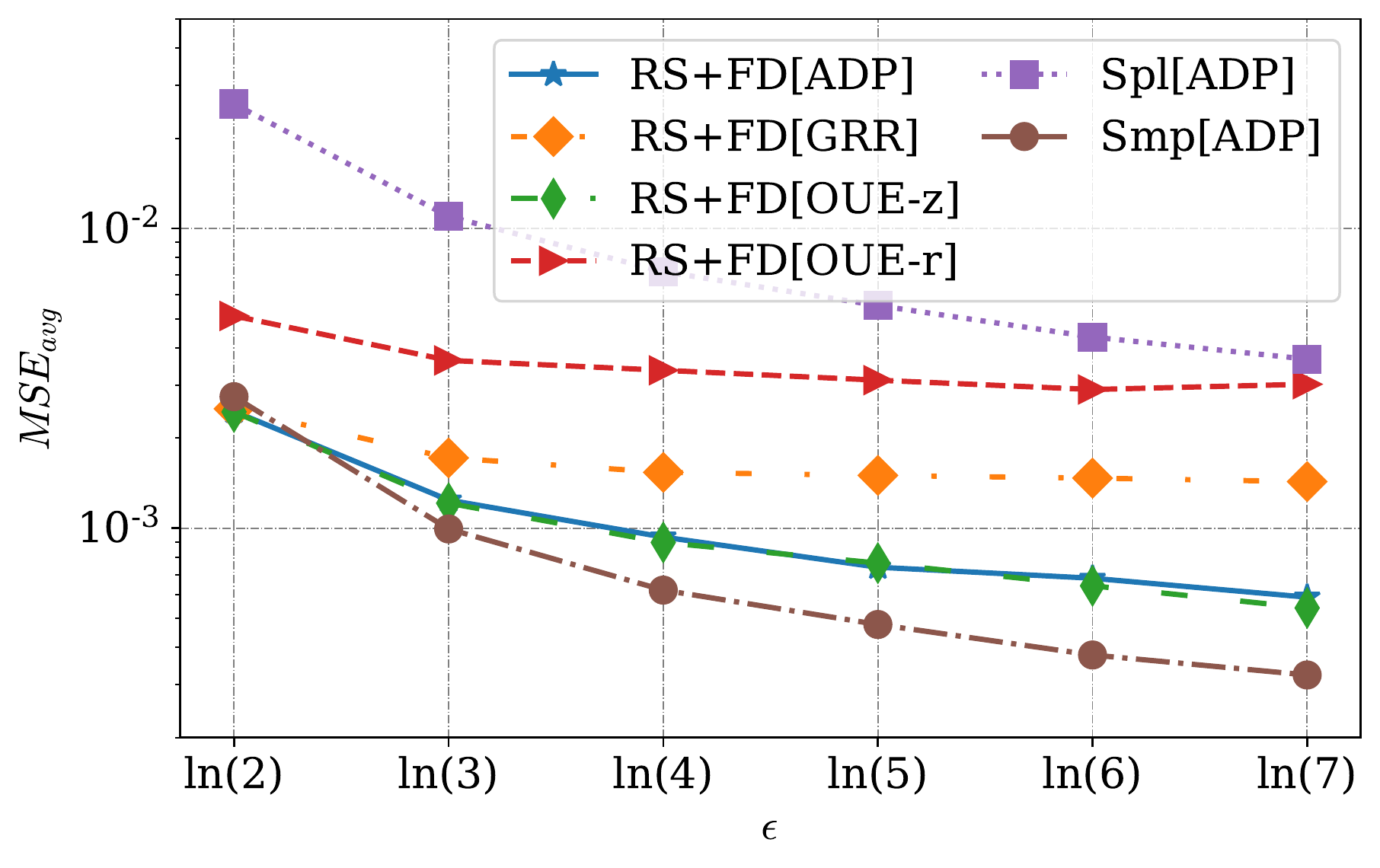}
    \caption{Averaged MSE varying $\epsilon$ on the \textit{Nursery} dataset with $n=12960$, $d=9$, and domain size $\textbf{c}=[3, 5, 4, 4, 3, 2, 3, 3, 5]$.}\label{fig:results_nursery}
\end{figure}

\begin{figure}[!ht]
    \centering
    \includegraphics[width=0.625\linewidth]{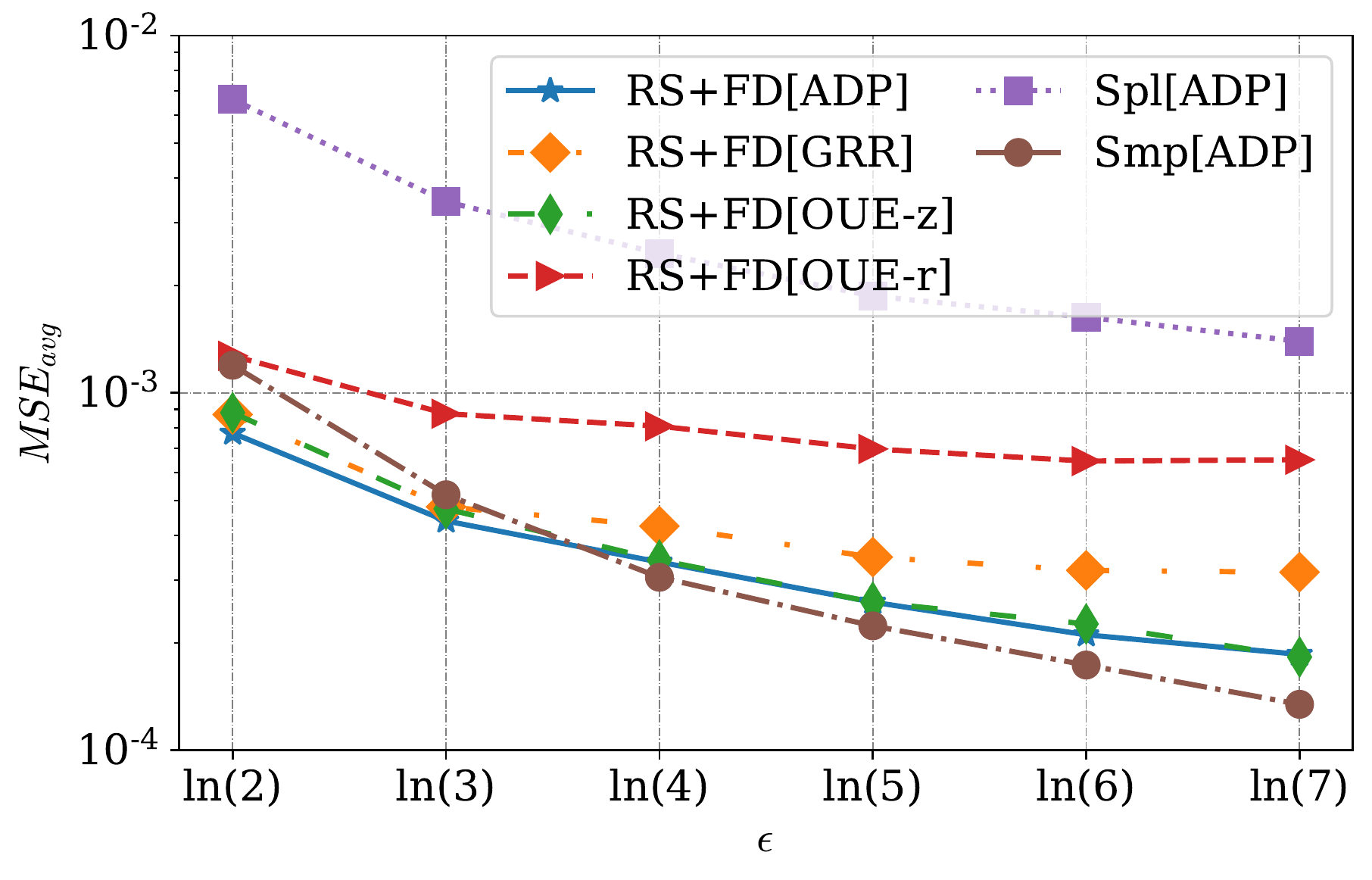}
    \caption{Averaged MSE varying $\epsilon$ on the \textit{Adult} dataset with $n=45222$, $d=9$, and domain size $\textbf{c}=[7, 16, 7, 14, 6, 5, 2, 41, 2]$.}\label{fig:results_adults}
\end{figure}

\begin{figure}[!ht]
    \centering
    \includegraphics[width=0.625\linewidth]{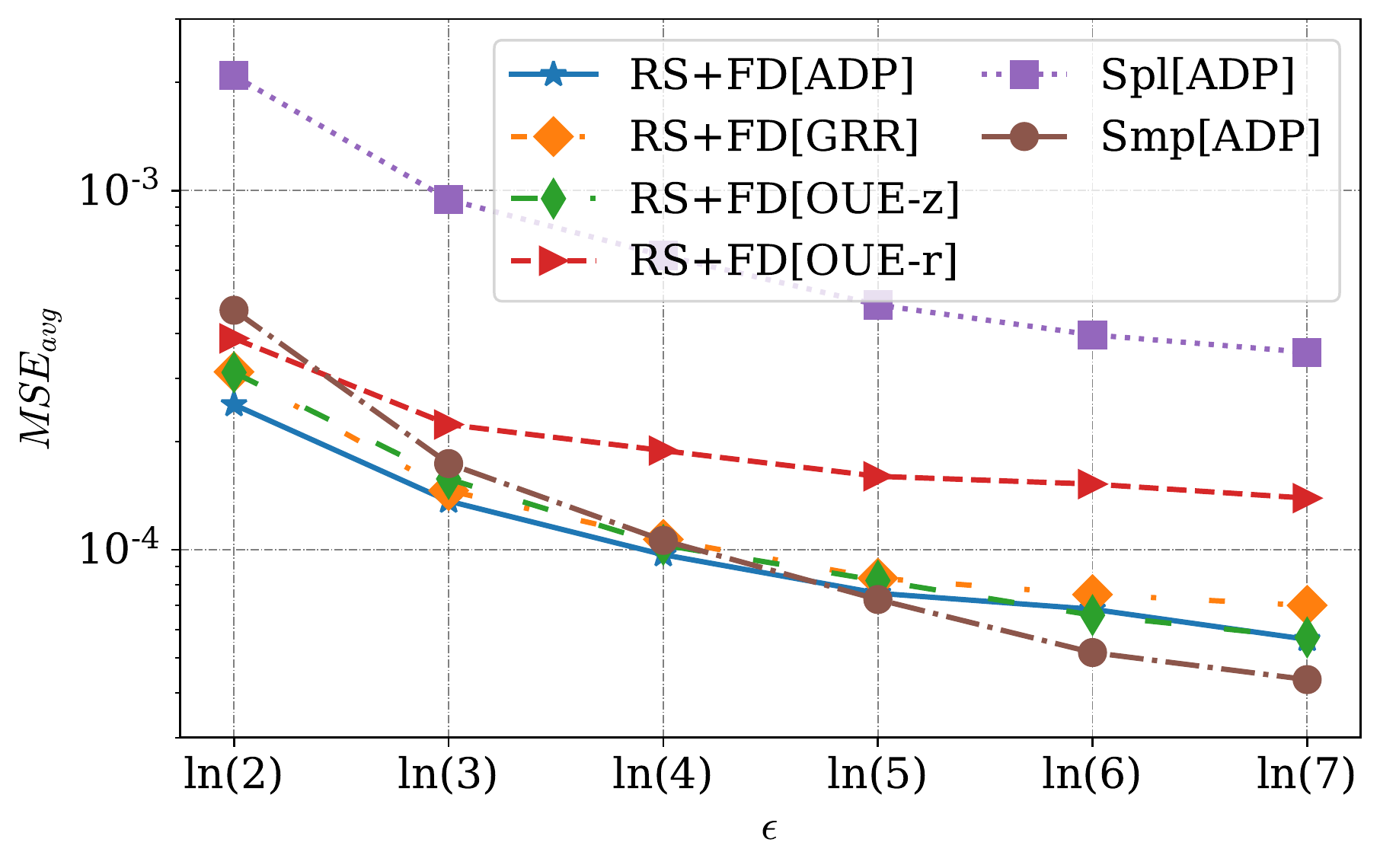}
    \caption{Averaged MSE varying $\epsilon$ on the \textit{MS-FIMU} dataset with $n=88935$, $d=6$, and domain size $\textbf{c}=[3, 3, 8, 12, 37, 11]$.}\label{fig:results_vhs}
\end{figure}

\begin{figure}[!ht]
    \centering
    \includegraphics[width=0.625\linewidth]{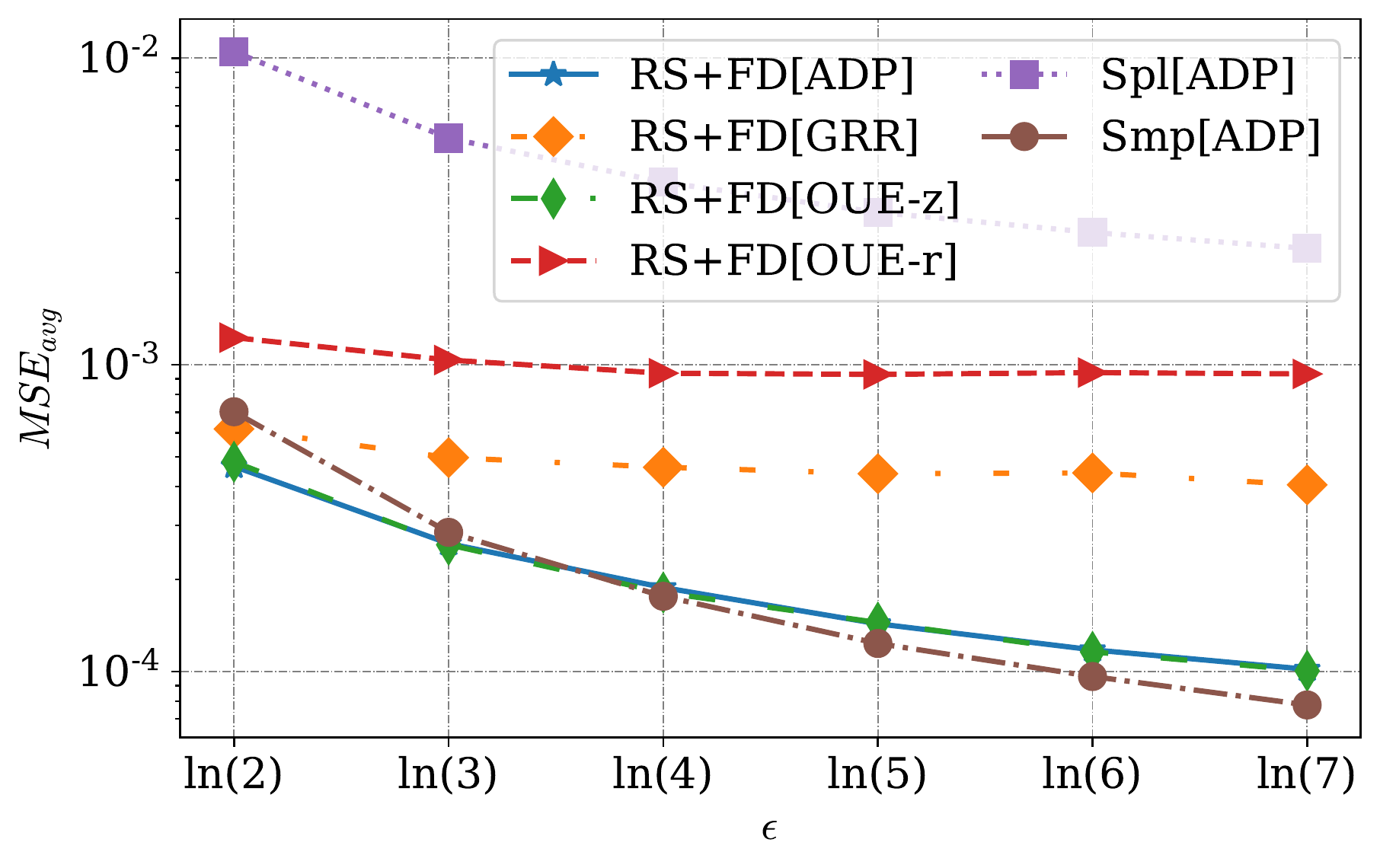}
    \caption{Averaged MSE varying $\epsilon$ on the \textit{Census-Income} dataset with $n=299285$, $d=33$, and domain size \begin{math}\textbf{c}=[9, 52, 47, 17,  3,  ..., 43, 43, 43,  5,  3,  3,  3,  2]\end{math}.}\label{ch6:fig_results_census}
\end{figure}

The results with real-world datasets follow similar behavior than with synthetic ones. For all tested datasets, one can observe that the $MSE_{avg}$ of our proposed protocols with RS+FD is still smaller than the \textit{Spl} solution with a best-effort adaptive mechanism Spl[ADP]. As also highlighted in the literature~\cite{wang2019,xiao2,tianhao2017,Arcolezi2021,Wang2021_b} and in Chapters~\ref{chap:chapter7} and~\ref{chap:chapter5}, privacy budget splitting is sub-optimal, which leads to higher estimation error.  

On the other hand, for both \textit{Adult} and \textit{MS-FIMU} datasets, our solutions RS+FD[GRR], RS+FD[OUE-z], and RS+FD[ADP] achieve nearly the same $MSE_{avg}$ (sometimes smaller $MSE_{avg}$ on high privacy regimes, i.e., for low $\epsilon$) than the \textit{Smp} solution with the best-effort adaptive mechanism Smp[ADP]. For the \textit{Nursery} dataset, with small number of users $n$, only RS+FD[OUE-z] and RS+FD[ADP] are competitive with Smp[ADP]. Lastly, for the \textit{Census} dataset, with a large number of attributes $d=33$, increasing the privacy parameter $\epsilon$ resulted in a small gain on data utility for our solutions RS+FD[GRR] and RS+FD[OUE-r]. On the other hand, both of our solutions RS+FD[OUE-z] and RS+FD[ADP] achieve nearly the same or smaller $MSE_{avg}$ scores than Smp[ADP]. 

Moreover, one can notice that using the \textit{approximate variance} in Eq.~\eqref{ineq:variance} led RS+FD[ADP] to achieve an improved performance over our RS+FD[GRR] and RS+FD[OUE-z] protocols applied individually. For instance, for the \textit{Adult} dataset, with RS+FD[ADP] it was possible to outperform Smp[ADP] 3x more than with RS+FD[GRR] or RS+FD[OUE-z] (similarly, 1x more for the \textit{MS-FIMU} dataset). Besides, for the \textit{Census-Income} dataset, RS+FD[ADP] improves the performance of the other protocols applied individually on high privacy regimes while accompanying the RS+FD[OUE-z] curve on the lower privacy regime cases.

\textbf{In general, these results help us answering the problematic of this chapter (cf. Section~\ref{ch6:sec_introduction}) that for the same privacy parameter $\epsilon$, one can achieve nearly the same or better data utility with our RS+FD solution than when using the state-of-the-art \textit{Smp} solution. Besides, RS+FD enhances users' privacy by ``hiding" the sampled attribute and its $\epsilon$-LDP value among fake data.} On the other hand, there is a price to pay on computation, in the generation of fake data, and on communication cost, which is similar to the \textit{Spl} solution, i.e., send a value per attribute.

\section{Discussion and Related Work} \label{ch6:discussion_rel_work}

As reviewed in Section~\ref{ch5:discussion_allomfree}, most studies for collecting multidimensional data with LDP mainly focused on numerical data~\cite{xiao2,Duchi2018,wang2019,Wang2021_b} or other complex tasks with categorical data, e.g., marginal estimation~\cite{Shen2021,Peng2019,Zhang2018,Ren2018,Fanti2016} and analytical/range queries~\cite{Jianyu2020,Xu2020,Gu2019,Cormode2019}. Regarding multidimensional frequency estimates, in Chapters~\ref{chap:chapter7} and~\ref{chap:chapter5}, we prove that for GRR, SUE, and OUE, sending a single attribute with the whole privacy budget $\epsilon$ results in less variance than splitting the privacy budget for all attributes, which is a common result in LDP literature~\cite{tianhao2017,Jianyu2020,Wang2021,erlingsson2020encode,bassily2017practical}.

However, in the aforementioned works~\cite{xiao2,wang2019,tianhao2017,Duchi2018,Wang2021_b} as well as in Chapters~\ref{chap:chapter7} and~\ref{chap:chapter5}, the sampling result is known by the aggregator. That is, each user samples a single attribute $j$, applies a local randomizer to $v_j$, and sends to the aggregator the tuple $y=\langle j, LDP(v_j) \rangle$ (i.e., \textit{Smp}). While one can achieve higher data utility (cf. Figs.~\ref{ch6:fig_results_syn1_syn2}-~\ref{ch6:fig_results_census}) with \textit{Smp} than splitting the privacy budget among $d$ attributes (\textit{Spl}), we argue that \textit{Smp} might be "unfair" with some users. More precisely, users whose sampled attribute is socially "more" sensitive (e.g., disease or location), might hesitate to share their data as the probability bound $e^{\epsilon}$ is "less" restrictive than $e^{\epsilon/d}$. For instance, assume that GRR is used with k=2 (HIV positive or negative) and the privacy budget is $\epsilon=ln(7) \sim 2$, the user will report the true value with probability as high as $p \sim 87\%$ (even with $\epsilon=1$, this probability is still high $p \sim 73\%$). On the other hand, if there are $d=10$ attributes (e.g., nine demographic and HIV test), with \textit{Spl}, the probability bound is now $e^{\epsilon/10}$ and $p \sim 55\%$.

Motivated by this privacy-utility trade-off between the solutions \textit{Spl} and \textit{Smp}, we proposed a solution named random sampling plus fake data (RS+FD), which generates uncertainty over the sampled attribute in the view of the aggregator. In this context, since the sampling step randomly selects an attribute with sampling probability $\beta=\frac{1}{d}$, there is an amplification effect in terms of privacy, a.k.a. amplification by sampling~\cite{Chaudhuri2006,Li2012,balle2018privacy,balle2020privacy,first_ldp}. A similar privacy amplification for sampling a random item of a single attribute has been noticed in~\cite{Wang2018} for frequent itemset mining in the LDP model too. Indeed, \textit{amplification} is an active research field on DP literature, which aims at finding ways to measure the privacy introduced by non-compositional sources of randomness, e.g., sampling~\cite{Chaudhuri2006,Li2012,balle2018privacy,balle2020privacy,first_ldp}, iteration~\cite{Feldman2018}, and shuffling~\cite{Balle2019,Erlingsson2019,erlingsson2020encode,Wang2020,li2021privacy}. 

\section{Conclusion} \label{ch6:sec_conclusion}

In this chapter, we proposed a solution, namely, RS+FD for multidimensional frequency estimates under $\epsilon$-LDP, which is generic to be used with any existing LDP mechanism developed for single-frequency estimation. More precisely, with RS+FD, the client-side has two steps: local randomization and fake data generation (cf. Fig.~\ref{ch6:fig_system_overview} and Alg.~\ref{alg:rs+fd}). First, an LDP mechanism preserves privacy for the data of the sampled attribute. Second, the fake data generator provides fake data for each $d-1$ non-sampled attribute. This way, the sanitized data is ``hidden" among fake data and, hence, the sampling result is not disclosed along with the users' report (and statistics). 

What is more, we notice that RS+FD can enjoy privacy amplification by sampling~\cite{Chaudhuri2006,Li2012,balle2018privacy,balle2020privacy,first_ldp}, detailed in Section~\ref{ch2:sub_sampling}. That is, if one randomly sample a dataset without replacement using a sampling rate $\beta < 1$, it suffices to use a privacy budget $\epsilon' \geq \epsilon$ to satisfy $\epsilon$-DP, where $\frac{e^{\epsilon'}-1}{e^{\epsilon}-1} = \frac{1}{\beta}$~\cite{Li2012}. This way, given that the sampled dataset for each attribute has non-overlapping users, i.e., each user selects an attribute with sampling probability $\beta=\frac{1}{d}$, to satisfy $\epsilon$-LDP, each user can apply an LDP mechanism with $\epsilon'=\ln{\left( d \cdot (e^{\epsilon} - 1) + 1 \right)} \geq \epsilon$.

Moreover, we integrated two state-of-the-art LDP mechanisms, namely, GRR~\cite{kairouz2016discrete} and OUE~\cite{tianhao2017}, within RS+FD to develop four protocols: RS+FD[GRR], RS+FD[OUE-z], RS+FD[OUE-r], and RS+FD[ADP]. We analyze these four protocols analytically and experimentally through a comprehensive and extensive set of experiments on both synthetic and real-world open datasets. With our experiments, we can conclude that \textbf{under the same privacy guarantee, our proposed protocols with RS+FD achieve similar or better utility (measured with the $MSE_{avg}$ metric) than using the state-of-the-art \textit{Smp} solution (see Figs.~\ref{ch6:fig_results_syn1_syn2} -- \ref{ch6:fig_results_census}).} Besides these \textbf{utility results}, RS+FD also generates \textit{uncertainty} over the sampled attribute in the view of the aggregator, which enhances \textbf{users' privacy}.

%% file: chapters/chapter91.tex
\chapter{Forecasting Mobility Data With Differentially Private Deep Learning} \label{chap:chapter91}

In Chapters~\ref{chap:chapter7}-\ref{chap:chapter6} we have focused and contributed on \textbf{statistical learning} with the local DP model. From this Chapter~\ref{chap:chapter91} until Chapter~\ref{chap:chapter92}, we concentrate our efforts on \textbf{differentially private machine learning}. As mentioned in Chapter~\ref{chap:chapter1}, we aim to solve real-world problems using machine learning, assuming centralized data owners (e.g., MNOs and EMS) that collect sensitive information from individuals for both billing and/or legal purposes. This way, we consider settings applying either centralized DP algorithms (Chapters~\ref{chap:chapter91} and~\ref{chap:chapter92}) or LDP algorithms (Chapters~\ref{chap:chapter8} and~\ref{chap:chapter9}) to sanitize the data on the server-side, which is $\epsilon$-DP for users. However, besides sanitizing the data, extracting meaningful predictions is also of great interest, thus, requiring a proper evaluation of the privacy-utility trade-off.

Moreover, in Chapters~\ref{chap:chapter1} and~\ref{chap:chapter4}, we have reviewed mobility reports published by OBS Flux Vision system~\cite{fluxvision1} and in Chapter~\ref{chap:chapter7} we have proposed an LDP-based CDRs processing system as a stronger alternative to ``anonymity on-the-fly", i.e., with ``sanitization on-the-fly". \textbf{In this chapter, we assume that besides generating mobility reports, MNOs (or any involved entity) could also be interested in forecasting aggregate human mobility statistics.} Therefore, in this chapter, we will assume the existence of two settings for privacy-preserving human mobility analytics using CDRs. The first scenario, \textbf{S$_1$}, considers that aggregated mobility statistics are published following the anonymity ``on-the-fly" model of MNOs CDRs processing systems (e.g., as in~\cite{fluxvision1}). The second setting, \textbf{S$_2$}, considers that besides anonymity ``on-the-fly", a centralized DP algorithm (e.g., Laplace or Gaussian mechanisms from Section~\ref{ch2:sub_dp}) is used to sanitize the aggregate mobility statistics before public release.

In other words, this corresponds to evaluating the privacy-utility trade-off of applying centralized DP algorithms to the current anonymity-based statistics. This is the core contribution of this chapter, in which we evaluate differentially private deep learning models for multivariate time-series forecasting of aggregate human mobility data. Notice that while the previous chapters considered multiple attributes, we will focus on a single attribute here, namely, the number of people per several given regions.

\section{Introduction}\label{ch91:introduction}

As reviewed in Chapters~\ref{chap:chapter1},~\ref{chap:chapter4}, and~\ref{chap:chapter7}, on analyzing mobility data, some studies have shown that humans follow particular patterns with predictability~\cite{deMontjoye2013} and, hence, \textit{users' privacy is a major concern}~\cite{deMontjoye2013,deMontjoye2018,Mir2013,app_blip,Acs2014,Zang2011,Buckee2014,Pyrgelis2017,Pyrgelis2020,Tu2018,Xu2017}. Because of these privacy issues, MNOs tend to publish aggregated mobility data~\cite{deAlarcon2021,Xu2017,Vespe2021,Tu2018,fluxvision1}, e.g., the number of users in given areas at a given timestamp, which, in other words, represents a \textbf{multivariate time series dataset}. 

However, as recent studies have shown, even aggregated mobility data (e.g., heatmaps) can be subject to membership inference attacks~\cite{Pyrgelis2017,Pyrgelis2020} and users' trajectory recovery attack~\cite{Tu2018,Xu2017}. More precisely, the later authors in~\cite{Tu2018,Xu2017} showed that their attack reaches accuracies as high as $73\% \sim 91\%$, suggesting generalization and perturbation through DP~\cite{Dwork2006,Dwork2006DP,dwork2014algorithmic} as a means to mitigate this attack. 

With these elements in mind, this chapter contributes with a comparative analysis between adding DP guarantees into two different steps of training deep learning (DL) models to \textbf{forecasting multivariate aggregated human mobility data}. On the one hand, we consider using \textbf{gradient perturbation}, which can be achieved by training DL models over original time-series data with the DP-SGD~\cite{DL_DP,tf_privacy,pytorch_privacy} algorithm. This case corresponds to collecting data following the scenario \textbf{S$_1$} mentioned at the beginning of this chapter and training a differentially private DL model. On the other hand, we consider using \textbf{input data perturbation}, i.e., training DL models with differentially private time series data. This corresponds to collecting data following the scenario \textbf{S$_2$} also mentioned at the beginning of this chapter and training any non-private DL model on it. We have briefly presented both gradient and input perturbation settings in Section~\ref{ch3:sub_DP_ML}. 

We carried out our experiments with the real-world mobility dataset collected by OBS~\cite{fluxvision1} named Paris-DB described in Section~\ref{ch3:paris_db}. In this chapter, \textbf{we aim at forecasting the future number of people at the next $30$-min interval in each of the 6 regions}. That is, given $X_{(t_1,t_{\tau})}$, the goal is to forecast $X_{(t_{\tau + 1})}$, i.e., \textbf{one-step-ahead forecasting}, which is unknown at time $\tau$. Therefore, we benchmark four state-of-the-art DL models (i.e., recurrent neural networks) with the Paris-DB, providing a first comparative evaluation on the impact of differential privacy guarantees when training DL models in both input and gradient perturbation settings. Indeed, we intend that from this study, other classical multivariate time series forecasting, ML, and privacy-preserving ML techniques can be tested and compared. We invite the interested reader to also visit the \textbf{Github page (\url{https://github.com/hharcolezi/ldp-protocols-mobility-cdrs}), in which we release the dataset and codes we used for our experiments.}

The remainder of this chapter is organized as follows. In Section~\ref{ch91:sec_results}, we present the experimental setup, our results and its discussion, and we review related work. Lastly, in Section~\ref{ch91:sec_conclusion}, we present the concluding remarks and future directions. The experiments and results in Sections~\ref{ch91:sec_results} and~\ref{ch91:sec_conclusion} were submitted as part of a full article to the Neural Computing and Applications journal.

\section{Experimental Validation} \label{ch91:sec_results}

We divide this section in the following way. First, we describe general settings for our experiments (Section~\ref{ch91:sub_gen_setup}). Next, we present the development and evaluation of non-private DL models (Section~\ref{ch91:sub_non_private_DL}). Lastly, we present the development of differentially private DL models, which include both gradient and input perturbation settings (Section~\ref{ch91:sub_private_DL}).

\subsection{General setup of experiments} \label{ch91:sub_gen_setup}

\textbf{Environment.} All algorithms were implemented in Python 3.8.8 with Keras~\cite{keras} and Tensorflow Privacy (TFP)~\cite{tf_privacy} libraries. 

\textbf{Dataset.} In this chapter, we only utilize the second period of the Paris-DB from Section~\ref{ch3:paris_db}, which has aggregated mobility data for $72$ days (from 2020-08-24 to 2020-11-04). We split the Paris-DB into exclusively divided learning (first $65$ days, i.e., $n_l=3120$ intervals of $30$-min) and testing (last $7$ days, i.e., $n_t=336$ intervals of $30$-min) sets. Table~\ref{tab:statistics_data} presents descriptive statistics about both dataset with the following measures per region (labeled as R1 - R6): min, max, mean, standard deviation (std), and median.. Fig.~\ref{fig:train_test_split} exemplifies the data separation into train and test sets for region R1.

\setlength{\tabcolsep}{5pt}
\renewcommand{\arraystretch}{1.4}
\begin{table}[!ht]
    \scriptsize
    \centering
    \caption{Descriptive statistics for the multivariate time series dataset on the number of users per coarse region.}
    \begin{tabular}{c c c c c c c}
    \hline
       \textbf{Statistic} &\textbf{R1}  &\textbf{R2}   &\textbf{R3}   &\textbf{R4}   &\textbf{R5}   &\textbf{R6}   \\\hline
         Min &   56,937 &   1,996 &   1,429 &    255 &   252 &    347 \\
        Max &  165,405 &  21,980 &  28,990 &  25,184 &  7,961 &  27,637 \\
        Mean &  116,777 &  14,307 &  16,274 &  11,758 &  4,166 &  11,559 \\
        Std &   17,947 &   2,803 &   3,915 &   3,682 &  1,450 &   5,136 \\
        Median &  121,488 &  14,808 &  16,661 &  12,134 &  4,495 &  12,542 \\\hline
    \end{tabular}
    \label{tab:statistics_data}
\end{table}

\begin{figure}[!ht]
    \centering
    \includegraphics[width=1\linewidth]{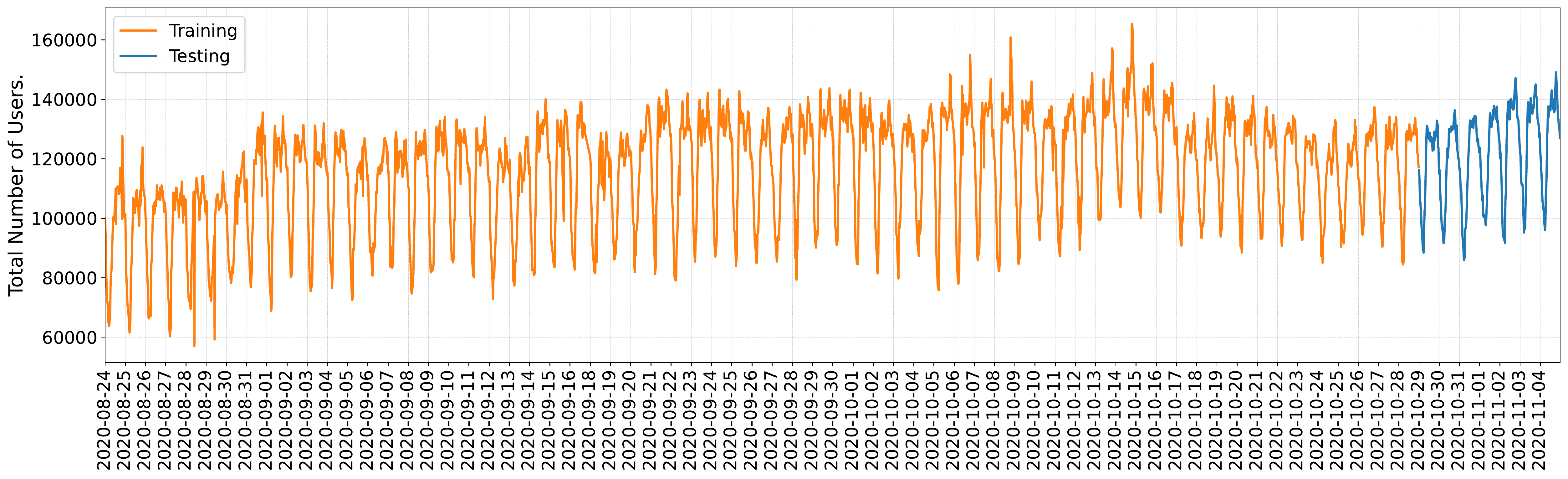}
    \caption{Example of data separation into training and testing sets for region R1.}
    \label{fig:train_test_split}
\end{figure}

\textbf{Temporal features.} We added the time of the day and the time of the week as cyclical features to help models recognizing low and high peak values of human mobility patterns.

\textbf{Forecasting methodology.} We used 6 prior time steps (i.e., lag values), which showed autocorrelation higher than 0.5 to predict a single step ahead in the future (i.e., short forecasting horizon). More specifically, the forecasting models will take into account the number of people in each region from $3$ hours to make predictions \textit{one-step-ahead} for each region in the next $30$-min interval. And in the end, we compute the performance metrics.

\textbf{Performance metrics.} All models were evaluated with standard time-series metrics, namely, RMSE and MAE, both explained in Section~\ref{ch3:sub_metrics}. RMSE was the primary metric to select the final DL models. As a multi-output scenario (i.e., $6$ regions), we present the metrics per region as well as its averaged values. In all experiments, due to randomness, we report the results of the model with the lowest RMSE over 10 runs. 

\subsection{Non-private DL forecasting models} \label{ch91:sub_non_private_DL}

\textbf{Baseline model.} We established a naive forecasting technique a.k.a. ``persistence model", which for each region, it returns the current number of people at time $t$ as the forecasted value, i.e., $\textbf{x}_{t+1}=\textbf{x}_t$. Notice that this is a quite accurate baseline since, in general, the number of people per region varies slowly by 30-min (i.e., walking people may take more time to move from one area to another).

\textbf{Methods evaluated.} To predict the number of users in each region in a multivariate time series forecasting framework, we compared the performance of four state-of-the-art DL models, i.e., recurrent neural networks: LSTM~\cite{LSTM}, GRU~\cite{GRU}, and their Bidirectional~\cite{BI_RNN} architectures, i.e., BiLSTM and BiGRU. These methods have been briefly presented in Section~\ref{ch3:dl_models}. 

\textbf{Model selection.} To optimize the hyperparameters per DL method, we used Bayesian optimization~\cite{hyperopt2013} (explained in Section~\ref{ch3:optimization}) with $100$ iterations to minimize $loss=RMSE_{avg} + RMSE_{std}$; the subscripts \textit{avg} and \textit{std} indicates the averaged and standard deviation values of the RMSE metric considering the 6 regions. For each method, we only used a single hidden layer followed by a dense layer (output), since RNNs generally perform well with a low number of hidden layers~\cite{Hewamalage2021}. So, we searched the following hyperparameters: number of neurons ($h_1$), batch size ($bs$), and learning rate ($\eta$). All models used ``relu" (rectified linear unit) as activation function, which resulted in better performance than the default ``tanh" activation in prior tests. Lastly, models were trained using the adam (adaptive moment estimation) optimizer during $100$ epochs by minimizing the MAE loss function. Table~\ref{ch91:tab_hyper_non_private} exhibits the hyperparameters' search space and the final value used per DL method.

\setlength{\tabcolsep}{5pt}
\renewcommand{\arraystretch}{1.4}
\begin{table}[!ht]
    \centering
    \scriptsize
    \caption{Search space for hyperparameters and the final configuration obtained by DL method.}
    \label{ch91:tab_hyper_non_private}
    \begin{tabular}{c c c c c c} 
    \hline
        \textbf{Hyperparameter's range} & \textbf{Step} & \textbf{LSTM} & \textbf{BiLSTM}  & \textbf{GRU}  & \textbf{BiGRU} \\ 
    \hline
         $h_1$: [25 -- 500]   & 25    &225     &500    &75    &175    \\
         $bs$: [5 -- 40]      & 5      &10     &10     &5    &5   \\
         $\eta$: [1e-5 -- 3e-3]  & --  &0.002233     &0.002303     &0.001725    &0.000289  \\
    \hline
    \end{tabular}
\end{table}

\textbf{Results and analysis.} Table~\ref{ch91:tab_results_non_private} present the performance of the developed DL models in comparison with the Baseline model based on RMSE and MAE metrics per region and the resulting mean. \textit{Notice that the metrics are in the real scale according to the number of users per region (cf. Table~\ref{tab:statistics_data}). That said, although R1 presents higher metric values, it does not necessarily mean worse results.} One solution could be normalizing the data. Besides, Fig.~\ref{fig:results_pred_non_private} illustrates for each region forecasting results for the last day of our testing set, which includes the real number of people and the predicted ones by each RNN: LSTM, GRU, BiLSTM, and BiGRU.

\setlength{\tabcolsep}{5pt}
\renewcommand{\arraystretch}{1.4}
\begin{table}[!ht]
    \centering
    \scriptsize
    \caption{Performance of the Baseline model and non-private DL models based on RMSE and MAE metrics per region and the resulting mean values.}
    \label{ch91:tab_results_non_private}
    \begin{tabular}{c c c c c c c c c} 
    \hline
    \textbf{Model} & \textbf{Metric} & \textbf{R1} & \textbf{R2} & \textbf{R3}  & \textbf{R4}  & \textbf{R5}   & \textbf{R6}  & \textbf{Mean}  \\ 
    \hline
    \multirow{2}{*}{Baseline}   & RMSE &3461.6 &  1131.8 &  1517.9 &  986.5 &  561.3 &  1362.3  &1503.6    \\
                                & MAE  &2597.5 &   839.4 &  1105.8 &  744.1 &  434.3 &   921.5  &1107.1    \\ \hline
    \multirow{2}{*}{LSTM}       & RMSE &2667.2    &1007.3    &1291.6    &887.2    &536.3    &1135.6    &1254.2    \\
                                & MAE  &2053.8    &758.1    &969.8    &662.6    &432.3    &786.0    &943.8    \\ \hline
    \multirow{2}{*}{BiLSTM}     & RMSE &2572.7    &1033.3    &1276.4    &872.7    &528.1    &1166.7    &1241.6     \\
                                & MAE  &1954.7    &781.5    &965.5    &660.8    &419.4    &808.2    &931.7    \\ \hline
    \multirow{2}{*}{GRU}        & RMSE &2539.1    &973.0     &1296.0     &953.5    &499.9    &1185.1    &1241.1    \\
                                & MAE  &1949.7    &722.8    &939.6    &740.2    &396.4    &829.1    &929.6    \\ \hline
    \multirow{2}{*}{BiGRU}      & RMSE &2560.3    &968.3    &1282.6    &832.1    &478.9    &1163.7    &\textbf{1214.3}    \\
                                & MAE  &1957.2    &717.0    &955.3    &623.0    &382.7    &807.5    &\textbf{907.1}    \\ \hline 
    \end{tabular}
\end{table}

\begin{figure}[H]
    \centering
    \includegraphics[width=0.97\linewidth]{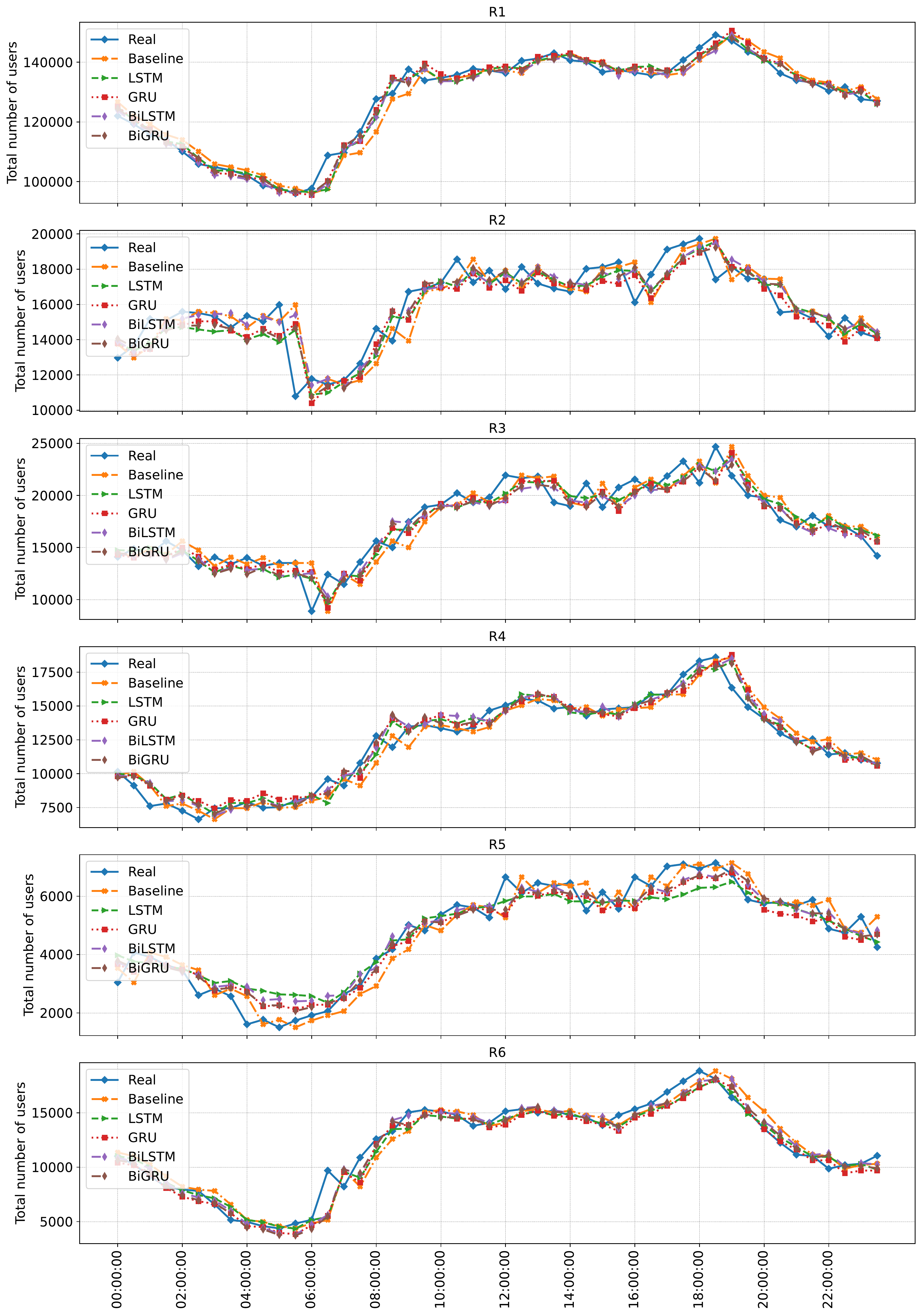}
    \caption{Multivariate time series forecast for the last day of the test set for the number of users per coarse region (R1 -- R6) by the following models: Baseline, LSTM, GRU, BiLSTM, and BiGRU.}
    \label{fig:results_pred_non_private}
\end{figure}

As one can notice, all DL models consistently outperform the Baseline model. On average, the BiGRU model outperformed all other forecasting methods, with results highlighted \textbf{in bold}. Indeed, for each region, the BiGRU consistently and considerably outperformed the Baseline model, showing the worthiness of developing DL models for this multivariate forecasting task. Similar scores were achieved by the GRU and BiLSTM models with an average RMSE around 1241. The least performing DL method in our dataset was the LSTM model. Extending the architectures, hyperparameters range, lag values (i.e., test with less or more input time steps), dropout layers, for example, could probably improve our models and change the resulting most performing technique. However, we will focus our attention on a comparative analysis of privacy-preserving DL methods in the next subsection and, thus, these possible extensions are left as future work.%

\subsection{Privacy-preserving DL forecasting models} \label{ch91:sub_private_DL}

\textbf{Methods evaluated.} We consider two privacy-preserving ML settings presented in Section~\ref{ch2:sub_dp}, namely, input perturbation (IP) and gradient perturbation (GP). Thus, we selected only the DL method that presented the smallest RMSE with original data, i.e., BiGRU (cf. Table~\ref{ch91:tab_results_non_private}). We will use BiGRU[IP] and BiGRU[GP] to indicate a BiGRU trained under input and gradient perturbation, respectively. 

For the model selection stage, \textbf{we first start with BiGRU[GP] since it allows defining a range of $\epsilon$, which is dependent on several hyperparameters of DP-SGD. For a fair comparison between both settings, we utilize the given range of $\epsilon$ to develop BiGRU[IP] models too.} Notice, however, that in both scenarios, ($\epsilon,\delta$)-DP can be ensured to each time series data sample. On the other hand, this also means that the same user may have contributed to all $n_l=3120$ training samples and, thus, in the worst case, the sequential composition theorem~\cite{dwork2014algorithmic} applies. With these elements in mind, we considered high privacy regimes ($\epsilon \ll 1$) such that the maximum $\check{ \epsilon} =  \sum_{i=1}^{n_l} \epsilon_i$ is compatible with real-world DP deployed systems~\cite{desfontaines_dp_real_world}. \textbf{This way, $\epsilon$ corresponds to the lower bound (the user appears in a single data point) and $\check{ \epsilon}$ represents the upper bound (the user appears in all data points).}

\textbf{BiGRU[GP] model selection.} In addition to standard hyperparameters $h_1$, $bs$, and $\eta$ (cf. Section~\ref{ch91:sub_non_private_DL}), we also included the TFP hyperparameters in the Bayesian optimization with $100$ iterations to minimize $loss=(RMSE_{avg} + RMSE_{std})e^{\epsilon}$; \textbf{the multiplicative factor $e^{\epsilon}$ is a penalization on high values of $\epsilon$, which varies depending on the hyperparameters used per iteration}. More specifically, given the number of training samples $n_l=3120$, we fix the following hyperparameters: the number of epochs equal $100$, $num\_microbatches=5$, $noise\_multiplier$ equal $\{35, 70, 140, 500\}$, respectively, and $\delta=10^{-7}$, which respects $\sum_{i=1}^{n_l} \delta_i < 1/n_l$~\cite{dwork2014algorithmic}. This way, we varied $h_1$, $bs$, $\eta$, and $l2\_norm\_clip$ according to Table~\ref{ch91:tab_hyper_GP}, which exhibits the hyperparameters' search space, the final value used per BiGRU[GP] model, \textbf{and the resulting privacy guarantee $\epsilon$ calculated with the \texttt{compute\_dp\_sgd\_privacy} function}~\cite{tf_privacy}, and the overall $\check{ \epsilon} =  \sum_{i=1}^{n_l} \epsilon_i$. Lastly, all BiGRU[GP] models also used ``relu" as activation function and were trained using the differentially private adam optimizer by minimizing the MAE loss function. 

\setlength{\tabcolsep}{5pt}
\renewcommand{\arraystretch}{1.4}
\begin{table}[ht]
    \scriptsize
    \centering
    \caption{Search space for standard and TFP hyperparameters, the final configuration per BiGRU[GP] model, the final privacy guarantee $\epsilon$ per time-series sample, and the maximum $\check{ \epsilon}$ following the sequential composition theorem~\cite{dwork2014algorithmic}.}
    \begin{tabular}{c c c c c}
    \hline
         \textbf{Hyperparameter} & \textbf{BiGRU[GP]$_1$} & \textbf{BiGRU[GP]$_2$} & \textbf{BiGRU[GP]$_3$}  & \textbf{BiGRU[GP]$_4$}  \\\hline
         $h_1$: [25 -- 500]                &500       &425         &275      &475 \\
         $bs$: [5 -- 40]                   &5         &5           &10        &5       \\
         $\eta$: [1e-5 -- 3e-3]            &0.002229  &0.000455     &0.000291  &0.001235    \\
         $l2\_norm\_clip$ : \{1, 1.5, 2, 2.5\}  &2.5      &2           &1      &2.5  \\
         $noise\_multiplier$ : \texttt{fixed}     &35      &70    &140       &500    \\\hline  
         \multirow{2}{*}{ \textbf{Privacy guarantee}}    &$\epsilon_1=0.0650$   &$\epsilon_2=0.0399$   &$\epsilon_3=0.0357$    &$\epsilon_4=0.0317$\\
         & $\check{ \epsilon}_1=202.8$       & $\check{ \epsilon}_2=124.488$       &  $\check{ \epsilon}_3=111.384$      & $\check{ \epsilon}_4=98.904$       \\
         \hline
    \end{tabular}
    \label{ch91:tab_hyper_GP}
\end{table}

\textbf{BiGRU[IP] model selection.} We fix $\delta=10^{-7}$ and we apply the Gaussian mechanism~\cite{dwork2014algorithmic}, by varying $\epsilon$ according to Table~\ref{ch91:tab_hyper_GP} (with their respective upper bound $\check{ \epsilon}$), \textbf{to the whole time series data, as it would be done if such system had been deployed in real life. The metrics, however, are computed in comparison with original raw time series data}. Because input perturbation allows using any post-processing techniques, we used the same model selection methodology as for non-private BiGRU models to optimize the hyperparameters for BiGRU[IP] models. The resulting values per $\epsilon=[0.0650,0.0399,0.0357,0.0317]$, respectively, are: $\textrm{BiGRU[IP]}_1 :\{h_1=200, bs=5, \eta=0.001993\}$, $\textrm{BiGRU[IP]}_2 :\{h_1=275, bs=5, \eta=0.001182\}$,  $\textrm{BiGRU[IP]}_3 :\{h_1=200, bs=10, \eta=0.001333\}$, and $\textrm{BiGRU[IP]}_4 :\{h_1=200, bs=10, \eta=0.000842\}$.

\textbf{Privacy-preserving results and analysis.} Table~\ref{ch91:tab_results_private} presents the performance of differentially private BiGRU models trained under input and gradient perturbation regarding the RMSE and MAE metrics per region and the resulting mean values. We also included in Table~\ref{ch91:tab_results_private} the utility loss of differentially private BiGRU models in comparison with non-private ones, for both RMSE and MAE averaged metrics $\mathscr{E}$, calculated as:

\begin{equation}\label{ch91:eq_acc_loss}
    \mathscr{U} = \frac{ \mathscr{E}_{DP} - \mathscr{E}_{NP}} { \mathscr{E}_{NP}} \textrm{,}
\end{equation}

in which $\mathscr{E}_{NP}$ is the result of Non-Private BiGRU (cf. averaged metric values \textbf{in bold} from Table~\ref{ch91:tab_results_non_private}) and $\mathscr{E}_{DP}$ refers to the results of either BiGRU[GP] or BiGRU[IP] models. Indeed, Eq.~\eqref{ch91:eq_acc_loss} will be positive unless the differentially private model outperforms the non-private one (which is not the case in our results). 

\setlength{\tabcolsep}{5pt}
\renewcommand{\arraystretch}{1.4}
\begin{table}[ht]
    \centering
    \scriptsize
    \caption{Performance of differentially private BiGRU models based on RMSE and MAE metrics per region and the resulting mean values. The last column $\mathscr{U}$ exhibits the utility loss of differentially private BiGRU models in comparison with non-private ones, for both RMSE and MAE averaged metrics expressed in $\%$.}
    \label{ch91:tab_results_private}
    \begin{tabular}{c c c c c c c c c c c} 
    \hline
    \textbf{$\epsilon,\check{ \epsilon}$ values}& \textbf{Model} & \textbf{Metric} & \textbf{R1} & \textbf{R2} & \textbf{R3}  & \textbf{R4}  & \textbf{R5}   & \textbf{R6}  & \textbf{Mean} & \textbf{$\mathscr{U}$} \\ 
    \hline
    \multirow{2}{*}{$\epsilon_1=0.0650$}&   \multirow{2}{*}{BiGRU[GP]$_1$}   & RMSE &2561.4    &1027.3    &1254.7  &866.7    &498.5    &1145.7    &1225.7  &0.9378 \\
                                      && MAE  &1973.4    &773.8    &925.     &644.1    &397.9    &781.5    &916.0  &0.9776   \\ \cline{2-11}
    \multirow{2}{*}{$\check{ \epsilon}_1=202.8$}&  \multirow{2}{*}{BiGRU[IP]$_1$} & RMSE &2600.9 &997.1  &1304.0  &852.7  &483.8 &1175.2  &1235.6  &1.7531  \\
                                &      & MAE  &1966.0    &737.5    &957.1    &645.4    &385.1    &821.1    &918.7  &1.2753  \\ \hline
    \multirow{2}{*}{$\epsilon_2=0.0399$}&   \multirow{2}{*}{BiGRU[GP]$_2$}   & RMSE &2600.2    &956.0     &1268.5    &841.5    &515.0  &1146.3  &\underline{1221.2} & \underline{0.5672} \\
                                      && MAE  &1978.9    &709.2    &944.4    &643.3    &417.4    &769.9    &910.5  &0.3713  \\ \cline{2-11}
    \multirow{2}{*}{$\check{ \epsilon}_2=124.488$}& \multirow{2}{*}{BiGRU[IP]$_2$}   & RMSE &2592.2  &978.4  &1251.5  &854.2   &495.6    &1158.6  &\underline{1221.8}  &\underline{0.6166}  \\
                                &      & MAE  &1986.1    &737.1    &910.9    &653.9    &393.2    &813.9    & 915.9  &0.9666  \\ \hline   
    \multirow{2}{*}{$\epsilon_3=0.0357$}&   \multirow{2}{*}{BiGRU[GP]$_3$}   & RMSE &2580.5  &990.0    &1268.5    &854.5    &504.9    &1154.3  &\textbf{1225.5}  &\textbf{0.9213}  \\
                                      && MAE  &1938.8    &753.0    &942.8    &659.7    &406.3    &773.6    &912.4  &0.5808  \\ \cline{2-11}
    \multirow{2}{*}{$\check{ \epsilon}_3=111.384$}& \multirow{2}{*}{BiGRU[IP]$_3$}   & RMSE &2587.8 &1004.7 &1262.3  &843.2    &512.8    &1186.2  &\textbf{1232.9}  &\textbf{1.5307}  \\
                                &      & MAE  &1963.1    &755.8    &957.5    &636.9    &414.6    &811.8    & 923.3  &1.7824  \\ \hline   
    \multirow{2}{*}{$\epsilon_4=0.0317$}&   \multirow{2}{*}{BiGRU[GP]$_4$}   & RMSE &2560.8    &978.3    &1322.5    &836.1   &494.4    &1195.4 &1231.3  &1.3990  \\
                                      && MAE  &1956.2    &715.1    &989.2    &633.6   &392.0    &821.6   &917.9   &1.1871  \\ \cline{2-11}
    \multirow{2}{*}{$\check{ \epsilon}_4=98.904$}& \multirow{2}{*}{BiGRU[IP]$_4$} & RMSE &2562.2    &1012.2    &1351.2    &862.9  &533.5  &1168.8  &1248.4  &2.8072 \\
                                &      & MAE  &1955.6    &756.8    &1027.6  &650.1    &423.9    &826.8      &940.2  &3.6454 \\ \hline   
    \end{tabular}
\end{table}

We remarked in our experiments that since there is a sufficient number of users per time series sample (cf. Table~\ref{tab:statistics_data}), it was still possible to make accurate forecasts in both privacy-preserving ML settings with the experimented range of ($\epsilon,\delta$)-DP. Indeed, from Table~\ref{ch91:tab_results_private}, one can notice that all differentially private BiGRU models achieved averaged RMSE lower than 1250, in which the worst result achieved by BiGRU[IP]$_4$ is just $2.8072\%$ less precise than the non-private BiGRU model, comparing the utility metric for RMSE. What is more, in both gradient and input perturbation settings, differentially private BiGRU models achieved smaller error metrics than non-private LSTM, BiLSTM, and GRU models (cf. Table~\ref{ch91:tab_results_non_private}). For instance, both BiGRU[GP]$_2$ and BiGRU[IP]$_2$ reached similar scores in comparison with the non-private BiGRU model, with a utility loss of about $0.57\%$ and $0.62\%$ (for RMSE), respectively. These results are highlighted in \underline{underlined} font, which represents our best results in terms of utility, with differentially private BiGRU models. 

Interestingly, the accuracy (measured with the RMSE metric) of differentially private BiGRU models did not necessarily decrease according to more strict $\epsilon$, i.e., lower values. One can note that results with $\epsilon_2$ and $\epsilon_3$ were more accurate than with $\epsilon_1$. This way, in terms of a satisfactory \textit{privacy-utility trade-off}, both BiGRU[GP]$_3$ ($0.92\%$ less accurate) and BiGRU[IP]$_3$ ($1.53\%$ less accurate) presented adequate metrics scores while satisfying a low value of $\epsilon$, with results highlighted \textbf{in bold}. Indeed, in the worst-case scenario, a user that was present in each data point would have leaked $\check{ \epsilon}_3=111.384$ at the end of 65 days (i.e., $\epsilon \sim 1.7$ per day), which follows real-world DP systems deployed by industry nowadays~\cite{linkedin,desfontaines_dp_real_world}.

The contribution of this research is significant for those involved in urban planning and decision-making~\cite{deMontjoye2018}, providing a solution to the human mobility multivariate forecast problem through RNNs and differentially private BiGRUs. In addition, we point out the research community to the Github page mentioned in the introduction section, in which we release the mobility dataset used in this paper for further experimentation with time series, machine learning, and privacy-preserving methods. The related literature to our work includes the generation of synthetic mobility data~\cite{Ouyang2018,Mir2013,Arcolezi2020}, the development of Markov models to infer travelers’ activity pattern~\cite{Yin2018}, and the development of privacy-preserving methods to analyze CDRs-based data~\cite{app_blip,Zang2011,Mir2013,Acs2014}. Besides, the work in~\cite{luca2020deep} surveys non-private deep learning applications to mobility datasets in general. Concerning differentially private deep learning, one can find the application of gradient perturbation-based DL models for load forecasting~\cite{UstundagSoykan2019}, an evaluation of differentially private DL models in federated learning for health stream forecasting~\cite{Imtiaz2020}, the proposal of locally differentially private DL architectures~\cite{Chamikara2020}, practical libraries for differentially private DL~\cite{tf_privacy,pytorch_privacy}, and theoretical research works~\cite{DL_DP,Shokri2015}.

Lastly, Fig.~\ref{ch91:fig_results_pred} illustrates for each region forecasting results for the last day of our testing set, which includes the real number of people and the predicted ones by the following models: Baseline, non-private BiGRU, BiGRU[GP]$_3$, and BiGRU[IP]$_3$. As one can notice, similar forecasting results were achieved by both non-private and DP-based BiGRU models, which clearly outperforms the Baseline model. Lastly, between both input and gradient perturbation settings, BiGRU[GP] models took more time to execute than BiGRU[IP] models due to DP-SGD. In terms of accuracy, BiGRU[GP] models consistently outperformed BiGRU[IP] models for the same ($\epsilon, \delta$)-DP privacy level in our experiments. Nevertheless, BiGRU[GP] is trained over non-DP time-series data, which might be subject to, e.g., data leakage~\cite{data_breaches}, membership inference attacks~\cite{Pyrgelis2017,Pyrgelis2020}, and users' trajectory recovery attacks~\cite{Tu2018,Xu2017}.

\section{Conclusion and Perspectives} \label{ch91:sec_conclusion}

In this chapter, we assumed the existence of two privacy-preserving MNOs CDRs processing system that collect and release multivariate aggregate human mobility data, described at the beginning of this chapter. However, along with collecting time-series data, extracting meaningful forecasts is also of great interest~\cite{hyndman2018forecasting}. Thus, this chapter evaluated differentially private DL models in both input and gradient perturbation settings to forecast multivariate aggregated mobility time series data.

\begin{figure}[H]
    \centering
    \includegraphics[width=0.97\linewidth]{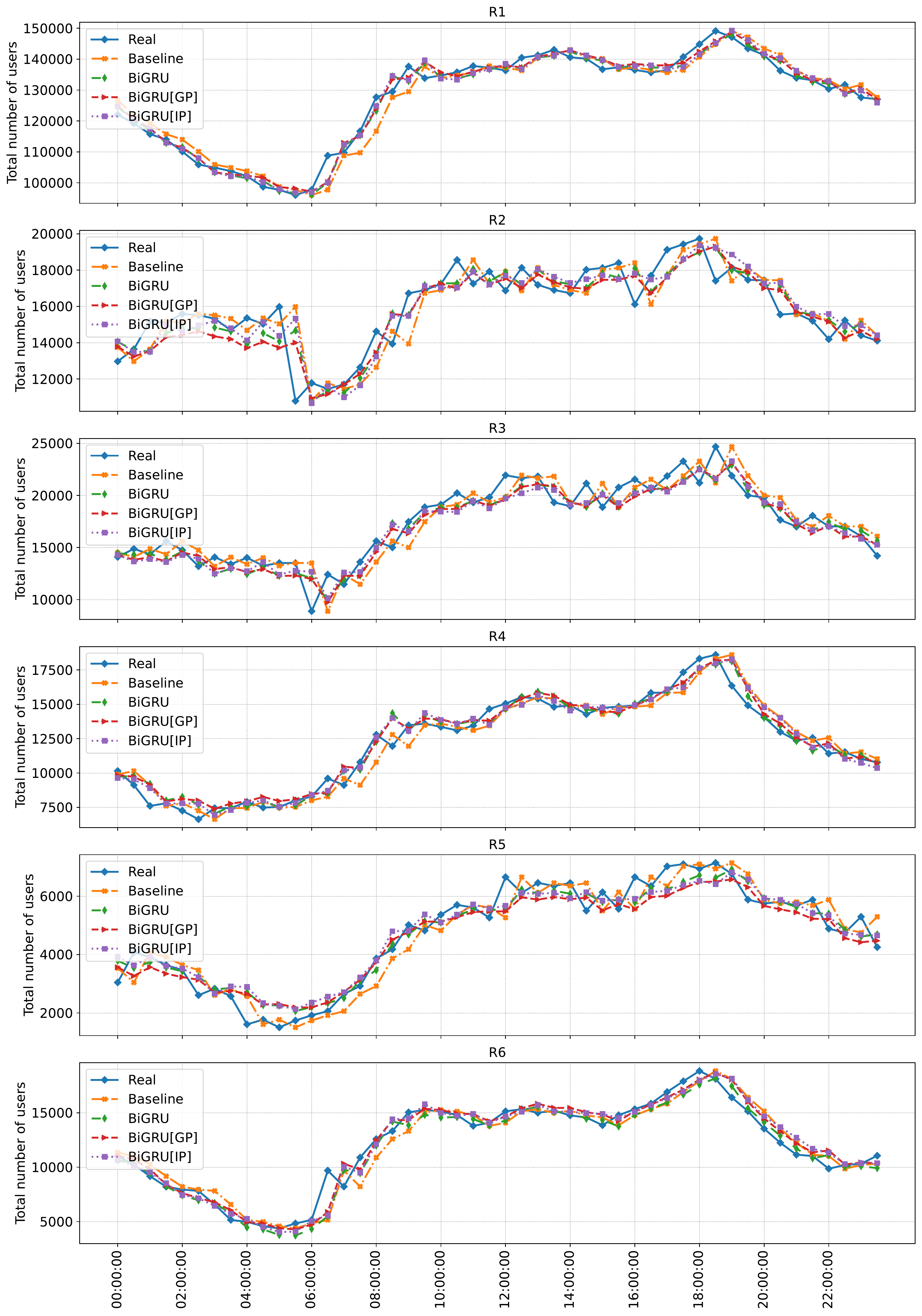}
    \caption{Multivariate time series forecast for the last day of the test set for the number of users per coarse region (R1 -- R6) by the following models: Baseline, non-private BiGRU, BiGRU[GP]$_3$, and BiGRU[IP]$_3$.}
    \label{ch91:fig_results_pred}
\end{figure}

Experiments were carried out on the dataset named Paris-DB from Section~\ref{ch3:paris_db}. First, we compared the performance of four non-private DL models (i.e., LSTM, GRU, BiLSTM, and BiGRU). Since the BiGRU model provided the highest utility, we selected it for building privacy-preserving models. Under gradient and input perturbation settings, i.e., BiGRU[GP] and BiGRU[IP], respectively, four values of $\epsilon \ll 1$ were evaluated. As shown in the results, differentially private BiGRU models achieve nearly the same performance as non-private BiGRU models, with utility loss related to the RMSE metric varying between $0.57\%$ -- $2.8\%$. 

\textbf{Thus, we conclude that it is still possible to have accurate multivariate forecasts in both privacy-preserving ML settings, favoring the gradient perturbation setting in terms of accuracy and the input perturbation setting in terms of privacy protection.} We believe that the input perturbation setting provides encouraging results for adding DP guarantees to MNOs CDRs processing systems, i.e., following the setting \textbf{S$_2$} mentioned at the beginning of this chapter. Indeed, besides being useful for forecasting tasks, DP would also add a layer of protection against, e.g., data breaches~\cite{data_breaches}, membership inference attacks~\cite{Pyrgelis2017,Pyrgelis2020}, and users' trajectory recovery attacks~\cite{Tu2018,Xu2017}. 

Some limitations and prospective directions of this chapter are described in the following. For differentially private BiGRU models, we only provided lower $\epsilon$ and upper $\check{\epsilon}$ bounds for the privacy guarantee of each sample in the time-series data. Using, however, advanced composition theorems~\cite{dwork2014algorithmic} to account for the final privacy budget for each user was out of the scope of this chapter since the Paris-DB dataset does not contain users' IDs. Besides, although the developed DL models outperform the Baseline model ($\textbf{x}_{t+1}=\textbf{x}_t$), there is plenty of room for improvements to be carried out on hyperparameters optimization (e.g., accounting for the overall privacy budget~\cite{papernot2021hyperparameter}), data scaling, the number of lag values, etc. For instance, some high-peak values were missed by both non-private and DP-based DL models (see Fig.~\ref{ch91:fig_results_pred}). In addition, we fixed the number of lagged values to 6 to predict a single step-ahead in the future (i.e., the forecasting horizon), in which the former can be tuned for performance improvement and the latter can be increased for multi-step forecasting tasks.

%% file: chapters/chapter8.tex
\chapter{Forecasting Firemen Demand by Region With LDP-Based Data} \label{chap:chapter8}

In Chapter~\ref{chap:chapter91}, we have started to evaluate the privacy-utility trade-off of differentially private machine learning models on a real-world problem concerning human mobility. From this Chapter~\ref{chap:chapter8} until Chapter~\ref{chap:chapter92}, we will focus on our second motivating project (cf. Section~\ref{ch1:motivation_objectives}), which concerns emergency medical services (EMS), in particular, using SDIS 25~\cite{sdis25} processed data by Selene Cerna. Similar to Chapter~\ref{chap:chapter91}, this chapter also focuses on multivariate time-series forecasting but is related to the number of firefighters' interventions per region (referred to as firemen demand by region throughout this chapter). While there are several examples of EMS publicly sharing their data~\cite{seattle,datagouv2,Lian2019}, we believe that more attention should be given to their victims' \textbf{privacy}. Indeed, the first question one may ask is if an intervention is a sensitive attribute. The answer is certainly yes because EMS would not have been called if the situation had not been severe enough (e.g., cardiac arrest, respiratory distress, ...). While the intention of the aforementioned EMS is laudable on publishing open-source data, there are many ways for misusing this information (e.g., discrimination in health insurance), which can jeopardize users' privacy. 

Therefore, \textbf{in collaboration with Selene Cerna}, we propose in this chapter \textbf{a methodology based on generalization and LDP, which allows EMS to properly sanitize all their data row-by-row (i.e., independently). Thus, thanks to the post-processing properties of DP~\cite{dwork2014algorithmic}, EMS could use and/or share the sanitized data with third parties to develop ML-based decision-support tools}. Indeed, our solution envisaged both \textbf{statistical learning} and \textbf{machine learning} tasks (i.e., frequency estimation and multivariate time-series forecasting of firemen demand by region). We invite the reader to refer to Chapter~\ref{chap:chapter2} for the background on LDP. Lastly, similar to Chapter~\ref{chap:chapter7}, we highlight that although we refer to our proposal as \textit{LDP-based}, this is a centralizer data owner (i.e., EMS) that applies the LDP protocol on its servers, thus, providing $\epsilon$-DP guarantees for users.

\section{Introduction} \label{ch8:introduction}

We start by recalling two recent publications of EMS French data on \texttt{data.gouv.fr}, a government site dedicated to initiatives in open data. The first concerns the 2007-2017 interventions of the Service Départemental d'Incendies et de Secours de Saône-et-Loire (SDIS 71), containing the number of interventions by type and by city~\cite{datagouv1}\footnote{Currently, this data is inaccessible on the referenced webpage.}. The second concerns the same type of data for SDIS 91 (Essonne department) for the period 2010-2018~\cite{datagouv2}. In each case, anonymization was done by aggregation: monthly for the first dataset and weekly for the second. Tables~\ref{ch8:tab_sdis71} and~\ref{ch8:tab_sdis91} exhibits five random samples of these datasets.

\setlength{\tabcolsep}{5pt}
\renewcommand{\arraystretch}{1.4}
\begin{table}[!ht]
    \scriptsize
    \centering
    \caption{Five random samples from the SDIS 71 dataset.}
    \begin{tabular}{c c c c c c c}
    \hline
     & Year\_Month &  ZIP Code &                 City &   Aid to Person  &  Fire & \ldots \\
    \hline
    1 &    2010-10 &  71093 &  LA CHAPELLE ST SAUVEUR &   1 &  0 & \ldots\\
    2 &     2012-04 &  71283 &                  MARNAY &   1 &  0 & \ldots \\
    3 &    2012-11 &  71499 &     SANVIGNES LES MINES &  11 &  2 & \ldots\\
    4  &    2013-10 &  71221 &                   GIVRY &  10 &  3 & \ldots\\
    5 &     2014-08 &  71017 &                  BALLORE &  1 &  0 & \ldots\\
    \hline
    \end{tabular}
    \label{ch8:tab_sdis71}
\end{table}

\setlength{\tabcolsep}{5pt}
\renewcommand{\arraystretch}{1.4}
\begin{table}[!ht]
    \scriptsize
    \centering
    \caption{Five random samples from the SDIS 91 dataset (cf.~\cite{datagouv2}).}
    \begin{tabular}{c c c c c c c}
    \hline
    {} &  Year &  Week &  ZIP Code &        City & Category &  Number of interventions \\
    \hline
    1  &       2018 &           24 &           91109 &  BRIERES-LES-SCELLES &          Aid to Person &       1 \\
    2  &       2018 &           39 &           91156 &        CHEPTAINVILLE &          Aid to Person &       1 \\
    3 &       2019 &           17 &           91215 &   EPINAY-SOUS-SENART &          Fire in Urban Place &       1 \\
    4 &       2019 &           28 &           91425 &            MONTLHERY &          Aid to Person &      10 \\
    5 &       2019 &           36 &           91629 &          VALPUISEAUX &          Aid to Person &       1 \\
    \hline
    \end{tabular}
    \label{ch8:tab_sdis91}
\end{table}

From Tables~\ref{ch8:tab_sdis71} and~\ref{ch8:tab_sdis91}, one can notice that the applied anonymization method is both \textit{too strong} and \textit{too weak}. \textbf{Too strong}, first of all, because carrying out one aggregation per month results in the loss of useful information by summarizing the interventions at a cloud of 120 samples (12 per year), for which only a simple linear regression remains possible: it is hard to envisage machine learning with such a dataset - this is true, to a lesser extent, for data aggregated weekly. Then, \textbf{too weak}, because this aggregation per month, or per week, was done in a blind and generalized way: if some cities have a sufficiently large number of interventions, which allows a simple temporal aggregation to achieve anonymization of the data, others conversely do not have enough. In the case of monthly aggregation, for example, there are more than 600 situations where there has been only one intervention in a city in a given month: at this level, the simple 2-anonymity~\cite{samarati1998protecting,SWEENEY2002} is no longer satisfied, and the information leakage is obvious. Such information leaks are also numerous in the case of weekly data, and \textbf{anonymization has failed for both sets of data}. For example, on analyzing the last row of Table~\ref{ch8:tab_sdis71}, we learn, for example, that in the city of Ballore (FR-71220), an intervention took place in August 2014. Considering that the city has 86 inhabitants, it would not be very difficult to find the person who received help this month. 

Moreover, a more risky case of possible data leakage concerns the Seattle Fire Department~\cite{seattle}, which \textbf{displays live EMS response information with the precise hour, location, and reason for the incident}. While the intention of publishing precise EMS data is creditable, there are many ways for (mis)using this information, which can jeopardize users' privacy. For instance, if an attacker knows that one intervention took place in front of the house of a debilitated person, they may accurately infer that this person received care and (mis)use this information for their own good.

Therefore, the objective of this chapter is to propose a methodology that allows sanitizing each interventions' data with DP guarantees while still making it possible to aggregate them and make accurate forecasts. More specifically, \textbf{our proposal also utilizes generalization}, which, first, allows generalizing the precise hour to, e.g., days, and, secondly, to generalize the space of several small cities to bigger regions. Then, besides agglomeration of cities to regions, \textbf{we also propose that each intervention's region be $\epsilon$-DP, which corresponds to applying an $\epsilon$-LDP protocol row-by-row}. This way, EMS could share the sanitized dataset containing the generalized timestamp plus the $\epsilon$-DP interventions' location data. \textbf{On the analyst side, one could, therefore, aggregate the sanitized data through frequency estimation (cf. Section~\ref{ch2:sub_ldp}) by different periods of interest}, e.g., daily, 4-days period, weekly, monthly, and so on, which would correspond to a \textbf{multivariate time-series dataset}.

In this chapter, we carried out our experiments with a real-world dataset collected by SDIS 25~\cite{sdis25} named \textbf{Interv-DB} from Section~\ref{ch3:interv_dbs}. The Interv-DB has information about $382046$ \textbf{interventions} attended by the fire brigade from 2006 to 2018 inside their department. Our experiments are separated into two parts. The first concerns \textbf{statistical learning}, i.e., to estimate the frequency of firemen demand by region in different periods using the state-of-the-art OUE~\cite{tianhao2017} protocol for single attribute frequency estimation (cf. Section~\ref{ch2:sub_ldp}). The second part is dedicated to \textbf{one-step-ahead forecasting} of firemen demand by region using daily estimated frequencies with the state-of-the-art XGBoost~\cite{XGBoost} ML technique.

The remainder of this chapter is organized as follows. Section~\ref{ch8:sec_methodology} introduces our proposed methodology to sanitize EMS interventions' location data. Section~\ref{ch8:sec_freq_est} presents experiments on frequency estimation of firemen demand by region in different aggregation periods. Section~\ref{ch8:sec_forecasting} evaluates the privacy-utility trade-off of our methodology through training differentially private ML models over the sanitized multivariate time-series data. In Section~\ref{ch8:conclusion}, we conclude this work. A preliminary version of the proposed LDP-based methodology in Section~\ref{ch8:sec_methodology} with its results was published in a full article~\cite{Arcolezi2020} in the Computers \& Security journal.

\section{Proposed LDP-Based Methodology}~\label{ch8:sec_methodology}

Fig.~\ref{ch8:fig_system_overview} illustrates the overview of the proposed privacy-preserving methodology that allows \textit{EMS to sanitize their data}. We summarize our proposal in the following. 

\begin{figure}[!htb]
	\centering
	\scriptsize
	\includegraphics[width=1\linewidth]{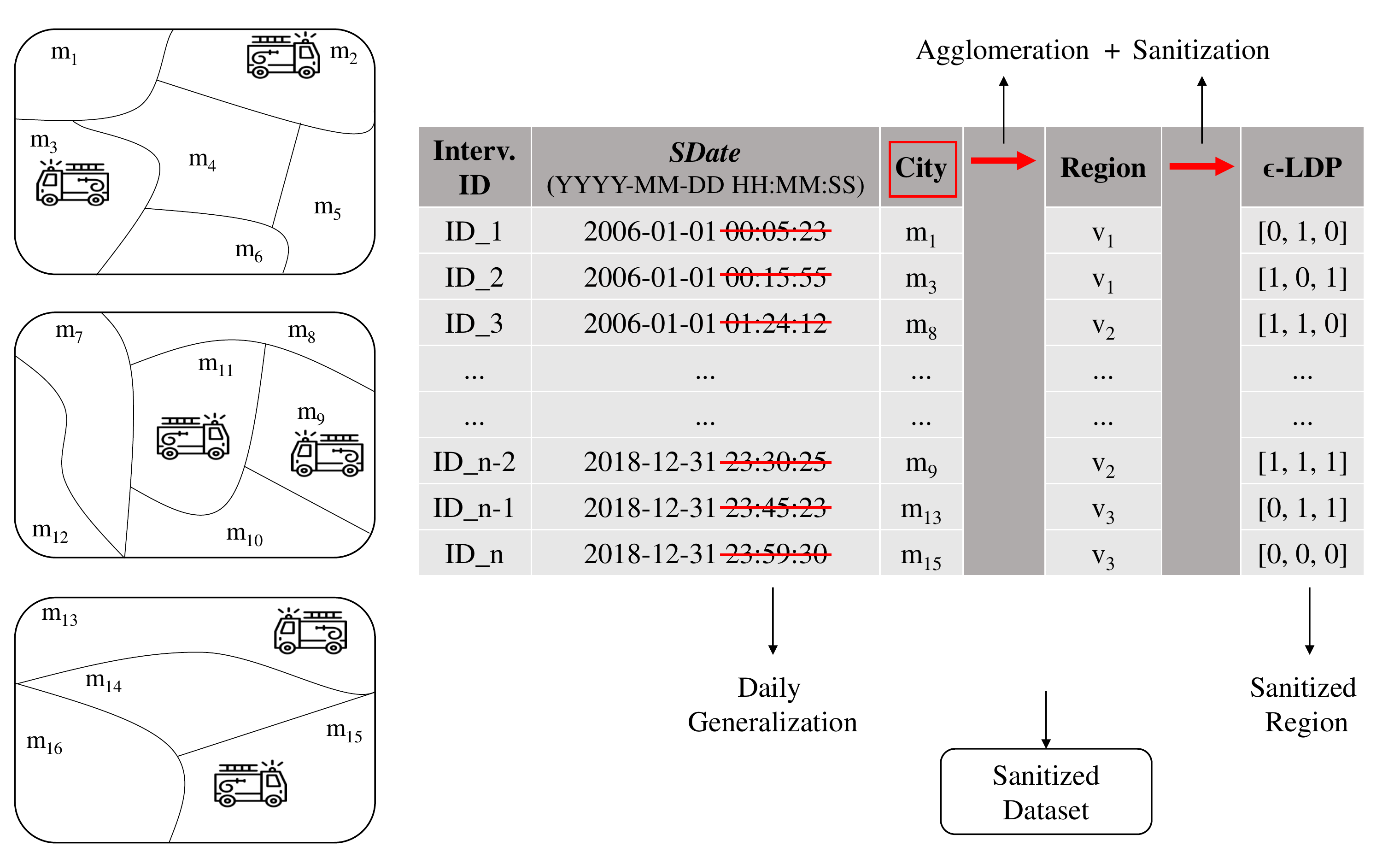}
	\caption{Overview of our LDP-based methodology to sanitize each firemen intervention independently.}%
	\label{ch8:fig_system_overview}
\end{figure}

\textbf{Sanitization of Interventions' Data (EMS Side).} From Fig.~\ref{ch8:fig_system_overview}, the first step to guarantee the privacy of each interventions' data is through \textbf{generalization} of the starting date of the intervention (i.e., \textit{SDate}) and the \textbf{agglomeration} of the location attribute. For the former, we adopt a \textbf{daily generalization} (not so strict as the datasets of Tables~\ref{ch8:tab_sdis91} and~\ref{ch8:tab_sdis71}). For the latter, we propose a strong \textbf{agglomeration of small cities to bigger regions} in order to obtain events that are sufficiently representative in number. For example, one can notice in the left side of Fig.~\ref{ch8:fig_system_overview} that a set of $M=\{m_1,m_2,...,m_{15},m_{16}\}$ small cities are grouped to a set $A=\{v_1,v_2,v_3\}$ of $3$ regions. 

In this chapter, using the data at our disposal, \textbf{$608$ cities where interventions happened in the Doubs department were generalized to $17$ regions} using the public dataset available in~\cite{comcom}. The set $A=\{v_1,v_2,...,v_{17}\}$ of 17 regions are: (1) CA du Grand Besançon, (2) CA Pays de Montbéliard Agglomeration, (3) CC Altitude 800, (4) CC de Montbenoit, (5) CC des Deux Vallées Vertes, (6) CC des Lacs et Montagnes du Haut-Doubs, (7) CC des Portes du Haut-Doubs, (8) CC du Doubs Baumois, (9) CC du Grand Pontarlier, (10) CC du Pays d'Héricourt, (11) CC du Pays de Maîche, (12) CC du Pays de Sancey-Belleherbe, (13) CC du Plateau de Frasne et du Val Rasne et du Val de Drugeon (CFD), (14) CC du Plateau de Russey, (15) CC du Val de Morteau, (16) CC du Val Marnaysien, (17) CC Loue-Lison. Fig.~\ref{ch8:fig_map_doubs} illustrates the department of Doubs with the respective cities and their agglomeration to regions.

\begin{figure}[!ht]
	\centering
	\subfloat{\includegraphics[width=0.4\linewidth]{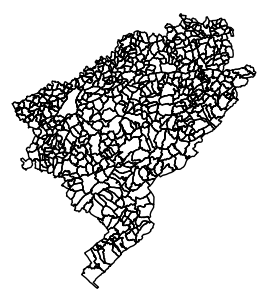}}
    \subfloat{\includegraphics[width=0.4\linewidth]{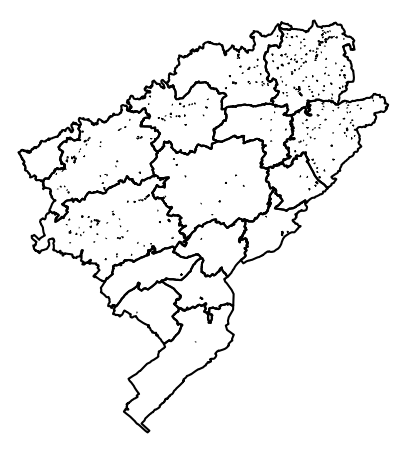}}
	\caption{cities in the department of Doubs agglomerated by regions.}
	\label{ch8:fig_map_doubs}
\end{figure}

Up to now, the generalized dataset might still have unique events for a single day in a given region. However, \textbf{to improve the level of privacy for each intervention}, we propose to apply the OUE~\cite{tianhao2017} LDP protocol row-by-row since it is independent on the number of regions. As presented in Section~\ref{ch2:sub_ldp}, given the original region $B=Encode(v)$ and the privacy budget $\epsilon$, OUE reports a sanitized bit-vector $B'$ where $\Pr[B_i'=1] =p=\frac{1}{2}$ if $B_i=1$ and $\Pr[B_i'=1] =q=\frac{1}{e^{\epsilon+1}}$ if $B_i=0$. Therefore, the final sanitized dataset (see the right side of Fig.~\ref{ch8:fig_system_overview}) would be both generalized in terms of time and space and, besides that, it would be $\epsilon$-DP, thus enhancing the privacy of each interventions' region.

\textbf{Generating Synthetic Multivariate Time Series Datasets (Analyst Side).} At this point, we assume that the analyst possess a sanitized dataset with generalized timestamp information and the $\epsilon$-DP sanitized region (see the bottom right of Fig.~\ref{ch8:fig_system_overview}). Since $\epsilon$ is a public parameter, \textbf{the analyst can define specific aggregation periods of their choice} (e.g., day, 3-days period, week, ...) and estimate the frequency $f(v_i)$ of firemen demand by region $v_i$ with $\hat{f}(v_i) = \frac{N_i - nq}{n(p - q)}$, where $N_i$ is the number of times the bit $i$ has been reported and $n$ is the number of interventions (cf. Eq.~\eqref{eq:est_pure}). In this context, a synthetic \textbf{multivariate time-series dataset} will be built with all the estimated frequencies, which is considered as a non-interactive case of DP~\cite{Zhu2017,dwork2014algorithmic}. Therefore, both \textit{statistical learning} and \textit{forecasting tasks} (following the input perturbation settings of Section~\ref{ch3:input_perturbation}) could be performed with the aggregated dataset, while preserving privacy of the individuals concerned.

\section{Frequency Estimation of Firemen Demand by Region}\label{ch8:sec_freq_est}

\subsection{Setup of experiments} \label{ch8:sub_setup}

\textbf{Environment.} All algorithms were implemented in Python 3.8.8 with NumPy 1.19.5 and Numba 0.53.1 libraries. In all experiments, we report average results over 100 runs as LDP algorithms are randomized.

\textbf{Dataset.} We applied our proposed LDP-based methodology from Section~\ref{ch8:sec_methodology} to the Interv-DB. The initial transformed dataset has both daily temporal information (the \textbf{when}) and the regions (the \textbf{where}) $382,046$ interventions took place from $2006$ until $2018$ (e.g., see the right-side of Fig.~\ref{ch8:fig_system_overview}). This transformed \textbf{original dataset} will be used for frequency estimation experiments in this section and for forecasting tasks in the next Section~\ref{ch8:sec_forecasting_results}.

\textbf{Methods evaluated.} In this chapter, we only applied the state-of-the-art OUE~\cite{tianhao2017} LDP protocol for single attribute frequency estimation.

\textbf{Evaluation and metrics.} We considered \textbf{three different aggregation periods} in our experiments. The first aggregation period we analyze is with \textbf{yearly} data ($\tau=13$ frequency estimates), which allows at the beginning of a year the fire brigade to better distribute their budget around its centers according to the firemen demand by region. Next, a \textbf{monthly} scenario ($\tau=156$ frequency estimates) is considered. And, similar to before, the fire brigade can have high-utility statistics from a third-party company to reorganize budgets and personnel each month. Lastly, a \textbf{daily} scenario ($\tau=4748$ frequency estimates) is taken into consideration such that ML algorithms could be applied in this amount of data.

For each aggregation period, similar to Chapter~\ref{chap:chapter6}, we vary the privacy parameter in a logarithmic range as $\epsilon=[\ln (2),\ln (3),...,\ln (7)]$. We use the MSE metric averaged by the number of aggregation periods $\tau$ to evaluate our results. Thus, for each time interval $t\in [1,\tau]$, we compute for each value $v_i \in A$ the estimated frequency $\hat{f}(v_i)$ and the original one $f(v_i)$ and calculate their differences. More precisely,

\begin{equation}
    MSE_{avg} = \frac{1}{\tau} \sum_{t \in [1,\tau]} \frac{1}{|A|} \sum_{v_i \in A}(f(v_i) - \hat{f}(v_i) )^2 \textrm{.}
\end{equation}

\subsection{Frequency Estimation Results} \label{ch8:sub_results_freq_est}

Fig.~\ref{ch8:fig_mse_time_scenarios} shows the relationship between $MSE_{avg}$ and $\epsilon$ for all three aggregation periods, i.e., daily, monthly, and yearly. Moreover, for the sake of illustration, Fig.~\ref{ch8:fig_ex_freq_est} exhibits the estimate frequency of firemen demand by region for the year $2013$, a month of $2017$, and a given day in 2016, with the three values for $\epsilon=[\ln(7),\ln(4),\ln(2)]$ (i.e., a low, a medium, and a high privacy guarantee). All three specific dates were chosen at random for illustration purposes.

\begin{figure}[!ht]
	\centering
	\includegraphics[width=0.85\linewidth]{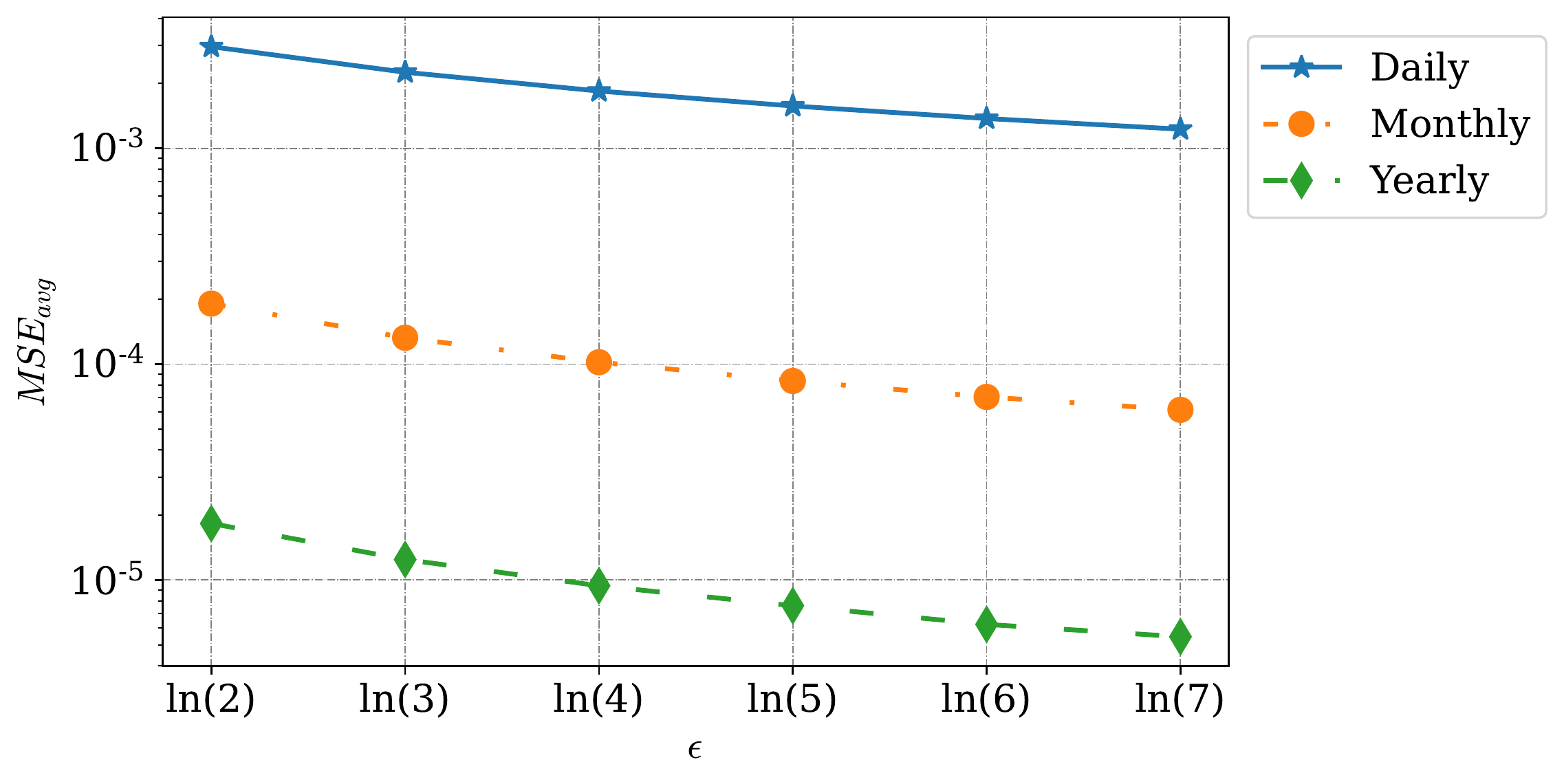}
	\caption{Analysis of $MSE_{avg}$ (y-axis) varying $\epsilon$ (x-axis) for each aggregation period: daily, monthly, and yearly.}
	\label{ch8:fig_mse_time_scenarios}
\end{figure}

\begin{figure}[!ht]
	\centering
	\subfloat{\includegraphics[width=1\linewidth]{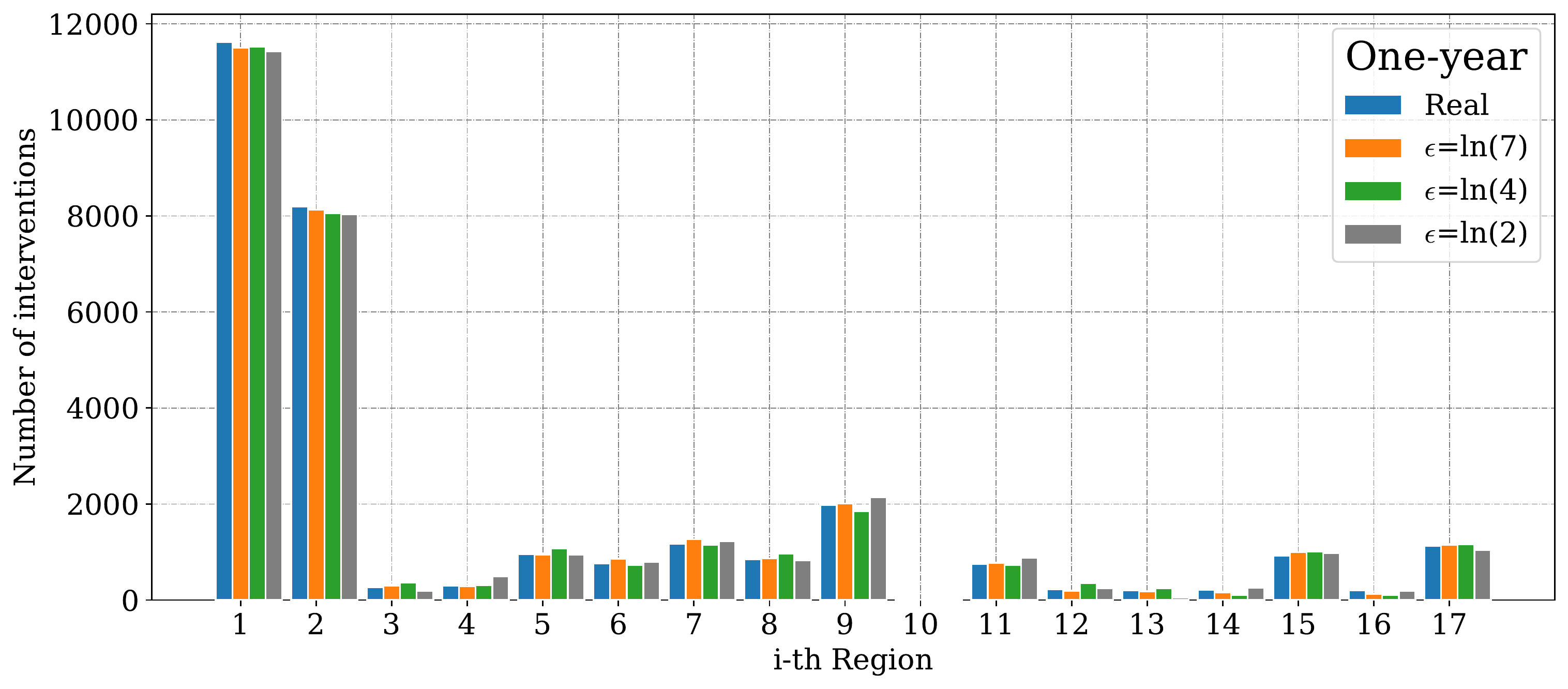}}\\
    \subfloat{\includegraphics[width=1\linewidth]{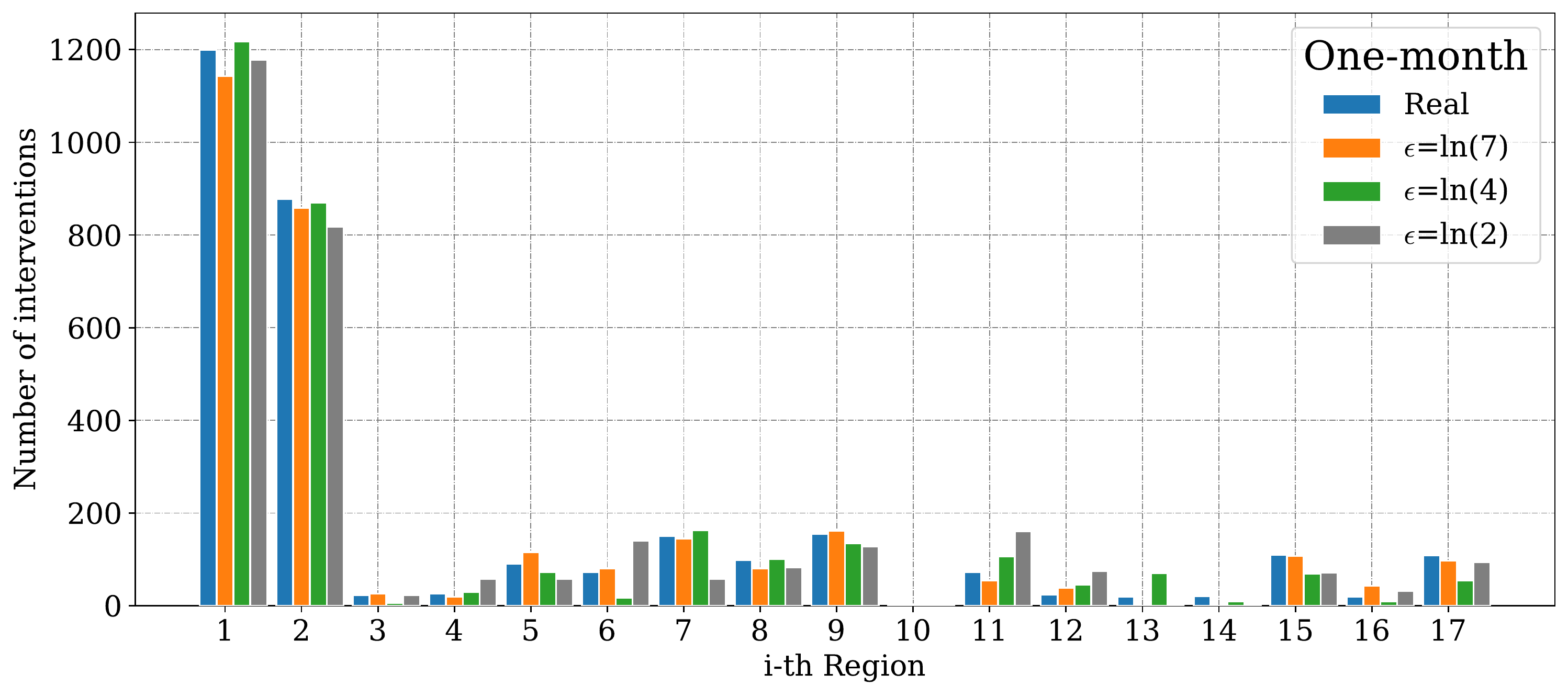}}\\
    \subfloat{\includegraphics[width=1\linewidth]{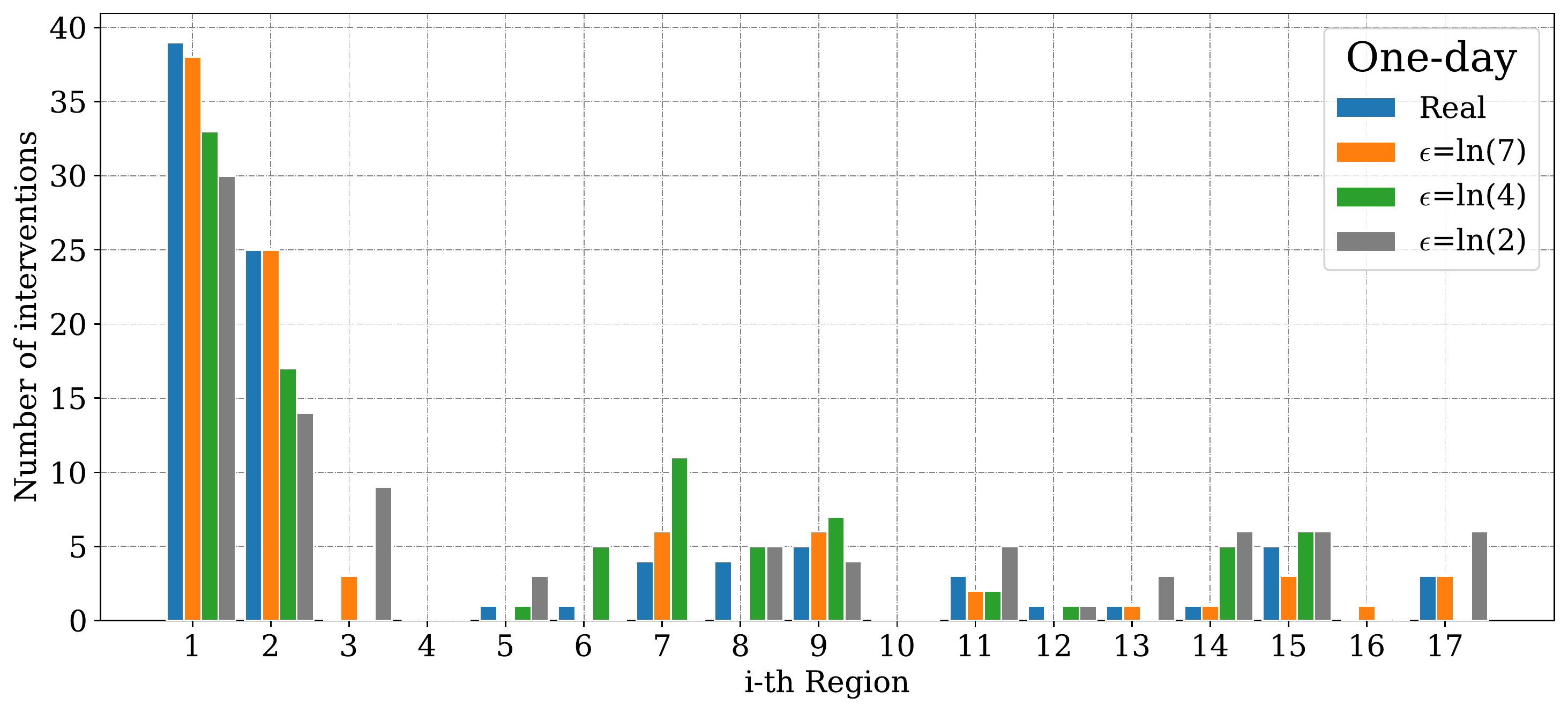}}\\
	\caption{Comparison between the original and estimated firemen demand by region for one-year, one-month, and one-day periods, respectively.}
	\label{ch8:fig_ex_freq_est}
\end{figure}

Notice that such kind of frequency estimation experiments allow evaluating the relationship between $MSE_{avg}$ versus data size (i.e., period of analysis) according to $\epsilon$ in order to find the best privacy-utility trade-off for different applications. For instance, each scenario allows the fire brigade to have a sanitized database of intervention's region where third party companies or the human resources department itself could acquire high-utility statistics. 

More specifically, as one can notice in Figures~\ref{ch8:fig_mse_time_scenarios} and~\ref{ch8:fig_ex_freq_est}, estimating the frequencies of firemen demand by region with OUE can be achieved high accuracy for different aggregation periods. Intuitively, the $MSE_{avg}$ decreases as the data size increases, with a difference of 1 order of magnitude for each aggregation scenario. In fact, this is because the variance of OUE is inversely proportional to the number of users $n$ (cf. Eq.~\eqref{eq:var_oue}). For example, for a one-year analysis, the number of interventions is at least $17333$ in 2006 (and higher the other years), while the average per day is just $47$ for the same year. For this reason, the utility of the data decreases for small aggregation periods. 

Hence, one has to balance the application of the sanitized data. For instance, if one intends to acquire statistics per year, results are very accurate with good privacy guarantees. However, if one intends to apply machine learning tasks to this data (as presented in the next section), a one-day scenario is more appropriate but with a higher estimation error. For instance, in Fig.~\ref{ch8:fig_mse_time_scenarios}, one can see the estimated frequencies for each period, where there are small estimation errors for the one-year scenario but considerable ones for both one-month and one-day scenarios.

Lastly, as also highlighted in the literature, the choice of $\epsilon$ depends on several factors (data size, the application domain) and one has to appropriately balance it considering the privacy of users and utility of data. In our case, as $608$ cities were generalized to $17$ regions, privacy could be slightly decreased (higher $\epsilon$ values) to acquire good utility for generating statistics. In the literature, common values to $\epsilon$ are within the range \(0.01-10\)~\cite{Hsu}. 

\section{Differentially Private Forecasting Firemen Demand by Region}\label{ch8:sec_forecasting}

The main purpose of this section is to evaluate the privacy-utility trade-off of training a state-of-the-art machine learning algorithm, namely XGBoost~\cite{XGBoost}, over $\epsilon$-DP estimated frequencies from Section~\ref{ch8:sec_freq_est}, to forecast the firemen demand by $17$ regions.

\subsection{Setup of Forecasting Experiments} \label{ch8:sub_setup_forecasting}

\textbf{Environment.} All algorithms were implemented in Python 3.8.8 with XGBoost~\cite{XGBoost} and Scikit-learn~\cite{scikit} libraries. In all experiments, we report average results over 10 runs as LDP algorithms are randomized (i.e., the sanitized datasets).

\textbf{Dataset.} There are $7$ datasets of frequency demand by region: the original dataset and the $6$ sanitized datasets from Section~\ref{ch8:sec_freq_est}, which guarantees $\epsilon$-DP in the range $\epsilon=[\ln (2),\ln (3),...,\ln (7)]$). More formally, each dataset $X_{(t_1,t_{\tau})}$ aggregates the number of interventions per $17$ regions and corresponding time period $t \in [1, \tau]$ of \textbf{daily intervals}. That is, $X_{(t_1,t_{\tau})} = [\langle t_1, \textbf{x}_{1}\rangle, \langle t_2, \textbf{x}_{2} \rangle, ..., \langle t_{\tau}, \textbf{x}_{\tau} \rangle ]$, where $\textbf{x}_{t}$ is a vector of size $17$ in which each position represents the number interventions per region at time $t \in[1,\tau]$. We exclusively divided our datasets into \textbf{learning (from 2006-2017)} and \textbf{testing (the year 2018)} sets. 

\textbf{Temporal features.} For both original and sanitized datasets, we added temporal features such as: year, month, day, weekday, year day, values (1 for `yes', 0 for `no') to indicate leap years, first or last day of the month, and first or last day of the year as attributes. 

\textbf{Forecasting methodology.} In this chapter, we aim at forecasting the future firemen demand by region in the \textbf{next day}. Thus, given $X_{(t_1,t_{\tau})}$, the goal is to forecast $X_{(t_{\tau + 1})}$, i.e., \textbf{one-step-ahead forecasting}, which is unknown at time $\tau$. More precisely, we only used a \textbf{single lag value}, i.e., we used the current frequency of firemen demand by region at time $t$ as an input to predict the future frequency at time $t+1$.

\textbf{Baseline model.} We established a naive forecasting technique that describes the \textbf{average} number of interventions in each day of the week per region.

\textbf{Methods evaluated.} In order to make a multi-forecast of firemen demand by region, the ``MultiOutputRegressor" from the Scikit-learn library~\cite{scikit} is applied. In this regard, one regressor per target (region) is fitted using the XGBoost~\cite{XGBoost} regressor with the parameters: max\_depth=3, learning\_rate=0.8, and n\_estimators=100. For all other parameters, we used their default values. We tuned these hyperparameters through a random search~\cite{bergstra2012random} optimization methodology with the following ranges per parameter: max\_depth=$\{1,2,3,...,12\}$, learning\_rate=[0.1,0.9], and n\_estimators=$\{50,100,150,...,1000\}$. \textbf{Seven models were built:} \textbf{One XGBoost model trained over original data} and \textbf{six XGBoost (input perturbation-based) models trained over sanitized data} considering the aforementioned $\epsilon$-DP range to predict the firemen demand by region for all days of 2018. 

\textbf{Performance metrics.} All models were evaluated with standard time-series metrics, namely, RMSE and MAE, both explained in Section~\ref{ch3:sub_metrics}. Moreover, as it is a multi-output scenario, we only present their averaged values.

\subsection{Forecasting Results}\label{ch8:sec_forecasting_results}

Fig.~\ref{ch8:fig_metrics_forecasting} illustrates the relationship between the RMSE and MAE metrics (y-axis) with $\epsilon$ (x-axis) for the Baseline model and XGBoost ones trained over original and sanitized datasets. Lastly, Fig.~\ref{ch8:fig_ex_forecasts} illustrates the best prediction results of a single day according to the $\epsilon=\ln(2)$-DP model. In Fig.~\ref{ch8:fig_ex_forecasts}, it is illustrated the original frequency of firemen demand by region in comparison with the predicted ones by XGBoost models trained with the original and sanitized datasets for a single day of 2018.

\begin{figure}[!htb]
	\centering
	\subfloat{\includegraphics[width=0.5\linewidth]{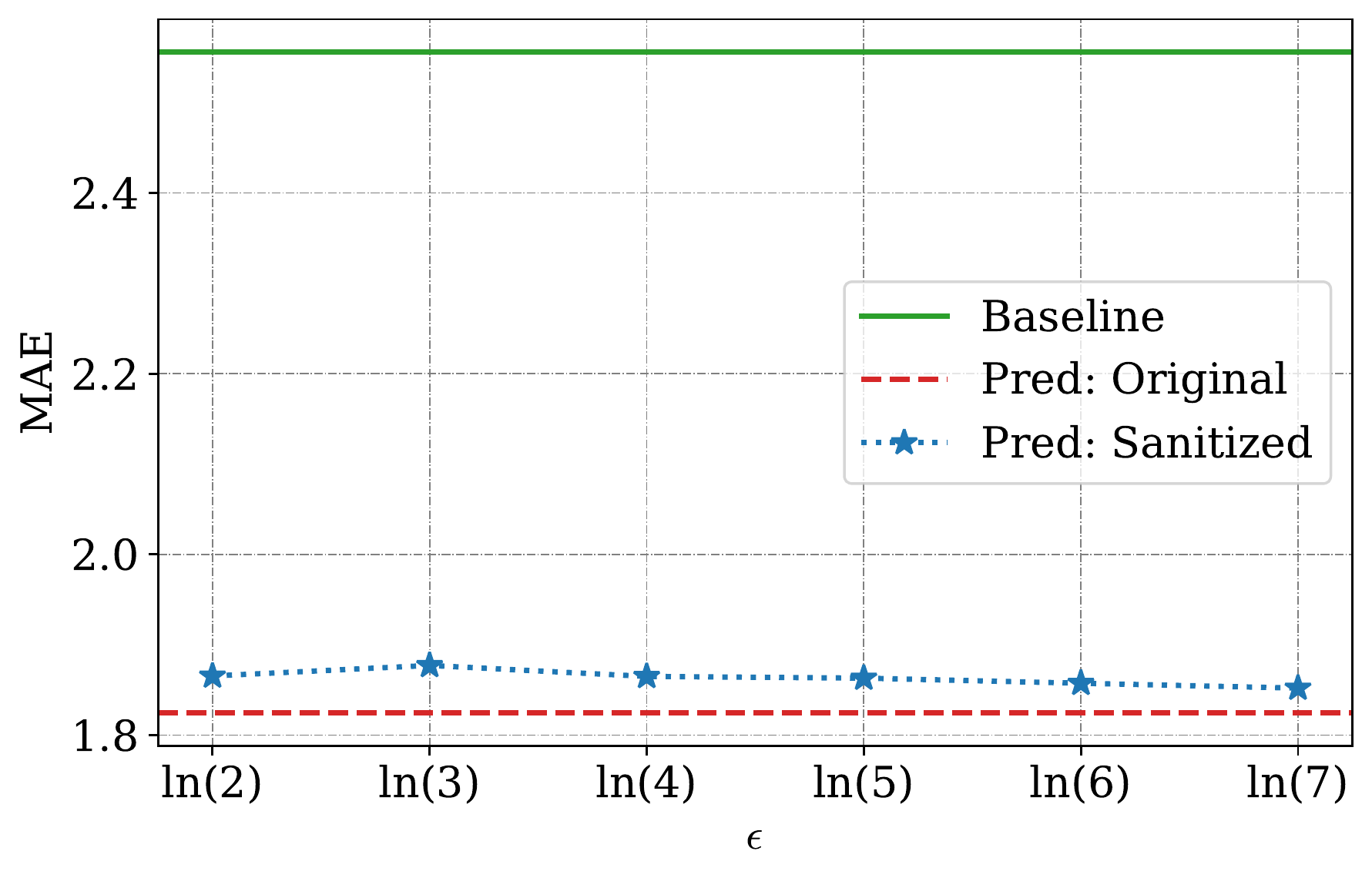}}
    \subfloat{\includegraphics[width=0.5\linewidth]{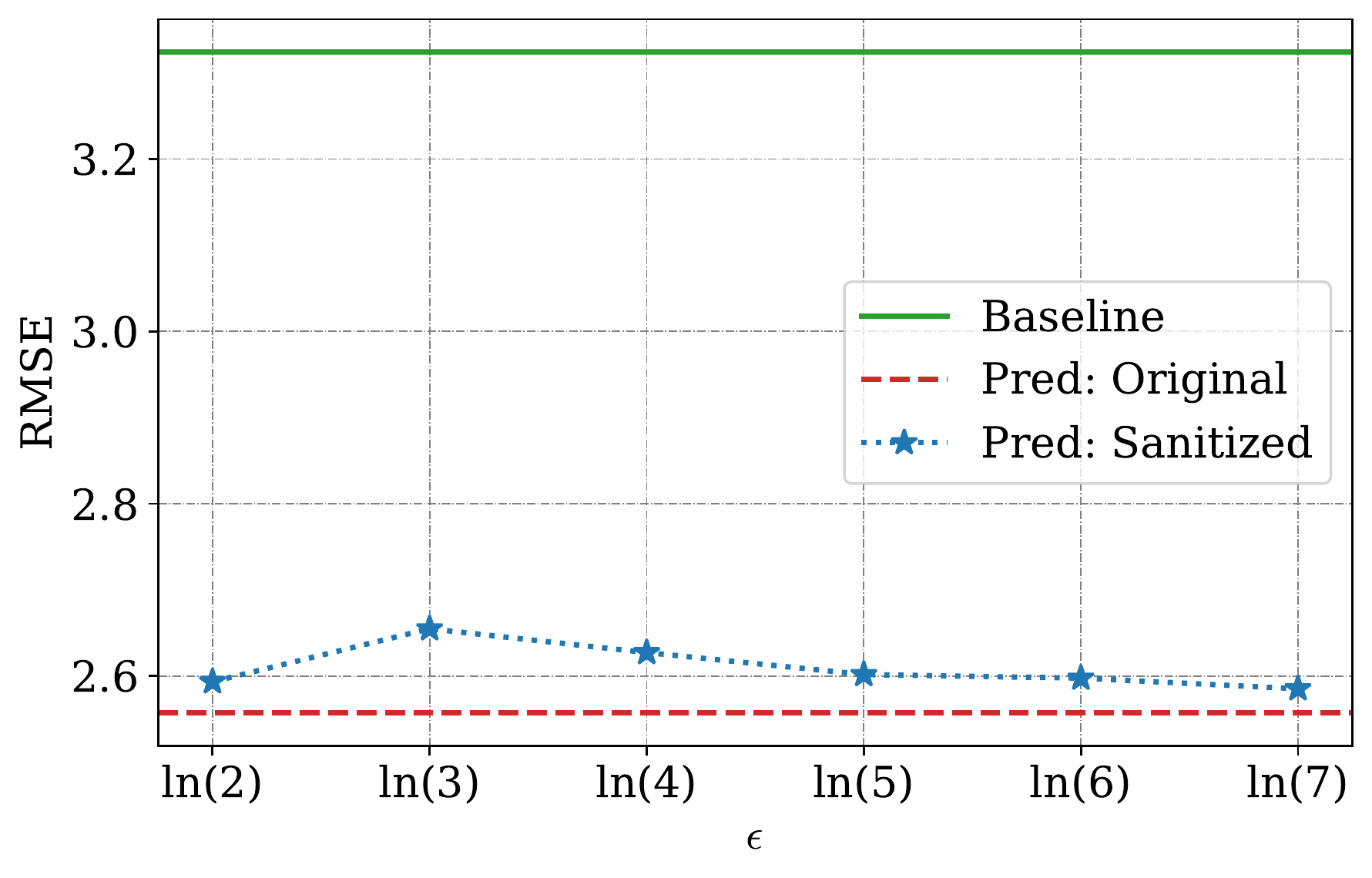}}
	\caption{MAE and RMSE metrics for the prediction models: Baseline and XGBoost trained over original and sanitized datasets.}
	\label{ch8:fig_metrics_forecasting}
\end{figure}

\begin{figure}[H]
	\centering
	\subfloat{\includegraphics[width=0.5\linewidth]{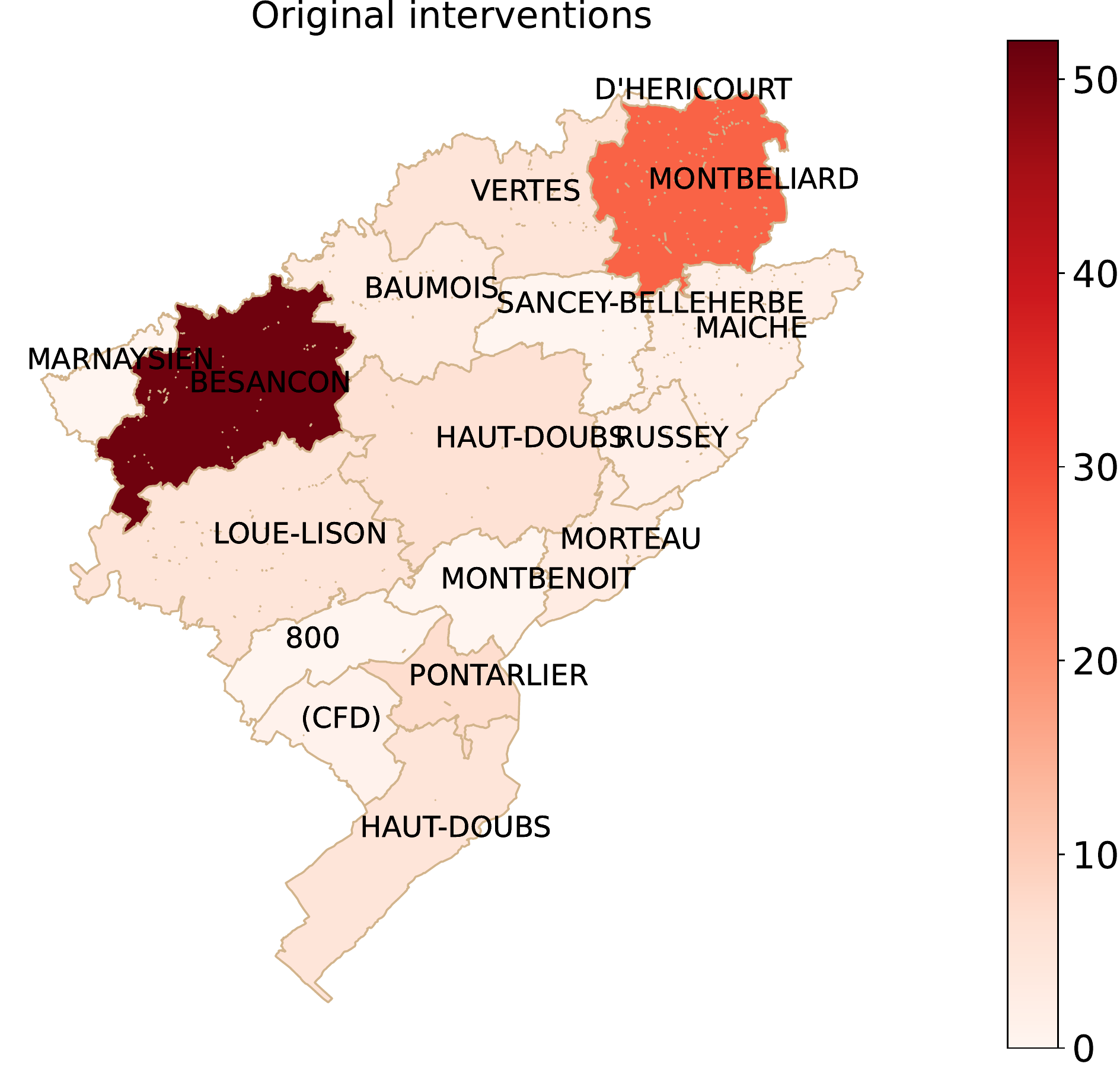}}\\
    \subfloat{\includegraphics[width=0.5\linewidth]{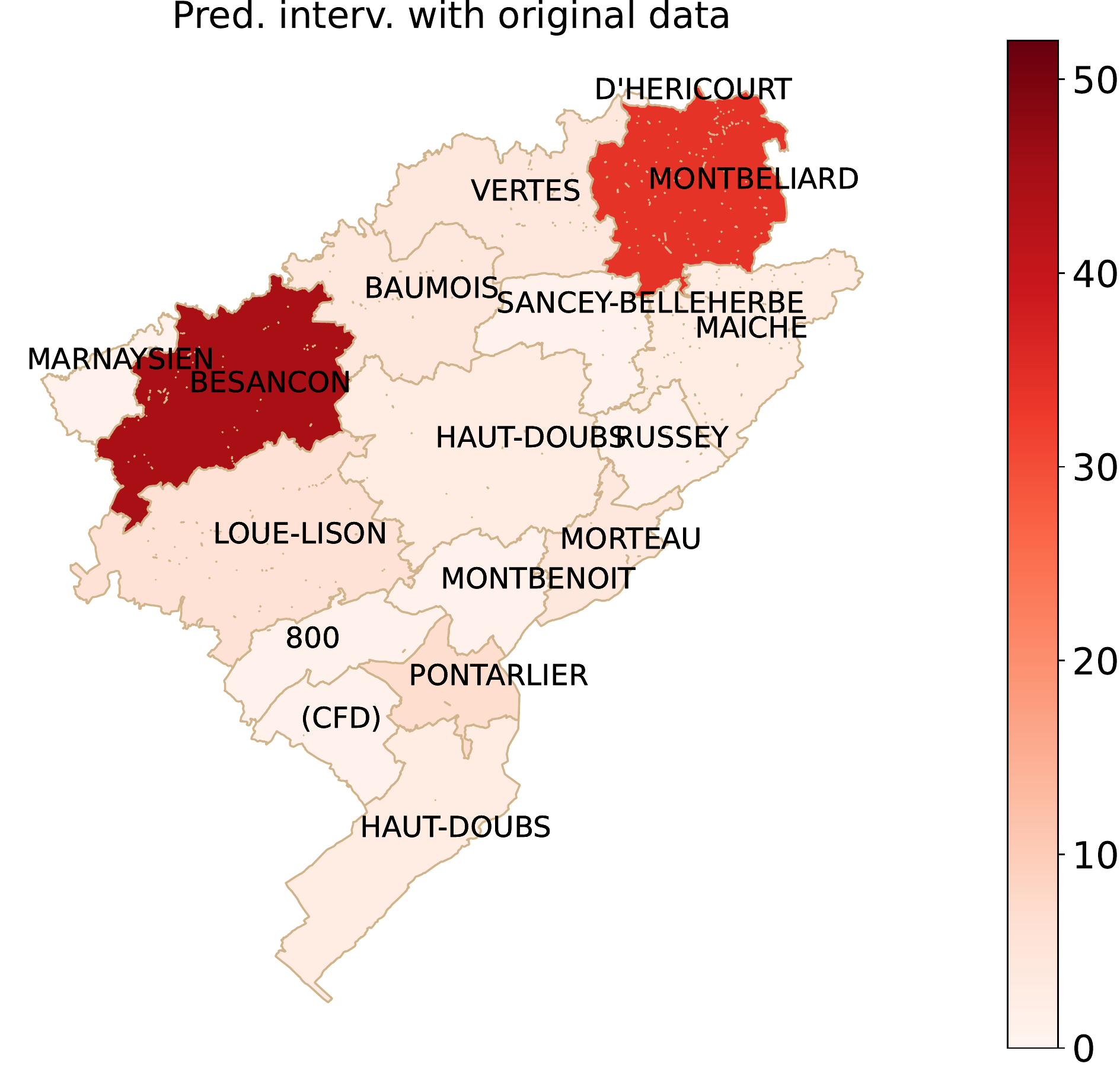}}
    \subfloat{\includegraphics[width=0.5\linewidth]{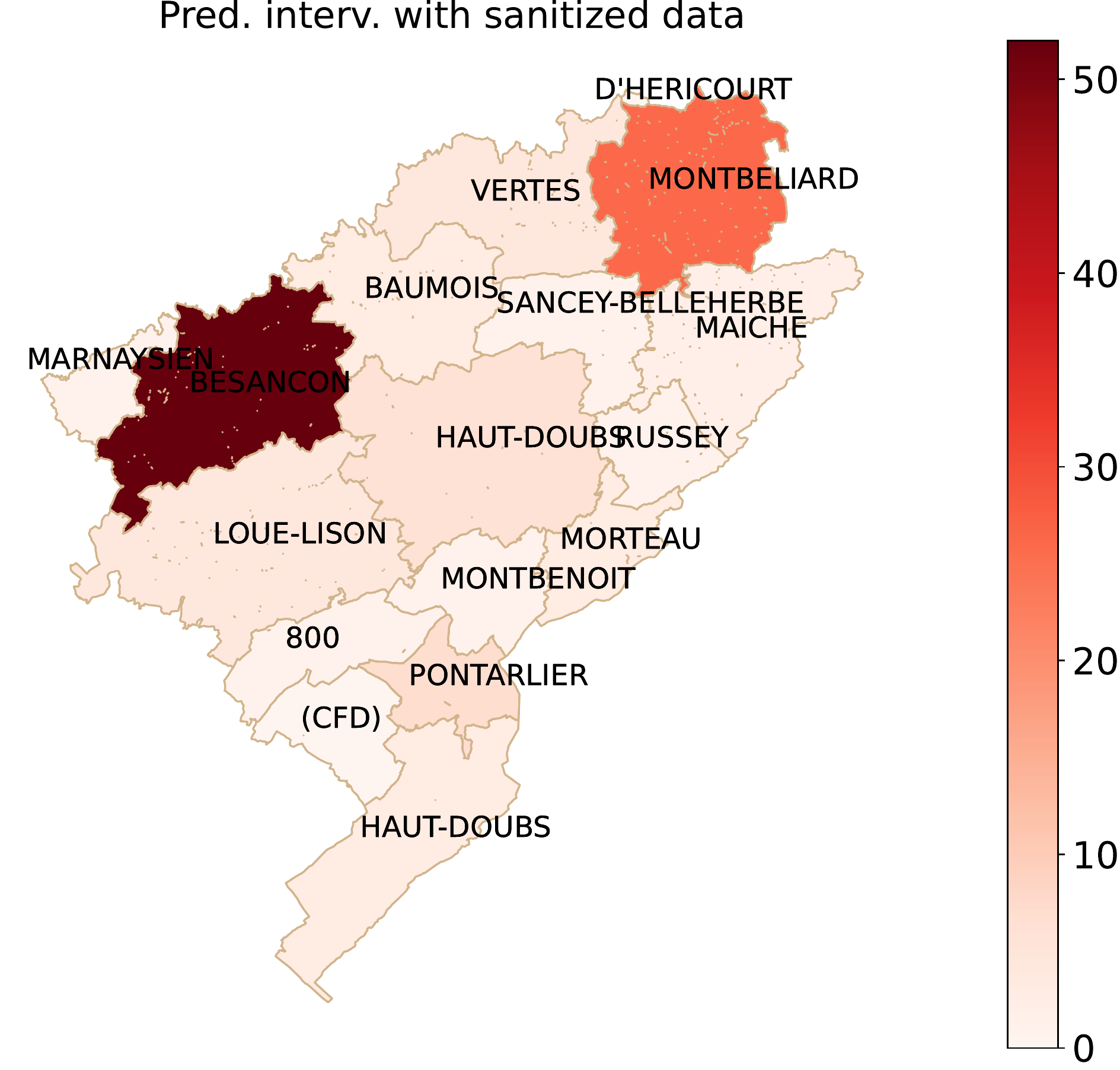}}\\
    \subfloat{\includegraphics[width=0.75\linewidth]{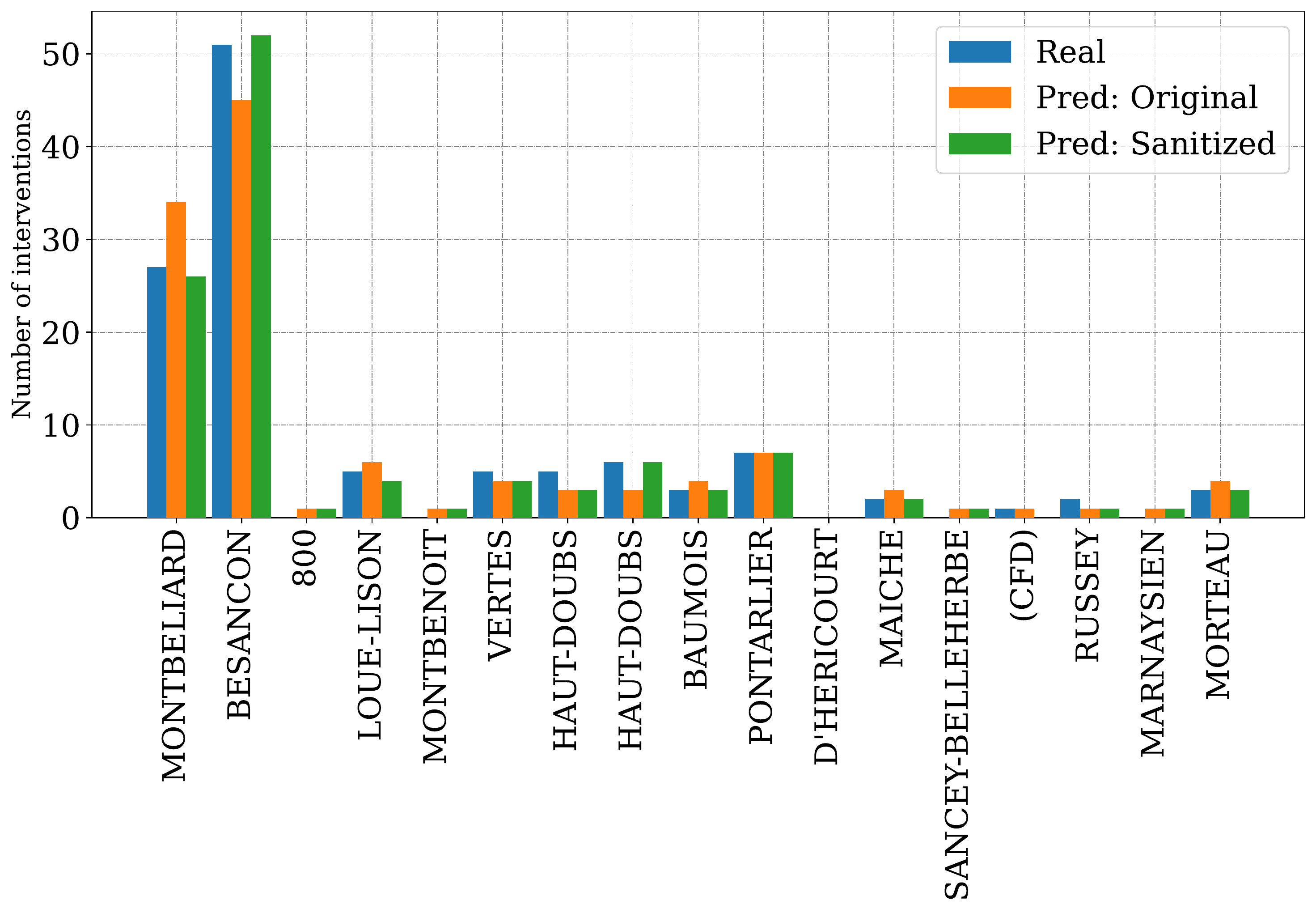}}
	\caption{Comparison of the original and predicted firemen demand by region for a single day using XGBoost trained over original data and an $\epsilon=\ln(2)$-DP one.}
	\label{ch8:fig_ex_forecasts}
\end{figure}

First, as one can notice in Fig.~\ref{ch8:fig_metrics_forecasting}, it is remarkable the improvement of the scores achieved by the XGBoost models for such complex task rather than developing a simple prediction model as the baseline (mean) assumed in this chapter. In addition, from both Figs.~\ref{ch8:fig_metrics_forecasting} and~\ref{ch8:fig_ex_forecasts}, one can notice that XGBoost models trained with sanitized data can also guarantee a good utility of the data for prediction purposes since they did not lose much utility in comparison with the model trained over original data. Indeed, this is true for the whole range of $\epsilon$ evaluated, thus, proving the effectiveness of our proposed solution, which relies on input perturbation allowing both statistical learning and forecasting tasks.

More precisely, in Fig.~\ref{ch8:fig_ex_forecasts}, it is shown for a given day of 2018 the comparison of the original and predicted firemen demand by region using the original dataset and a sanitized one with the strongest $\epsilon=\ln(2)$ tested in our experiments. As one can notice, accurate multivariate forecasts could be achieved even with a strongly sanitized dataset. With such forecasting results, the fire brigade could efficiently prepare themselves for short-, middle-, and long-term scenarios. For example, knowing that certain regions are more prospect to happen incidents, the fire brigade can better allocate the human and machinery resources as well as planning the construction of new barracks. All of these could be achieved while providing strong privacy guarantees for each intervention, using our proposed methodology.

Indeed, forecasting the operational demand is one main goal of Fire brigades (and EMS in general) to optimize their services~\cite{Grekousis2019,Couchot2019,Lin2020,Arcolezi2020,Chen2016,Couchot2019,Cerna2020_b,Pirklbauer2019,Lstm_Cerna2019,Cerna2020_boosting,EliasMallouhy2021,Guyeux2019}. For instance, in~\cite{Chen2016}, the authors identified that shorter ambulance response time is associated with a higher survival rate and predicted the demand of ambulances to allow their reallocation. Besides, in previous works of our research group~\cite{Couchot2019,Cerna2020_b,Lstm_Cerna2019,Cerna2020_boosting,EliasMallouhy2021,Guyeux2019}, several classical time-series forecasting, ML, and DL techniques have been employed to forecast the \textit{total} number of interventions considering the whole Doubs region, in different time-slots (e.g., 1 hour, 3 hours, ...). Besides, in~\cite{Couchot2019}, our group also proposed to forecast the operational demand of two main regions of Doubs and by motive, by slots of 3 hours. In that work, input perturbation was considered through applying \textit{k}-anonymity~\cite{samarati1998protecting,SWEENEY2002}, \textit{l}-diversity~\cite{Machanavajjhala2006}, and centralized DP~\cite{Dwork2006,Dwork2006DP,dwork2014algorithmic} algorithms to sanitize the original dataset. The main difference between the work in~\cite{Couchot2019} with this chapter, is that an LDP mechanism is used to sanized row-by-row independently, which also permits the data analyst to aggregate by different slots of time (cf. Section~\ref{ch8:sec_freq_est}).

\section{Conclusion} \label{ch8:conclusion}

In this chapter, we proposed a privacy-preserving methodology based on generalization and LDP, which would allow EMS to use and/or share the sanitized database for both \textbf{statistical learning} and \textbf{forecasting tasks} on the frequency of firemen demand by region. Indeed, while there are examples of EMS open data publication~\cite{datagouv1,datagouv2}, we believe that more attention should be given to the privacy of the victims concerned. For instance, EMS should not blindly generalized the number of interventions per month/week and city, as there could be many cases of unique interventions (e.g., see Tables~\ref{ch8:tab_sdis91} and~\ref{ch8:tab_sdis71}). Moreover, in a more non-private case, we argue that publishing the \textbf{precise information about the time, the location, and the reason of the emergency} (e.g., as in~\cite{seattle}), \textbf{could increase the possibility of breaching someone's privacy}.

Therefore, in our solution, we propose that both the time and the location be generalized, the former by \textit{day} and the latter by big regions (e.g., we generalize $608$ cities to $17$ regions in Fig.~\ref{ch8:fig_map_doubs}). In addition to generalization, we also propose that an $\epsilon$-LDP protocol (cf. Section~\ref{ch2:sub_ldp}) be applied to each interventions' region (i.e., row-by-row sanitization), thus, enhancing the privacy of users. In this chapter, we used the state-of-the-art OUE~\cite{tianhao2017} protocol for single attribute frequency estimation. As shown in the results of Section~\ref{ch8:sec_freq_est}, the OUE mechanism can adequately estimate the firemen demand by region with a good level of privacy guarantees for all three aggregation periods (see Fig.~\ref{ch8:fig_ex_freq_est}). 

Moreover, as shown in Section~\ref{ch8:sec_forecasting}, it is possible to forecast the future firemen demand by region with sanitized data as well as with the original data (cf. Fig~\ref{ch8:fig_metrics_forecasting}). More specifically, the work in this chapter shows that EMS data, which is sensitive but can be very useful, can be properly sanitized to avoid information leakage while remaining useful for both statistical learning (cf. Fig.~\ref{ch8:fig_ex_freq_est}) and forecasting (cf. Fig.~\ref{ch8:fig_ex_forecasts}) purposes. Lastly, while this chapter focused on aggregate information, thus, applying generalization and LDP protocols for frequency estimation, the next Chapter~\ref{chap:chapter9} investigates how to sanitize the coordinates (i.e., latitude and longitude) of the emergency's location, focusing on a different problem, i.e., predicting the response time of each ambulance.

%% file: chapters/chapter9.tex
\chapter{Preserving Emergency's Location Privacy to Predict Response Time} \label{chap:chapter9}

In Chapters~\ref{chap:chapter1}, we have reviewed our second motivating project concerning the SDIS 25 (i.e., an EMS in France) and in Chapter~\ref{chap:chapter8}, we have proposed an LDP-based methodology focusing on both \textit{statistical learning} and \textit{machine learning forecasting} on the frequency of firemen demand by region. In this Chapter~\ref{chap:chapter9} and in Chapter~\ref{chap:chapter92}, \textbf{following our collaboration with Selene Cerna}, we focus our attention on a different problem, which concerns the \textbf{response time of EMS to each emergency}. Indeed, many victims require care within adequate time (e.g., cardiac arrest) and, thus, improving response time is vital. In this context, the \textbf{location} of the emergency is a determinant factor of EMS response time since it defines, e.g., the distance between the EMS center and the emergency scene. 

With these elements in mind, we asked ourselves, is the \textbf{precise location really necessary to be used as a predictor of an ML model that predicts ambulance response times?} In fact, we still consider that EMS intend to share a sanitized version of their data, such that third parties could build decision-support tools to optimize the EMS service. So, \textbf{in collaboration with Selene Cerna}, we propose in this chapter to use the geo-indistinguishability~\cite{Andrs2013} LDP model to sanitize each emergency scene independently (i.e., row-by-row). In addition, there are many other predictors that may also be `perturbed', e.g., the calculated distance between both the EMS center and the emergency scene; the estimated travel time, the city, and so on. Thus, thanks to the post-processing properties of DP~\cite{dwork2014algorithmic}, EMS could use and/or share the sanitized data with third parties to develop ML-based decision-support tools. We invite the reader to refer to Chapter~\ref{chap:chapter2} for the background on LDP and geo-indistinguishability. Lastly, similar to Chapter~\ref{chap:chapter7} and~\ref{chap:chapter8}, we highlight that although we apply an \textit{LDP-based} mechanism, this is a centralizer data owner (i.e., EMS) that applies the geo-indistinguishability protocol on its servers, thus, providing centralized privacy guarantees for users.

\section{Introduction} \label{ch9:introduction}

Ambulance response time (ART) is a key component for evaluating pre-hospital EMS operations. ART refers to the period between the EMS notification and the moment an ambulance arrives at the emergency scene~\cite{Brger2018,Byrne2019}. In many urgent situations (e.g., cardiovascular emergencies, trauma, or respiratory distress), the victims need first-aid treatment within adequate time to increase survival rate~\cite{Do2012,Byrne2019,Holmn2020,Brger2018,Chen2016,Lee2019} and, hence, improving ART is vital.

One important factor of ART is the \textit{location} of the intervention~\cite{Nehme2016,Do2012,Lian2019,Byrne2019,aladdini2010ems}, e.g., in dense urban areas, the distance may be short, but the travel time may be longer due to traffic congestion. On the other hand, travel distance and travel time may be longer for rural areas. In other words, the location information is of great importance for the prediction of travel time and, naturally, ART~\cite{Nehme2016,aladdini2010ems}. As also mentioned in Chapter~\ref{chap:chapter8}, \textbf{the location of an emergency}, on the other hand, \textbf{is considered sensitive information} since it might identify who received assistance and for what purpose. For example, attackers with auxiliary information may correctly deduce that a weakened person activated the EMS if they know that one intervention took place in front of their residence. The attackers may then exploit this knowledge for their own benefit.

In this chapter, we propose to sanitize, independently, each emergency location data with geo-indistinguishability (GI)~\cite{Andrs2013} (cf. Section~\ref{ch2:sub_geo_ind}), which is based on the state-of-the-art DP~\cite{Dwork2006,Dwork2006DP,dwork2014algorithmic} model. Indeed, we aim to evaluate the effectiveness of several values of $\epsilon$ (i.e., the privacy budget) to sanitize emergency location data with GI and train ML-based models to predict ART. In other words, this is a practical evaluation of GI on a real-world EMS task. This way, EMS would only use and/or share sanitized data with third parties to \textit{train} and develop ML-based decision support systems, thus, protecting their victims if there are data leakages~\cite{data_breaches} or if the built ML model is subject to membership inference attacks and data reconstruction attacks~\cite{Song2017,Shokri2017,Carlini2019}. 

In our context, besides the own location, with the exact coordinates of both SDIS 25 centers and the emergency scenes, one can retrieve important features such as the distance and estimated travel time. However, if the location is sanitized via GI, many other explanatory variables (e.g., distance, travel time, city) would be `perturbed' too. As reviewed in Section~\ref{ch3:input_perturbation}, training ML models with sanitized data is also known as \textit{input perturbation}~\cite{first_ldp,ppdm}. 

We perform our experiments on the SDIS 25 preprocessed dataset named \textbf{ART-DB} from Section~\ref{ch3:art_dbs}. The ART-DB contains information about $186130$ \textbf{dispatched ambulances} from SDIS 25 centers that attended $182700$ EMS interventions from 2006 up to June 2020. To the author's knowledge, this is the first work to assess the impact of geo-indistinguishability on sanitizing the location of emergency scenes when training the ML model for such an important task.

The remainder of this chapter is organized as follows. In Section~\ref{ch9:sec_materials_methods}, we describe the sanitization of emergency scenes with GI and the experimental setup. In Section~\ref{ch9:sec_results_discussion}, we present the results of our experiments and its discussion including related work. Lastly, in Section~\ref{ch9:sec_conclusion}, we present the concluding remarks. The development, results and discussion presented in this chapter were published in a full article~\cite{Arcolezi2021_geo} in the Mathematical and Computational Applications journal.

\section{Materials and Methods} \label{ch9:sec_materials_methods}
In this section, we present the GI-based sanitization of emergency location data (Section~\ref{ch9:sub_gi_location}) and the experimental setup (Section~\ref{ch9:sub_experiments}).

\subsection{Preserving emergency location privacy with geo-indistinguishability} \label{ch9:sub_gi_location}

To preserve geo-indistinguishability of each emergency scene, we apply the polar Laplace mechanism in Alg.~\ref{alg:GI_location} presented in Chapter~\ref{chap:chapter2} to the \textit{Location} attribute of \textbf{each intervention}. The codes we developed and used for all experiments are available in a Github repository\footnote{\url{https://github.com/hharcolezi/ldp-protocols-mobility-cdrs}.}. More specifically, even if the ART-DB is per ambulance dispatch (i.e., $186130$ ambulances), we used the same sanitized value per intervention (i.e., $182700$ unique interventions). Although in~\cite{Andrs2013} the authors propose two further steps to Alg.~\ref{alg:GI_location}, i.e., discretization and truncation, both steps can be neglected in our context. This is, first, because SDIS 25 may also help other EMS outside the Doubs region, and second, we assume that any location in the continuous plane can be an emergency scene. While reporting an approximate location in the middle of a river may not have much sense in location-based services, in an emergency dataset with approximate locations, this may indicate an urgency for someone who drowned in the river, for example.

We used five different levels for the privacy budget $\epsilon=l/r$, where $l$ is the privacy level we want within a radius $r$. Table~\ref{tab:epsilon_gi} exhibits the five different levels of privacy, selected similar to the original GI paper~\cite{Andrs2013}. For the sake of illustration, Figure~\ref{fig:real_VS_anon} exhibits three maps of the Doubs region with the points of original location (left-hand plot), $\epsilon=0.005493$-GI location (middle plot), and $\epsilon=0.002747$-GI location (right-hand plot). As one can notice, with an intermediate privacy level ($l=\ln{(3)},r=400$), locations are more spread throughout the map while with a lower privacy level ($l=\ln{(3)},r=200$), locations approximate the real clusters.

\setlength{\tabcolsep}{5pt}
\renewcommand{\arraystretch}{1.4}
\begin{table}[H]
    \centering
    \scriptsize
    \begin{tabular}{|c|c|c|}
    \hline
    $\epsilon=l/r$ &    $l$ &    $r$ (meters) \\ \hline
        $0.005493$ &  $\ln{(3)}$ &    $200$ \\ \hline
        $0.002747$ &  $\ln{(3)}$ &    $400$ \\ \hline
        $0.001155$ &  $\ln{(2)}$ &    $600$ \\ \hline
        $0.000866$ &  $\ln{(2)}$ &    $800$ \\ \hline
        $0.000693$ &  $\ln{(2)}$ &    $1,000$ \\ \hline
    \end{tabular}
    \caption{Values of $\epsilon=l/r$ for sanitizing emergency location data with GI.}
    \label{tab:epsilon_gi}
\end{table}

\begin{figure}[!ht]
    \centering
    \includegraphics[width=1\linewidth]{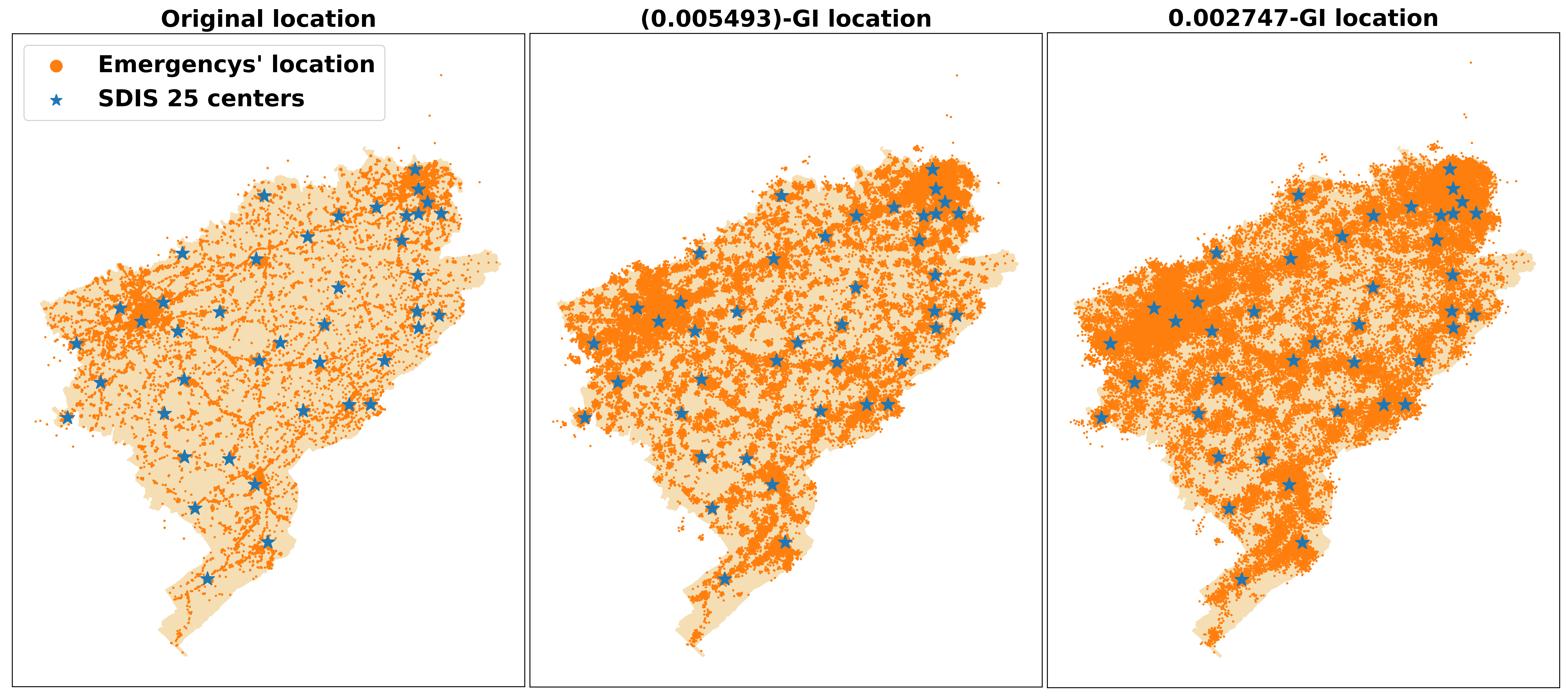}
    \caption{Emergency locations and SDIS 25 centers throughout the Doubs region: original data (left-hand plot), $\epsilon=0.005493$-GI data (middle plot), and $\epsilon=0.002747$-GI data (right-hand plot).}
    \label{fig:real_VS_anon}
\end{figure}

Moreover, with the new \textit{Location} values of each intervention, we also reassigned the city, the district, and the zone \textbf{when applicable}. In addition, we recalculated the following features associated with it: the great-circle distance~\cite{great_circle_dist}, the estimated driving distance, and estimated travel time. The latter two features were recalculated with the open source routing machine (OSRM) API~\cite{OSRM}, which only considers roads, i.e., if the obfuscated location is in the middle of a farm, the closest route estimates the driving distance and travel time until the closest road. We also highlight that if the new coordinates of the emergency scene indicate a location closer to another SDIS 25 center, even in real life, it would not imply that this center took charge of the intervention. Therefore, the \textit{center} attribute \textbf{was not} `perturbed'. 

To show the impact of the noise added to the \textit{Location} attribute, Table~\ref{ch9:tab_var_features_GI} exhibits the percentage of time that categorical attributes (zone, city, and district) were `perturbed' (i.e., reassigned); the mean and std values of the great-circle distance attribute (considering the SDIS 25 center and the emergency scene) and its Pearson correlation coefficient~\cite{correlation} with the ART variable (Corr. ART). In Table~\ref{ch9:tab_var_features_GI}, we report the mean(std) values since we repeated our experiments with 10 different seeds (i.e., DP algorithms are randomized). Although we did not include the estimated driving distance and estimated travel time from OSRM API in this analysis, in preliminary tests, we noticed that these two features follow a similar pattern as the great-circle distance attribute.

From Table~\ref{ch9:tab_var_features_GI}, one can notice that many features are perturbed due to sanitization of emergency's location with GI. With high levels of $\epsilon$ (i.e., less private), the city and the zone suffer low `perturbation'. On the other hand, district is reassigned many times as it is geographically smaller than the others. For example, in the fourth row Table~\ref{ch9:tab_var_features_GI}, when $\epsilon=0.000866$, the city is already reassigned more than $50\%$ of the time and the district about $74\%$ of the time. Moreover, one can notice that the mean and std values of the great-circle distance increase as the $\epsilon$ parameter decreases (i.e., more private). Because $\epsilon = l/r$, making $l$ smaller and/or $r$ higher, the stricter $\epsilon$ becomes, and therefore more noise is added to the original locations. Besides, the Pearson correlation coefficient between the great-circle distance with the ART variable decreases as $\epsilon$ becomes smaller. 

\setlength{\tabcolsep}{5pt}
\renewcommand{\arraystretch}{1.4}
\begin{table}[!ht]
    \centering
    \scriptsize
    \begin{tabular}{| c|c|c|c|c|c|c |}
    \hline
         \multirow{2}{*}{\textbf{Data}}    & \textbf{Zone} & \textbf{City}    & \textbf{District}    & \multicolumn{3}{c|}{\textbf{Great-circle Dist. (km)}}\\ \cline{2-7}
                     & \multicolumn{3}{c|}{`perturbation' (\%)}     & Mean  &std &  Corr. ART \\ \hline
         Original            & - & - & - &3.44    &3.72    &0.369  \\ \hline
         $\epsilon=0.005493$ &5.20(0.05) &7.68(0.06) &25.8(0.05) &3.48(1e-3) &3.72(7e-4) &0.367(2e-4) \\ \hline
         $\epsilon=0.002747$ &11.3(0.05) &17.6(0.10) &41.5(0.12) &3.57(1e-3) &3.72(1e-3) &0.362(2e-4)\\ \hline
         $\epsilon=0.001155$ &28.1(0.06) &42.3(0.10) &66.2(0.09) &4.03(3e-3) &3.74(3e-3) &0.335(5e-4)  \\ \hline
         $\epsilon=0.000866$ &35.5(0.10) &52.4(0.11) &74.0(0.11) &4.38(3e-3) &3.81(4e-3) &0.313(1e-3)  \\ \hline
         $\epsilon=0.000693$ &41.4(0.12) &60.3(0.09) &79.4(0.05) &4.77(6e-3) &3.92(5e-3) &0.288(1e-3)  \\ \hline
    \end{tabular}
    \caption{Percentage of perturbation for categorical attributes (city, zone, and district) according to $\epsilon$ and statistical properties  (mean and std values and correlation with ART) of the original and GI-based datasets for the great-circle distance attribute. Mean(std) values are reported since we repeated our experiments with 10 different seeds.}
    \label{ch9:tab_var_features_GI}
\end{table}

\subsection{Setup of Experiments} \label{ch9:sub_experiments}

Four state-of-the-art ML techniques have been used in our experiments, \textbf{to predict the scalar ART outcome} in a regression framework. More precisely, we compared the performance of two state-of-the-art ML techniques based on decision trees, which are known for their high performance (and speed) with tabular data (i.e., LGBM~\cite{XGBoost} and XGBoost~\cite{NIPS2017_6907}); a traditional and well-known deep learning (i.e., MLP~\cite{goodfellow2016deep,LeCun2015}), and a classical statistical method that can perform both variable selection and regularization (i.e., LASSO~\cite{lasso}). All these methods have been revised in Section~\ref{ch3:ML_models} 

Because in Table~\ref{ch9:tab_var_features_GI} there are low variations (i.e., small std values) on all features that depend on the sanitized location, we ran our experimental validation only once. As detailed in Section~\ref{ch3:art_dbs}, in our experiments, \textbf{each sample corresponds to one ambulance dispatch}, in which there are \textbf{temporal features} (e.g., hour, day), , \textbf{traffic data} (i.e., indicators from~\cite{bisonFute}), hourly \textbf{weather data} (e.g., temperature, pressure, ..., from~\cite{meteoFrance}), \textbf{location-based features} (latitude, longitude, district, city, and the zone), and \textbf{computed features} (e.g., the distance between the SDIS 25 center and the emergency scene, estimated travel time, estimated driving distance, where the two latter are from~\cite{OSRM}). \textbf{The scalar target variable is the ART in minutes, which is the time measured from the SDIS 25 notification to the ambulance's arrival on-scene}. 

In addition, the ART-DB was preprocessed by Selene Cerna as follows. All numerical features (e.g., temperature) were standardized using the \textit{StandardScaler} function from the Scikit-learn library~\cite{scikit}. Categorical features (e.g., center, zone, hour) were encoded using mean encoding, i.e., the mean value of the ART variable with respect to each feature (considering the training set only). \textbf{The target variable, namely ART,} was kept in its original format (minutes) since no remarkable improvement was achieved with scaling. 

With these elements in mind, we divided the ART-DB into \textbf{training (years 2006-2019)} and \textbf{testing (six months of 2020)} sets to evaluate our models. Thus, five models per ML technique (i.e., XGBoost, LGBM, MLP, and LASSO) were built to predict ART on each month of 2020 using the sanitized (training) datasets with different levels of $\epsilon$-GI location data (\textit{cf.} Table~\ref{tab:epsilon_gi}). All models were trained continuously, i.e., at the end of each month, the new known data were added to the training set after sanitization with $\epsilon$-GI. \textbf{Lastly, all models were tested with original data. On the one hand, this would prevent having in real-life a sanitized location that would compromise the EMS response time. On the other hand, each time the model is re-fitted (or retrained), the new known data should also be sanitized with $\epsilon$-GI.} In addition, for comparison purposes, we also trained and evaluated one additional model per ML technique with original data. In this chapter, the models were evaluated using the following regression metrics: RMSE, MAE, MAPE, and $R^2$, all presented in Section~\ref{ch3:sub_metrics}.

Results for each metric were calculated using data from the 6 months evaluation period. The RMSE metric was also used during the hyperparameters tuning process via Bayesian optimization (BO), explained in Section~\ref{ch3:optimization}. To this end, we used the HYPEROPT library~\cite{hyperopt2013} with $100$ iterations for each model. Table~\ref{ch9:tbl_hyperparameters} displays the range of each hyperparameter used in the BO, as well as the final configuration used to train and test the models. 

\setlength{\tabcolsep}{2.5pt}
\renewcommand{\arraystretch}{1.4}
\begin{table}[!ht]
\scriptsize
\centering
    \begin{tabular}{c c c c c c c c}
    \hline
    \multirow{2}{*}{\textbf{Model}} & \multirow{2}{*}{\textbf{Search space}} & \multicolumn{6}{c}{\textbf{Final configuration per dataset}} \\
    & & Original & $\epsilon=0.005493$ & $\epsilon=0.002747$ &$\epsilon=0.001155$ &$\epsilon=0.000866$ &$\epsilon=0.000693$\\
    \hline
    \multirow{9}{*}{XGBoost}
         & max\_depth: [1, 10]           & 9   &9   &6  &6  &9  &9 \\
         & n\_estimators: [50, 500]     &  465  &465  &130  &235  &465 &465\\
         & learning\_rate: [0.001, 0.5]    & 0.0265  & 0.0265   &0.0858  &0.0486  &0.0265  &0.0265\\
         & min\_child\_weight: [1, 10]           &5  & 5  &7  &7  &5  &5 \\
         & max\_delta\_step: [1, 11]   & 4  & 4 &3  &4  &4  &4\\
         & gamma: [0.5, 5]        & 3 & 3   &0  &2  &3  &3  \\
         & subsample: [0.5, 1]            & 0.8  & 0.8  &1   &1   &0.8  &0.8  \\
         & colsample\_bytree: [0.5, 1]            & 0.5  & 0.5  &0.5   &0.5   &0.5  &0.5  \\
         & alpha: [0, 5]            & 2  &  2 &1   &2   &2  &2  \\\hline
    
    \multirow{7}{*}{LGBM}
         & max\_depth: [1, 10]           & 7   & 8  &10  &8  &8  &6 \\
         & n\_estimators: [50, 500]     & 355   &326  &477  & 250 &80 &441\\
         & learning\_rate: [1e-4, 0.5]    & 0.0188  & 0.0098   &0.0164  &0.0285  &0.0586  &0.0300\\
         & subsample: [0.5, 1]           & 0.54066 & 0.5228  &0.6138  &0.6699  & 0.6732 &0.5812 \\
         & colsample\_bytree: [0.5, 1]   & 0.5160  & 0.5575 &0.5204  &0.6870  &0.5507  &0.5451\\
         & num\_leaves: [31, 400]        & 400 & 192   &245  & 398 &132  & 95 \\
         & reg\_alpha: [0, 5]            & 4  & 0  & 5  & 0  & 1 & 4 \\
    \hline
    
    \multirow{7}{*}{MLP}
         & Dense layers: [1, 7]                      &7  & 3 & 4    &6   &6   &6   \\
         & Number of neurons: [$2^8$, $2^{13}$]      &$2^{10}$  & $2^{12}$  & $2^{12}$   &$2^{9}$   &$2^{12}$   &$2^{9}$   \\
         & Batch size: [32, 168]                &140   & 80 & 48   &82   &70   &44   \\
         & Learning rate: [1e-5, 0.01]   &0.00265  & 0.00124 & 0.0099     &0.0099   &0.0094   &0.0077   \\
         & Optimizer: Adam          &Adam & Adam & Adam & Adam & Adam & Adam \\
         & Epochs: 100                  & 100 & 100 & 100 & 100 & 100 & 100  \\
         & Early stopping: 10                   & 10 & 10 & 10 & 10 & 10 & 10 \\
    \hline

    \multirow{1}{*}{LASSO}
         & alpha: [0.01, 2]                    & 0.0205 & 0.0307 &0.0105 &0.0100 &0.0112 &0.0107\\\hline
    \end{tabular}
    \caption{Search space for hyperparameters by ML model and the final configuration obtained for predicting ARTs per dataset.}
    \label{ch9:tbl_hyperparameters}
\end{table}

\section{Results and Discussion} \label{ch9:sec_results_discussion}

In this section, we present the results of our experimental validation (Section~\ref{ch9:sub_results}) and a general discussion (Section~\ref{ch9:sub_discussion}) including related work and limitations.

\subsection{Privacy-preserving ART prediction} \label{ch9:sub_results}

Figure~\ref{fig:results_metrics} illustrates the impact of the level of GI for each ML model to predict ART according to each metric. As one can notice in this figure, for XGBoost, LGBM, and LASSO, there were minor differences between training models with original location data or sanitized ones. On the other hand, models trained with MLP performed poorly with GI-based data. In addition, by analyzing models trained with original data, while the smaller RMSE for LASSO is about 5.65, for more complex ML-based models, RMSE is less than 5.6, achieving 5.54 with XGBoost and LGBM. In comparison with the results of existing literature, lower $R^2$ scores and similar RMSE and MAE results were achieved in~\cite{Lian2019} to predict ART while using original location data only.

Indeed, among the four tested models, LGBM and XGBoost achieve similar metric results while favoring the LGBM model. Thus, Figure~\ref{ch9:fig_pred_LGBM} illustrates the BO iterative process for LGBM models trained with original and sanitized data according to the RMSE metric (left-hand plot); and ART prediction results for 50 dispatched ambulances in 2020 out of 8,709 ones (right-hand plot) with an LGBM model trained with original data (Pred: original) and with two LGBM models trained sanitized data, i.e., with $\epsilon=0.005493$ (low privacy level) and with $\epsilon=0.000693$ (high privacy level). 

\begin{figure}[!ht]
    \centering
    \includegraphics[width=1\linewidth]{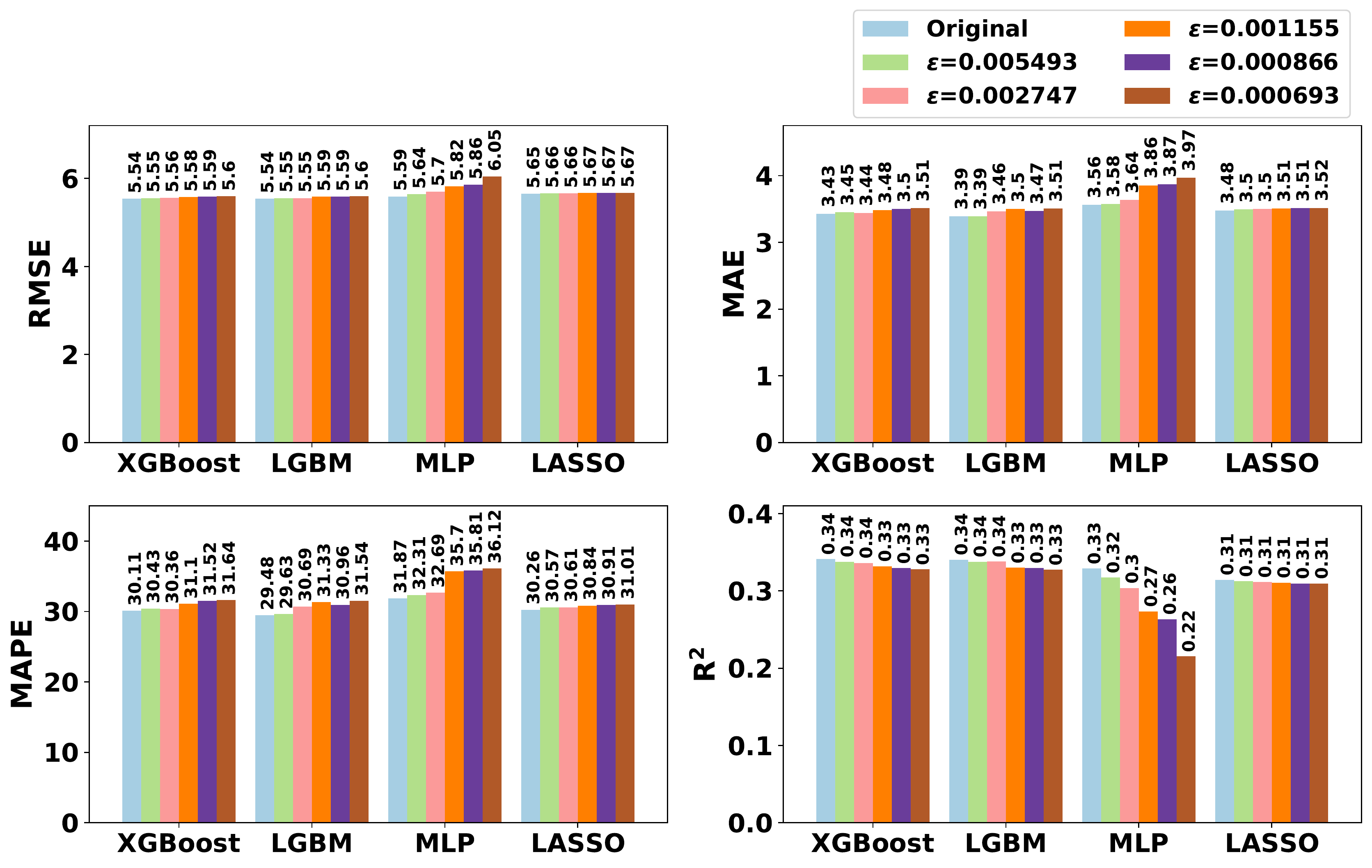}
    \caption{Impact of the level of $\epsilon$-geo-indistinguishability for each ML model to predict ART according to each metric.}
    \label{fig:results_metrics}
\end{figure}

\begin{figure}[!ht]
    \centering
    \includegraphics[width=1\linewidth]{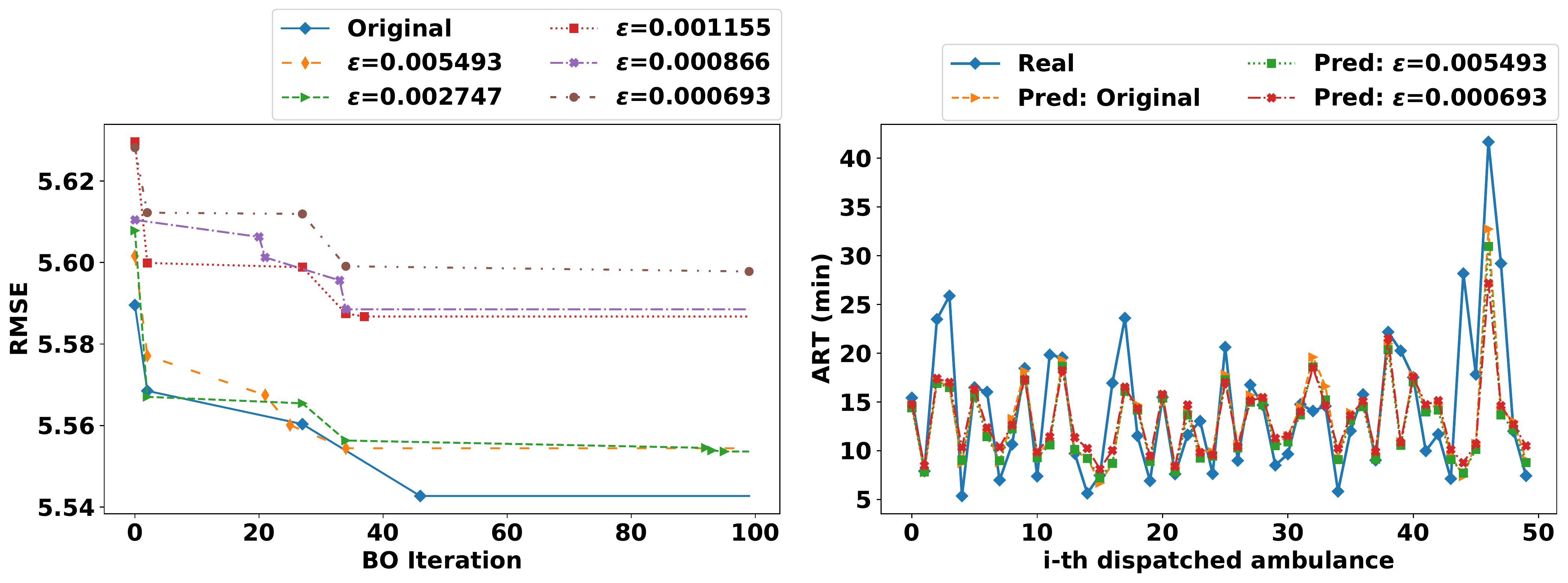}
    \caption{The left-hand plot illustrates the hyperparameters tuning process via Bayesian optimization with 100 iterations for LGBM models trained with original data and sanitized ones. The right-hand plot illustrates the prediction of ARTs with LGBM models trained with original data and with sanitized ones.}
    \label{ch9:fig_pred_LGBM}
\end{figure}

As one can notice in the left-hand plot of Figure~\ref{ch9:fig_pred_LGBM}, once data are sanitized with different levels of $\epsilon$-GI, the hyperparameters optimization via BO is also perturbed. This way, local minimums were achieved in different steps of the BO (i.e., the last marker per curve indicates the local minimum). For instance, even though $\epsilon=0.002747$ is more strict than $\epsilon=0.005493$, results were still better for the former since, in the last steps of BO, three better local minimums were found. Besides, prospective predictions were achieved with either original or sanitized data. For instance, in the right-hand plot of Figure~\ref{ch9:fig_pred_LGBM}, even for the high peak-value of ART around 40 minutes, LGBM's prediction achieved some reasonable estimation. Although several features were perturbed due to the sanitization of the emergency scene (e.g., city, zone, etc), the models could still achieve similar predictions as the model trained with original location data.

\subsection{Discussion and Related Work} \label{ch9:sub_discussion}

The medical literature has mainly focused attention on the analysis of ART~\cite{Do2012,Austin2003,Silverman2007} and its association with trauma~\cite{Pons2002,Byrne2019} and cardiac arrest~\cite{Brger2018,Lee2019,Holmn2020}, for example. To reduce ART, some works propose reallocation of ambulances~\cite{Carvalho2020,Chen2016}, operation demand forecasting~\cite{Grekousis2019,Couchot2019,Lin2020,Arcolezi2020,Chen2016,Couchot2019,Cerna2020_b,Lstm_Cerna2019,Cerna2020_boosting,EliasMallouhy2021,Guyeux2019}, travel time prediction~\cite{aladdini2010ems}, simulation models~\cite{Aboueljinane2012,Peleg2004}, and EMS response time predictions~\cite{aladdini2010ems,Lian2019}. The work in~\cite{Lian2019} propose a real-time system for predicting ARTs for the San Francisco fire department, which closely relates to our work in this chapter. The authors processed about $4.5$ million EMS calls utilizing original location data to predict ART using four ML models, namely linear regression, linear regression with elastic net regularization, decision tree regression, and random forest. However, no privacy-preserving experiment was performed because the main objective of their paper was proposing a scalable, ML-based, and real-time system for predicting ART. Besides, we also included weather data that the authors in~\cite{Lian2019} did not consider in their system, which could help to recognize high ARTs due to bad weather conditions, for example.

Because most of EMS data are personal and confidential (e.g., location, reason), there is a need for privacy-preserving techniques for processing and using these data. In this chapter, even if the intervention's \textit{reason} could be an indicator of the call urgency, we did not consider this sensitive attribute in our data analysis nor privacy-preserving prediction models. This is because, for SDIS 25, the ARTs limits are defined by the zone~\cite{Cerna2020}. Additionally, we also did not include the victims' personal data (e.g., gender, age) in our predictions or analysis since, during the calls, the operator may not acquire such information, e.g., when a third party activates the SDIS 25 for unidentified victims. This way, we focused our attention on the \textbf{location privacy of each intervention}.

Indeed, location privacy is an emerging and active research topic in the literature~\cite{Primault2019,elsalamouny_gambs,Shokri2011,Chatzikokolakis2017,Andrs2013,Yang2021} as publicly exposing users' location raises major privacy issues. To address this location privacy issue, in this chapter, we sanitized each emergency location using the state-of-the-art GI~\cite{Andrs2013} model. As highlighted in~\cite{Andrs2013}, attackers in LSBs may have side information about the user's reported location, e.g., knowing that the user is probably visiting the Eiffel Tower instead of swimming in the Seine river. However, this does not apply in our context because someone may have drowned and EMS had to intervene. Similarly, even for the dataset with intermediate (and high) privacy in which locations are spread out in the Doubs region (\textit{cf.} map with $0.005493$-GI location in Figure~\ref{fig:real_VS_anon}), someone may have been lost in the forest and EMS would have to interfere. For these reasons, using (or sharing datasets with) approximate emergency locations (e.g., sanitized with GI) is a prospective direction since many locations are possible emergency scenes. Indeed, we are not interested in hiding the emergency's location completely since some approximate information is required in order to retrieve other features (e.g., city, zone, estimated distance) to use for predicting ART.

With the differentially private input perturbation setting adopted in this chapter, data are protected from data leakage and are more difficult to reconstruct, for example. For instance, the authors in~\cite{kang2020input,Fukuchi2017} investigate how input perturbation through applying controlled Gaussian noise on data samples can guarantee $(\epsilon,\delta)$-DP on the final ML model. This means, since ML models are trained with perturbed data, there is a perturbation on the gradient and on the final parameters of the model too.

In this chapter, rather than Gaussian noise, the emergency scenes were sanitized with Alg.~\ref{alg:GI_location} explained in Chapter~\ref{chap:chapter2}, i.e., adding two-dimensional Laplacian noise centered at the exact user location $x \in \mathbb{R}^2$. In addition, this sanitization also perturb other associated and calculated features such as: city, district, zone (e.g., urban or not), great-circle distance, estimated driving distance, and estimated travel time (\textit{cf.} Table~\ref{ch9:tab_var_features_GI}). As well as the optimization of hyperparameters, i.e., once data are differentially private, one can apply any function on it and, therefore, we also noticed perturbation on the BO procedure. Yet, as shown in the results, prospective ART predictions were achieved with either original or sanitized data. What is more, even with a high level of sanitization ($\epsilon=0.000693$) there was an adequate privacy-utility trade-off. According to~\cite{lewis1982}, if the mean absolute percentage error (i.e., MAPE) is greater than 20\% and less than 50\%, the forecast is reasonable, which is the results we have in this chapter with MAPE around 30\%. 

\section{Conclusion} \label{ch9:sec_conclusion}

In this chapter, we aimed to predict the response time that each center equipped with ambulances had to an event, which could be used as an intelligent decision-support system to dynamically select the center to deploy ambulances. However, we also took into consideration that the emergency locations are sensitive data, requiring proper sanitization. Therefore, this interdisciplinary work aimed to evaluate the effectiveness of predicting ARTs with ML models trained over sanitized location data with different levels of $\epsilon$-geo-indistinguishability. 

As shown in the results, the sanitization of location data and the perturbation of its associated features (e.g., city, distance) had certain impact on data utility (see Table~\ref{ch9:tab_var_features_GI}) but no considerable effect on predicting ART (see Fig.~\ref{ch9:fig_pred_LGBM}). With these findings, EMS may prefer using and/or sharing sanitized datasets to avoid possible data leakages, membership inference attacks, or data reconstruction attacks, for example~\cite{data_breaches,Song2017,Shokri2017,Carlini2019}. \textbf{In conclusion, while predicting ART might allow EMS to save more lives, we notice that it is also possible to do so while preserving the victims' location privacy.}

Lastly, on the one hand, this chapter focused on \textbf{response time} taking into consideration the \textbf{recommended times} SDIS 25 ambulances should arrive on-scene (e.g., for Z1 the ART should be $\leq 10$ minutes)~\cite{Cerna2020}. The next Chapter~\ref{chap:chapter92} proposes a privacy-preserving solution to \textbf{response time} taking into consideration the urgency level of the intervention through predicting if each victim will die.

%% file: chapters/chapter92.tex
\chapter{Privacy-Preserving Prediction of Victim's Mortality} \label{chap:chapter92}

In Chapter~\ref{chap:chapter9}, we have started to focus on EMS \textbf{response time} to emergencies. In this last chapter of contribution, we continue in this direction from another perspective: \textbf{Although SDIS 25 ARTs depend mainly on the zone~\cite{Cerna2020}, is there a way to recognize or of being aware that an emergency will require \textit{priority attention}?} To answer this question, \textbf{with Selene Cerna}, we proposed a methodology based on ML techniques to predict the victims' mortality \textbf{using data gathered from the start of the emergency call until the SDIS 25 is notified}. Within this interval of interest, there are data about the call processing times, operators’ and victims’ personal data; the location of the emergency, and so on. In other words, there are two entities we will be concerned with, namely, \textbf{call center operators} and \textbf{victims} regarding \textbf{privacy}. Similar to Chapters~\ref{chap:chapter8} and~\ref{chap:chapter9}, we still consider that EMS intend to share an anonymized version of their data, such that third parties could build decision-support tools to optimize the EMS service. However, \textbf{differently of a single sensitive attribute} (i.e., \textit{only location} in Chapters~\ref{chap:chapter8} and~\ref{chap:chapter9}), \textbf{there are several personal attributes concerning victims and operators}. Therefore, \textbf{in this chapter}, we evaluated the \textbf{privacy-utility trade-off} of ML models trained over anonymized data using either the \textit{k}-anonymity model (cf. Section~\ref{ch2:sub_k_anon}) or of a differentially private algorithm~\cite{Bild2018} that produces truthful data output. \textbf{Throughout this chapter, we slightly abuse of our notation and use the terms \textit{anonymized/anonymization} for both \textit{k-anonymity} and DP (instead of anonymized/sanitization) guarantees} (cf. Section~\ref{ch2:introduction}). We invite the reader to refer to Chapter~\ref{chap:chapter2} for the background on both \textit{k}-anonymity~\cite{samarati1998protecting,SWEENEY2002} and DP~\cite{Dwork2006,Dwork2006DP,dwork2014algorithmic} models.

\section{Introduction} \label{sec:introduction}

As reviewed in Chapters~\ref{chap:chapter1},~\ref{chap:chapter3},~\ref{chap:chapter8}, and~\ref{chap:chapter9}, EMS are a key component of healthcare systems around the world. An important measurement of their quality is their response time, which is measured from the time the EMS is notified to the time an ambulance arrives at the emergency scene (cf. Chapter~\ref{chap:chapter9}). In fact, shorter ambulance response times are potential contributors to higher survival rates~\cite{Holmn2020,Byrne2019,Brger2018,Chen2016,Pons2002,Lee2019} since every second is a matter of life. For instance, \textbf{the response time also depends on how and by whom the call is processed} in the EMS center~\cite{Penn2016}. For this reason, there is a need to optimize these services and take advantage of plenty of data gathered throughout the years in hospitals and EMS.

In this chapter, we consider as \textit{interval of interest} the period comprising the time where the SDIS 25 call center's phone starts to ring until some center(s) is notified to handle the intervention or the call ends. With all accessible data within this interval (e.g., victims and operators data, call processing times, ...), \textbf{the purpose of this chapter is to evaluate the privacy-utility trade-off of training ML models over anonymized data to predict the victims' mortality}. Therefore, there are two entities we are concerned with, namely, \textit{call center operators} and \textit{victims} with regard to privacy. 

For instance, with the raw dataset containing direct identifiers (e.g.,  names), one straightforward question as: ``\textit{Is there any operator linked with an increased ratio of victims' death}?" can be easily computed, which compromises the operators' privacy and can lead to social and/or economical damages. Similarly, one can easily access the reason for the intervention (e.g., cardiac arrest) and use this information to jeopardize the victims' privacy through discrimination in health insurance, for example. Besides, as reviewed in Section~\ref{ch2:sub_k_anon}, even by excluding direct identifiers, both victims' and operators' identities are still at risk of being retrieved~\cite{sweeney2015only}. Indeed, attributes such as gender, age, and ZIP code (a.k.a. quasi-identifiers -- QIDs) can be combined with public data to reidentify individuals~\cite{samarati1998protecting,SWEENEY2002}. 

For instance, on analyzing the Vic\_Mort-DB from Section~\ref{ch3:calls_vic_ope_dbs}, considering \textit{victims}, by combining three available QIDs (gender, age, and city), one can find about $22000$ cases with the trivial $k=1$-anonymity level~\cite{samarati1998protecting,SWEENEY2002}. This means, in some cities with low population density, it would not be difficult to find out the person who needed help by knowing their gender and age. Similarly, combining four QIDs considering \textit{operators} (gender, age, grade/career, and seniority) leads to a similar output with many unique rows. One exception is that there is a set of operators, and each row represents an \textit{event} of who treated the emergency call. This reinforces the need for applying privacy-preserving techniques to protect the users' privacy.

Therefore, in this chapter, we assessed the effectiveness of anonymizing the Vic\_Mort-DB from Section~\ref{ch3:calls_vic_ope_dbs} with two state-of-the-art privacy techniques, namely, \textit{k}-anonymity~\cite{samarati1998protecting,SWEENEY2002} and DP~\cite{Dwork2006,Dwork2006DP,dwork2014algorithmic} before training any ML model to predict the victims' mortality. The Vic\_Mort-DB has information about $177883$ \textbf{victims} that the SDIS 25 attended from January 2015 to December 2020. To the author's knowledge, this is the first work to assess the impact of privacy-preserving techniques on predicting the victims' mortality. Indeed, while these predictions may allow SDIS 25 (or EMS in general) to save more lives, we notice that it is also possible to do so with anonymized datasets, which preserves both victims' and operators' privacy. 

The remainder of this chapter is organized as follows. In Section~\ref{ch92:sec_results}, we present the experimental setup, our results, discussion, and we review related work. Lastly, in Section~\ref{ch92:sec_conclusion}, we present the concluding remarks. The proposed privacy-preserving methodology to predict the victims' mortality (and their transportation to health facilities not approached here) developed with Selene Cerna, part of the results/discussion of Section~\ref{ch92:sec_results} were accepted as a full article~\cite{Arcolezi2021_vic_mortransp} in the Transactions on Industrial Informatics journal. 

\section{Experimental Validation} \label{ch92:sec_results}

We divide this section in the following way. First, we describe general settings for our experiments (Section~\ref{ch92:sub_gen_setup}). Next, we present the development and evaluation of privacy-preserving ML models (Section~\ref{ch92:sub_private_ML}). Lastly, we discuss our work and review related work (Section~\ref{ch92:sec_disc_rel_work}).

\subsection{General setup of experiments} \label{ch92:sub_gen_setup}

\textbf{Environment.} All algorithms were implemented in Python 3.8.8 with XGBoost~\cite{XGBoost} and Scikit-learn~\cite{scikit} libraries. The anonymization methods were implemented with the ARX\footnote{\url{https://arx.deidentifier.org/}} tool~\cite{Prasser2015}.

\textbf{ML model evaluated.} With Selene Cerna, two ML models and two DL models have been compared with the original data, i.e., with Vic\_Mort-DB. The most performing method was XGBoost. Therefore, \textbf{only XGBoost} will be used in this chapter to evaluate its privacy-utility trade-off of being trained over anonymized data.

\textbf{Dataset.} We utilize the Vic\_Mort-DB from Section~\ref{ch3:calls_vic_ope_dbs} divided into exclusively \textbf{learning} (from 2015-2019 with $n_l=149321$ victims) and \textbf{testing} (the year 2020 with $n_t=28562$ victims) sets.

\textbf{Privacy models evaluated.} We compared the effectiveness of both \textit{k}-anonymity~\cite{samarati1998protecting,SWEENEY2002} and DP~\cite{Dwork2006,Dwork2006DP,dwork2014algorithmic} models, both presented in Chapter~\ref{chap:chapter2}. The differentially private model of ARX was proposed in~\cite{Bild2018}, namely, \textit{SafePub}, which produces truthful data output. More precisely, DP is ensured by sampling, in which the sampling probability depends on $\epsilon$, and data are released in a generalized form that also satisfies \textit{k}-anonymity (where \textit{k} depends on $\epsilon$ and $\delta$). Also, we highlight that \textit{both privacy models were applied only in the learning set} and, hence, the \textit{testing set was transformed} using the final generalization hierarchies. 

In our experiments, we vary the $\epsilon$ parameter in the range $\epsilon=[0.2,0.4,0.6,0.8,1.0]$ and we fix $\delta=10^{-6} \ll 1/n_l$. \textit{With these parameters, the differentially private learning sets} were sub-sampled from $n_l=149321$ to $n_l^*=[24621,45038,61841,76418,88663]$ samples and, \textit{besides DP guarantees}, \textit{k}-anonymity is also satisfied with $\textbf{k}=[62,62,65,70,74]$, respectively. Thus, \textit{for a fair comparison between the two privacy models}, we also set $\textbf{k}=[62,62,65,70,74]$ when applying the \textit{k}-anonymity model.

\textbf{Generalization approach.} The following (generalization) transformations were considered to anonymize each information concerning the victim (Vic.) and operator (Ope.):

\begin{itemize}
    \item \textit{Age (Vic. and Ope.)} by intervals of growing amplitude: 10, 20, 40, 80, total suppression (*);
    \item \textit{Gender (Vic. and Ope.)} by total suppression (*);
    \item \textit{Vic. City ID} by masking (five) digits: 2222*, 222**, 22***, 2****, total suppression (*);
    \item \textit{Ope. Seniority (in days)} by intervals of growing amplitude (about 6 months): 180, 360, 720, ..., total suppression (*);
    \item \textit{Ope. Grade} by total suppression (*).
\end{itemize}

\textbf{Experimental evaluation.} \textbf{Eleven} models were built. \textbf{One} XGBoost model trained over original data, \textbf{five} XGBoost (input perturbation-based) models trained over DP data considering the aforementioned $\epsilon$-DP range, and \textbf{five} XGBoost (input perturbation-based) models trained over \textit{k}-anonymous data considering the aforementioned $\textbf{k}$ range. To optimize XGBoost hyperparameters, we applied Bayesian optimization~\cite{hyperopt2013} (explained in Section~\ref{ch3:optimization}) with $100$ iterations, with the following specification: n\_estimators [50-1000], learning\_rate [0.001-0.5], max\_depth [1-20], colsample\_by\_tree [0.2-1], and scale\_pos\_weight [20-60]. For all other parameters, we used their default values. 

\textbf{Performance metrics.} All XGBoost models were evaluated with standard binary classification metrics, namely, ACC (accuracy) and MF1 (macro f1-score), both explained in Section~\ref{ch3:sub_metrics}. The MF1 metric was also used as the objective function for the Bayesian optimization.

\subsection{Privacy-Preserving Binary Classification of Victims' Mortality} \label{ch92:sub_private_ML}

Fig.~\ref{ch92:fig_metrics} illustrates the relationship between the MF1 and ACC metrics (y-axis) with $\epsilon$ (x-axis) for XGBoost models trained over original and anonymized datasets (i.e., differentially private and \textit{k}-anonymous). For each value of $\epsilon$, the corresponding \textit{k}-anonymity guarantee is $\textbf{k}=[62,62,65,70,74]$, respectively.

\begin{figure}[!ht]
    \centering
    \includegraphics[width=1\linewidth]{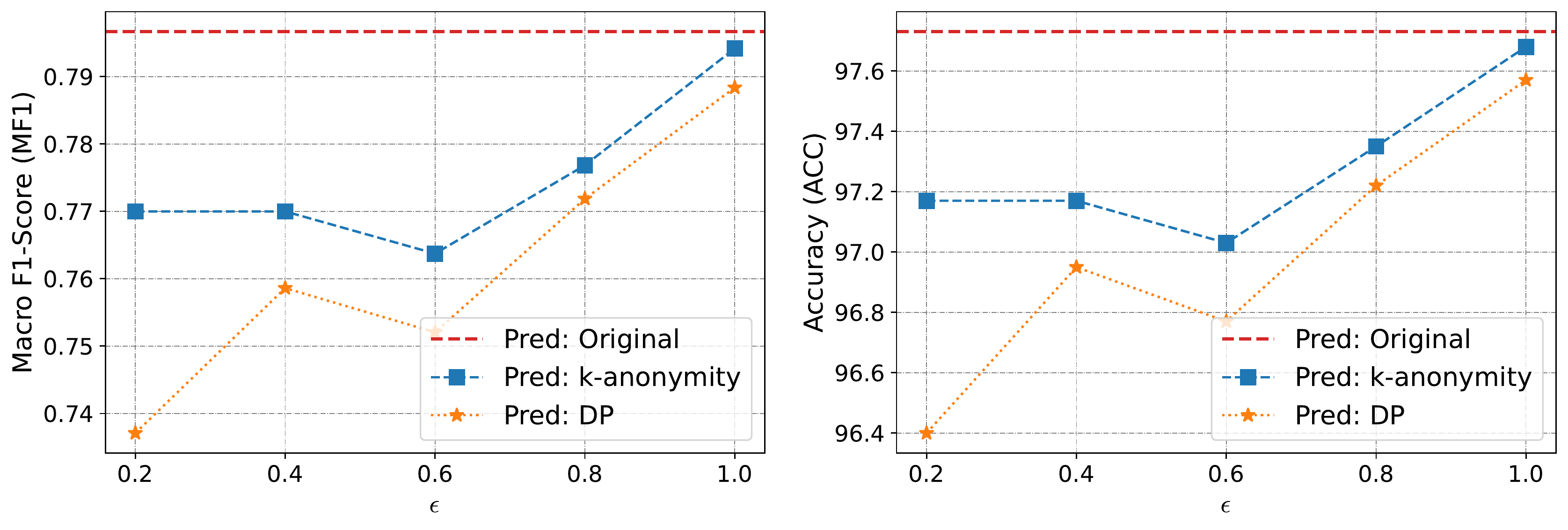}
    \caption{MF1 and ACC metrics (y-axis) for XGBoost models trained over original, differentially private, and \textit{k}-anonymous datasets. For each value of $\epsilon$ (x-axis), the corresponding \textit{k}-anonymity guarantee is $\textbf{k}=[62,62,65,70,74]$, respectively.}
    \label{ch92:fig_metrics}
\end{figure}

One can notice from Fig.~\ref{ch92:fig_metrics} that XGBoost models trained with anonymized data can also guarantee a good privacy-utility trade-off for the binary classification of victims' mortality. Overall, the results with the \textit{k}-anonymous datasets are still close to the results with original data while the results with DP decreased more. This could be due to DP applying both sub-sampling of the learning set as well as the generalization and/or suppression of QIDs to also satisfy \textit{k}-anonymity.

On the other hand, with the commonly used $\epsilon=1$ privacy guarantee~\cite{dwork2014algorithmic,Bild2018} that also satisfies $\textit{k}=74$-anonymity, the MF1 and ACC scores of both XGBoost models trained over DP and \textit{k}-anonymous data had no considerable loss of utility. Indeed, selecting $\epsilon=1$ has also been suggested in~\cite{Bild2018} as a good parameterization value. Thus, considering $\epsilon=1$ and $k=74$, Table~\ref{tab:generalization_QIDs} exhibits the final generalization approach for each QID we considered of each entity (Victim -- Vic. and Operator -- Ope.) and privacy model (\textit{k}-anonymity and DP). The symbol $*$ in Table~\ref{tab:generalization_QIDs} indicates full suppression for an attribute or masking of a digit (for Vic. City). On the one hand, the transformed/suppressed features limit the data analysis one can carry on, e.g., to find correlation between features. On the other hand, as one can notice from Fig.~\ref{ch92:fig_metrics}, although some features (the QIDs) suffered transformation and/or suppression, XGBoost models were still able to classify victims' mortality as good as with the original dataset, while providing privacy guarantees for both victims and operators. 

\setlength{\tabcolsep}{5pt}
\renewcommand{\arraystretch}{1.4}
\begin{table}[!ht]
    \centering
    \scriptsize
    \caption{Final generalization hierarchy for each QID of each entity, namely, victim and call center operator.}
    \label{tab:generalization_QIDs}
    \begin{tabular}{c c c} \hline
        \multirow{2}{*}{Attribute} & \multicolumn{2}{c}{Final Generalization Hierarchy} \\ \cline{2-3}
        & \textit{k}-anonymity & Differential Privacy \\ \hline
         Vic. Age &[0, 40[, [40, 80[, [80, 101[      &  [0, 20[, [20, 40[, ..., [80, 101[     \\ 
         Vic. Gender & Feminine, Masculine, Not registered      &  Feminine, Masculine, Not registered, *      \\ 
         Vic. City & 21***, 22***, 23***     &  212**, 222**, 233**...     \\ \hline
         Ope. Age & *     &  [22, 32[, [32, 42[, ..., [52, 62[     \\ 
         Ope. Gender & Feminine, Masculine     &  Feminine, Masculine, *     \\ 
         Ope. Seniority (in days) & [0, 2880[, [2880, 5760[, ..., [8640, 10204[     &   *    \\ 
         Ope. Grade & *     &   *    \\ \hline
    \end{tabular}
\end{table}

These results suggest that some patterns were still kept even with the transformed features. So, Fig.~\ref{ch92:fig_feature_importance} illustrates 10 features with the highest impact considering the most performing XGBoost model trained over original data, $\epsilon=1$-DP, and $\textit{k}=74$-anonymity, considering the type ``Gain" feature importance algorithm. This algorithm is based on the relative contribution of each feature to improve the accuracy in the division of a branch. In Fig.~\ref{ch92:fig_feature_importance}, the following prefixes are used: ``PROBA" for probability, ``INT" for intervention, and ``VIC" for victim, which corresponds to the features of the Vic\_Mort-DB from Section~\ref{ch3:calls_vic_ope_dbs}.

From Fig.~\ref{ch92:fig_feature_importance}, we can notice that the calculated variables from probabilities (PROBA\_MORT\_MOT and PROBA\_MORT\_AGE) and the type of intervention (type, subtype, and motive) have a great impact on the creation of the models. Besides, the victims' age and gender showed more importance than the alert diffusion time (INT\_D\_DIFF\_ALERT) and duration of the call. Lastly, in our experiments, operators' personal data did not show much importance for any XGBoost model, in this way, for upcoming works, we consider not using such predictors as there would be a need for preserving their privacy. 

\begin{figure}[!ht]
    \centering
    \includegraphics[width=0.9\linewidth]{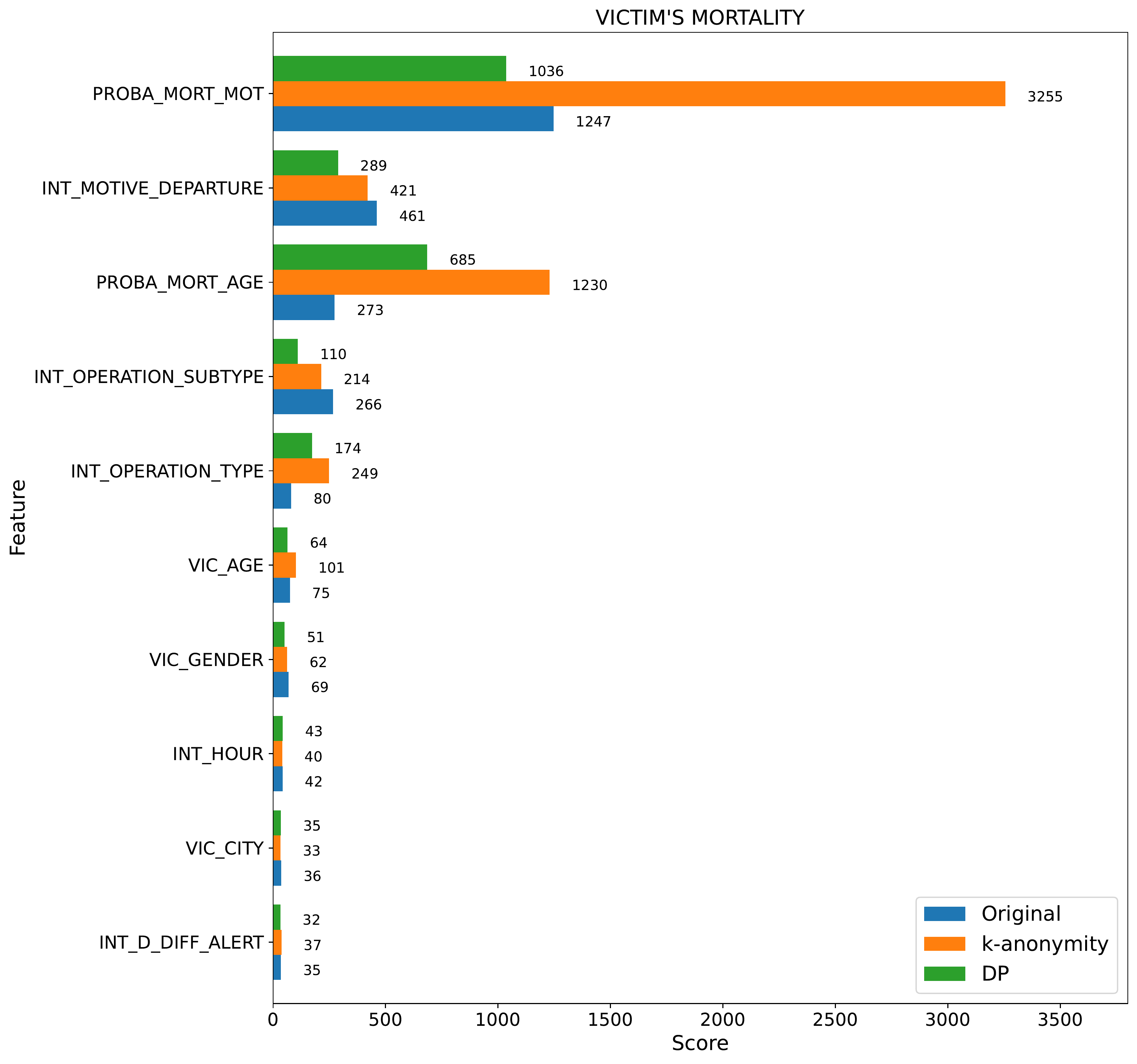}
    \caption{Feature importance from the XGBoost models trained over original, $\epsilon=1$-DP, and $\textit{k}=74$-anonymity, considering the type ``Gain" as score and the first 10 variables.} 
    \label{ch92:fig_feature_importance}
\end{figure}

\subsection{Discussion and Related Work} \label{ch92:sec_disc_rel_work}

As reviewed in recent survey works~\cite{ShafafMalek2019,Tang2021,Stewart2021}, several decision-support systems based on ML techniques have been proposed for application in emergency medicine. Indeed, in the context of this chapter, for EMS, there are many interests in using ML methods for tasks such as: identifying possible medical conditions before arrival on emergency departments~\cite{Kang2020}, to predict ambulances' demand to allow their reallocation~\cite{Chen2016}, to predict operation demand~\cite{Grekousis2019,Couchot2019,Lin2020,Arcolezi2020,Chen2016,Couchot2019,Cerna2020_b,Lstm_Cerna2019,Cerna2020_boosting,EliasMallouhy2021,Guyeux2019}, to predict ambulance response time~\cite{Arcolezi2021,Lian2019}, to predict the ambulances' turnaround time in hospitals~\cite{Cerna2021}, to predict clinical outcomes~\cite{Kwon2019}, to early identify clinical conditions on emergency calls~\cite{Blomberg2019}, to recognize and predict service disruptions~\cite{Cerna2020}, and so on. However, to our knowledge, we are the first group investigating \textbf{privacy-preserving ML solutions to EMS}.

Indeed, although the collection of medical data allows investigations to propose improved ML-based decision-support tools, on the other hand, there is a problem with the disclosure of personal and sensitive information. In the privacy-preserving data mining literature, there are few alternatives, e.g., objective perturbation~\cite{chaudhuri2011differentially}, gradient perturbation~\cite{DL_DP,pytorch_privacy,tf_privacy}, and input data perturbation~\cite{first_ldp,ppdm}, that can help to mitigate these problems. This chapter also adopted the \textit{input perturbation} setting (cf. Section~\ref{ch3:input_perturbation}) because it allows using any ML and post-processing techniques in contrast with gradient~\cite{DL_DP} or objective perturbation~\cite{chaudhuri2011differentially}. Furthermore, input perturbation is consistent with real-world applications in which EMS would only utilize and/or share anonymized data with third parties to train and improve ML-based decision support systems. This way, because each sample in the dataset is anonymized, data are protected from data leakage and are more difficult to reconstruct when the ML model receives attacks~\cite{Song2017,Shokri2017,Carlini2019}, for example.

\section{Conclusion} \label{ch92:sec_conclusion}

In this chapter, we aimed to predict the victims' mortality using data collected from the emergency call until an SDIS 25 center is notified about the intervention. More precisely, with all data available within this time interval (e.g., call processing times, operators’ and victims’ personal data, location, etc), \textbf{we sought to predict if victims will die}. This way, the SDIS 25 (or EMS in general) can quickly dispatch ambulances depending on the level of urgency. However, we also take into consideration both \textit{victims'} and \textit{call center operators'} \textit{privacy} when training the ML models. Therefore, this interdisciplinary work aimed to evaluate the effectiveness of predicting victims' mortality with XGBoost models trained over differentially private and \textit{k}-anonymous data with different levels of $\epsilon$ and \textit{k}. 

As shown in the results, even with anonymized datasets ($\textit{k}=74$-anonymous and $\epsilon=1$-differentially-private), mortality could be predicted with accuracy as high as $97\%$ with MF1 scores of about $79\%$. These results showed (again) the potential of privacy-preserving ML solutions for EMS, which can be used as a decision-support tool to early identify mortality while preserving the users' privacy and, thus, help EMS to save more lives. As a result of these findings, EMS may consider utilizing and/or sharing anonymised datasets to prevent data leakages, membership inference attacks, and data reconstruction attacks~\cite{data_breaches,Song2017,Shokri2017,Carlini2019}.

Lastly, some limitations of this chapter are described in the following. First, on anonymizing the datasets, there is a clear difference in the type of privacy we provided for each \textit{entity}. On the one hand, because \textit{victims} were unique in our dataset, DP and \textit{k}-anonymity provided \textit{user-level}~\cite{dwork2014algorithmic} privacy. On the other hand, there is a unique set of \textit{operators} that treated many emergency calls and, thus, DP and \textit{k}-anonymity provided \textit{event-level}~\cite{dwork2014algorithmic} privacy. Also, we considered an ideal case where the information of all victims in the testing set was acquired during the call. However, this may not always occur in real life, e.g., when someone activates EMS for unidentified victims.

%% file: chapters/Conclusion.tex
\chapter{Conclusion \& Perspectives} \label{chap:conclusion}

\section{General Conclusion}

In this manuscript, we approached several aspects of privacy-preserving data collection and publishing, as well as privacy-preserving machine learning. This manuscript is separated into four parts. In the \textbf{first part}, we introduced the context, the motivating projects, and the objectives. The \textbf{second part} started presenting the data privacy and ML techniques our works depend on. We finished the second part by presenting the datasets we experiment on. 

The \textbf{third part} contains our contributions to privacy-preserving statistical learning, mainly with the LDP model. In the first chapter of the third part, i.e., Chapter~\ref{chap:chapter4}, we proposed an approach to infer and recreate synthetic data that provides a precise mobility scenario based on one-week statistical data of \textit{unions of consecutive days} made available by~\cite{fluxvision1}. The generated and open dataset (\url{https://github.com/hharcolezi/OpenMSFIMU}) named MS-FIMU can be used to evaluate new privacy-preserving techniques as well as ML tasks. For instance, in Chapters~\ref{chap:chapter5} and~\ref{chap:chapter6}, the MS-FIMU dataset has been used to evaluate the effectiveness of our proposed LDP protocols for multidimensional and longitudinal frequency estimates. The MS-FIMU dataset has also been used in Chapter~\ref{chap:chapter7}, in which \textit{we proposed an LDP-based CDRs processing system to publish multidimensional mobility reports throughout time}. We also prove in Chapter~\ref{chap:chapter7} that for collecting multidimensional data with GRR~\cite{kairouz2016discrete} the utility loss sending a single attribute with $\epsilon$-LDP (i.e., \textit{Smp} solution) is lower than splitting the privacy budget over the number of attributes. This proof extends to two other protocols named SUE~\cite{rappor} and OUE~\cite{tianhao2017}, as shown in Chapter~\ref{chap:chapter5}.

We then abstracted the problem of Chapter~\ref{chap:chapter7} and thus, in Chapter~\ref{chap:chapter5}, we focused on \textit{improving the utility of LDP protocols for longitudinal and multidimensional frequency estimates}. Indeed, the combination of both ``multi" settings (i.e., numerous attributes and longitudinal data collection) presents several problems, for which this manuscript provides the first solution called ALLOMFREE under $\epsilon$-LDP. Under the same privacy guarantee, our studies revealed that ALLOMFREE consistently and significantly outperforms the state-of-the-art protocols, namely, L-SUE (a.k.a. Basic-RAPPOR~\cite{rappor}) and L-OUE (i.e., OUE~\cite{tianhao2017} with \textit{memoization}), with an average accuracy increase ranging from $10\%$ up to $55\%$.

We start the last chapter of the third part, i.e., Chapter~\ref{chap:chapter6}, by arguing that the state-of-the-art \textit{Smp} solution for multidimensional frequency estimates might be ``unfair" with some users since the reported value uses the whole privacy budget $\epsilon$, which is what is currently accepted. \textit{We thus propose RS+FD, which may be utilized with any current LDP protocol designed for single-frequency estimation}. More precisely, with RS+FD, the client-side has two steps: local randomization and fake data generation (cf. Fig.~\ref{ch6:fig_system_overview} and Alg.~\ref{alg:rs+fd}). With our experiments, we concluded that \textit{under the same privacy guarantee, our proposed protocols with RS+FD achieve similar or smaller estimation error than using the state-of-the-art \textit{Smp} solution while enhancing all users' privacy.}

The \textbf{fourth and last part} comprises our contributions to differentially private machine learning predictions. Indeed, \textit{the main goal of this fourth part was to evaluate the privacy-utility trade-off of training ML and DL models over differentially private data} (a.k.a. input perturbation~\cite{ppdm,first_ldp}). So, in Chapter~\ref{chap:chapter91}, we developed and assessed the performance of DL models on two privacy-preserving ML settings, namely, input and gradient perturbation. For the former, we applied the Gaussian mechanism~\cite{dwork2014algorithmic} to each sample before training any DL model, and for the latter, we trained DL models with the DP-SGD~\cite{DL_DP,tf_privacy,pytorch_privacy} algorithm. Both settings were compared on a multivariate time series forecasting task of aggregate human mobility data. \textit{We concluded that it is still possible to have accurate multivariate forecasts in both privacy-preserving ML settings}. \textit{In terms of accuracy} (measured with the RMSE metric), \textit{the gradient perturbation setting surpasses input perturbation}. However, \textit{input perturbation provides higher privacy protection than gradient perturbation} as it might also protect the aggregated mobility data against known threats (e.g., data breaches~\cite{data_breaches}, membership inference attacks~\cite{Pyrgelis2017,Pyrgelis2020}, and trajectory recovery attack~\cite{Tu2018,Xu2017}).

Next, we started to focus on our second motivating project with a collaboration with Selene Cerna, which concerns the SDIS 25 (i.e., an EMS in France). \textit{Our assumption is that EMS intends to deploy decision-support systems to optimize their services but only shares sanitized data with the development team (i.e., third parties)}. So, in Chapter~\ref{chap:chapter8}, \textit{we proposed an LDP-based methodology to allow EMS to properly sanitize interventions' data}. With several experiments in frequency estimation and input perturbation-based ML forecasting, we concluded that interventions data can be properly sanitized to avoid leakage of information while remaining useful for both statistical learning (cf. Fig.~\ref{ch8:fig_ex_freq_est}) and forecasting (cf. Fig.~\ref{ch8:fig_ex_forecasts}) purposes. 

Moreover, in our two last contribution chapters, i.e., Chapter~\ref{chap:chapter9} and~\ref{chap:chapter92}, \textit{we evaluated the privacy-utility trade-off of solutions based on ML and DP} with a focus on \textit{optimizing EMS response time to emergencies}. More precisely, in Chapter~\ref{chap:chapter9}, we proposed to evaluate the effectiveness of several values of $\epsilon$ (i.e., the privacy budget) to sanitize emergency location data with geo-indistinguishability~\cite{Andrs2013} and train ML-based models to predict ambulance response times. We concluded that the sanitization of location data with geo-indistinguishability and the perturbation of its associated features (e.g., city, distance) had a certain impact on data utility (see Table~\ref{ch9:tab_var_features_GI}) but no considerable effect on predicting ARTs (see Fig.~\ref{ch9:fig_pred_LGBM}). Finally, in Chapter~\ref{chap:chapter92}, we concentrate our attention on each \textit{victim} by using several sensitive attributes of both victims and call center operators (cf. Section~\ref{ch3:calls_vic_ope_dbs}). Our objective was to use data collected within the time of the emergency call until an SDIS 25 center is notified about the intervention to predict the victims' mortality. Once more, we focused on anonymizing/sanitizing the dataset and, thus, we assessed the effectiveness of both \textit{k}-anonymity~\cite{samarati1998protecting,SWEENEY2002} and DP~\cite{Dwork2006,Dwork2006DP,dwork2014algorithmic} models. That chapter concludes that even with anonymized datasets, victims' mortality could be predicted with high accuracy and macro f1-scores, which could be used as a decision-support tool by EMS to early identify high urgent situations.

\section{Perspectives}

The research fields in privacy and privacy-preserving ML are broad and promising. For instance, it is possible and interesting to investigate the following topics in the \textbf{short term}:

\begin{itemize}
    \item To integrate the proposed LDP protocols from Chapters~\ref{chap:chapter5} and~\ref{chap:chapter6} into the LDP-based CDRs processing system of Chapter~\ref{chap:chapter7} in order to improve the privacy of MNOs clients.

    \item To investigate how to combine the optimal longitudinal LDP protocols from Chapter~\ref{chap:chapter5} (i.e., L-GRR and L-OSUE) and our proposed RS+FD solution from Chapter~\ref{chap:chapter6} is also a planned and indicated direction.

    \item For both Chapters~\ref{chap:chapter91} and~\ref{chap:chapter8}, for future work, we suggest and intend to investigate a more complex DL architecture to improve the results of DL/ML models proposed in this manuscript for their respective multivariate time series forecasting task. 
    
    \item Regarding the work on Chapter~\ref{chap:chapter9}, the intended future works are to extend the analysis and predictions to different operation times of EMS such as the pre-travel delay (i.e., gathering personnel and ambulances) and travel times (e.g., from the center to the emergency scene, from the emergency scene to hospitals), \textbf{while respecting users' privacy}. 

\end{itemize}

In addition, for the \textbf{long term}, we list below some perspectives of the works in this manuscript:

\begin{itemize}
    \item To extend the proposed LDP-based CDRs processing system of Chapter~\ref{chap:chapter7} to the shuffle DP model~\cite{Balle2019,Erlingsson2019,erlingsson2020encode,Wang2020,li2021privacy}, which could provide strong privacy guarantees as well as accurate multidimensional frequency estimates (i.e., mobility reports throughout time).
    
    \item To investigate our proposed RS+FD solution from Chapter~\ref{chap:chapter6} on generating synthetic data from $\epsilon$-LDP multidimensional frequency estimates for classification/regression tasks (e.g., as in~\cite{Cyphers2017}) in two perspectives: \textbf{performance} and \textbf{privacy-protection} (e.g., against membership inference attacks). 
    
    \item To study if given a reported tuple $\textbf{y}$ one can state which attribute value is ``fake" or not by seeing the estimated frequencies reported with our RS+FD solution from Chapter~\ref{chap:chapter6}. Indeed, we suggest investigating this phenomenon in both single-time collection and longitudinal studies (i.e., throughout time). 
    
    \item Concerning multivariate time-series forecasting (i.e., Chapters~\ref{chap:chapter91} and~\ref{chap:chapter8}), investigating the data leakage through membership inference attacks of both differentially private input and gradient perturbation settings is also a prospective and intended direction.

    \item Some future work for Chapter~\ref{chap:chapter92} would be to investigate a uniform notion of privacy for both entities (i.e., a set of operators linked to many unique victims). In addition, we intend to evaluate privacy-preserving ML models with randomly excluded data from victims (i.e., sex and age) since these data might not be acquired during the calls, in order to assess the models' robustness. Besides, another prospective direction would be working with the \textit{text observations} registered by operators during calls, which could be treated with natural language processing techniques, while preserving the privacy of the individuals concerned (e.g.,~\cite{qu2021privacy,yu2021differentially}).
\end{itemize}

%% file: chapters/publications.tex
\chapter{Publications}

During the period of this thesis, the author has published the following papers and resources. The superscript $^*$ highlights equal contribution for co-first authors \textbf{in bold}.

\section*{JOURNAL PAPERS}
\begin{itemize}
    \item \textbf{$^*$Arcolezi, H. H.}, \textbf{$^*$Cerna, S.}, Couchot, J.-F, Guyeux, C., \& Makhoul, A. (2021) Privacy-Preserving Prediction of Victim's Mortality and Their Need for Transportation to Health Facilities. \textbf{IEEE Transactions on Industrial Informatics}, Early Access~\cite{Arcolezi2021_vic_mortransp}.
    
    \item \textbf{Arcolezi, H. H.}, Cerna, S., Guyeux, C., \& Couchot, J.-F. (2021). Preserving Geo-Indistinguishability of the Emergency Scene to Predict Ambulance Response Time. \textbf{Mathematical and Computational Applications}, 26(3), 56~\cite{Arcolezi2021_geo}.
    
    \item \textbf{Arcolezi, H. H.}, Couchot, J.-F., Cerna, S., Guyeux, C., Royer, G., Al Bouna, B., \& Xiao, X. (2020). Forecasting the Number of Firefighters Interventions per Region with Local-Differential-Privacy-Based Data. \textbf{Computers \& Security}, 96, 101888~\cite{Arcolezi2020}.
\end{itemize}

\section*{CONFERENCE PAPERS}
\begin{itemize}

\item \textbf{Arcolezi, H. H.}, Couchot, J.-F., Al Bouna, B., \& Xiao, X. (2021). Random Sampling Plus Fake Data: Multidimensional Frequency Estimates With Local Differential Privacy. In Proceedings of the 30th ACM \textit{International Conference on Information and Knowledge Management} (CIKM ’21), November, Virtual Event, QLD, Australia~\cite{Arcolezi2021_rs+fd}.

\item \textbf{Arcolezi, H. H.}, Couchot, J.-F., Al Bouna, B., \& Xiao, X. (2020). Longitudinal Collection and Analysis of Mobile Phone Data with Local Differential Privacy. In Proceedings of the 15th IFIP \textit{International Summer School on Privacy and Identity Management}, September, 40-57. Springer, Cham~\cite{Arcolezi2021}.

\item \textbf{Arcolezi, H. H.}, Couchot, J.-F., Baala, O., Contet, J.-M., Al Bouna, B., \& Xiao, X. (2020). Mobility modeling through mobile data: generating an optimized and open dataset respecting privacy. In Proceedings of the 16th \textit{International Wireless Communications and Mobile Computing} (IWCMC), June, 1689–1694~\cite{ms_fimu}.

\end{itemize}

\section*{SUBMITTED PAPERS}

\begin{itemize}
    
    \item \textbf{Arcolezi, H. H.}, Couchot, J.-F., Al Bouna, B., \& Xiao, X. Improving the Utility of Locally Differentially Private Protocols for Longitudinal and Multidimensional Frequency Estimates. \textbf{Digital Communications and Networks}. Submitted in August 2021~\cite{Arcolezi2021_allomfree}.
    
    \item \textbf{Arcolezi, H. H.}, Couchot, J.-F., Renaud, D., Al Bouna, B., \& Xiao, X. Differentially Private Multivariate Time Series Forecasting of Aggregated Human Mobility With Deep Learning: Input or Gradient Perturbation? \textbf{Neural Computing and Applications}. Submitted in September 2021.

\end{itemize}

\section*{CO-AUTHORED PAPERS}

Furthermore, the author also participated as a co-author in the following published papers.

\begin{itemize}

    \item Cerna, S., \textbf{Arcolezi, H. H.}, Guyeux, C., Royer-Fey, G., \& Chevallier, C. (2021). Machine learning-based forecasting of firemen ambulances’ turnaround time in hospitals, considering the COVID-19 impact. \textbf{Applied Soft Computing}, 109, 107561~\cite{Cerna2021}.
    
    \item Cisneros, L. L., \textbf{Arcolezi, H. H.}, Cerna, S., Brandão, J.L., Santos, G.C., Navarro, T.P., \& Carvalho, A.A. (2021). Machine Learning Algorithms to Predict In-Hospital Mortality in Patients with Diabetic Foot Ulceration. In Proceedings of the \textit{XXIII Congresso da Sociedade Brasileira de Diabetes}.
    
    \item Cerna, S., Guyeux, C., \textbf{Arcolezi, H. H.}, Couturier, R., \& Royer, G. (2020). A comparison of LSTM and XGBoost for predicting firemen interventions. In Proceedings of the 8th \textit{World Conference on Information Systems and Technologies} (WorldCIST), April, 424–434~\cite{Cerna2020_b}.
    
    \item Cerna, S., Guyeux, C., \textbf{Arcolezi, H. H.}, \& Royer, G. (2020). Boosting Methods for Predicting Firemen Interventions. In Proceedings of the 11th \textit{International Conference on Information and Communication Systems} (ICICS), 001–006~\cite{Cerna2020_boosting}.

\end{itemize}

\section*{RESOURCES \& CODES}

The generated MS-FIMU dataset of Chapter~\ref{chap:chapter4} is fully available on the following GitHub page: 

$\bullet$ \url{https://github.com/hharcolezi/OpenMSFIMU}.

Lastly, the author also maintains a list of DP and LDP experiments of the work carried out in Chapters~\ref{chap:chapter7},~\ref{chap:chapter5},~\ref{chap:chapter6},~\ref{chap:chapter91}, and~\ref{chap:chapter9} on the following GitHub page: 

$\bullet$ \url{https://github.com/hharcolezi/ldp-protocols-mobility-cdrs}.